%% file: main.tex
\def\confversion{0}
\def\ifconf{\ifnum\confversion=1}
\def\ifnotconf{\ifnum\confversion=0}

\documentclass[11 pt]{article}

\def\showauthornotes{0}
\def\showkeys{0}
\def\showdraftbox{0}

\input{macros}

\usepackage{tikz}
\usepackage{float}

\usepackage{thm-restate}
\usepackage{titlesec}
\usepackage{comment}

\usepackage{scalerel,stackengine}
\stackMath
\newcommand\reallywidehat[1]{%
	\savestack{\tmpbox}{\stretchto{%
			\scaleto{%
				\scalerel*[\widthof{\ensuremath{#1}}]{\kern-.6pt\bigwedge\kern-.6pt}%
				{\rule[-\textheight/2]{1ex}{\textheight}}
			}{\textheight}%
		}{0.5ex}}%
	\stackon[1pt]{#1}{\tmpbox}%
}
\newcommand{\RR}{\mathbb{R}}
\newcommand{\NN}{\mathbb{N}}

\newcommand{\EE}{\mathop{\mathbb{E}}}
\newcommand{\GN}{\mathcal{N}}

\newcommand{\sig}{\sigma}

\newcommand{\al}{\alpha}
\newcommand{\lda}{\lambda}
\newcommand{\Lda}{\Lambda}
\newcommand{\Gam}{\Gamma}
\newcommand{\gam}{\gamma}
\newcommand{\Del}{\Delta}

\newcommand{\tens}[2] {\ensuremath{#1 ^ {\otimes #2}}}

\def\<{{\langle}} \def\>{{\rangle}}       
\newcommand{\iprod}[1]{\langle #1 \rangle}
\DeclareMathOperator{\Tr}{Tr}

\newtheorem{propn}[theorem]{Proposition}

\renewcommand{\E}{\mathop{\mathbb{E}}}

\newcommand{\pE}{\tilde{\EE}}

\newcommand{\psdmass}{\normalfont{(PSD mass) }}
\newcommand{\middleshapebounds}{\normalfont{(Middle shape bounds) }}
\newcommand{\intersectionbounds}{\normalfont{(Intersection term bounds) }}
\newcommand{\truncationbounds}{\normalfont{(Truncation error bounds) }}

\newcommand{\middleshapeboundstwo}{Middle shape bounds}
\newcommand{\intersectionboundstwo}{Intersection term bounds}
\newcommand{\truncationboundstwo}{Truncation error bounds}

\renewcommand*{\circle}[1]{\scalebox{0.85}{\footnotesize
    \tikz[baseline=(char.base)]{
        \node[shape=circle,draw,inner sep=1.5pt](char) {   \ifx&#1&
        \color{white} $i$
        \else
        $#1$
        \fi};
    }}}
\renewcommand*{\square}[1]{\scalebox{0.85}{\footnotesize
    \tikz[baseline=(char.base), square/.style={regular polygon,regular polygon sides=4}]{
        \node[draw,square, inner sep=0.5pt](char) {
        \ifx&#1&
        \color{white} $t$
        \else
        $#1$
        \fi};
    }}}

\usepackage[top=1in, bottom=1in, left=1in, right=1in]{geometry}

\usepackage{amsmath}
\usepackage{amssymb}
\usepackage{amsthm}
\usepackage{ytableau}
\usepackage{mathtools}
\usepackage{blkarray}
\usepackage{bigstrut}

\usepackage{amsmath,amsthm,amssymb,amscd,epsfig,url}
\usepackage{fullpage}
\usepackage{times}
\usepackage{xcolor}

\usepackage{pgfplots}
\pgfplotsset{compat=1.15}
\usepackage{mathrsfs}
\usetikzlibrary{arrows}

\begin{document}

\title{Machinery for Proving Sum-of-Squares Lower Bounds on Certification Problems}

\author{
  Aaron Potechin\thanks{{\tt University of Chicago}. {\tt potechin@uchicago.edu}. Supported in part by NSF grant CCF-2008920.}
  \and
  Goutham Rajendran\thanks{{\tt University of Chicago}. {\tt goutham@uchicago.edu}. Supported in part by NSF grants CCF-1816372 and CCF-2008920.}
}

\date{\today}

\maketitle
\thispagestyle{empty}

{\abstract{In this paper, we construct general machinery for proving Sum-of-Squares lower bounds on certification problems by generalizing the techniques used by Barak \etal~\cite{BHKKMP16} to prove Sum-of-Squares lower bounds for planted clique. Using this machinery, we prove degree $n^{\epsilon}$ Sum-of-Squares lower bounds for tensor PCA, the Wishart model of sparse PCA, and a variant of planted clique which we call planted slightly denser subgraph.
}}

	\newpage
	\ifnotconf
	\pagenumbering{roman}
	\setcounter{tocdepth}{2}
	\clearpage
	\fi

	\pagenumbering{arabic}
	\setcounter{page}{1}

	\section{Introduction}\label{sec: intro}

	\input{intro}

	\section{Preliminaries}\label{sec: prelim}

	\input{prelim}

	\section{Informal Description of our Machinery}\label{sec: informal_statement}

	\input{informal_statement}

	\section{Application: Planted slightly denser subgraph}\label{sec: plds_qual}

	We first pseudocalibrate with respect to the random and planted distributions, enabling us to write the moment matrix in terms of graph matrices. Then, we show the qualitative bounds needed in our machinery.

	\input{planted_ds_qual}

	\section{Application: Tensor PCA}\label{sec: tpca_qual}

	Just as in planted slightly denser subgraph, we pseudocalibrate, decompose the moment matrix and then show qualitative machinery bounds.

	\input{tensor_pca_qual}

	\section{Application: Sparse PCA}\label{sec: spca_qual}

	Just as in earlier applications, we pseudocalibrate with respect to the random and planted distributions, enabling us to write the moment matrix in terms of graph matrices and then show the qualitative bounds needed in our machinery.

	\input{sparse_pca_qual}

	\section{Definitions and Main Theorem Statement}\label{sec: technical_def_and_main_theorem}

	\input{technical_def_and_main_theorem}

	\section{Proof of the Main Theorem}\label{sec: proof_of_main}

	\input{proof_of_main}

	\section{Choosing the functions $B_{norm}(\alpha)$, $B(\gamma)$, $N(\gamma)$, and $c(\alpha)$}\label{sec: choosing_funcs}

	\input{choosing_funcs}

	\section{Bounding truncation error}\label{sec: showing_positivity}

	\input{showing_positivity}

    \section{Planted slightly denser subgraph: Full verification}\label{sec: plds_quant}

	\input{planted_ds_quant}

	\section{Tensor PCA: Full verification}\label{sec: tpca_quant}

	\input{tensor_pca_quant}

	\section{Sparse PCA: Full verification}\label{sec: spca_quant}

	\input{sparse_pca_quant}

	\section{Conclusion}\label{sec: conclusion}

	\input{conclusion}

	\subsection*{Acknowledgements}
	We thank Sam Hopkins, Pravesh Kothari, Prasad Raghavendra, Tselil Schramm, David Steurer and Madhur Tulsiani for helpful discussions. We also thank Sam Hopkins and Pravesh Kothari for assistance in drafting the informal description of the machinery (Section \ref{sec: informal_statement}). Parts of this work have also appeared in \cite{potechin2022sub, rajendran2022nonlinear}.

	\bibliographystyle{alphaurl}
	\bibliography{madhur}

	\input{appendix}


\end{document}

%% file: macros.tex
\usepackage{xspace,xcolor,enumerate}
\usepackage{amsmath,amssymb}
\usepackage{amsthm}
\usepackage[toc,page]{appendix}
\usepackage{thmtools}
\usepackage{thm-restate}
\usepackage{color,graphicx}
\usepackage{boxedminipage}
\usepackage{makecell}

\ifnum\showkeys=1
\usepackage[color]{showkeys}
\fi

\definecolor{darkred}{rgb}{0.5,0,0}
\definecolor{darkgreen}{rgb}{0,0.35,0}
\definecolor{darkblue}{rgb}{0,0,0.55}

\usepackage[pdfstartview=FitH,pdfpagemode=UseNone,colorlinks,linkcolor=darkblue,filecolor=darkred,citecolor=darkgreen,urlcolor=darkred,pagebackref]{hyperref}

\usepackage[capitalise,nameinlink]{cleveref}
\usepackage[T1]{fontenc}
\usepackage{mathtools,dsfont,bbm}
\usepackage{mathpazo}

\ifnum\showauthornotes=1
\newcommand{\Authornote}[2]{{\sf\small\color{red}{[#1: #2]}}}
\newcommand{\Authorcomment}[2]{{\sf \small\color{gray}{[#1: #2]}}}
\newcommand{\Authorfnote}[2]{\footnote{\color{red}{#1: #2}}}
\else
\newcommand{\Authornote}[2]{}
\newcommand{\Authorcomment}[2]{}
\newcommand{\Authorfnote}[2]{}
\fi

\ifnum\showdraftbox=1

\else

\fi

\newtheorem{theorem}{Theorem}[section]

\newtheorem{definition}[theorem]{Definition}

\newtheorem{lemma}[theorem]{Lemma}
\newtheorem{remark}[theorem]{Remark}
\newtheorem{proposition}[theorem]{Proposition}
\newtheorem{corollary}[theorem]{Corollary}
\newtheorem{claim}[theorem]{Claim}
\newtheorem{fact}[theorem]{Fact}

\newtheorem{remk}[theorem]{Remark}
\newtheorem{example}[theorem]{Example}

\newtheorem{algo}[theorem]{Algorithm}

\def\FullBox{\hbox{\vrule width 6pt height 6pt depth 0pt}}

\def\qed{\ifmmode\qquad\FullBox\else{\unskip\nobreak\hfil
\penalty50\hskip1em\null\nobreak\hfil\FullBox
\parfillskip=0pt\finalhyphendemerits=0\endgraf}\fi}

\def\qedsketch{\ifmmode\Box\else{\unskip\nobreak\hfil
\penalty50\hskip1em\null\nobreak\hfil$\Box$
\parfillskip=0pt\finalhyphendemerits=0\endgraf}\fi}

\def\to{\rightarrow}
\def\eps{\varepsilon}
\def\epsilon{\varepsilon}

\def\eps{\epsilon}

\def\phi{\varphi}
\def\cal{\mathcal}

\newcommand{\defeq}{:=}

\newcommand{\etal}{et al.\xspace}

\newcommand{\R}{{\mathbb R}}
\newcommand{\E}{{\mathbb E}}

\newcommand{\N}{{\mathbb{N}}}

\usepackage{nicefrac}

\newcommand{\norm}[1]{\ensuremath{\left\lVert #1 \right\rVert}}

\newcommand{\ip}[2] {\ensuremath{\langle #1 , #2 \rangle}}

\makeatletter

\def\ProbabilityRender#1#2{
  \@ifnextchar\bgroup%
  {\renderwithdist{#1}{#2}}
   {\singlervrender{#1}{#2}}
}
\def\singlervrender#1#2{%
   \ensuremath{\mathchoice
       {{#1}\left[ #2 \right]}
       {{#1}[ #2 ]}
       {{#1}[ #2 ]}
       {{#1}[ #2 ]}
   }
}
\def\renderwithdist#1#2#3{%
   \@ifnextchar\bgroup
   {\superfancyrender{#1}{#2}{#3}}
   {\ensuremath{\mathchoice
      {\underset{#2}{#1}\left[ #3 \right]}
      {{#1}_{#2}[ #3 ]}
      {{#1}_{#2}[ #3 ]}
      {{#1}_{#2}[ #3 ]}
     }
   }
}
\def\superfancyrender#1#2#3#4#5{
   \ensuremath{\mathchoice
      {\underset{#1}{{#1}}\left#4 #3 \right#5}
      {{#1}_{#2}#4 #3 #5}
      {{#1}_{#2}#4 #3 #5}
      {{#1}_{#2}#4 #3 #5}
   }
}
\makeatother

\newfont{\inhead}{eufm10 scaled\magstep1}

\newcommand{\calD}{{\cal D}}
\newcommand{\calE}{{\cal E}}
\newcommand{\calF}{{\cal F}}

\newcommand{\calI}{{\cal I}}

\newcommand{\calL}{{\cal L}}
\newcommand{\calM}{{\cal M}}

\newcommand{\calP}{{\cal P}}

\newcommand{\poly}{{\mathrm{poly}}}

\newcommand{\littleoh}{\operatorname{o}}

\newcommand{\Erdos}{Erd\H{o}s\xspace}
\newcommand{\Renyi}{R\'enyi\xspace}



%% file: intro.tex
The Sum-of-Squares (SoS) hierarchy is an optimization technique that harnesses the power of semidefinite programming to solve optimization tasks. First independently investigated by Shor \cite{shor1987approach}, Nesterov \cite{nesterov2000squared}, Parillo \cite{parrilo2000structured}, Lasserre \cite{lasserre2001global} and Grigoriev \cite{grigoriev2001complexity, Grigoriev01}, the SoS hierarchy offers a sequence of convex relaxations for polynomial optimization problems and is parameterized by an integer called the degree of the SoS hierarchy. As we increase the degree $d$ of the hierarchy, we get progressively stronger convex relaxations which are solvable in $n^{O(d)}$ time\footnote{There is a caveat, see \cite{o2017sos}}. This has paved the way for the SoS hierarchy to be a powerful tool in algorithm design both in the worst case and the average case settings. Indeed, there has been tremendous success in using the SoS hierarchy to obtain efficient algorithms for combinatorial optimization problems (e.g., \cite{GW94, AroraRV04, GuruswamiS11, raghavendra2017strongly}) as well as problems stemming from Statistics and Machine Learning (e.g., \cite{BarakBHKSZ12, bks15, HopSS15, pot17, kothari2017outlier}). In fact, SoS achieves the state-of-the-art approximation guarantees for many fundamental problems such as Sparsest Cut \cite{AroraRV04}, MaxCut \cite{GW94}, Tensor PCA \cite{HopSS15} and all Max-$k$-CSPs \cite{Raghavendra08}. Moreover, for a large class of problems, it's been shown that SoS relaxations are the most efficient among all semidefinite programming relaxations \cite{lrs15}.

On the flip side, some problems have remained intractable beyond a certain threshold even by considering higher degrees of the SoS hierarchy \cite{BHKKMP16, KothariMOW17, MRX20, sklowerbounds}. For example, consider the Planted Clique problem where we have to distinguish a random graph sampled from the \Erdos-\Renyi model $G(n, \frac{1}{2})$ from a random graph which is obtained by first sampling a graph from $G(n, \frac{1}{2})$ and then planting a clique of size $n^{\frac{1}{2} - \eps}$ for a small constant $\eps > 0$. It was shown in \cite{BHKKMP16} that with high probability degree $o(\log n)$ SoS fails to solve this distinguishing problem.

There are many reasons for why studying lower bounds against the SoS hierarchy is important. Firstly, SoS is a generic proof system that captures a broad class of algorithmic reasoning \cite{FKP19}. In particular, under mild conditions, SoS captures statistical query algorithms and algorithms based on low degree polynomials. Therefore, SoS lower bounds indicate to the algorithm designer the intrinsic hardness of the problem and suggest that if they want to break the algorithmic barrier, they need to search for algorithms that are not captured by SoS. Secondly, in average case problem settings, standard complexity theoretic assumptions such as P $\neq$ NP have not been shown to give insight into the limits of efficient algorithms. Instead, lower bounds against powerful proof systems such as SoS have served as strong evidence of computational hardness \cite{hop17, hop18, kunisky2021spectral}. Thus, understanding the power of the SoS hierarchy on these problems is an important step towards understanding the approximability of these problems.

\subsection{Our contributions}

In this paper, we consider the following general category of problems: Given a random input, can we certify that it does not contain a given structure?
Some important examples of this kind of problem are as follows.
\begin{enumerate}
    \item Planted clique: Can we certify that a random graph does not have a large clique?
    \item Tensor PCA: Given an order $k$ tensor $T$ with random independent Gaussian entries, can we certify that there is no unit vector $x$ such that $\ip{T}{x \otimes\ldots\otimes x}$ is large?
    \item Wishart model of sparse PCA: Given an $m \times d$ matrix $S$ with random independent Gaussian entries (which corresponds to taking $m$ samples from $\GN(0, I_d)$), can we certify that there is no $k$-sparse unit vector $x$ such that $\norm{Sx}$ is large?
\end{enumerate}

These kinds of problems, known as certification problems, are closely related to their optimization or estimation variants. A certification algorithm is required to produce a proof/certificate of a bound that holds for \textit{all} inputs, as opposed to most inputs. The Sum-of-Squares hierarchy provides such certificates in a canonical way for a wide variety of such problems, so analyzing SoS paves the way towards understanding the certification complexity of these problems. We investigate the following question.

\begin{quote}
    \em{For certification problems, what are the best bounds that SoS can certify?}
\end{quote}

In this work, we build general machinery for proving probabilistic Sum of Squares lower bounds on certification problems. To build our machinery, we generalize the techniques pioneered by \cite{BHKKMP16} for proving Sum of Squares lower bounds for planted clique. We start with the standard framework for proving probabilistic Sum of Squares lower bounds:
\begin{enumerate}
    \item Construct candidate pseudo-expectation values $\pE$ and the corresponding moment matrix $\Lambda$ (see \cref{subsec: sos}).
    \item Show that with high probability, $\Lambda \succeq 0$.
\end{enumerate}
For planted clique, \cite{BHKKMP16} constructed $\pE$ and the corresponding moment matrix $\Lambda$ by introducing the pseudo-calibration technique (see \cref{subsec: pseudocalibration}). They then showed through a careful and highly technical analysis that with high probability $\Lambda \succeq 0$.

In this paper, we investigate how generally the techniques used for planted clique can be applied. We show that by considering coefficient matrices based on the coefficients obtained by pseudo-calibration, we can give a general framework which captures most of the technical analysis for the sum of squares lower bound for planted clique \cite{BHKKMP16}. Using this framework, we show that there are relatively simple conditions on these coefficient matrices which are sufficient to ensure that the moment matrix $\Lda$ is PSD with high probability. This allows us to focus on the structure which is specific to each problem rather than the technical details of the analysis.

We exhibit the usefulness of our machinery by using it to achieve strong SoS lower bounds for the problems of Tensor PCA, Sparse PCA and a problem closely related to the Planted Clique problem that we call Planted Slightly Denser Subgraph. We do this with relative ease once the machinery is in place. The Sparse PCA lower bounds complement a long line of work on algorithmic guarantees that stretches for over two decades (the most recent one being \cite{sparse_pca_focs20}), thereby giving a complete picture (up to polylogarithmic factors) for the approximability-inapproximability thresholds for Sparse PCA.

We believe that our machinery is a promising approach for proving a variant of the low-degree conjecture (see \cite{hop18, holmgren2020counterexamples} for background on the low-degree conjecture). The low-degree conjecture is a fascinating open problem which says if there is a low-degree polynomial lower bound for a problem which is sufficiently symmetric then there is a sum of squares lower bound for a noisy version of the problem. If we can show that having a low-degree polynomial lower bound implies that the coefficient matrices obtained via pseudo-calibration (for the noisy version of the problem) obey the conditions required for our machinery, then this would prove this variant of the low-degree conjecture.

\subsubsection{A brief summary of pseudo-calibration}
A natural way to prove lower bounds on a certification problem is as follows.
\begin{enumerate}
    \item Construct a ``maximum entropy'' planted distribution of inputs which has the given structure.
    \item Show that we cannot distinguish between the random and planted distributions and thus cannot certify that a random input does not have the given structure.
\end{enumerate}
Based on this idea, the pseudo-calibration technique introduced  by \cite{BHKKMP16} constructs candidate pseudo-expectation values $\pE$ so that as far as low degree tests are concerned, $\pE$ for the random distribution mimics the behavior of the given structure for the planted distribution (for details, see \cref{subsec: pseudocalibration}). This gives a candidate moment matrix $\Lambda$ which we can then analyze with our machinery. A majority of known high-degree average-case SoS lower bounds in the literature have pseudo-expectation values that were either obtained by, or could be obtained by pseudocalibration, e.g., Planted Clique \cite{BHKKMP16}, Max-$k$-CSPs \cite{KothariMOW17}, Max-Cut on regular graphs \cite{MRX20}, Sherrington-Kirkpatrick problem \cite{sklowerbounds, MRX20}. It has also been successful for Densest-$k$-subgraph but for the weaker Sherali-Adams Hierarchy \cite{chlamtavc2018sherali}.

Naturally, pseudocalibration is the starting point for our SoS lower bounds. That said, our machinery is quite general and can be applied even if the candidate moment matrix $\Lambda$ is not obtained via pseudo-calibration.

\subsubsection{Our results on Tensor PCA, Sparse PCA, and Planted Slightly Denser Subgraph}
In this section, we formally state the main hardness theorems we show by applying our machinery. We defer discussing prior work, how we improve on them, and other related work to \cref{sec: prior_work}.

We describe the planted distributions we use to show our SoS lower bounds for planted slightly denser subgraph, tensor PCA, and the the Wishart model of sparse PCA. We also state the random distributions for completeness and for contrast. We then state our results.

\paragraph{Planted slightly denser subgraph}
We use the following distributions.
\begin{restatable}{itemize}{PLDSdistributions}
    \item Random distribution: Sample $G$ from $G(n, \frac{1}{2})$
    \item Planted distribution: Let $k$ be an integer and let $p > \frac{1}{2}$. Sample a graph $G'$ from $G(n, \frac{1}{2})$. Choose a random subset $S$ of the vertices, where each vertex is picked independently with probability $\frac{k}{n}$. For all pairs $i, j$ of vertices in $S$, rerandomize the edge $(i, j)$ where the probability of $(i, j)$ being in the graph is now $p$. Set $G$ to be the resulting graph.
\end{restatable}
In \cref{sec: plds_qual}, we compute the candidate moment matrix $\Lambda$ obtained by pseudo-calibration.

\begin{restatable}{theorem}{PLDSmain}\label{thm: plds_main}
Let $C_p > 0$. There exists a constant $C > 0$ such that for all sufficiently small constants $\eps > 0$, if $k \le n^{\frac{1}{2} - \eps}$ and $p =  \frac{1}{2} + \frac{n^{-C_p\eps}}{2}$, then with high probability, the candidate moment matrix $\Lda$ given by pseudo-calibraton for degree $n^{C\eps}$ Sum-of-Squares is PSD.
\end{restatable}

\begin{corollary}\label{cor: plds_main}
Let $C_p > 0$. There exists a constant $C > 0$ such that for all sufficiently small constants $\eps > 0$, if $k \le n^{\frac{1}{2} - \eps}$ and $p =  \frac{1}{2} + \frac{n^{-C_p\eps}}{2}$, then with high probability, degree $n^{C\eps}$ Sum-of-Squares cannot certify that a random graph $G$ from $G(n, \frac{1}{2})$ does not have a subgraph of size $\approx k$ with edge density $\approx p$.
\end{corollary}

\paragraph{Tensor PCA}
Let $k \ge 2$ be an integer. We use the following distributions.
\begin{restatable}{itemize}{TPCAdistributions}
    \item Random distribution: Sample $A$ from $\GN(0, I_{[n]^k})$.
	\item Planted distribution: Let $\lda,\Delta > 0$. Sample $u$ from $\{-\frac{1}{\sqrt{\Delta n}}, 0, \frac{1}{\sqrt{\Delta n}}\}^n$ where the values are taken with probabilites $\frac{\Delta}{2}, 1 - \Delta, \frac{\Delta}{2}$ respectively. Then sample $B$ from $\GN(0, I_{[n]^k})$. Set $A = B + \lda \tens{u}{k}$.
\end{restatable}

In \cref{sec: tpca_qual}, we compute the candidate moment matrix $\Lambda$ obtained by pseudo-calibration.

\begin{restatable}{theorem}{TPCAmain}\label{thm: tpca_main}
    Let $k \ge 2$ be an integer. There exist constants $C,C_{\Del} > 0$ such that for all sufficiently small constants $\eps > 0$, if $\lda \le n^{\frac{k}{4} - \eps}$ and $\Del = n^{-C_{\Del}\eps}$ then with high probability, the candidate moment matrix $\Lambda$ given by pseudo-calibration for degree $n^{C\eps}$ Sum-of-Squares is PSD.
\end{restatable}
\begin{corollary}\label{cor: tpca_main}
    Let $k \ge 2$ be an integer. There exists a constant $C > 0$ such that for all sufficiently small constants $\eps > 0$, if $\lda \le n^{\frac{k}{4} - \eps}$, then with high probability, degree $n^{C\eps}$ Sum-of-Squares cannot certify that for a random tensor $A$ from $\GN(0, I_{[n]^k})$, there is no vector $u$ such that $\norm{u} \approx 1$ and $\ip{A}{\underbrace{u \otimes\ldots\otimes u}_{\text{k times}}} \approx \lambda$.
\end{corollary}

\paragraph{Wishart model of Sparse PCA}
We use the following distributions.
\begin{restatable}{itemize}{SPCAdistributions}
    \item Random distribution: $v_1, \ldots, v_m$ are sampled from $\GN(0, I_d)$ and we take $S$ to be the $m \times d$ matrix with rows $v_1, \ldots, v_m$.
	\item Planted distribution: Sample $u$ from $\{-\frac{1}{\sqrt{k}}, 0, \frac{1}{\sqrt{k}}\}^d$ where the values are taken with probabilites $\frac{k}{2d}, 1 - \frac{k}{d}, \frac{k}{2d}$ respectively. Then sample $v_1, \ldots, v_m$ as follows. For each $i \in [m]$, with probability $\Delta$, sample $v_i$ from $\GN(0, I_d + \lda uu^T)$ and with probability $1 - \Delta$, sample $v_i$ from $\GN(0, I_d)$. Finally, take $S$ to be the $m \times d$ matrix with rows $v_1, \ldots, v_m$.
\end{restatable}

In \cref{sec: spca_qual}, we compute the candidate moment matrix $\Lambda$ obtained by pseudo-calibration.

\begin{restatable}{theorem}{SPCAmain}\label{thm: spca_main}
    There exists a constant $C > 0$ such that for all sufficiently small constants $\eps > 0$, if $m \le \frac{d^{1 - \eps}}{\lda^2}, m \le \frac{k^{2 - \eps}}{\lda^2}$, and there exists a constant $A$ such that $0 < A < \frac{1}{4}$, $d^{4A} \le k \le d^{1 - A\eps}$, and $\frac{\sqrt{\lda}}{\sqrt{k}} \le d^{-A\eps}$, then with high probability, the candidate moment matrix $\Lambda$ given by pseudo-calibration for degree $d^{C\eps}$ Sum-of-Squares is PSD.
\end{restatable}

\begin{corollary}\label{cor: spca_main}
    There exists a constant $C > 0$ such that for all sufficiently small constants $\eps > 0$, if $m \le \frac{d^{1 - \eps}}{\lda^2}, m \le \frac{k^{2 - \eps}}{\lda^2}$, and there exists a constant $A$ such that $0 < A < \frac{1}{4}$, $d^{4A} \le k \le d^{1 - A\eps}$, and $\frac{\sqrt{\lda}}{\sqrt{k}} \le d^{-A\eps}$, then with high probability, the degree $d^{C\eps}$ degree Sum-of-Squares cannot certify that for a random $m \times d$ matrix $S$ with Gaussian entries, there is no vector $u$ such that $u$ has $\approx k$ nonzero entries, $\norm{u} \approx 1$, and $\norm{Su}^2 \approx m + m{\Delta}\lambda$.
\end{corollary}

\begin{remark}
Note that our planted distributions only approximately satisfy constraints such as having a subgraph of size $k$, having a unit vector $u$, and having $u$ be $k$-sparse. While we would like to use planted distributions which satisfy such constraints exactly, these distributions bring additional technical difficulties. This same issue appeared in the SoS lower bounds for planted clique \cite{BHKKMP16}, which was recently resolved by Pang~\cite{Pang21}. Resolving this issue in general is a subtle but important open problem.
\end{remark}

\subsection{Relation to prior work on Planted Clique/Dense Subgraph, Tensor PCA, and Sparse PCA}\label{sec: prior_work}

\subsubsection{Planted Dense Subgraph}\label{sec: plds}

In the planted dense subgraph problem, we are given a random graph $G$ where a dense subgraph of size $k$ has been planted and we are asked to find this planted dense subgraph.
This is a natural generalization of the $k$-clique problem \cite{karp1972reducibility} and has been subject to a long line of work over the years (e.g. \cite{feige1997densest, feige2001dense, khot2006ruling, bhaskara2010detecting, bhaskara2012polynomial, braverman2017eth, manurangsi2017almost}).
In this work, we consider the following certification variant of planted dense subgraph.

\begin{quote}
\em{Given a random graph $G$ sampled from the \Erdos-\Renyi model $G(n, \frac{1}{2})$, certify an upper bound on the edge density of the densest subgraph on $k$ vertices.}
\end{quote}

For many different parameter regimes of the random and planted distributions (an example being planting $G_{k, q}$ in $G_{n, p}$ for constants $p < q$), and when $k = o(\sqrt{n})$, the hardness of the easier distinguishing version of planted dense subgraph problem has been posed as formal conjecture (often referred to as the PDS conjecture) before in the literature (see e.g., \cite{hajek2015computational, chen2014statistical, brennan2018reducibility, brennan2019universality}). This has also led to many reductions to other problems \cite{brennan2019optimal}, although it's not clear if these reductions can be made in the SoS framework without loss in the parameter dependence.

In our case, we consider the slightly planted denser subgraph version where for $k \le n^{\frac{1}{2} - \eps}$, we plant a subgraph of density $\frac{1}{2} + \frac{1}{n^{O(\eps)}}$, i.e. $p = \frac{1}{2}, q = \frac{1}{2} + \frac{1}{n^{O(\eps)}}$. This has been widely believed to require sub-exponential time. Our work provides strong evidence towards this by exhibiting unconditional lower bounds against the powerful SoS hierarchy, even if we consider $n^{O(\eps)}$ levels, which corresponds to $n^{n^{O(\eps)}}$ running time! We expect this to lead to this problem being used as a natural starting point for reductions to show sub-exponential time hardness for various problems.

Within the SoS literature, \cite{BHKKMP16} show that for $k \le n^{\frac{1}{2} - \eps}$ for a constant $\eps > 0$, the degree $o(\log n)$ Sum-of-Squares cannot distinguish between a fully random graph sampled from $G(n, \frac{1}{2})$ from a random graph which has a planted $k$-clique. This implies that degree $o(\log n)$ SoS cannot certify an edge density better than $1$ for the densest $k$-subgraph if $k \le n^{\frac{1}{2} - \eps}$.

In \cref{cor: plds_main}, we show that for $k \le n^{\frac{1}{2} - \eps}$ for a constant $\eps > 0$, degree $n^{\Omega(\eps)}$ SoS cannot certify an edge density better than $\frac{1}{2} + \frac{1}{n^{O(\eps)}}$. The degree of SoS in our setting, $n^{\Omega(\eps)}$ is vastly higher than the earlier known result which uses degree $o(\log n)$. To the best of our knowledge, this is the first result that proves such a high degree lower bound. 

We remark that when we take $k = n^{\frac{1}{2} - \eps}$,  the true edge density of the densest $k$-subgraph is $\frac{1}{2} + \frac{\sqrt{\log(n/k)}}{\sqrt{k}} + \littleoh(\frac{1}{\sqrt{k}}) \approx \frac{1}{2} + \frac{1}{n^{1/4 - \eps/2}}$ as was shown in \cite[Corollary 2]{gamarnik2019landscape} whereas, by \cref{cor: plds_main}, the SoS optimum is as large as $\frac{1}{2} + \frac{1}{n^{\eps}}$. This highlights a significant difference in the optimum value.

\subsubsection{Tensor PCA}

The Tensor Principal Component Analysis problem, originally introduced by \cite{richard2014statistical}, is a generalization of the PCA problem from machine learning to higher order tensors.
Tensor PCA is a remarkably useful technique to exploit higher order moments of the data.
It was originally applied in latent variable modeling \cite{anandkumar2014tensor, kivva2021learning, kivva2022identifiability, kivvaidentifiability, anandkumar2014analyzing} but it has now found applications in topic modeling, video processing, collaborative filtering,  community detection, etc. (see e.g. \cite{hsu2012spectral, anandkumar2014guaranteed, richard2014statistical, anandkumar2014tensor, anandkumar2014analyzing, duchenne2011tensor, li2010tensor}).
Formally, given an order $k$ tensor of the form $\lda u^{\otimes k} + B$ where $u \in \RR^n$ is a unit vector and $B \in \RR^{[n]^k}$ has independent Gaussian entries, we would like to recover $u$. Here, $\lda$ is known as the signal-to-noise ratio.

This can be equivalently considered to be the problem of optimizing a homogenous degree $k$ polynomial $f(x)$, with random Gaussian coefficients over the unit sphere $\norm{x} = 1$. In general, polynomial optimization over the unit sphere is a fundamental primitive with a lot of connections to other areas of optimization (e.g. \cite{frieze2008new, brubaker2009random,brandao2017quantum, barak2014rounding, bks15, bhattiprolu2017weak}). Tensor PCA is an average case version of the above problem and has been studied before in the literature \cite{richard2014statistical, HopSS15, tensorpca16, hop17}. In this work, we consider the certification version of this average case problem.

\begin{quote}
\em{For an integer $k \ge 2$, given a random tensor $A \in \RR^{[n]^k}$ with entries sampled independently from $\GN(0, 1)$, certify an upper bound on $\ip{A}{x^{\otimes k}}$ over unit vectors $x$.}
\end{quote}

In \cite{tensorpca16}, it was shown that $q \le n$ levels of SoS certifies an upper bound of $\frac{2^{O(k)} (n \cdot \text{polylog}(n))^{k/4}}{q^{k/4 - 1/2}}$ for the Tensor PCA problem. When $q = n^{\eps}$ for sufficiently small $\eps$, this gives an upper bound of $n^{\frac{k}{4} - O(\eps)}$. \cref{cor: tpca_main} shows that this is tight.

In \cite{hop17}, they state a theorem similar to \cref{cor: tpca_main} and observe that it can be proved by applying the techniques used to prove the SoS lower bounds for planted clique. However, they do not give an explicit proof. Also, while they consider the setting where the random distribution has entries from $\{-1, 1\}$, we work with the more natural setting where the distribution is $\GN(0, 1)$. We remark that our machinery can also easily recover their result with the entries being restricted to $\{-1, 1\}$.

When $k = 2$, the maximum value of $\ip{\tens{x}{k}}{A}$ over the unit sphere $\norm{x}^2 = 1$ is precisely the largest eigenvalue of $(A + A^T)/2$ which is $\Theta(\sqrt{n})$ with high probability. For any integer $k \ge 2$, the true maximum of $\ip{\tens{x}{k}}{A}$ over $\norm{x}^2 = 1$ is $O(\sqrt{n})$ with high probability \cite{tomioka2014}. In contrast, by \cref{cor: tpca_main}, the optimum value of the degree $n^{\eps}$ SoS is as large as $n^{\frac{k}{4} - O(\eps)}$. This exhibits an integrality gap of $n^{\frac{k}{4} - \frac{1}{2} - O(\eps)}$.

\subsubsection{Wishart model of Sparse PCA}

The Wishart model of Sparse PCA, also known as the Spiked Covariance model, was originally proposed by \cite{johnstone_lu2009}. In this problem, we observe $m$ vectors $v_1, \ldots, v_m \in \RR^d$ from the distribution $\GN(0, I_d + \lda uu^T)$ where $u$ is a $k$-sparse unit vector, and we would like to recover $u$. Here, the sparsity of a vector is the number of nonzero entries and $\lda$ is known as the signal-to-noise ratio.

Sparse PCA is a fundamental routine that has applications in a diverse range of fields such as medicine, economics, image and signal processing and finance (e.g. \cite{wang2012online, naikal2011informative, majumdar2009image, tan2014classification, chun2009expression, allen2011sparse}).
We remark that in many of these applications, to learn models with sparse structure, heuristics are often used, such as greedy algorithms (e.g. \cite{johnson2012high, liu2014forward, jalali2011learning, zhang2008adaptive}) and score-based algorithms (e.g. \cite{chickering2002optimal, nandy2018high, rajendran2021structure}) but Sparse PCA is a more principled framework.
It's known that vanilla PCA does not yield good estimators in high dimensional settings \cite{baik2005phase, paul2007asymptotics, johnstone_lu2009}. A large volume of work has gone into studying Sparse PCA and it's variants, both from an algorithmic perspective (e.g. \cite{amini_wainwright2008, ma2013sparse, krauthgamer2015, deshpande2016, wang2016statistical}) as well as from an inapproximability perspective (e.g. \cite{berthet2013complexity, ma_wigderson_15, diakonikolas2017statistical, hop17, brennan2019optimal}).

Given the decades of research on this problem and how fundamental it is for a multitude of applications and disciplines, understanding the computational threshold behavior of the Wishart model of Sparse PCA is an extremely important research topic in statistics. In particular, prior works have explored statistical query lower bounds, SDP lower bounds, lower bounds by reductions from widely believed conjectures, etc. On the other hand, there have only been two prior works on lower bounds against SoS, specifically only for degree $2$ and degree $4$ SoS, which can be attributed to the difficulty in proving such lower bounds. In this paper, we vastly strengthen these lower bounds and show almost-tight lower bounds for the SoS hierarchy of degree $d^{\eps}$ which corresponds to a running time of $d^{d^{O(\eps)}}$.

Between this work and prior works, we completely understand the parameter regimes where sparse PCA is easy or conjectured to be hard up to polylogarithmic factors. In \cref{fig: spca_thresholds1} and \cref{fig: spca_thresholds2}, we assign the different parameter regimes into the following categories.
\begin{itemize}
	\item Diagonal thresholding: In this regime, Diagonal thresholding \cite{johnstone_lu2009, amini_wainwright2008} recovers the sparse vector. Covariance thresholding \cite{krauthgamer2015, deshpande2016} and SoS \cite{sparse_pca_focs20} can also be used in this regime. Covariance thresholding has better dependence on logarithmic factors and SoS works in the presence of adversarial errors.
	\item Vanilla PCA: Vanilla PCA can recover the vector, i.e. we do not need to use the fact that the vector is sparse (see e.g. \cite{berthet2013, sparse_pca_focs20}).
	\item Spectral: An efficient spectral algorithm recovers the sparse vector (see e.g. \cite{sparse_pca_focs20}).
	\item Can test but not recover: A simple spectral algorithm distinguishes the planted distribution from the random distribution but it is information theoretically impossible to recover the sparse vector \cite[Appendix E]{sparse_pca_focs20}.
	\item Hard: A regime where it is conjectured to be hard to distinguish between the random and the planted distributions. We discuss this in more detail below.
\end{itemize}

In \cref{fig: spca_thresholds1} and \cref{fig: spca_thresholds2}, the regimes corresponding to Diagonal thresholding, Vanilla PCA and Spectral are dark green, while the regimes corresponding to Spectral* and Hard are light green and red respectively. The hard regime is the one studied in this work.

\begin{figure}[!ht]
    \centering
    \includegraphics[scale=.6, trim={0 0 0 0},clip]{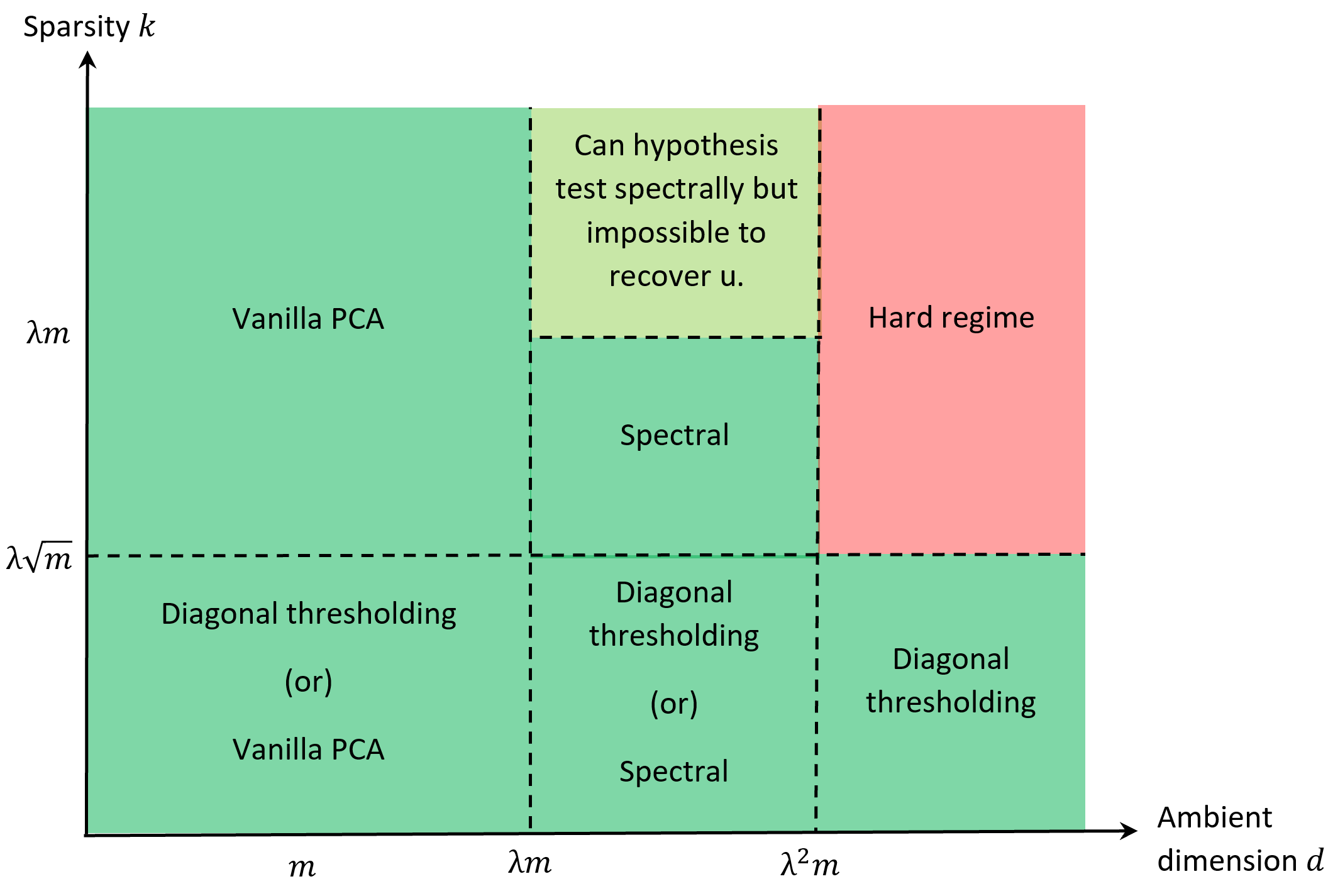}
    \caption{The computational barrier diagram when $\lda \ge 1$}
    \label{fig: spca_thresholds1}
\end{figure}

\begin{figure}[!ht]
    \centering
    \includegraphics[scale=.6, trim={0 0 0 0},clip]{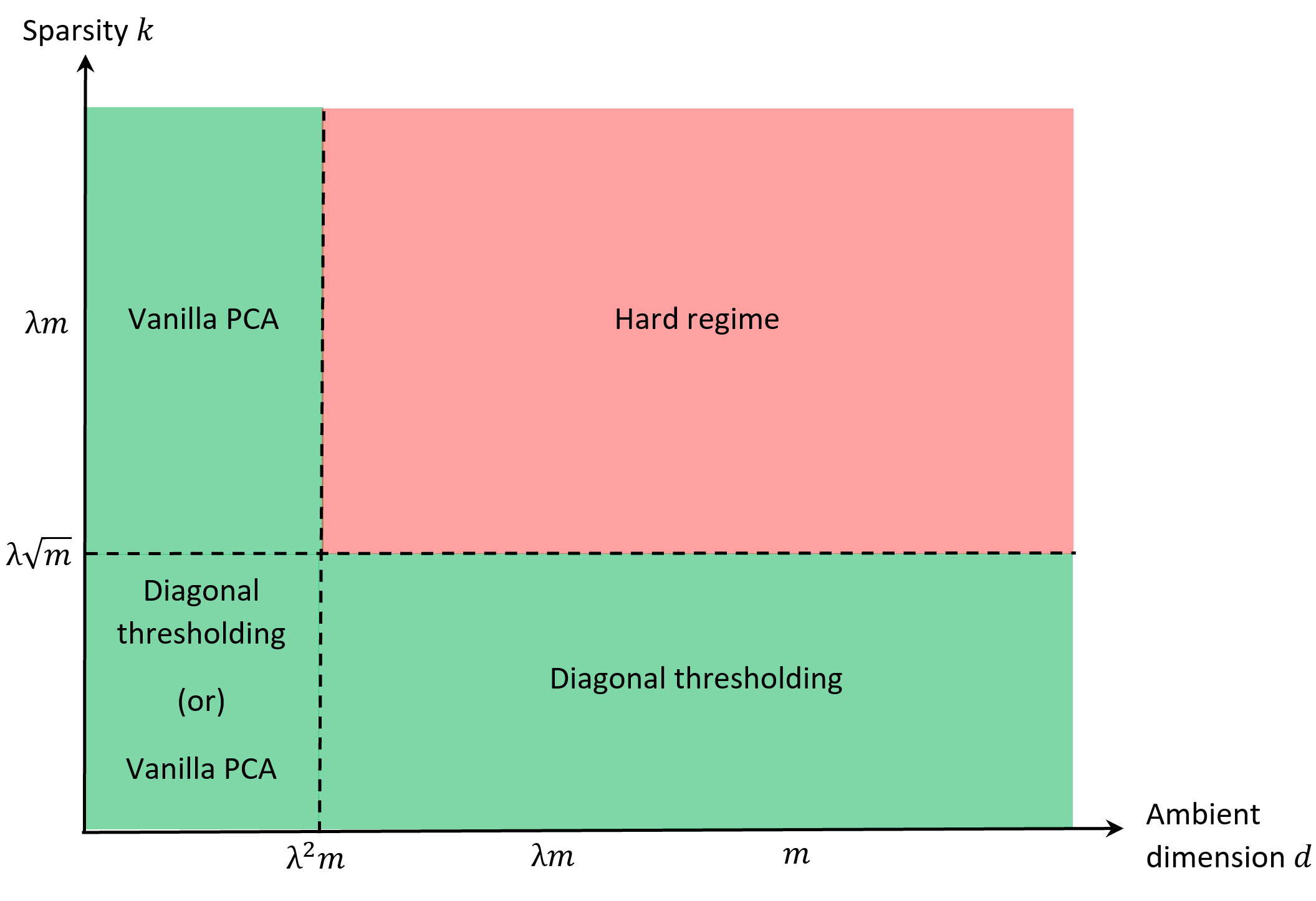}
    \caption{The computational barrier diagram when $\lda < 1$}
    \label{fig: spca_thresholds2}
\end{figure}

In the \textit{Hard} parameter regime where $m \ll \frac{k^2}{\lda^2}$ and $m \ll \frac{d}{\lda^2}$, degree $2$ and (weak) degree $4$ SoS lower bounds have been explored in prior works, while we handle degree $d^{O(\eps)}$. In particular, the works \cite{krauthgamer2015, berthet2013} obtain degree $2$ SoS lower bounds. \cite{ma_wigderson_15} obtain degree $4$ SoS lower bounds using an ad-hoc construction. It's not clear if their construction can be generalized for higher degrees. Moreover, the bounds they obtain are tight up to polylogarithmic factors when $\lda$ is a constant but are not tight when $\lda$ is not a constant, so we improve their bounds even in the degree $4$ case. We subsume all these earlier known results in this work with \cref{cor: spca_main}. This is a vast improvement over prior known sum of squares lower bounds and provides compelling evidence for the hardness of Sparse PCA in this parameter range.

The work \cite{hop17} considers the related but qualitatively different Wigner model of Sparse PCA and they state degree $d^{\eps}$ SoS lower bounds, without explicitly proving these bounds. The techniques in that work do not recover our results because the matrix formed by the random samples in the Wishart model is asymmetric, and handling it correctly is far from being a mere technicality. On the other hand, our machinery can recover tight results on the Wigner model as well, though we only analyze the Wishart model in this paper.

In \cite{sparse_pca_focs20}, they prove that if $m \le \frac{d}{\lda^2}$ and $m \le \left(\frac{k^2}{\lda^2}\right)^{1 - \Omega(\eps)}$, then degree $n^{\eps}$ polynomials cannot distinguish the random and planted distributions. 
\cref{cor: spca_main} says that under mildly stronger assumptions, degree $n^{\eps}$ Sum-of-Squares cannot distinguish the random and planted distributions, so we confirm that SoS is no more powerful than low degree polynomials in this setting.

There have also been direct reductions from planted clique to Sparse PCA \cite{brennan2019optimal}, and it's natural to ask if these reductions can obtain SoS lower bounds on Sparse PCA from the known SoS lower bounds on planted clique \cite{BHKKMP16}. To the best of our knowledge, no such reduction is known and constructing such a reduction would be challenging as it would have to be captured by SoS and avoid losing too much in the parameters. Still, it may well be possible to construct such a reduction.

\subsection{An overview of the machinery: All three results from a single main theorem}

Now that we have described our results,
we want to highlight that all three results are obtained via applications of one main theorem, which we call the machinery. \cref{thm: plds_main}, \cref{thm: tpca_main} and \cref{thm: spca_main} all essentially boil down to showing that a large moment matrix $\Lda$ is PSD. To show this, the machinery constructs certain \emph{coefficient matrices} from $\Lda$ and gives conditions on these coefficient matrices which are sufficient to guarantee that $\Lda$ is PSD with high probability. In this section, we give an informal sketch of the machinery and how it generalizes the techniques used to prove the SoS lower bound for planted clique \cite{BHKKMP16}. We also motivate some of the conditions that arise in the machinery.
\subsubsection{Shapes and graph matrices}
Before we can describe how our machinery works, we need to describe shapes and graph matrices, which were originally introduced by \cite{BHKKMP16, medarametla2016bounds} and later generalized in \cite{AMP20}. Both the planted clique analysis and our analysis use shapes and graph matrices.


Shapes $\al$ are graphs that contain extra information about the vertices. Corresponding to each shape $\al$, there is a matrix-valued function (i.e. a matrix whose entries depend on the input) $M_{\al}$ that we call a graph matrix. Graph matrices are analogous to a Fourier basis, but for matrix-valued functions that exhibit a certain kind of symmetry. In our setting, $\Lda$ will be such a matrix-valued function, so we can decompose $\Lda$ as a linear combination of graph matrices.

Shapes and graph matrices have several properties which make them very useful to work with. First, $\norm{M_{\al}}$ can be bounded with high probability in terms of simple combinatorial properties of the shape $\al$. Second, if two shapes $\alpha$ and $\beta$ match up in a certain way, we can combine them to form a larger shape $\alpha \circ \beta$. We call this operation shape composition. Third, each shape $\al$ has a canonical decomposition into three shapes, the left, middle and right parts of $\al$, which we call $\sigma$, $\tau$, and ${\sigma'}^T$. For this canonical decomposition, we have that $\alpha = \sigma \circ \tau \circ {\sigma'}^T$ and $M_{\alpha} \approx M_{\sigma}M_{\tau}M_{{\sigma'}^T}$ \footnote{Actually, due to a technical issue related to automorphism groups, this equation is off by a multiplicative constant. For details, see Lemma \ref{lm:morthsimplereexpression}.}. This decomposition turns out to be crucial for both the planted clique analysis and our analysis.
\subsubsection{Summary of the SoS lower bound for planted clique and the machinery}
We now give a brief summary of the techniques for the SoS lower bound for planted clique and for our machinery. We elaborate on these steps in \cref{ideadescriptionsubsection}, we formally describe these steps in \cref{sec: informal_statement}, and we carry out these steps in full generality in \cref{sec: technical_def_and_main_theorem} and \cref{sec: proof_of_main}.





For planted clique, the SoS lower bound analysis works as follows
\begin{enumerate}
	\item Using the technique of pseudo-calibration, construct a candidate moment matrix $\Lda$.
	\item Decompose the moment matrix $\Lda$ as a linear combination $\Lda = \sum_{\text{shapes } \alpha}{\lambda_{\alpha}M_{\alpha}}$ of graph matrices $M_{\alpha}$.
	\item For each shape $\alpha$, decompose $\alpha$ into a left part $\sigma$, a middle part $\tau$, and a right part ${\sigma'}^T$. 
	We then have that $M_{\alpha} \approx M_{\sigma}M_{\tau}M_{{\sigma'}^T}$.
	\item Using the approximate decompositions $M_{\alpha} \approx M_{\sigma}M_{\tau}M_{{\sigma'}^T}$, give an approximate decomposition $\Lda \approx LQL^T$ of $M$ where $Q \succeq 0$ with high probability.
	\item Show that with high probability, $\Lda = LQL^T - (LQL^T - M) \succeq 0$ by carefully analyzing the difference $LQL^T - M$ using similar techniques.
\end{enumerate}

For our machinery, we use a similar framework. The key innovation of our machinery is that we introduce coefficient matrices (step 4) and carry out the analysis in terms of these coefficient matrices.
\begin{enumerate}
	\item Construct a candidate moment matrix $\Lda$. This can be done either using pseudo-calibration or in a more ad-hoc manner.
	\item Decompose the moment matrix $\Lda$ as a linear combination $\Lda = \sum_{\text{shapes } \alpha}{\lambda_{\alpha}M_{\alpha}}$ of graph matrices $M_{\alpha}$.
	\item For each shape $\alpha$, decompose $\alpha$ into a left part $\sigma$, a middle part $\tau$, and a right part ${\sigma'}^T$.
	\item Based on the coefficients $\lambda_{\alpha}$ and the decompositions of the shapes $\alpha$ into left, middle, and right parts, construct coefficient matrices $H_{Id_U}$ and $H_{\tau}$.
	\item Based on the coefficient matrices $H_{Id_U}$ and $H_{\tau}$, obtain an approximate PSD decomposition of $\Lda$.
	\item Show that the error terms (which we call intersection terms) can be bounded by the approximate PSD decomposition of $\Lda$.
\end{enumerate}
We show that this analysis will succeed as long as three conditions on the coefficient matrices are satisfied (see \cref{sec: informal_statement} for a qualitative statement of these conditions and \cref{sec: technical_def_and_main_theorem} for the precise statement of these conditions). Thus, in order to use our machinery to prove sum of squares lower bounds, it is sufficient to do the following.
\begin{enumerate}
    \item Construct a candidate moment matrix $\Lda$.
    \item Decompose the moment matrix $\Lda$ as a linear combination $\Lda = \sum_{\text{shapes } \alpha}{\lambda_{\alpha}M_{\alpha}}$ of graph matrices $M_{\alpha}$ (akin to Fourier decomposition) and find the corresponding coefficient matrices.
    \item Verify the required conditions on the coefficient matrices.
\end{enumerate}

\subsubsection{A sketch of the intuition behind the machinery conditions}\label{ideadescriptionsubsection}

\paragraph{Giving an approximate PSD factorization}
As discussed above, we decompose the moment matrix $\Lda$ as a linear combination $\Lda = \sum_{\text{shapes } \alpha}{\lambda_{\alpha}M_{\alpha}}$ of graph matrices $M_{\alpha}$. We then decompose each $\alpha$ into left, middle, and right parts $\sigma$, $\tau$, and ${\sigma'}^T$. We now have that
\[
\Lda = \sum_{\alpha = \sigma \circ \tau \circ {\sigma'}^T}{\lambda_{\sigma \circ \tau \circ {\sigma'}^T}M_{\sigma \circ \tau \circ {\sigma'}^T}}
\]

We first consider the terms $\sum_{\sig, \sig'} \lda_{\sig \circ \sig'^T}M_{\sig \circ \sig'^T} \approx \sum_{\sig, \sig'} \lda_{\sig \circ \sig'^T}M_{\sig} M_{\sig'^T}$ where $\tau$ corresponds to an identity matrix and can be ignored.

If there existed real numbers $v_{\sig}$ for all left shapes $\sig$ such that $\lda_{\sig \circ \sig'^T} = v_{\sig}v_{\sig'}$, then we would have
\[
\sum_{\sig, \sig'} \lda_{\sig \circ \sig'^T}M_{\sig} M_{\sig'^T} = \sum_{\sig, \sig'} v_{\sig}v_{\sig'}M_{\sig} M_{\sig'^T} = (\sum_{\sig} v_{\sig}M_{\sig})(\sum_{\sig} v_{\sig}M_{\sig})^T \succeq 0
\]
which shows that the contribution from these terms is positive semidefinite. In fact, this turns out to be the case for the planted clique analysis. However, this may not hold in general. To handle this, we note that the existence of $v_{\sig}$ can be relaxed as follows: Let $H$ be the matrix with rows and columns indexed by left shapes $\sig$ such that $H(\sig, \sig') = \lda_{\sig \circ \sig'^T}$. Up to scaling, $H$ will be one of our coefficient matrices. If $H$ is positive semidefinite then the contribution from these terms will also be positive semidefinite. In fact, this will be
the PSD mass condition of our main theorem, the qualitative version of which can be found in \cref{informalmaintheoremstatement}.

\paragraph{Handling terms with a non-trivial middle part}

Unfortunately, we also have terms $\lda_{\sig \circ \tau \circ \sig'^T}M_{\sig \circ \tau \circ \sig'^T}$ where $\tau$ is non-trivial. Our strategy will be to charge these terms to other terms.

For the sake of simplicity, we will describe how to handle one term. A starting point is the following inequality. For a left shape $\sig$, a middle shape $\tau$, a right shape $\sig'^T$, and real numbers $a, b$,
\[(a M_{\sig} - bM_{\sig'}M_{\tau^T})(a M_{\sig} - bM_{\sig'}M_{\tau^T})^T \succeq 0\]
which rearranges to
\begin{align*}
ab(M_{\sig}M_{\tau}M_{\sig'^T} + (M_{\sig}M_{\tau}M_{\sig'^T})^T) &\preceq a^2M_{\sig}M_{\sig^T} + b^2M_{\sig'}M_{\tau^T}M_{\tau}M_{\sig'^T}\\
&\preceq a^2M_{\sig}M_{\sig^T} + b^2\norm{M_{\tau}}^2M_{\sig'}M_{\sig'^T}
\end{align*}

If $\lda_{\sig \circ \tau \circ \sig'^T}^2\norm{M_{\tau}}^2 \le \lda_{\sig \circ \sig^T}\lda_{\sig' \circ \sig'^T}$, then we can choose $a, b$ such that $a^2 \le \lda_{\sig \circ \sig^T}, b^2 \norm{M_{\tau}}^2 \le \lda_{\sig' \circ \sig'^T}$ and $ab = \lda_{\sig \circ \tau \circ \sig'^T}$. This will approximately imply
\[\lda_{\sig \circ \tau \circ \sig'^T}(M_{\sig \circ \tau \circ \sig'^T} + M_{\sig \circ \tau \circ \sig'^T}^T) \preceq \lda_{\sig \circ \sig^T}M_{\sig \circ \sig^T} + \lda_{\sig' \circ \sig'^T}M_{\sig' \circ \sig'^T}\]
which will give us a way to charge terms with a nontrivial middle part against terms with a trivial middle part.

While we could try to apply this inequality term by term, it is not strong enough to give us our results. Instead, we generalize this inequality to work with the entire set of shapes $\sig, \sig'$ for a fixed $\tau$. This will lead us to
the middle shape bounds condition of our main theorem, the qualitative version of which can be found in \cref{informalmaintheoremstatement}.

\paragraph{Handing intersection terms}

There's one important technicality in the above heuristic calculations. Whenever we decompose $\alpha$ into left, middle, and right parts $\sigma$, $\tau$, and ${\sigma'}^T$, $M_{\sigma}M_{\tau}M_{{\sigma'}^T}$ is only approximately equal to $M_{\alpha} = M_{\sigma \circ \tau \circ {\sigma'}^T}$. All the other error terms have to be carefully handled in our analysis. We call these terms intersection terms.

These intersection terms themselves turn out to be graph matrices and our strategy is to now recursively decompose them into $\sig_2 \circ \tau_2 \circ \sig_2'^T$ and apply the previous ideas. To do this methodically, we employ several ideas such as the notion of intersection patterns and the generalized intersection tradeoff lemma (see \cref{sec: proof_of_main}). Properly handling the intersection terms is one of the most technically intensive parts of our work. 
This analysis leads us to the intersection term bounds condition of our main theorem, the qualitative version of which can be found in \cref{informalmaintheoremstatement}.

\paragraph{Applying the machinery}

To apply the machinery to our problems of interest, we verify the spectral conditions that our coefficients should satisfy and then we can use our main theorem. The planted slightly denser subgraph application is straightforward and will serve as a good warmup to understand our machinery. In the applications to Tensor PCA and Sparse PCA, the shapes corresponding to the graph matrices with nonzero coefficients have nice structural properties that will be crucial for our analysis. We exploit this structure and use novel charging arguments to verify the conditions of our machinery.

\subsection{Comparison to Other Sum-of-Squares Lower Bounds on Certification Problems}
\cite{BHKKMP16} proved sum of squares lower bounds for the planted clique problem. Our machinery vastly generalizes the techniques of their paper. However, for a specific technical reason, our machinery actually doesn't recover the sum of squares lower bound for planted clique (See \cref{rmk: planted_clique_failure}). That said, this is because we have attempted to keep our framework as general as possible so that it is applicable to other problems, at the cost of losing a specific technicality that's needed for planted clique in particular.

\cite{hop17} remarked that the techniques used in \cite{BHKKMP16} can be used to give Sum-of-Squares lower bounds for $\pm{1}$ variants of tensor PCA and sparse PCA, though this is not made explicit. In this paper, we use our machinery to make these lower bounds explicit. We also handle the Wishart model of sparse PCA, which is more natural and significantly harder to prove lower bounds for. In particular, the bounds we prove do not follow solely from the techniques used in prior works.

\cite{KothariMOW17} proved that for random constraint satisfaction problems (CSPs) where the predicate has a balanced pairwise independent distribution of solutions, with high probability, degree $\Omega(n)$ SoS is required to certify that these CSPs do not have a solution. While they don't state it in this manner, the pseudo-expectation values used by \cite{KothariMOW17} can also be derived using pseudo-calibration \cite{rajendran2018combinatorial, brown2020extended}. That said, their analysis for showing that the moment matrix is PSD is very different. It is an interesting question whether or not it is possible to unify these analyses.

\cite{MRX20} showed that it's possible to lift degree $2$ SoS solutions to degree $4$ SoS solutions under suitable conditions, and used it to obtain degree $4$ SoS lower bounds for average case $d$-regular Max-Cut and the Sherrington Kirkpatrick problem. Their construction is inspired by pseudo-calibration and their analysis also goes via graph matrices.

In a joint work with others~\cite{sklowerbounds}, we proved degree $n^{\eps}$ SoS lower bounds for the Sherrington-Kirkpatrick problem via an intermediate problem which we called Planted Affine Planes. In that work, the construction and analysis also goes via pseudo-calibration and graph matrices, but since the constructed moment matrix had a nontrivial nullspace,
we had to do a certain kind of preprocessing to handle this nullspace. After this preprocessing, the moment matrix was dominated by its expected value, so the graph matrix norm bounds \cite{AMP20} were sufficient to prove our result and we did not need to use this machinery, which would have been overkill.

\cite{kunisky2020} recently proposed a technique to lift degree $2$ SoS lower bounds to higher levels and applied it to construct degree $6$ lower bounds for the Sherrington-Kirkpatrick problem. Interestingly, their construction does not go via pseudo-calibration.

In a joint work with other collaborators~\cite{jones2021sum}, we recently obtained SoS lower bounds for the problem of independent set on sparse \Erdos-\Renyi random graphs. To do this, we introduced a variant of psuedo-calibration which we called connected truncation. In that work, the analysis was similar to the planted clique analysis \cite{BHKKMP16} and the analysis in this machinery. Indeed, several techniques from the analysis of our machinery were used, such as shifting to shapes by incurring factors of sizes of automorphism groups and bounding sums by defining the $c(\al)$ function. However, we proved our main lower bound directly in that work because we were trying to push the limits of our machinery and in general, the limits of known SoS lower bounds. In particular, we were working with sparse inputs where graph  matrices behave differently (\cite{jones2021sum, RT20}). While this can be accommodated by our machinery by modifying certain parameters, we had other technical barriers such as missing edge indicators that were problem specific and hence, not handled here. Finally, while we have a slack of $\poly(d \log n)$ in our bounds in the machinery, in \cite{jones2021sum} we obtained a tighter bound in \cite{jones2021sum} which is tight up to a factor of $O(\poly(d) \log n)$ (where $d$ is the SoS degree).

\subsection{Related Algorithmic Techniques}\label{subsec: related_techniques}

\paragraph{Low degree polynomials} Consider a problem where the input is sampled from one of two distributions and we would like to identify which distribution it was sampled from. Usually, one distribution is the completely random distribution while the other is a planted distribution that contains a given structure not present in the random distribution. In this setting, a closely related method is to use low-degree polynomials to try and distinguish the two distributions. More precisely, if there is a low-degree polynomial such that its expected value on the random distribution is very different than its expected value on the planted distribution, this distinguishes the two distributions.
Recently, this method has been shown to be an excellent heuristic, as it recovers the conjectured hardness thresholds for several problems and is considerably easier to analyze \cite{hop17, hop18, kunisky19notes}.

Under mild conditions, the SoS hierarchy is at least as powerful as low degree polynomials \cite{hop17}. However, it is unknown whether low degree polynomials generally have the same power as the SoS hierarchy or if there are situations where the SoS hierarchy is significantly more powerful. The low-degree conjecture \cite{hop17, hop18} says that if there is a low-degree polynomial lower bound for a problem which is sufficiently symmetric then if we consider a noisy version of the problem, no polynomial time algorithm can solve this problem. In particular, low-degree sum of squares cannot solve this problem.

So far, the low-degree conjecture has stood up well. \cite{holmgren2020counterexamples} tested the low-degree conjecture by exploring problems where there is a low-degree polynomial lower bound yet the problem can be solved in polynomial time. They found that while such problems exist, so far all of these problems either rely on an asymmetric structure such as an error correcting code or can be made harder by adding some noise \footnote{As shown by \cite{holmgren2020counterexamples}, we have to be careful about what kind of noise we add. In particular, replacing a fraction of the input with random noise is insufficient. Instead, we should add some noise to every entry of the input.}. Thus, a variant of the low-degree conjecture may well be true. It is
a fascinating open problem to obtain SoS lower bounds directly from low-degree polynomial lower bounds.

In this paper, we confirm that for tensor PCA and the Wishart model of sparse PCA with slightly adjusted planted distributions, the SoS hierarchy is no more powerful than low-degree polynomials. Our machinery is also a potential approach for obtaining SoS lower bounds from low-degree polynomial lower bounds. In particular, if we could show that low-degree polynomial bounds for a problem imply that the conditions required by our machinery are satisfied for a noisy version of the problem, this would give the desired SoS lower bounds.

\paragraph{The Statistical Query Model}
The statistical query model is another popular restricted class of algorithms introduced by \cite{kearns1998efficient}. In this model, for an underlying distribution, we can access it by querying expected value of functions of the distribution. Concretely, for a distribution $D$ on $\RR^n$, we have access to it via an oracle that given as query a function $f: \RR^n \to [-1, 1]$ returns $\EE_{x \sim D} f(x)$ up to some additive adversarial error. SQ algorithms capture a broad class of algorithms in statistics and machine learning and have been used to study information-computation tradeoffs. There has also been significant work trying to understand the limits of SQ algorithms (e.g. \cite{feldman2017statistical, feldman2018complexity, diakonikolas2017statistical}).

An important distinction to note is that SQ algorithms do not take into account the complexity of the oracle, so SQ algorithms may have more power than SoS. On the other hand, SoS can work directly with the input samples while SQ algorithms can only query expected values of functions on the input, which may give SoS more power than SQ algorithms. Thus, in general, SQ algorithms and SoS are incomparable. That said, the recent work \cite{brennan2020statistical} showed that low-degree polynomials and statistical query algorithms have equivalent power under mild conditions. Under these conditions, SoS lower bounds give strictly stronger evidence of hardness.

\subsection{Organization of the paper}

In this work, we occasionally distinguish between the qualitative and quantitative versions of theorem statements. Qualitative theorem statements capture the essence of the inequalities we prove, and serve to illustrate the main forms of the bounds we desire, which are helpful to build intuition. Quantitative theorems on the other hand build on their qualitative counterparts by stating the precise bounds that are needed.

The remainder of this paper is organized as follows. In \cref{sec: prelim}, we give some preliminaries. In particular, we describe the Sum-of-Squares hierarchy and present a brief overview of the machinery and some proof techniques that we use. In \cref{sec: informal_statement}, we present the informal statement of the main theorem. In \cref{sec: plds_qual}, \cref{sec: tpca_qual} and \cref{sec: spca_qual}, we qualitatively verify the conditions of the machinery for planted slightly denser subgraph, tensor PCA, and sparse PCA respectively. While these sections only verify the qualitative conditions, the results in these sections are precise and will be reused in \cref{sec: plds_quant}, \cref{sec: tpca_quant} and \cref{sec: spca_quant} to fully verify the conditions of the machinery.

In \cref{sec: technical_def_and_main_theorem}, we introduce the formal definitions and state the main theorem in full generality. We recall the definitions in the simpler case and also show how it generalizes (denoted by an asterisk *). We leave the exposition choice to the reader. This section also contains many examples to illustrate the definitions.
In \cref{sec: proof_of_main}, we prove the main theorem while abstracting out the choice of several functions. In \cref{sec: choosing_funcs}, we choose these functions so that that they satisfy the conditions needed for our main theorem. In Section \ref{sec: showing_positivity}, we give a general strategy to bound truncation error. In \cref{sec: plds_quant}, \cref{sec: tpca_quant} and \cref{sec: spca_quant}, we prove \cref{thm: plds_main}, \cref{thm: tpca_main} and \cref{thm: spca_main} respectively. We conclude in \cref{sec: conclusion}.

%% file: prelim.tex
\subsection{The Sum of Squares Hierarchy}\label{subsec: sos}

The SoS hierarchy is a powerful class of algorithms parameterized by it's degree. As we increase the degree, we get progressively stronger algorithms (with longer running times). It's been shown formally to obtain the state-of-the art guarantees for many problems both in the worst case and the average case setting.
For constant degree SoS, the hierarchy can be optimized in polynomial time\footnote{There is a caveat, see \cite{o2017sos}}. In general, for degree-$d$ SoS, we can solve it in $n^{O(d)}$ time. Most of our applications in this paper focus on showing hardness for the SoS hierarchy when the degree is $n^{\eps}$, which corresponds to a subexponential running time.

We now formally describe the sum of squares hierarchy.

\begin{definition}[Pseudo-expectation values]\label{def: pseudoexpectation}
Given polynomial constraints $g_1 = 0$,\ldots,$g_m = 0$, degree $d$ pseudo-expectation values are a linear map $\pE$ from polynomials of degree at most $d$ to $\mathbb{R}$ satisfying the following conditions:
  \begin{enumerate}
    \item $\pE[1] = 1$, \label{pe:normalized}
    \item $\pE[f \cdot g_i] = 0$ for every $i \in [m]$ and polynomial $f$ such that $\deg(f \cdot g_i) \leq d$. \label{pe:feasible}
    \item $\pE[f^2] \geq 0$ for every polynomial $f$ such that $\deg(f^2) \le d$. \label{pe:psdness}
  \end{enumerate}
\end{definition}
The intuition behind pseudo-expectation values is that the conditions on the pseudo-expectation values are conditions that would be satisfied by any actual expected values over a distribution of solutions, so optimizing over pseudo-expectation values gives a relaxation of the problem. Moreover, the conditions on pseudo-expectation values can be captured by a semidefinite program.
In particular, \cref{pe:psdness} in \cref{def: pseudoexpectation} can be reexpressed in terms of a matrix called the moment matrix.

\begin{definition}[Moment Matrix of $\pE$]
Given degree $d$ pseudo-expectation values $\pE$, define the associated moment matrix $\Lda$ to be a matrix with rows and columns indexed by monomials $p$ and $q$ such that the entry corresponding to row $p$ and column $q$ is
  \[
  \Lda[p, q] \defeq \pE\left[pq\right].
  \]
\end{definition}

It is easy to verify that \cref{pe:psdness} in~\cref{def: pseudoexpectation} equivalent to $\Lda \succeq 0$.

For our setting, we are investigating the following kind of question. Given polynomial constraints $g_1 = 0$,\ldots,$g_m = 0$, can degree $d$ SoS certify that some other polynomial $h$ has value at most $c$?

If there do not exist pseudo-expectation values $\pE$ satisfying the conditions in~\cref{def: pseudoexpectation} such that $\pE[h] > c$ then degree $d$ SoS certifies that $\pE[h] \leq c$. More precisely, by duality, there exists a degree $d$ SoS/Positivstellensatz proof that $h \leq c$.
On the other hand, if there exist degree $d$ pseudo-expectation values $\pE$ satisfying the conditions in~\cref{def: pseudoexpectation} such that $\pE[h] > c$ then degree $d$ SoS fails to certify that $h \leq c$. This is what we need to show in order to prove SoS lower bounds on certification problems.

While originally introduced in the context of combinatorial optimization where it remains an effective technique \cite{GW94, AroraRV04, GuruswamiS11, raghavendra2017strongly}, SoS algorithms have recently revolutionized robust machine learning, where we study learning algorithms for noisy data, where the noise could be either random or adversarial.
Robust machine learning has found a variety of important safety-critical applications, e.g. in computer vision \cite{szegedy2013intriguing, goodfellow2014explaining, xie2019feature, hendrycks2021natural, sebe2013robust, xie2020adversarial, fischer2017adversarial, kurakin2016adversarial} and speech recognition \cite{hsu2021robust, wang2022wav2vec, rajendran2022analyzing, ravanelli2020multi, li2015robust, alzantot2018did, neekhara2019universal, olivier2022recent}.
In this important field, SoS has recently lead to breakthrough algorithms for
long-standing open problems \cite{bakshi2020robustly, liu2021settling, hopkins2020mean, klivans2018efficient, FKP19, kothari2017outlier, bakshi2020outlier, bakshi2020list, schramm2017fast}. Highlights include robustly learning mixtures of Gaussians (see \cite{bakshi2020robustly, liu2021settling} and references therein), efficient robust algorithms for regression \cite{klivans2018efficient}, moment estimation \cite{kothari2017outlier}, clustering \cite{bakshi2020outlier} and subspace recovery \cite{bakshi2020list}. See also the works \cite{BarakBHKSZ12, bks15, HopSS15, pot17}) for more.

\subsection{Pseudo-calibration}\label{subsec: pseudocalibration}

To obtain SoS integrality gaps on random instances, we need to construct valid pseudo-expectation values for a random input instance of an optimization problem. Naturally, these pseudo-expectation values will depend on the input. Psuedo-calibration is a heuristic introduced by \cite{BHKKMP16} to construct such candidate pseudo-expectation values almost mechanically by considering a planted distribution supported on instances of the problem with large objective value and using this planted distribution as a guide to construct the pseudo-expectation values.

Unfortunately, psuedo-calibration doesn't guarantee feasibility of these candidate pseudo-expectation values and the corresponding moment matrix and this has to be verified separately for different problems. This verification of feasibility is relatively easy except for the PSDness condition, which often leads to highly technical and involved analyses. The machinery attempts to mitigate this problem by providing easily verifiable conditions to prove PSDness, regardless of whether the moment matrix was obtained via pseudo-calibration.

For our applications, psuedocalibration is used to obtain a candidate pseudoexpectation operator $\pE$ and a corresponding moment matrix $\Lda$
from the random vs planted problem. This will be the starting point for all our applications. Here, we do not attempt to motivate and describe it in great detail. Instead, we will briefly describe the heuristic, the intuition behind it and show an example of how to use it. A detailed treatment can be found in \cite{BHKKMP16}.

Let $\nu$ denote the random distribution and $\mu$ denote the planted distribution. Let $v$ denote the input and $x$ denote the variables for our SoS relaxation. The main idea is that, for an input $v$ sampled from $\nu$ and any polynomial $f(x)$ of degree at most the SoS degree, pseudo-calibration proposes that for any low-degree test $g(v)$, the correlation of $\pE[f]$ should match in the planted and random distributions. That is,
\[\EE_{v \sim \nu}[\pE[f(x)]g(v)] = \EE_{(x, v) \sim \mu}[f(x)g(v)]\]

Here, the notation $(x, v) \sim \mu$ means that in the planted distribution $\mu$, the input is $v$ and $x$ denotes the planted structure in that instance. For example, in planted clique, $x$ would be the indicator vector of the clique. If there are multiple, pick an arbitrary one.

Let $\calF$ denote the Fourier basis of polynomials for the input $v$. By choosing different basis functions from $\calF$ as choices for $g$ such that the degree is at most $n^{\eps}$ (hence the term low-degree test), we get all lower order Fourier coefficients for $\pE[f(x)]$ when considered as a function of $v$. Furthermore, the higher order coefficients are set to be $0$ so that the candidate pseudoexpectation operator can be written as
\[\pE f(x) = \sum_{\substack{g \in \calF\\deg(g) \le n^{\eps}}} \EE_{v \sim \nu}[\pE[f(x)]g(v)] g(v) = \sum_{\substack{g \in \calF\\deg(g) \le n^{\eps}}} \EE_{(x, v) \sim \mu}[[f(x)]g(v)] g(v)\]

The coefficients $\EE_{(x, v) \sim \mu}[[f(x)]g(v)]$ can be explicitly computed in many settings, which therefore gives an explicit pseudoexpectation operator $\pE$.

One intuition for pseudo-calibration is as follows. The planted distribution is usually chosen to be a maximum entropy distribution which still has the planted structure. This conforms to the philosophy that random instances are hard for SoS, such as the uniform Bernoulli distribution for planted clique or the Gaussian distribution for Tensor PCA. By conditioning on the lower order moments matching such a planted distribution, pseudo-calibration can be interpreted as sort of interpolating between the random and planted distributions by only looking at lower order Fourier characters. This intuition has proven to be successful, since pseudo-calibration been successfully exploited to construct SoS lower bounds for a wide variety of dense as well as sparse problems.

An advantage of pseudo-calibration is that this construction automatically satisfies some nice properties that the pseudoexpectation $\pE$ should satisfy. It's linear in $v$ by construction. For all polynomial equalities of the form $f(x) = 0$ that is satisfied in the planted distribution, it's true that $\pE[f(x)] = 0$. For other polynomial equalities of the form $f(x, v) = 0$ that are satisfied in the planted distribution, the equality $\pE[f(x, v)] = 0$ is approximately satisfied. In most cases, $\pE$ can be mildly adjusted to satisfy these exactly.

The condition $\pE[1] = 1$ is not automatically satisfied but in most applications, we usually require that $\pE[1] = 1 \pm \littleoh(1)$. Indeed, this has been the case for all known successful applications of pseudo-calibration. Once we have this, we simply set our final pseudoexpectation operator to be $\pE'$ defined as $\pE'[f(x)] = \pE[f(x)] / \pE[1]$.

We remark that the condition $\pE[1] = 1 \pm \littleoh(1)$ corresponds to having a low-degree polynomial lower bound and has been quite successful in predicting the right thresholds between approximability and inapproximability \cite{hop17, hop18, kunisky19notes}.

\paragraph{Example: Planted Clique}
As an warmup, we review the pseudo-calibration calculation for planted clique. Here, the random distribution $\nu$ is $G(n, \frac{1}{2})$.
The planted distribution $\mu$ is as follows. For a given integer $k$, first sample $G'$ from $G(n, \frac{1}{2})$, then choose a random subset $S$ of the vertices where each vertex is picked independently with probability $\frac{k}{n}$. For all pairs $i, j$ of distinct vertices in $S$, add the edge $(i, j)$ to the graph if not already present. Set $G$ to be the resulting graph.

The input is given by $G \in \{-1, 1\}^{\binom{[n]}{2}}$ where $G_{i, j}$ is $1$ if the edge $(i, j)$ is present and $-1$ otherwise. Let $x_1, \ldots, x_n$ be the boolean variables for our SoS program such that $x_i$ indicates if $i$ is in the clique.

\begin{definition}
Given a set of vertices $V \subseteq [n]$, define $x_V = \prod_{v \in V}{x_v}$.
\end{definition}
\begin{definition}
Given a set of possible edges $E \subseteq \binom{[n]}{2}$, define $\chi_E = (-1)^{|E \setminus E(G)|} = \prod_{(i, j) \in E}G_{i, j}$.
\end{definition}

Pseudo-calibration says that for all small $V$ and $E$,
\[
\EE_{G \sim \nu}\left[\tilde{E}[x_V]\chi_E\right] = \EE_{\mu}\left[x_V{\chi_E}\right]
\]
Using standard Fourier analysis, this implies that if we take
$
c_E = \EE_{\mu}\left[x_V{\chi_E}\right] = \left(\frac{k}{n}\right)^{|V \cup V(E)|}
$
where $V(E)$ is the set of the endpoints of the edges in $E$, then for all small $V$,
\[
\pE[x_V] = \sum_{E:E \text{ is small}}{{c_E}\chi_E} = \sum_{E:E \text{ is small}}{\left(\frac{k}{n}\right)^{|V \cup V(E)|}\chi_E}
\]

Since the values of $\pE[x_V]$ are known, by multi-linearity, this can be naturally extended to obtain values $\pE[f(x)]$ for any polynomial $f$ of degree at most the SoS degree.

%% file: informal_statement.tex
In this section, we informally describe our machinery for proving sum of squares lower bounds on planted problems. Our goal for this section is to qualitatively state the conditions under which we can show that the moment matrix $\Lambda$ is PSD with high probability (see \cref{informalmaintheoremstatement}). For simplicity, in this section we restrict ourselves to the setting where the input is $\{-1,1\}^{\binom{n}{2}}$ (i.e. a random graph on $n$ vertices). We also defer the proofs of several important facts until \cref{sec: technical_def_and_main_theorem}. In \cref{sec: technical_def_and_main_theorem}, we give the general definitions, fill in the missing proofs, and give the full, quantitative statement of our main result (see \cref{generalmaintheorem}).

\subsection{Fourier analysis for matrix-valued functions: ribbons, shapes, and graph matrices}
For our machinery, we need the definitions of ribbons, shapes, and graph matrices from \cite{AMP20}.


\subsubsection{Ribbons}
\emph{Ribbons} lift the usual Fourier basis for functions $\{ f \, : \, \{ \pm 1\}^{n \choose 2} \rightarrow \R \}$ to  matrix-valued functions.




\begin{definition}[Simplified ribbons -- see \cref{def: ribbons}]
	Let $n \in \N$.
	A ribbon $R$ is a tuple $(E_R,A_R,B_R)$ where $E_R \subseteq {[n] \choose 2}$ and $A_R,B_R$ are tuples of elements in $[n]$.
	$R$ thus specifies:
	\begin{enumerate}
		\item A Fourier character $\chi_{E_R}$.
		\item Row and column indices $A_R$ and $B_R$.
	\end{enumerate}
	We think of $R$ as a graph with vertices
	\[
	V(R) = \{ \text{ endpoints of $(i,j) \in E_R$ } \} \cup A_R \cup B_R
	\]
	and edges $E(R) = E_R$, where $A_R, B_R$ are distinguished tuples of vertices.
\end{definition}

\begin{definition}[Matrix-valued function for a ribbon $R$]
	Given a ribbon $R$, we define the matrix valued function $M_R \, : \, \{ \pm 1\}^{n \choose 2} \rightarrow \R^{\frac{n!}{(n - |A_R|)!} \times \frac{n!}{(n - |B_R|)!}}$ to have entries $M_R(A_R,B_R) = \chi_{E_R}$ and $M_R(A',B') = 0$ whenever $A' \neq A_R$ or $B' \neq B_R$.
\end{definition}

The following proposition captures the main property of the matrix-valued functions $M_R$ -- they are an orthonormal basis. We leave the proof to the reader.
\begin{proposition}
	The matrix-valued functions $M_R$ form an orthonormal basis for the vector space of matrix valued functions with respect to the inner product
	\[
	\iprod{M,M'} = \E_{G \sim \{ \pm 1\}^{n \choose 2} }\left[\Tr \left(M(G) (M'(G))^\top\right)\right].
	\]
\end{proposition}

We don't directly utilize this proposition in our work but this gives insight on to the structure of the matrix valued functions we define and motivates the definition of graph matrices, that we use extensively.


\begin{example}
    In \cref{fig: ribbon_shape}, consider the ribbon $R$ as shown. We have $A_R = (1, 3), B_R = (4), V(R) = \{1, 2, 3, 4\}, E_R = \{\{1, 2\}, \{3, 2\}, \{2, 4\}\}$. The Fourier character is $\chi_{E_R} = \chi_{1, 2}\chi_{3, 2}\chi_{2, 4}$. And finally, $M_R$ is a matrix with rows and columns indexed by tuples of length $|A_R| = 2$ and $|B_R| = 1$ respectively, with exactly one nonzero entry $M_R((1, 3), (4)) = \chi_{E_R}$. Succinctly, \[M_R =
  \begin{blockarray}{rl@{}c@{}r}
    & & \makebox[0pt]{column $(4)$} \\[-0.5ex]
    & & \,\downarrow \\[-0.5ex]
    \begin{block}{r(l@{}c@{}r)}
    & \makebox[3.1em]{\Large $0$\bigstrut[t]} & \vdots &\makebox[4.2em]{\Large $0$} \\[-0.2ex]
    \text{row }(1, 3) \to \mkern-9mu & \raisebox{0.5ex}{\makebox[3.2em][l]{\dotfill}} & \chi_{1, 2}\chi_{3, 2}\chi_{2, 4} & \raisebox{0.5ex}{\makebox[4.2em][r]{\dotfill}} \\[+0ex]
    & \makebox[3.1em]{\Large $0$} & \vdots &\makebox[4.2em]{\bigstrut\Large $0$} \\
    \end{block}
  \end{blockarray}\]
\end{example}

\begin{figure}[!h]
    \centering
    \includegraphics[scale=.6, trim={0 5cm 2 5cm},clip]{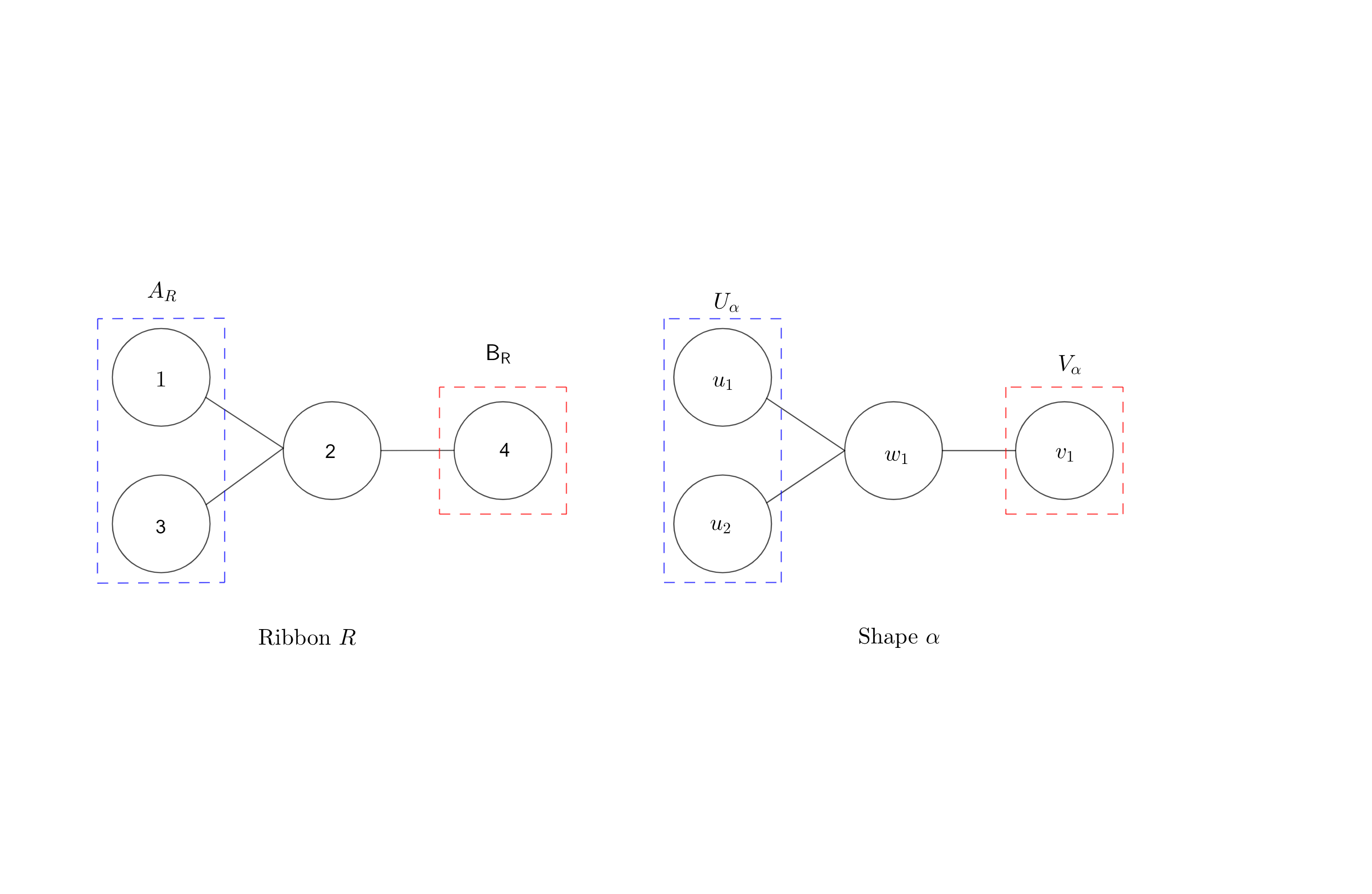}
    \caption{Example of a ribbon and a shape}
    \label{fig: ribbon_shape}
\end{figure}

\subsubsection{Shapes and Graph Matrices}
As described above, \emph{ribbons} are an orthonormal basis for matrix-valued functions. However, we will need an orthogonal basis for the subset of those functions which are symmetric with respect to the action of $S_n$. For this, we use \emph{graph matrices}, which are described by \emph{shapes}. The idea is that each ribbon $R$ has a shape $\alpha$ which is obtained by replacing the vertices of $R$ with unspecified indices. Up to scaling, the graph matrix $M_{\alpha}$ is the average of $M_{\pi(R)}$ over all permutations $\pi \in S_n$.



\begin{definition}[Simplified shapes -- see \cref{def: shapes}]
	Informally, a shape $\alpha$ is just a ribbon $R$ where the vertices are specified by variables rather than having specific values in $[n]$. More precisely, a shape $\alpha = (V(\alpha),E(\alpha),U_{\alpha},V_{\alpha})$ is a graph on vertices $V(\alpha)$, with
	\begin{enumerate}
		\item Edges $E(\alpha) \subseteq {{V(\alpha)} \choose 2}$
		\item Distinguished tuples of vertices $U_\alpha = (u_1,u_2,\dots)$ and $V_\alpha = (v_1,v_2,\dots)$, where $u_i,v_i \in V(\alpha)$.
	\end{enumerate}
	(Note that $V(\alpha)$ and $V_\alpha$ are not the same object!)
\end{definition}

\begin{definition}[Shape transposes]
	Given a shape $\alpha$, we define $\alpha^{\top}$ to be the shape $\alpha$ with $U_{\alpha}$ and $V_{\alpha}$ swapped i.e. $U_{\sigma^{\top}} = V_{\sigma}$ and $V_{\sigma^{\top}} = U_{\sigma}$.
	Note that $M_{\alpha^{\top}} = M_\alpha^{\top}$, where $M_\alpha^{\top}$ is the usual transpose of the matrix-valued function $M_\alpha$.
\end{definition}

\begin{definition}[Graph matrices]
	Let $\alpha$ be a shape.
	The graph matrix $M_{\alpha} \, : \, \{ \pm 1\}^{n \choose 2} \rightarrow \R^{\frac{n!}{(n - |U_{\alpha}|)!} \times \frac{n!}{(n - |V_{\alpha}|)!}}$ is defined to be the matrix-valued function with $A,B$-th entry
	\[
	M_{\alpha}(A,B) = \sum_{\substack{R \text{ s.t. } A_R = A, B_R = B \\ \exists \phi:V(\alpha) \to [n]: \\ \phi \text{ is injective}, \phi(\alpha) = R}}{\chi_{E_R}}
	\]
	In other words, $M_\alpha = \sum_{R} M_R$ where the sum is over ribbons $R$ which can be obtained by assigning each vertex in $V(\alpha)$ a label from $[n]$.
\end{definition}

\begin{example}
In \cref{fig: ribbon_shape}, consider the shape $\al$ as shown. We have $U_{\al} = (u_1, u_2), V_{\al} = (v_1), V(\al) = \{u_1, u_2, v_1, w_1\}$ and $E(\al) = \{\{u_1, w_1\}, \{u_2, w_1\}, \{w_1, v_1\}\}$. $M_{\al}$ is a matrix with rows and columns indexed by tuples of length $|U_{\al}| = 2$ and $|V_{\al}| = 1$ respectively. The nonzero entries will have rows and columns indexed by $(a_1, a_2)$ and $b_1$ respectively for all distinct $a_1, a_2, b_1$, with the corresponding entry being $M_{\al}((a_1, a_2), (b_1)) = \sum_{c_1 \in [n] \setminus \{a_1, a_2, b_1\}} \chi_{a_1, c_1}\chi_{a_2, c_1}, \chi_{c_1, b_1}$. Here, the injective map $\phi$ maps $u_1, u_2, w_1, v_1$ to $a_1, a_2, c_1, b_1$ respectively and we sum over all such maps. Succinctly, \[M_{\al} =
  \begin{blockarray}{rl@{}c@{}r}
    & & \makebox[0pt]{column $(b_1)$} \\[-0.5ex]
    & & \,\downarrow \\[-0.5ex]
    \begin{block}{r(l@{}c@{}r)}
    &  & \vdots & \\[-0.2ex]
    \text{row }(a_1, a_2) \to \mkern-9mu & \raisebox{0.5ex}{\makebox[3.2em][l]{\dotfill}} & \sum_{c_1 \in [n] \setminus \{a_1, a_2, b_1\}} \chi_{a_1, c_1}\chi_{a_2, c_1}\chi_{c_1, b_1} & \raisebox{0.5ex}{\makebox[4.2em][r]{\dotfill}} \\[+.5ex]
    &  & \vdots & \\
    \end{block}
  \end{blockarray}\]
\end{example}

\begin{remk}
    The fact that we are summing over all "free" vertices in $V(\al) \setminus (U_{\al} \cup V_{\al})$ is how we are incorporating symmetry into the definition of these graph matrices.
\end{remk}

The following examples illustrate that simple matrices such as the adjacency matrix of a graph and the identity matrix are also graph matrices.

\begin{example}[Adjacency matrix]\label{ex:adj-matrix}
	Let $\alpha$ be the shape with two vertices $V(\alpha) = \{u_1,v_1\}$ and a single edge $E(\alpha) = \{\{ u_1,v_1\}\}$. The tuples $U_\alpha, V_\alpha$ are $(u_1), (v_1)$, respectively.
	Then $M_\alpha$ has entries $(M_\alpha)_{i,j}(G) = G_{ij}$ if $i \neq j$ and $(M_\alpha)_{i,i} = 0$.
	If $G \in \{ \pm 1\}^{n \choose 2}$ is thought of as a graph, then $M_\alpha$ is precisely its $\pm 1$ adjacency matrix with zeros on the diagonal.
\end{example}

\begin{example}[Identity matrix]
	If $V(\alpha) = \{u\}$ is a singleton, $E(\alpha) = \emptyset$, and $U_{\alpha} = V_{\alpha} = (u)$, then $M_\alpha(G)$ is identically equal to the $n \times n$ identity matrix, independent of $G$.
\end{example}



For more examples of graph matrices and why they can be a useful tool to work with, see \cite{AMP20}.
\begin{remark}
As noted in \cite{AMP20}, we index graph matrices by tuples rather than sets so that they are symmetric (as a function of the input) under permutations of $[n]$.
\end{remark}

\subsection{Factoring Graph Matrices and Decomposing Shapes into Left, Middle, and Right Parts}

A crucial idea in our analysis is the idea from \cite{BHKKMP16} of decomposing each shape $\alpha$ into left, middle, and right parts. This will allow us to give an approximate factorization of each graph matrix $M_{\alpha}$.

\subsubsection{Leftmost and Rightmost Minimum Vertex Separators and Decomposition of Shapes into Left, Middle, and Right Parts}
For each shape $\alpha$ we will identify three other shapes, which we denote by $\sigma,\tau,{\sigma'}^T$ and call (for reasons we will see soon) the \emph{left, middle, and right parts of $\alpha$}, respectively.
The idea is that $M_{\alpha} \approx M_{\sigma} M_{\tau} M_{{\sigma'}^T}$.
We obtain $\sigma, \tau$, and ${\sigma'}^T$ by splitting the shape $\alpha$ along the \emph{leftmost and rightmost minimum vertex separators}.

\begin{definition}[Vertex Separators]
	We say that a set of vertices $S$ is a vertex separator of $\alpha$ if every path from $U_{\alpha}$ to $V_{\alpha}$ in $\alpha$ (including paths of length $0$) intersects $S$. Note that for any vertex separator $S$, $U_{\alpha} \cap V_{\alpha} \subseteq S$.
\end{definition}



\begin{definition}[Minimum Vertex Separators]
	We say that $S$ is a minimum vertex separator of $\alpha$ if $S$ is a vertex separator of $\alpha$ and for any other vertex separator $S'$ of $\alpha$, $|S| \leq |S'|$.
\end{definition}

\begin{definition}[Leftmost and Rightmost Minimum Vertex Separators] \
	\begin{enumerate}
		\item We say that $S$ is the leftmost minimum vertex separator of $\alpha$ if $S$ is a minimum vertex separator of $\alpha$ and for every other minimum vertex separator $S'$ of $\alpha$, every path from $U_{\alpha}$ to $S'$ intersects $S$.
		\item We say that $T$ is the rightmost minimum vertex separator of $\alpha$ if $T$ is a minimum vertex separator of $\alpha$ and for every other minimum vertex separator $S'$ of $\alpha$, every path from $S'$ to $V_{\alpha}$ intersects $T$.
	\end{enumerate}
\end{definition}
It is not immediately obvious that leftmost and rightmost minimum vertex separators are well-defined. For the simplified setting we are considering here, this was shown by \cite{BHKKMP16}. We give a more general proof in \cref{separatorswelldefinedsection}.

We now describe how to split $\alpha$ into left, middle, and right parts $\sigma, \tau$, and ${\sigma'}^T$.

\begin{definition}[Decomposition Into Left, Middle, and Right Parts]
	Let $\alpha$ be a shape and let $S$ and $T$ be the leftmost and rightmost minimum vertex separators of $\alpha$. Given orderings $O_S$ and $O_T$ for $S$ and $T$, we decompose $\alpha$ into left, middle, and right parts $\sigma$, $\tau$, and ${\sigma'}^T$ as follows.
	\begin{enumerate}
		\item The left part $\sigma$ of $\alpha$ is the part of $\alpha$ reachable from $U_\alpha$ without passing through $S$. It includes $S$ but excludes all edges which are entirely within $S$.
		More formally,
		\begin{enumerate}
		    \item $V(\sigma) = \{u \in V(\alpha): \text{ there is a path } P \text{ from } U_{\alpha} \text{ to } u \text{ in } \alpha \text{ such that } (V(P) \setminus \{u\}) \cap S = \emptyset\}$
		    \item $U_\sigma = U_\alpha$ and $V_\sigma = S$ with the ordering $O_S$
		    \item $E(\sigma) = \{\{u,v\} \in E(\alpha): u,v \in V(\sigma), u \notin S \text{ or } v \notin S\}$
		\end{enumerate}
		\item 
		The right part ${\sigma'}^T$ of $\alpha$ is the part of $\alpha$ reachable from $V_\alpha$ without intersecting $T$ more than once. It includes $T$ but excludes all edges which are entirely within $T$.
		More formally,
		\begin{enumerate}
		    \item $V({\sigma'}^T) = \{u \in V(\alpha): \text{ there is a path } P \text{ from } V_{\alpha} \text{ to } u \text{ in } \alpha \text{ such that } (V(P) \setminus \{u\}) \cap T = \emptyset\}$
		    \item $U_{{\sigma'}^T} = T$ with the ordering $O_T$ and $V_{{\sigma'}^T} = V_\alpha$.
		    \item $E({\sigma'}^T) = \{\{u,v\} \in E(\alpha): u,v \in V({\sigma'}^T), u \notin T \text{ or } v \notin T\}$
		\end{enumerate}
		\item The middle part $\tau$ of $\alpha$ is, informally, the part of $\alpha$ between $S$ and $T$ (including $S$ and $T$ and all edges which are entirely within $S$ or within $T$).
		More formally, let $U_\tau = S$ with the ordering $O_S$, let $V_\tau = T$ with the ordering $O_T$, and let $E(\tau) = E(\alpha) \setminus (E(\sigma) \cup E(\sigma'))$ be all of the edges of $E(\alpha)$ which do not appear in $E(\sigma)$ or $E(\sigma')$.
		Then $V(\tau)$ is all of the vertices incident to edges in $E(\tau)$ together with $S, T$.
	\end{enumerate}
\end{definition}

\begin{example}
    \cref{fig: basic_shape_comp} illustrates an example decomposition.
    \begin{enumerate}
        \item If we start with the shape $\al$ denoted as $\sig \circ \sig'^T$, observe that there is a unique minimum vertex separator, which consists of the middle vertex of degree $5$, i.e. the one that's not in either $U_{\sig \circ \sig'^T}$ or $V_{\sig \circ \sig'^T}$.
        Then, $\al$ is decomposed in to the left part $\sig$, a trivial middle part $\tau$ (not shown in this figure) which has $V(\tau) = \{u\}, U_{\tau} = V_{\tau} = (u), E(\tau) = \emptyset$, and the right part $\sig'^T$.
        \item If we start with the shape $\al$ denoted as $\sig \circ \tau \circ \sig'^T$, then the leftmost minimum vertex separator is the vertex of degree $4$ and the rightmost minimum vertex separator is the vertex of degree $5$. Then, $\al$ is decomposed into the left part $\sig$, the middle part $\tau$ and the right part $\sig'^T$, which are all shown in this figure.
    \end{enumerate}
\end{example}

\begin{remark}
Note that the decomposition into left, middle, and right parts depends on the ordering for the vertices in $S$ and $T$. As we will discuss later (see Section \ref{fullcoefficientmatrixsubsection}), we will use all possible orderings simultaneously and then scale things by an appropriate constant.
\end{remark}

\begin{figure}[!h]
    \centering
    \includegraphics[scale=0.45, trim={4.5cm 2cm 0 2cm},clip]{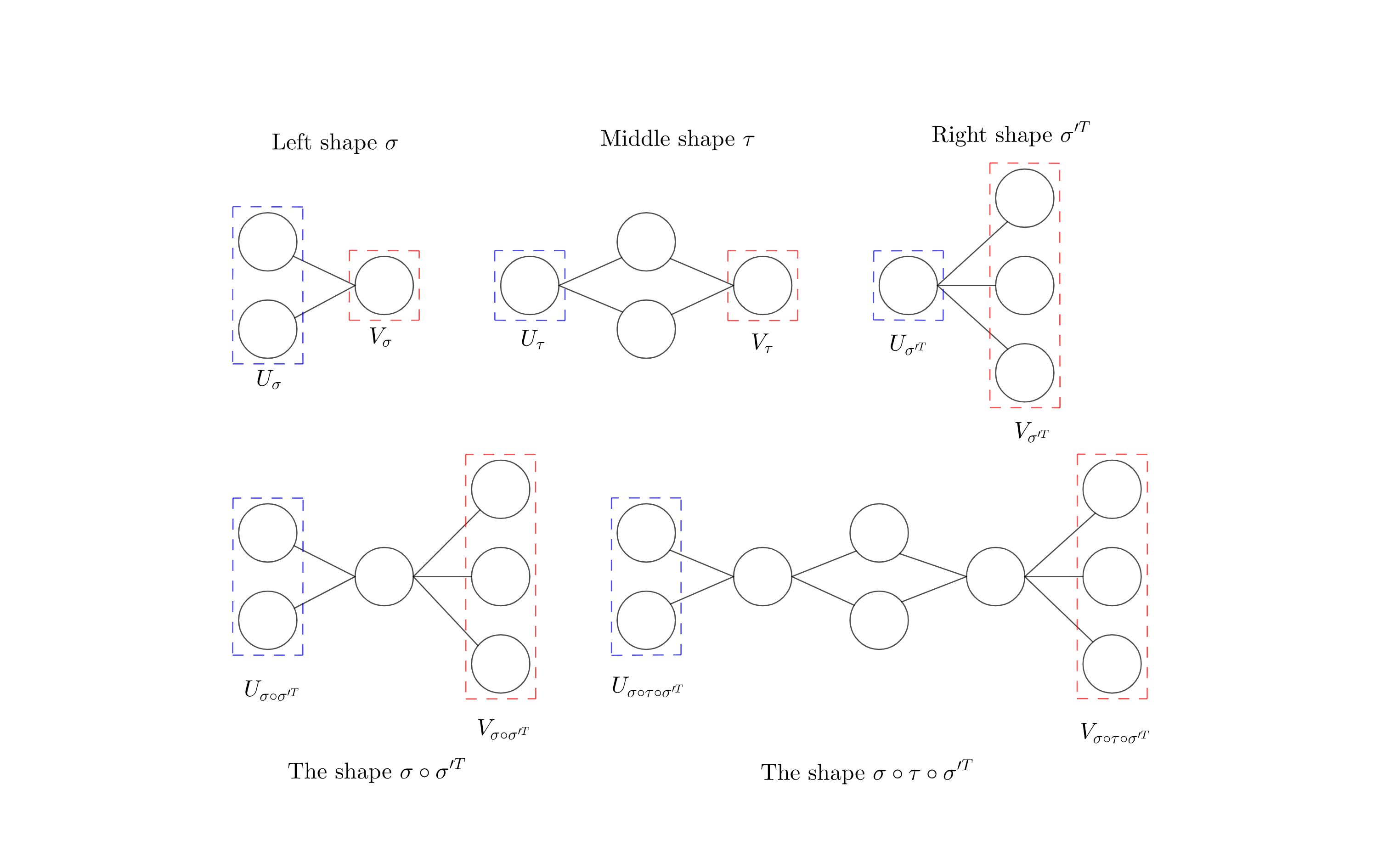}
    \caption{Illustration of shape composition and decomposition.}
    \label{fig: basic_shape_comp}
\end{figure}

\noindent Because of the minimality and leftmost/rightmost-ness of the vertex separators $S,T$ used to define $\sigma, \tau, \sigma'$, the shapes $\sigma, \tau, \sigma'$ have some special combinatorial structure, which we capture in the following proposition. We defer the proof until \cref{sec: technical_def_and_main_theorem} where we state a generalized version.

\begin{proposition}
	$\sigma$, $\tau$, and ${\sigma'}^{T}$ have the following properties:
	\begin{enumerate}
		\item $V_{\sigma} = S$ is the unique minimum vertex separator of $\sigma$.
		\item $S$ and $T$ are the leftmost and rightmost minimum vertex separators of $\tau$.
		\item $T = U_{{\sigma'}^T}$ is the unique minimum vertex separator of ${\sigma'}^T$.
	\end{enumerate}
\end{proposition}

\noindent Based on this, we define sets of shapes which can appear as left, middle, or right parts.

\begin{definition}[Left, Middle, and Right Parts] Let $\alpha$ be a shape.
	\begin{enumerate}
		\item We say that $\alpha$ is a left part if $V_{\alpha}$ is the unique minimum vertex separator of $\alpha$, all vertices of $\alpha$ are reachable from $U_{\alpha}$ without passing through $V_{\alpha}$, and $E(\alpha)$ has no edges which are entirely contained in $V_{\alpha}$.
		\item We say that $\alpha$ is a proper middle part if $U_{\alpha}$ is the leftmost minimum vertex separator of $\alpha$ and $V_{\alpha}$ is the rightmost minimum vertex separator of $\alpha$
		\item We say that $\alpha$ is a right part if $U_{\alpha}$ is the unique minimum vertex separator of $\alpha$, all vertices of $\alpha$ are reachable from $V_{\alpha}$ without passing through $U_{\alpha}$, and $E(\alpha)$ has no edges which are entirely contained in $U_{\alpha}$.
	\end{enumerate}
\end{definition}
\begin{remark}
	For technical reasons, later on we will need to consider improper middle parts $\tau$ where $U_{\tau}$ and $V_{\tau}$ are not the leftmost and rightmost minimum vertex separators of $\tau$, which is why we make this distinction here.
\end{remark}

\noindent The following proposition is also straightforward from the definitions.

\begin{proposition}
	A shape $\sigma$ is a left part if and only if $\sigma^{T}$ is a right part
\end{proposition}

\subsubsection{Products of Graph Matrices}
We now analyze what happens when we take the products of graph matrices. Roughly speaking, we will have that if $\alpha$ can be decomposed into left, middle, and right parts $\sigma$, $\tau$, and ${\sigma'}^{T}$ then $M_{\alpha} \approx M_{\sigma}M_{\tau}M_{{\sigma'}^T}$. However, this is only an approximation rather than an equality, and this will be the source of considerable technical difficulties.

We begin with a concatenation operation on ribbons.
\begin{definition}[Ribbon Concatenation]
	If $R_1$ and $R_2$ are two ribbons such that $V(R_1) \cap V(R_2) = B_{R_1} = A_{R_2}$ and either $R_1$ or $R_2$ contains no edges entirely within $B_{R_1} = A_{R_2}$ then we define $R_1 \circ R_2$ to be the ribbon formed by glueing together $R_1$ and $R_2$ along $B_{R_1} = A_{R_2}$.
	In other words,
	\begin{enumerate}
		\item $V(R_1 \circ R_2) = V(R_1) \cup V(R_2)$
		\item $E(R_1 \circ R_2) = E(R_1) \cup E(R_2)$
		\item $A_{R_1 \circ R_2} = A_{R_1}$ and $B_{R_1 \circ R_2} = B_{R_2}$.
	\end{enumerate}
\end{definition}

\noindent The following proposition is easy to check.

\begin{proposition}
	Whenever $R_1, R_2$ are ribbons such that $R_1 \circ R_2$ is defined, $M_{R_1}M_{R_2} = M_{R_1 \circ R_2}$
\end{proposition}

\noindent We have an analogous definition for concatenating shapes:

\begin{definition}[Shape Concatenation]
	If $\alpha_1$ and $\alpha_2$ are two shapes such that $V(\alpha_1) \cap V(\alpha_2) = V_{\alpha_1} = U_{\alpha_2}$ and either $\alpha_1$ or $\alpha_2$ contains no edges entirely within $V_{\alpha_1} = U_{\alpha_2}$ then we define $\alpha_1 \circ \alpha_2$ to be the shape formed by glueing together $\alpha_1$ and $\alpha_2$ along $V_{\alpha_1} = U_{\alpha_2}$. In other words,
	\begin{enumerate}
		\item $V(\alpha_1 \circ \alpha_2) = V(\alpha_1) \cup V(\alpha_2)$
		\item $E(\alpha_1 \circ \alpha_2) = E(\alpha_1) \cup E(\alpha_2)$
		\item $U_{\alpha_1 \circ \alpha_2} = U_{\alpha_1}$ and $V_{\alpha_1 \circ \alpha_2} = V_{\alpha_2}$.
	\end{enumerate}
\end{definition}

\begin{example}
    \cref{fig: basic_shape_comp} illustrates an example of shape composition. Observe how the shapes $\sig \circ \sig'^T$ and $\sig \circ \tau \circ \sig'^T$ are obtained from the shapes $\sig, \tau$ and $\sig'^T$.
\end{example}

\noindent The next proposition, again easy to check, shows that the shape concatenation operation respects the left/middle/right part decomposition.

\begin{proposition}
	If $\alpha$ can be decomposed into left, middle, and right parts $\sigma,\tau,{\sigma'}^{T}$ then $\alpha = \sigma \circ \tau \circ {\sigma'}^T$.
\end{proposition}

We now discuss why $M_\alpha = M_{\sigma \circ \tau \circ {\sigma'}^T} \approx M_\sigma M_\tau M_{{\sigma'}^T}$ is only an approximation rather than an equality. Consider the difference $M_{\sigma}M_\tau M_{{\sigma'}^T} - M_{\sigma \circ \tau \circ {\sigma'}^T}$. The graph matrix $M_{\sigma \circ \tau \circ {\sigma'}^T}$ decomposes (by definition) into a sum over injective maps $\phi \, : \, V(\sigma \circ \tau \circ {\sigma'}^T) \rightarrow [n]$. Also by expanding definitions, the product $M_{\sigma}M_\tau M_{{\sigma'}^T}$ expands into a sum over triples of injective maps $(\phi_1, \phi_2, \phi_3)$, where $\phi_1 \, : \, V(\sigma) \rightarrow [n], \phi_2 \, : \, V(\tau) \rightarrow [n], \phi_3 \, : \, V(\sigma') \rightarrow [n]$ where $\phi_1$ and $\phi_2$ agree on $V_{\sigma} = U_{\tau}$ and $\phi_2$ and $\phi_3$ agree on $V_{\tau} = U_{{\sigma'}^T}$.

If they are combined into one map $\phi: V(\sigma \cup \tau \cup {\sigma'} ) \to [n]$, the resulting $\phi$ may not be injective because $\phi_1(V(\sigma)), \phi_2(V(\tau)), \phi_3(V({\sigma'}^T))$ may have nontrivial intersection (beyond $\phi_1(V_\sigma)$ and $\phi_2(V_\tau)$).
We call the resulting terms \emph{intersection terms} and handling them properly is a major part of the technical analysis.
\begin{remark}
    Actually, the approximation $M_\alpha = M_{\sigma \circ \tau \circ {\sigma'}^T} \approx M_\sigma M_\tau M_{{\sigma'}^T}$ is also off by a multiplicative constant because there is also a subtle issue involving the automorphism groups of these shapes. For now, we ignore this issue. For details about this issue, see Lemma \ref{lm:morthsimplereexpression}.
\end{remark}


%

\subsection{Shape Coefficient Matrices}
The idea for our analysis is as follows. Given a matrix-valued function $\Lambda$ which is symmetric under permutations of $[n]$, we write $\Lambda = \sum_{\alpha}{\lambda_{\alpha}M_{\alpha}}$. We then break each shape $\alpha$ up into left, middle, and right parts $\sigma$, $\tau$, and ${\sigma'}^{T}$.

For this analysis, we use \emph{shape coefficient matrices} $H_{\tau}$ whose rows and columns are indexed by left shapes and whose entries depend on the coefficients $\lambda_{\alpha}$. We choose these matrices so that
\[
\Lambda = \sum_{\tau}{H_{\tau}(\sigma,\sigma')M_{\sigma \circ \tau \circ {\sigma'}^T}} \approx \sum_{\tau}{H_{\tau}(\sigma,\sigma')M_{\sigma}M_{\tau}M_{{\sigma'}^T}}
\]
To set this up, we separate the possible middle parts $\tau$ into groups based on the size of $U_{\tau}$ and whether or not they are trivial.
\begin{definition}
    We define $\mathcal{I}_{mid}$ to be the set of all possible $U_{\tau}$. Here $\mathcal{I}_{mid}$ is the set of tuples of unspecified vertices of the form $U = (u_1,\ldots,u_k)$ where $0 \leq k \leq d$.
\end{definition}
\begin{definition}
    We say that a proper middle shape $\tau$ is trivial if $E(\tau) = \emptyset$ and $|U_{\tau} \cap V_{\tau}| = |U_{\tau}| = |V_{\tau}|$ (i.e. $V_{\tau}$ is a permutation of $U_{\tau}$).
\end{definition}
For simplicity, the only proper trivial middle parts $\tau$ we consider are shapes $Id_U$ corresponding to identity matrices.
\begin{definition}
    Given a tuple of unspecified vertices $U = (u_1,\ldots,u_{|U|})$ We define $Id_U$ to be the shape where $V(Id_U) = U$, $U_{Id_U} = V_{Id_U} = U$, and $E(Id_U) = \emptyset$.
\end{definition}
We group all of the proper non-trivial middle parts $\tau$ into sets $\mathcal{M}_U$ based on the size of $U_{\tau}$.
\begin{definition}
    Given a tuple of unspecified vertices $U = (u_1,\ldots,u_{|U|})$, we define $\mathcal{M}_U$ to be the set of proper non-trivial middle parts $\tau$ such that $U_{\tau}$ and $V_{\tau}$ have the same size as $U$. Note that $U_{\tau}$ and $V_{\tau}$ may intersect each other arbitrarily.
\end{definition}
With these definitions, we can now define our shape coefficient matrices.
\begin{definition}
    Given $U \in \mathcal{I}_{mid}$, we define $\mathcal{L}_U$ to be the set of left shapes $\sigma$ such that $|V_{\sigma}| = |U|$.
\end{definition}
\begin{definition}
    For each $U \in \mathcal{I}_{mid}$, we define the shape coefficient matrix $H_{Id_U}$ to be the matrix indexed by left shapes $\sigma,\sigma' \in \mathcal{L}_{U}$ with entries $H_{Id_U}(\sigma,\sigma') = \frac{1}{|U|!}\lambda_{\sigma \circ {\sigma'}^T}$
\end{definition}
\begin{definition}
    For each $U \in \mathcal{I}_{mid}$, for each $\tau \in \mathcal{M}_{U}$, we define the shape coefficient matrix $H_{\tau}$ to be the matrix indexed by left shapes $\sigma,\sigma' \in \mathcal{L}_{U}$ with entries $H_{\tau}(\sigma,\sigma') = \frac{1}{(|U|!)^2}\lambda_{\sigma \circ \tau \circ {\sigma'}^T}$
\end{definition}
With these shape coefficient matrices, we have the following decomposition of $\Lambda = \sum_{\alpha}{\lambda_{\alpha}M_{\alpha}}$.
\begin{lemma}
$\Lambda = \sum_{U \in \mathcal{I}_{mid}}{\sum_{\sigma,\sigma' \in \mathcal{L}_U}{H_{Id_U}(\sigma,\sigma')M_{\sigma \circ {\sigma'}^T}}} + \sum_{U \in \mathcal{I}_{mid}}{\sum_{\tau \in \mathcal{M}_{U}}{\sum_{\sigma,\sigma' \in \mathcal{L}_U}{H_{\tau}(\sigma,\sigma')M_{\sigma \circ \tau \circ  {\sigma'}^T}}}}$
\end{lemma}
We defer the proof of this lemma to \cref{lm:determiningcoefficientmatrices}.

For technical reasons, we need to define one more operation to handle intersection terms. We call this operation \emph{the $-\gamma,\gamma$ operation.}

\begin{definition}
Given $U,V \in \mathcal{I}_{mid}$ where $|U| > |V|$, we define $\Gamma_{U,V}$ to be the set of left parts $\gamma$ such that $|U_{\gamma}| = |U|$ and $|V_{\gamma}| = |V|$.
\end{definition}
\begin{definition}
	Given $U,V \in \mathcal{I}_{mid}$ where $|U| > |V|$, a shape coefficient matrix $H_{Id_V}$, and a $\gamma \in \Gamma_{U,V}$, we define the shape coefficient matrix $H^{-\gamma,\gamma}_{Id_V}$ to be the matrix indexed by left shapes $\sigma,\sigma' \in \mathcal{L}_{U}$ with entries $H^{-\gamma,\gamma}_{Id_V}(\sigma,\sigma') = H(\sigma \circ \gamma, \sigma' \circ \gamma)$
\end{definition}

\subsection{Informal Theorem Statement}

We are now ready to state a simplified, qualitative version of our main theorem. For the full, quantitative version of our main theorem, see \cref{simplifiedmaintheorem}.

\begin{theorem}\label{informalmaintheoremstatement}
	There exist functions $f(\tau): \mathcal{M}_U \rightarrow \R$ and $f(\gamma): \Gamma_{U,V} \rightarrow \R$ depending on $n$ and other parameters such that if $\Lambda = \sum_{\alpha}{\lambda_{\alpha}M_{\alpha}}$ and the following conditions hold:
	\begin{enumerate}
		\item \psdmass For all $U \in \mathcal{I}_{mid}$,  $H_{Id_U} \succeq 0$
		\item \middleshapebounds For all $U \in \mathcal{I}_{mid}$ and all $\tau \in \mathcal{M}_{U}$,
		\[
		\left[ {\begin{array}{cc}
				H_{Id_U} & f(\tau)H_{\tau} \\
				f(\tau)H^T_{\tau} & H_{Id_U}
		\end{array}} \right] \succeq 0
		\]
		\item \intersectionbounds For all $U,V \in \mathcal{I}_{mid}$ such that $|U| > |V|$ and all $\gamma \in \Gamma_{U,V}$, $H^{-\gamma,\gamma}_{Id_{V_{\gamma}}} \preceq f(\gamma)H_{Id_{U_{\gamma}}}$
	\end{enumerate}
	then with probability at least $1-o(1)$ over $G \sim \{ \pm 1\}^{n \choose 2}$ it holds that $\Lambda(G) \succeq 0$.
\end{theorem}

\begin{remark}
	Condition $1$ of \cref{informalmaintheoremstatement} will follow from condition $2$ but we state it explicitly since it will correspond to the dominating terms of the approximate PSD decomposition. And in applications, it will be both easy to verify and will shed light on the structure of the coefficients which in turn will be useful for verifying conditions $2$ and $3$.
\end{remark}

\begin{remark}
As we will demonstrate in the remainder of this paper, our machinery works well when the coefficients $\lambda_{\alpha}$ has
some decay for each vertex or edge in the shape. In many settings, this can be done quite easily by adding noise to the distribution, such as resampling part of the input, or by lowering the parameters slightly, such as $m \le n^{k/4 - \eps}$ instead of $m \le n^{k/4}$.
\end{remark}

\subsubsection{Choice of functions $f(\tau)$ and $f(\gam)$}

In a rough sense, $f(\tau)$ measures the blow-up in the norm by using $M_{\tau}$ instead of $M_{Id_U}$ in the corresponding term of the Fourier decomposition. So we choose $f(\tau)$ to be $\norm{M_{\tau}}$, upto lower order terms. Our second condition verifies that the coefficients that arise because of this $\tau$ (which are encoded in $H_{\tau}$) are sufficiently small to overpower this norm blowup.

The fact that $\norm{M_{\tau}}$ is equal to $\tilde{O}(n^{\frac{|V(\tau)|-|U_{\tau}|}{2}})$ has been shown in previous works \cite{BHKKMP16, AMP20}. So, we choose $f(\tau)$ to be $\tilde{O}(n^{\frac{|V(\tau)|-|U_{\tau}|}{2}})$ where the problem instance is on $G_{n, 1/2}$. For problems with Gaussian or other inputs, similar forms of $f(\tau)$ can be used, which have been shown formally in the work of \cite{AMP20}. When we state the main theorem in general, we use a single $f(\tau)$ that incorporates all of these settings.

$f(\gam)$ is a bit trickier to describe. In our analysis, we roughly collect intersection terms from the approximate PSD decomposition and charge them to shapes of the form $\sig \circ \gam \circ \gam^T \circ \sig'^T$. Using the same idea as the previous step, we charge these to shapes with trivial middle shapes. Roughly, $f(\gam)$ for a fixed $\gam$ upper bounds the blowup from the norms of the original shape as compared to the new intersection shape. The third condition argues that the the original coefficients are sufficiently small to compensate for these blowups.

For problems on $G_{n, 1/2}$, we set $f(\gam) = \tilde{O}(n^{|V(\gamma) \setminus U_{\gamma}|})$. For problems with Gaussian inputs, we choose essentially the same function, but they fall under the umbrella of generalized graph matrices, where $V(\gam)$ and $U_{\gam}$ are defined accordingly. Indeed, in our main theorem, we encompass both these settings with a single choice of $f(\gam)$.

\subsection{An informal application to planted clique}
Before we move on to further definitions needed for a more complete statement of the main theorem, we present an informal example.

\begin{example}
	When the pseudo-calibration method is applied to prove an SoS lower bound for the planted clique problem in $n$ node graphs with clique size $k$, as in \cite{BHKKMP16}, the matrix-valued function which results is $\Lambda =  \sum_{\alpha \, : \, |V(\alpha)| \leq t}{\left(\frac{k}{n}\right)^{|V(\alpha)|}M_{\alpha}}$ where $t \approx \log(n)$.
	One may then compute that the matrices $H_{Id_U}$ and $H_{\tau}$ are as follows (at least so long as $|V(\sigma)|,|V(\tau)|,|V(\sigma')| \ll t$; we ignore this detail for now).
	For all $r \in [0,\frac{d}{2}]$,
	\begin{enumerate}
		\item For $U$ with $|U|  = r$, $H_{Id_U}(\sigma,\sigma') = \left(\frac{k}{n}\right)^{|V(\sigma)| + |V(\sigma')| - r}$
		\item For all proper, non-trivial middle shapes $\tau$ such that $|U_{\tau}| = |V_{\tau}| = r$,
		\[
		H_{\tau}(\sigma,\sigma') = \left(\frac{k}{n}\right)^{|V(\sigma)| + |V(\sigma')| + |V(\tau)|- 2r}
		\]
	\end{enumerate}
	Defining $v_r$ to be the vector such that $v_r(\sigma) = \left(\frac{k}{n}\right)^{|V(\sigma)| - \frac{r}{2}}$, we have that
	\begin{enumerate}
		\item For $U$ with $|U|  = r$, $H_{Id_U} = {v_{|U|}}{v^T_{|U|}}$
		\item For all proper, non-trivial middle shapes $\tau$ such that $|U_{\tau}| = |V_{\tau}| = r$, $H_{\tau} = \left(\frac{k}{n}\right)^{|V(\tau)|- r}{v_r}{v^T_r}$
		\item For all left parts $\gamma$, $H^{-\gamma,\gamma}_{Id_{V_{\gamma}}} = \left(\frac{k}{n}\right)^{2|V(\gamma)| - |U_{\gamma}| - |V_{\gamma}|}v_{|U_{\gamma}|}v^{T}_{|U_{\gamma}|}$
	\end{enumerate}
	It turns out in this setting that we can take $f(\tau)$ to be $\tilde{O}(n^{\frac{|V(\tau)|-|U_{\tau}|}{2}})$ and $f(\gamma)$ to be $\tilde{O}(n^{|V(\gamma) \setminus U_{\gamma}|})$. Thus, as long as $k \ll \sqrt{n}$,
	\begin{enumerate}
		\item For any $U$ and all $\tau$ such that $V_{\tau} \neq U_{\tau}$ with $|U_{\tau}| = |V_{\tau}| = |U|$, $f(\tau)H_{\tau} \preceq H_{Id_U}$.
		\item For all non-trivial left parts $\gamma$, $H^{-\gamma,\gamma}_{Id_{V_{\gamma}}} \preceq f(\gamma)H_{Id_{U_{\gamma}}}$
	\end{enumerate}
\end{example}
\begin{remark}\label{rmk: planted_clique_failure}
This does not quite satisfy the conditions of Theorem \ref{informalmaintheoremstatement} because there are $\tau$ such that $V_{\tau} = U_{\tau}$ but which are non-trivial because $E(\tau) \neq \emptyset$. For these $\tau$, condition 2 of Theorem \ref{informalmaintheoremstatement} fails. In order to prove their SoS lower bounds for planted clique, \cite{BHKKMP16} handle this issue by grouping together all of the $\tau$ where $V_{\tau} = U_{\tau}$ into the indicator function for whether $V_{\tau} = U_{\tau}$ is a clique.

Since this issue is specific to planted clique, we don't try to incorporate it into the machinery to avoid losing generality.

\end{remark}

For the sake of exposition, a detailed analysis with figures of all the shapes and all the coefficient matrices that appear for the degree-$4$ SoS lower bound for planted clique is given in \cref{sec: deg_4_planted_clique}. Note that we present these details purely for the interested reader and they are not needed to apply the machinery.

\subsection{Generalizing the machinery}
In this section, we restricted ourselves to the case when the input is from $\{-1,1\}^{\binom{n}{2}}$ for simplicity. However, for our results we will need to handle more general types of inputs. We now briefly describe which kinds of inputs we will need to handle and how we handle them.
\begin{enumerate}
    \item In general, the entries of the input may be labeled by more than $2$ indices. For example, for tensor PCA on order $3$ tensors, the entries of the input are indexed by $3$ indices. To handle this, we will have shapes which have hyperedges rather than edges.
    \item In general, the entries of the input will come from a distribution $\Omega$ rather than being $\pm{1}$. To handle this, we will take an orthonormal basis $\{h_k\}$ for $\Omega$. We will then give each edge/hyperedge a label $l$ to specify which polynomial $h_l$ should be applied to that entry of the input.
    \item In general, there may be $t$ different types of indices rather than just one type of index. In this case, the symmetry group will be $S_{n_1} \times \ldots \times S_{n_t}$ rather than $S_n$. To handle this, we will have shapes with different types of vertices.
\end{enumerate}
We formally make these generalizations in \cref{sec: technical_def_and_main_theorem}.

\subsection{Further definitions needed for our applications}

We will describe some more notations and definitions that will be useful to us to describe the qualitative bounds for our applications. For each of our applications, we will describe the corresponding modifications needed to the definitions already in place and present new definitions where necessary.

\subsubsection{Planted slightly denser subgraph}

Since the input is a graph $G \in \{-1, 1\}^{\binom{[n]}{2}}$, most of what we introduced already apply to this setting. To describe the moment matrix, we need to define the truncation parameter.

\begin{definition}[Truncation parameters]
For integers $D_{sos}, D_V \ge 0$, say that a shape $\al$ satisfies the truncation parameters $D_{sos}, D_V$ if
\begin{itemize}
    \item The degrees of the monomials that $U_{\al}$ and $V_{\al}$ correspond to, are at most $\frac{D_{sos}}{2}$
    \item The left part $\sig$, the middle part $\tau$ and the right part $\sig'$ of $\al$ satisfy $|V(\sig)|, |V(\tau)|, |V(\sig')| \le D_V$
\end{itemize}
\end{definition}

\subsubsection{Tensor PCA}

We consider the input to be a tensor $A \in \RR^{[n]^k}$. The input entries are now sampled from the distribution $\GN(0, 1)$ instead of $\{-1, 1\}$. So, we will work with the Hermite basis of polynomials.
Let the standard unnormalized Hermite polynomials be denoted as $h_0(x) = 1, h_1(x) = x, h_2(x) = x^2 - 1, \ldots$. Then, we work with the basis $h_a(A) \defeq \prod_{e \in [n]^k} h_e(A_e)$ over $a \in \NN^{[n]^k}$. Accordingly, we will modify the graphs that represent ribbons (and by extension, shapes), to have labeled hyperedges of arity $k$. So, an hyperedge $e$ with a label $t$ will correspond to the hermite polynomial $h_t(A_e)$.

\begin{definition}[Hyperedges]
Instead of standard edges, we will have labeled hyperedges of arity $k$ in the underlying graphs for our ribbons as well as shapes. The label for an hyperedge $e$, denoted $l_e$, is an element of $\NN$ which will correspond to the Hermite polynomial being evaluated on that entry.
\end{definition}

Note that our hyperedges are ordered since the tensor $A$ is not necessarily symmetric.

For variables $x_1, \ldots, x_n$, the rows and columns of our moment matrix will now correspond to monomials of the form $\prod_{i \le n} x_i^{p_i}$ for $p_i \ge 0$. To capture this, we use the notion of index shape pieces and index shapes. Informally, we split the above monomial product into groups based on their powers and each such group will form an index shape piece.

\begin{definition}[Index shape piece]
    An index shape piece $U_i= ((U_{i, 1}, \ldots, U_{i, t}), p_i)$ is a tuple of indices $(U_{i, 1}, \ldots, U_{i, t})$ along with a power $p_i \in \NN$. Let $V(U_i)$ be the set $\{U_{i, 1}, \ldots, U_{i, t}\}$ of vertices of this index shape piece. When clear from context, we use $U_i$ instead of $V(U_i)$.
\end{definition}

If we realize $U_{i, 1}, \ldots, U_{i, t}$ to be indices $a_1, \ldots, a_t \in [n]$, then, this realization of this index shape piece corresponds to the monomial $\prod_{j \le t} x_{a_j}^{p_i}$.

\begin{definition}[Index shape]
An index shape $U$ is a set of index shape pieces $U_i$ that have different powers. Let $V(U)$ be the set of vertices $\cup_i V(U_i)$. When clear from context, we use $U$ instead of $V(U)$.
\end{definition}

Observe that each realization of an index shape corresponds to a row or column of the moment matrix.

\begin{definition}
For two index shapes $U, V$, we write $U \equiv V$ if for all powers $p$, the index shape pieces of power $p$ in $U$ and $V$ have the same length.
\end{definition}

\begin{definition}
Define $\calI_{mid}$ to be the set of all index shapes $U$ that contain only index shape pieces of power $1$.
\end{definition}

In the definition of shapes, the distinguished set of vertices should now be replaced by index shapes.

\begin{definition}[Shapes]
Shapes are tuples $\al = (H_{\al}, U_{\al}, V_{\al})$ where $H_{\al}$ is a graph with hyperedges of arity $k$ and $U_{\al}, V_{\al}$ are index shapes such that $U_{\al}, V_{\al} \subseteq V(H_{\al})$.
\end{definition}

\begin{definition}[Proper shape]
A shape $\al$ is proper if it has no isolated vertices outside $U_{\al} \cup V_{\al}$, no multi-edges and all the edges have a nonzero label.
\end{definition}

To define the notion of vertex separators, we modify the notion of paths for hyperedges.

\begin{definition}[Path]
A path is a sequence of vertices $u_1, \ldots, u_t$ such that $u_i, u_{i + 1}$ are in the same hyperedge, for all $i \le t - 1$.
\end{definition}

The notions of vertex separator and decomposition into left, middle and right parts are identically defined with the above notion of hyperedges and paths. In \cref{sec: technical_def_and_main_theorem}, we will show that they are well defined.

In the definition of trivial shape $\tau$, we now require $U_{\tau} \equiv V_{\tau}$. For $U \in \calI_{mid}$, $\calM_U$ will be the set of proper non-trivial middle parts $\tau$ with $U_{\tau} \equiv V_{\tau} \equiv U$ and $\calL_U$ will be the set of left parts $\sig$ such that $V_{\sig} \equiv U$. Similarly, for $U, V \in \calI_{mid}$, $\calL_{U, V}$ will be the set of left parts $\gam$ such that $U_{\gam} \equiv U$ and $V_{\gam} \equiv V$.

In order to define the moment matrix, we need to truncate our shapes based on the number of vertices and the labels on our hyperedges. So, we make the following definition.

\begin{definition}[Truncation parameters]
For integers $D_{sos}, D_V, D_E \ge 0$, say that a shape $\al$ satisfies the truncation parameters $D_{sos}, D_V, D_E$ if
\begin{itemize}
    \item The degrees of the monomials that $U_{\al}$ and $V_{\al}$ correspond to, are at most $\frac{D_{sos}}{2}$
    \item The left part $\sig$, the middle part $\tau$ and the right part $\sig'^T$ of $\al$ satisfy $|V(\sig)|, |V(\tau)|, |V(\sig'^T)| \le D_V$
    \item For each $e \in E(\al)$, $l_e \le D_E$.
\end{itemize}
\end{definition}

\subsubsection{Sparse PCA}

We consider the $m$ vectors $v_1, \ldots, v_m \in \RR^d$ to be the input. Similar to Tensor PCA, we will work with the Hermite basis of polynomials since the entries are sampled from the distribution $\GN(0, 1)$.
In particular, if we denote the unnormalized Hermite polynomials by $h_0(x) = 1, h_1(x) = x, h_2(x) = x^2 - 1, \ldots$, then, we work with the basis $h_a(v) \defeq \prod_{i \in [m], j \in [n]} h_{a_{i, j}}(v_{i, j})$ over $a \in \NN^{m \times n}$. To capture this basis, we will modify the graphs that represent ribbons (and by extension, shapes), to be bipartite graphs with two types of vertices, and have labeled edges that go across vertices of different types. So, an edge $(i, j)$ with label $t$ between a vertex $i$ of type $1$ and a vertex $j$ of type $2$ will correspond to $h_t(v_{i, j})$.

\begin{definition}[Vertices]
We will have two types of vertices, the vertices corresponding to the $m$ input vectors that we call type $1$ vertices and the vertices corresponding to ambient dimension of the space that we call type $2$ vertices.
\end{definition}

\begin{definition}[Edges]
Edges will go across vertices of different types, thereby forming a bipartite graph. An edge between a type $1$ vertex $i$ and a type 2 vertex $j$ corresonds to the input entry $v_{i, j}$. Each edge will have a label in $\NN$ corresponding to the Hermite polynomial evaluated on that entry.
\end{definition}

We will have variables $x_1, \ldots, x_n$ in our SoS program, so we will work with index shape pieces and index shapes as in Tensor PCA, since the rows and columns of our moment matrix will now correspond to monomials of the form $\prod_{i \le n} x_i^{p_i}$ for $p_i \ge 0$. But since in our decompositions into left, right and middle parts, we will have type $2$ vertices as well in the vertex separators, we will define a generalized notion of index shape pieces and index shapes.

\begin{definition}[Index shape piece]
    An index shape piece $U_i= ((U_{i, 1}, \ldots, U_{i, t}), t_i, p_i)$ is a tuple of indices $(U_{i, 1}, \ldots, U_{i, t})$ along a type $t_i \in \{1, 2\}$ with a power $p_i \in \NN$. Let $V(U_i)$ be the set $\{U_{i, 1}, \ldots, U_{i, t}\}$ of vertices of this index shape piece. When clear from context, we use $U_i$ instead of $V(U_i)$.
\end{definition}

For an index shape piece $((U_{i, 1}, \ldots, U_{i, t}), t_i, p_i)$ with type $t_i = 2$, if we realize $U_{i_1}, \ldots, U_{i_t}$ to be indices $a_1, \ldots, a_t \in [n]$, then, this index shape pieces correspond this to the monomial $\prod_{j \le n} x_{a_j}^{p_i}$.

\begin{definition}[Index shape]
An index shape $U$ is a set of index shape pieces $U_i$ that have either have different types or different powers. Let $V(U)$ be the set of vertices $\cup_i V(U_i)$. When clear from context, we use $U$ instead of $V(U)$.
\end{definition}

Observe that each realization of an index shape corresponds to a row or column of the moment matrix. For our moment matrix, the only nonzero rows correspond to index shapes that have only index shape pieces of type $2$, since the only SoS variables are $x_1 \ldots, x_n$, but in order to do our analysis, we need to work with the generalized notion of index shapes that allow index shape pieces of both types.

\begin{definition}
For two index shapes $U, V$, we write $U \equiv V$ if for all types $t$ and all powers $p$, the index shape pieces of type $t$ and power $p$ in $U$ and $V$ have the same length.
\end{definition}

\begin{definition}
Define $\calI_{mid}$ to be the set of all index shapes $U$ that contain only index shape pieces of power $1$.
\end{definition}

Since we are working with standard graphs, the notion of path and vertex separator need no modifications, but we will now use the minimum weight vertex separator instead of the minimum vertex separator where we define the weight as follows.

\begin{definition}[Weight of an index shape]
Suppose we have an index shape $U = \{U_1, U_2\} \in \calI_{mid}$ where $U_1 = ((U_{1, 1}, \ldots, U_{1, |U_1|}), 1, 1)$ is an index shape piece of type $1$ and $U_2 = ((U_{2, 1}, \ldots, U_{2, |U_2|}), 2, 1)$ is an index shape piece of type $2$. Then, define the weight of this index shape to be $w(U) = \sqrt{m}^{|U_1|}\sqrt{n}^{|U_2|}$.
\end{definition}

We now give the modified definition of shapes.

\begin{definition}[Shapes]
Shapes are tuples $\al = (H_{\al}, U_{\al}, V_{\al})$ where $H_{\al}$ is a graph with two types of vertices, has labeled edges only across vertices of different types and $U_{\al}, V_{\al}$ are index shapes such that $U_{\al}, V_{\al} \subseteq V(H_{\al})$.
\end{definition}

\begin{definition}[Proper shape]
A shape $\al$ is proper if it has no isolated vertices outside $U_{\al} \cup V_{\al}$, no multi-edges and all the edges have a nonzero label.
\end{definition}

In \cref{sec: technical_def_and_main_theorem}, we will show that with this new definition of weight and shapes, any shape $\al$ has a unique decomposition into $\sig \circ \tau \circ \sig'^T$ where $\sig, \tau, \sig'^T$ are left, middle and right parts respectively. Here, $\tau$ may possibly be improper.

In the definition of trivial shape $\tau$, we now require $U_{\tau} \equiv V_{\tau}$. For $U \in \calI_{mid}$, $\calM_U$ will be the set of proper non-trivial middle parts $\tau$ with $U_{\tau} \equiv V_{\tau} \equiv U$ and $\calL_U$ will be the set of left parts $\sig$ such that $V_{\sig} \equiv U$. Similarly, for $U, V \in \calI_{mid}$, $\calL_{U, V}$ will be the set of left parts $\gam$ such that $U_{\gam} \equiv U$ and $V_{\gam} \equiv V$.

Finally, in order to define the moment matrix, we need to truncate our shapes based on the number of vertices and the labels on our edges. So, we make the following definition.

\begin{definition}[Truncation parameters]
For integers $D_{sos}, D_V, D_E \ge 0$, say that a shape $\al$ satisfies the truncation parameters $D_{sos}, D_V, D_E$ if
\begin{itemize}
    \item The degrees of the monomials that $U_{\al}$ and $V_{\al}$ correspond to, are at most $\frac{D_{sos}}{2}$
    \item The left part $\sig$, the middle part $\tau$ and the right part $\sig'^T$ of $\al$ satisfy $|V(\sig)|, |V(\tau)|, |V(\sig'^T)| \le D_V$
    \item For each $e \in E(\al)$, $l_e \le D_E$.
\end{itemize}
\end{definition}

\subsubsection{Relaxing the third condition}\label{sec: hgamma_qual}

In \cref{informalmaintheoremstatement}, the third qualitative condition we'd like to show is as follows:
For all $U,V \in \mathcal{I}_{mid}$ such that $|U| > |V|$ and all $\gamma \in \Gamma_{U,V}$, $H^{-\gamma,\gamma}_{Id_{V_{\gamma}}} \preceq f(\gamma)H_{Id_{U_{\gamma}}}$.
For technical reasons, we won't be able to show this directly. To handle this, we instead work with a slight modification of $H_{Id_{U_{\gamma}}}$, a matrix $H_{\gamma}'$ that's very close to $H_{Id_{U_{\gamma}}}$. So, what we will end up showing is:
For all $U,V \in \mathcal{I}_{mid}$ such that $|U| > |V|$ and all $\gamma \in \Gamma_{U,V}$, $H^{-\gamma,\gamma}_{Id_{V_{\gamma}}} \preceq f(\gamma)H'_{\gamma}$.

Let $D_V$ be the truncation parameter. A canonical choice for $H'_{\gamma}$ is to take
\begin{enumerate}
	\item $H'_{\gamma}(\sigma,\sigma') = H_{Id_U}(\sigma, \sigma')$ whenever $|V(\sigma \circ \gamma)| \leq D_V$ and $|V(\sigma' \circ \gamma)| \leq D_V$.
	\item $H'_{\gamma}(\sigma,\sigma') = 0$ whenever $|V(\sigma \circ \gamma)| > D_V$ or $|V(\sigma' \circ \gamma)| > D_V$.
\end{enumerate}

With this choice, $H_{\gamma}'$ is the same as $H_{Id_{U_{\gamma}}}$ upto truncation error. We will formally bound the errors in the quantitative sections after we introduce the full machinery.

%% file: planted_ds_qual.tex
\subsection{Pseudo-calibration}

We will pseudo-calibrate with respect the following pair of random and planted distributions which we denote $\nu$ and $\mu$ respectively.

\PLDSdistributions*

We assume that the input is given as $G_{i, j}$ for $i, j \in \binom{[n]}{2}$ where $G_{i, j}$ is $1$ if the edge $(i, j)$ is present in the graph and $-1$ otherwise. We work with the Fourier basis $\chi_E$ defined as $\chi_E(G) \defeq \prod_{(i, j) \in E} G_{i, j}$. For a subset $I \subseteq [n]$, define $x_I \defeq \prod_{i \in I} x_I$.

\begin{lemma}
Let $I \subseteq [n], E \subseteq \binom{[n]}{2}$. Then,
\[\EE_{\mu}[x_I \chi_E(G)] = \left(\frac{k}{n}\right)^{|I \cup V(E)|} (2p - 1)^{|E|}\]
\end{lemma}

\begin{proof}
When we sample $(G, S)$ from $\mu$, we condition on whether $I \cup V(E) \subseteq S$.
\begin{align*}
\EE_{(G, S)\sim \mu}[x_I \chi_E(G)] &= Pr_{(G, S) \sim \mu}[I \cup V(E) \subseteq S]\EE_{(G, S) \sim \mu}[x_I\chi_E(G)|I \cup V(E) \subseteq S]\\
&\qquad + Pr_{(G, S) \sim \mu}[I \cup V(E) \not\subseteq S]\EE_{(G, S) \sim \mu}[x_I\chi_E(G)|I \cup V(E) \not\subseteq S]
\end{align*}
We claim that the second term is $0$. In particular, $\EE_{(G, S) \sim \mu}[x_I\chi_E(G)|I \cup V(E) \not\subseteq S] = 0$ because when $I \cup V(E) \not\subseteq S$, either $S$ doesn't contain a vertex in $I$ or an edge $(i, j) \in E$ is outside $S$. If $S$ doesn't contain a vertex in $I$, then $x_I = 0$ and hence, the quantity is $0$. And if an edge $(i, j) \in E$ is outside $S$, since this edge is sampled with probability $\frac{1}{2}$, by taking expectations, the quantity $\EE_{(G, S) \sim \mu}[x_I\chi_E(G)|I \cup V(E) \not\subseteq S]$ is $0$.

Finally, note that $Pr_{(G, S) \sim \mu}[I \cup V(E) \subseteq S] = \left(\frac{k}{n}\right)^{|I \cup V(E)|}$ and
\[\EE_{(G, S) \sim \mu}[x_I\chi_E(G)|I \cup V(E) \subseteq S] = \EE_{(G, S) \sim \mu}[\chi_E(G)|V(E) \subseteq S] = (2p - 1)^{|E|}\]
The last equality follows because for each edge $e \in E$, since $e$ is present independently with probability $p$, the expected value of $\chi_e$ is $1\cdot p + (-1) \cdot (1 - p) = 2p - 1$.
\end{proof}


Define the degree of SoS to be $D_{sos} = n^{C_{sos}\eps}$ for some constant $C_{sos} > 0$ that we choose later. And define the truncation parameter to be $D_V = n^{C_V\eps}$ for some constant $C_V > 0$.

\begin{remk}[Choice of parameters]\label{rmk: choice_of_params1}
	We first set $\eps$ to be a sufficiently small constant. Based on this choice, we will set $C_V$ to be a sufficiently small constant to satisfy all the inequalities we use in our proof. Based on these choices, we can choose $C_{sos}$ to be sufficiently small to satisfy the inequalities we use.
\end{remk}

We will now describe the decomposition of the moment matrix $\Lda$.

\begin{definition}\label{def: plds_coeffs}
	If a shape $\alpha$ satisfies the following properties:
	\begin{itemize}
		\item $\alpha$ is proper,
		\item $\alpha$ satisfies the truncation parameter $D_{sos}, D_V$.
	\end{itemize}
	then define \[\lambda_{\alpha} = \left(\frac{k}{n}\right)^{|V(\al)|}  (2p - 1)^{|E(\al)|}\]
\end{definition}

\begin{corollary}
	$\Lambda = \sum \lda_{\al}M_{\al}$.
\end{corollary}

\subsection{Qualitative machinery bounds}

In this section, we will prove the PSD mass condition and the qualitative versions of the middle shape and intersection term bounds.

\begin{restatable}[PSD mass]{lemma}{PLDSone}\label{lem: plds_cond1}
	For all $U \in \calI_{mid}$, $H_{Id_U} \succeq 0$
\end{restatable}

While this is easy to prove directly, we would like to introduce appropriate notation so that this lemma as well as the qualitative bounds to follow are immediate. 
Therefore, we state the qualitative conditions next and then prove them all together. 
Now, we define the following quantities which capture the contribution of the vertices within $\tau, \gam$ to the Fourier coefficients.

\begin{restatable}{definition}{PLDSstau}\label{def: plds_stau}
	For $U \in \calI_{mid}$ and $\tau \in \calM_U$, define
	$S(\tau) = \left(\frac{k}{n}\right)^{|V(\tau)| - |U_{\tau}|}(2p - 1)^{|E(\tau)|}$.
	And for all $U, V \in \calI_{mid}$ where $w(U) > w(V)$ and $\gam \in \Gam_{U, V}$, define
	$S(\gam) = \left(\frac{k}{n}\right)^{|V(\gam)| - \frac{|U_{\gam}| + |V_{\gam}|}{2}}(2p - 1)^{|E(\gam)|}$.
\end{restatable}

We can now state our qualitative bounds, which we prove shortly.

\begin{restatable}[Qualitative middle shape bounds]{lemma}{PLDStwosimplified}\label{lem: plds_cond2_simplified}
	For all $U \in \calI_{mid}$ and $\tau \in \calM_U$,
	\[
	\begin{bmatrix}
		\frac{S(\tau)}{|Aut(U)|}H_{Id_U} & H_{\tau}\\
		H_{\tau}^T & \frac{S(\tau)}{|Aut(U)|}H_{Id_U}
	\end{bmatrix}
	\succeq 0
	\]
\end{restatable}

In the following qualitative intersection term bounds, we use the canonical definition of $H_{\gam}'$ from \cref{sec: hgamma_qual}.

\begin{restatable}[Qualitative intersection term bounds]{lemma}{PLDSthreesimplified}\label{lem: plds_cond3_simplified}
	For all $U, V \in \calI_{mid}$ where $w(U) > w(V)$ and all $\gam \in \Gam_{U, V}$,
	\[\frac{|Aut(V)|}{|Aut(U)|}\cdot\frac{1}{S(\gam)^2}H_{Id_V}^{-\gam, \gam} = H_{\gam}'\]
\end{restatable}

In order to prove these bounds, we define the following quantity to capture the contribution of the vertices within $\sig$ to the Fourier coefficients.

\begin{definition}
	For a shape $\sig\in \calL$, define
	$T(\sig) = \left(\frac{k}{n}\right)^{|V(\sig)| - \frac{|V_{\sig}|}{2}}(2p - 1)^{|E(\sig)|}$.
	For $U \in \calI_{mid}$, define $v_U$ to be the vector indexed by $\sig \in \calL$ such that $v_U(\sig) = T(\sig)$ if $\sig \in \calL_U$ and $0$ otherwise.
\end{definition}

The following propositions are immediate from \cref{def: plds_coeffs}.

\begin{propn}
	For all $U\in \calI_{mid}, \rho \in \calP_U$, $H_{Id_U} = \frac{1}{|Aut(U)|}v_Uv_U^T$.
\end{propn}





\begin{propn}
	For any $U \in \calI_{mid}$ and $\tau \in \calM_U$, $H_{\tau} = \frac{1}{|Aut(U)|^2} S(\tau) v_Uv_U^T$.
\end{propn}


The first proposition implies that for all $U \in \calI_{mid}$, $H_{Id_U} \succeq 0$, which is the PSD mass condition \cref{lem: plds_cond1}.
\cref{lem: plds_cond2_simplified} and \cref{lem: plds_cond3_simplified} also follow easily.

\begin{proof}[Proof of \cref{lem: plds_cond2_simplified}]
\begin{align*}
    \begin{bmatrix}
		\frac{S(\tau)}{|Aut(U)|}H_{Id_U} & H_{\tau}\\
		H_{\tau}^T & \frac{S(\tau)}{|Aut(U)|}H_{Id_U}
	\end{bmatrix} &= \begin{bmatrix}
		\frac{S(\tau)}{|Aut(U)|}v_Uv_U^T & \frac{S(\tau)}{|Aut(U)|^2}v_Uv_U^T\\
		\frac{S(\tau)}{|Aut(U)|^2}v_Uv_U^T & \frac{S(\tau)}{|Aut(U)|}v_Uv_U^T
	\end{bmatrix} \succeq 0
\end{align*}
\end{proof}




\begin{proof}[Proof of \cref{lem: plds_cond3_simplified}]
    Fix $\sig, \sig' \in \calL_{U}$ such that $|V(\sig \circ \gam)|, |V(\sig' \circ \gam)| \le D_V$. Note that $|V(\sig)| - \frac{|V_{\sig}|}{2} + |V(\sig')| - \frac{|V_{\sig'}|}{2} + 2(|V(\gam)| - \frac{|U_{\gam}| + |V_{\gam}|}{2}) = |V(\sig \circ \gam \circ \gam^T \circ \sig'^T)|$. Using \cref{def: plds_coeffs}, we can easily verify that $\lda_{\sig \circ \gam \circ \gam^T \circ \sig'^T} = T(\sigma)T(\sigma') S(\gam)^2$. Therefore, $H_{Id_V}^{-\gam, \gam}(\sig, \sig') = \frac{|Aut(U)|}{|Aut(V)|} S(\gam)^2 H_{Id_U}(\sig, \sig')$. Since $H'_{\gam}(\sig, \sig') = H_{Id_U}(\sig, \sig')$ whenever $|V(\sig \circ \gam)|, |V(\sig' \circ \gam)| \le D_V$, this completes the proof.
\end{proof}



%% file: tensor_pca_qual.tex
\subsection{Pseudo-calibration}

\begin{definition}[Slack parameter]
	Define the slack parameter to be $\Delta = n^{-C_{\Del}\eps}$ for a constant $C_{\Del} > 0$.
\end{definition}

We will pseudo-calibrate with respect the following pair of random and planted distributions which we denote $\nu$ and $\mu$ respectively.

\TPCAdistributions*

Let the Hermite polynomials be $h_0(x) = 1, h_1(x) = x, h_2(x) = x^2 - 1, \ldots$. For $a \in \NN^{[n]^k}$ and variables $A_e$ for $e \in [n]^k$, define $h_a(A) \defeq \prod_{e \in [n]^k} h_e(A_e)$. We will work with this Hermite basis.

\begin{lemma}
	Let $I \in \NN^n, a \in \NN^{[n]^k}$. For $i \in [n]$, let $d_i = \sum_{i \in e \in [n]^k} a_e$. Let $c$ be the number of $i$ such that $I_i + d_i$ is nonzero. Then, if $I_i + d_i$ are all even, we have
	\[\EE_{\mu}[u^I h_a(A)] = \Delta^c\left(\frac{1}{\sqrt{\Delta n}}\right)^{|I|} \prod_{e \in [n]^k} \left(\frac{\lda}{(\Del n)^{\frac{k}{2}}}\right)^{a_e}\]
	Else, $\EE_{\mu}[u^I h_a(v)] = 0$.
\end{lemma}

\begin{proof}
	When $A \sim \mu$, for all $e \in [n]^k$, we have $A_e = B_e + \lda \prod_{i \le k} u_{e_i}$. where $B_e \sim \GN(0, 1)$.
	Let's analyze when the required expectation is nonzero. We can first condition on $u$ and use the fact that for a fixed $t$, $\EE_{g \sim \GN(0, 1)}[h_k(g + t)] = t^k$ to obtain
	\[\EE_{(u_i, w_e) \sim \mu}[u^I h_a(A)] = \EE_{(u_i) \sim \mu}[u^I\prod_{e \in [n]^k}(\lda \prod _{i \le k}u_{e_i})^{a_e}] = \EE_{(u_i) \sim \mu}[\prod_{i \in [n]} u_i^{I_i + d_i}] \prod_{e \in [n]^k} \lda^{a_e}\]

	Observe that this is nonzero precisely when all $I_i + d_i$ are even, in which case \[\EE_{(u_i) \sim \mu}[\prod_{i \in [n]} u_i^{I_i + d_i}] = \Delta^c\left(\frac{1}{\sqrt{\Del n}}\right)^{\sum_{i \le n} I_i + d_i} =  \Delta^c\left(\frac{1}{\sqrt{\Delta n}}\right)^{|I|} \prod_{e \in [n]^k} \left(\frac{1}{(\Del n)^{\frac{k}{2}}}\right)^{a_e}\]
	where we used the fact that $\sum_{e \in [n]^k} a_e = k \sum_{i \in [n]} d_i$.
	This completes the proof.
\end{proof}


Define the degree of SoS to be $D_{sos} = n^{C_{sos}\eps}$ for some constant $C_{sos} > 0$ that we choose later. And define the truncation parameters to be $D_V = n^{C_V\eps}, D_E = n^{C_E\eps}$ for some constants $C_V, C_E > 0$.

\begin{remk}[Choice of parameters]\label{rmk: choice_of_params2}
	We first set $\eps$ to be a sufficiently small constant. Based on the choice of $\eps$, we will set the constant $C_{\Del} > 0$ sufficiently small so that the planted distribution is well defined. Based on these choices, just as in \cref{rmk: choice_of_params1} we choose $C_V, C_E, C_{sos}$ in that order.
\end{remk}

The underlying graphs for the graph matrices have the following structure; There will be $n$ vertices of a single type and the edges will be ordered hyperedges of arity $k$.
For the analysis of Tensor PCA, we will use the following notation.
\begin{itemize}
	\item For an index shape $U$ and a vertex $i$, define $deg^{U}(i)$ as follows: If $i \in V(U)$, then it is the power of the unique index shape piece $A \in U$ such that $i \in V(A)$. Otherwise, it is $0$.
	\item For an index shape $U$, define $deg(U) = \sum_{i \in V(U)} deg^U(i)$. This is also the degree of the monomial that $U$ corresponds to.
	\item For a shape $\alpha$ and vertex $i$ in $\alpha$, let $deg^{\alpha}(i) = \sum_{i \in e \in E(\alpha)} l_e$.
	\item For any shape $\alpha$, let $deg(\alpha) = deg(U_{\al}) + deg(V_{\al})$.
\end{itemize}

We will now describe the decomposition of the moment matrix $\Lda$.

\begin{definition}\label{def: tpca_coeffs}
	If a shape $\alpha$ satisfies the following properties:
	\begin{itemize}
		\item $deg^{\alpha}(i) + deg^{U_{\alpha}}(i) + deg^{V_{\alpha}}(i)$ is even for all $i \in V(\alpha)$,
		\item $\alpha$ is proper,
		\item $\alpha$ satisfies the truncation parameters $D_{sos}, D_V, D_E$.
	\end{itemize}
	then define \[\lambda_{\alpha} = \Delta^{|V(\al)|} \left(\frac{1}{\sqrt{\Delta n}}\right)^{deg(\alpha)}  \prod_{e \in E(\al)} \left(\frac{\lda}{(\Del n)^{\frac{k}{2}}}\right)^{l_e}\]
	Otherwise, define $\lambda_{\alpha} = 0$.
\end{definition}

\begin{corollary}
	$\Lambda = \sum \lda_{\al}M_{\al}$.
\end{corollary}

\subsection{Qualitative machinery bounds}

Just as in planted slightly denser subgraph, we prove the PSD mass condition and the qualitative middle shape and intersection term bounds, by first stating them and then introducing appropriate notation to prove them all in a unified manner.

\begin{restatable}[PSD mass]{lemma}{TPCAone}\label{lem: tpca_cond1}
	For all $U \in \calI_{mid}$, $H_{Id_U} \succeq 0$
\end{restatable}

We define the following quantities to capture the contribution of the vertices within $\tau, \gam$ to the Fourier coefficients.

\begin{restatable}{definition}{TPCAstau}\label{def: tpca_stau}
	For $U \in \calI_{mid}$ and $\tau \in \calM_U$, if $deg^{\tau}(i)$ is even for all vertices $i \in V(\tau) \setminus U_{\tau} \setminus V_{\tau}$, define
	\[S(\tau) = \Delta^{|V(\tau)| - |U_{\tau}|}\prod_{e \in E(\tau)}\left(\frac{\lda}{(\Del n)^{\frac{k}{2}}}\right)^{l_e}\]
	Otherwise, define $S(\tau) = 0$. 
	For all $U, V \in \calI_{mid}$ where $w(U) > w(V)$ and $\gam \in \Gam_{U, V}$, if $deg^{\gam}(i)$ is even for all vertices $i$ in $V(\gam) \setminus U_{\gam} \setminus V_{\gam}$, define
	\[S(\gam) = \Delta^{|V(\gam)| - \frac{|U_{\gam}| + |V_{\gam}|}{2}}\prod_{e \in E(\gam)}\left(\frac{\lda}{(\Del n)^{\frac{k}{2}}}\right)^{l_e}\]
	Otherwise, define $S(\gam) = 0$.
\end{restatable}

We now state the qualitative bounds in terms of these quantities.

\begin{restatable}[Qualitative middle shape bounds]{lemma}{TPCAtwosimplified}\label{lem: tpca_cond2_simplified}
	For all $U \in \calI_{mid}$ and $\tau \in \calM_U$,
	\[
	\begin{bmatrix}
		\frac{S(\tau)}{|Aut(U)|}H_{Id_U} & H_{\tau}\\
		H_{\tau}^T & \frac{S(\tau)}{|Aut(U)|}H_{Id_U}
	\end{bmatrix}
	\succeq 0
	\]
\end{restatable}



We again use the canonical definition of $H_{\gam}'$ from \cref{sec: hgamma_qual}.

\begin{restatable}[Qualitative intersection term bounds]{lemma}{TPCAthreesimplified}\label{lem: tpca_cond3_simplified}
	For all $U, V \in \calI_{mid}$ where $w(U) > w(V)$ and all $\gam \in \Gam_{U, V}$,
	\[\frac{|Aut(V)|}{|Aut(U)|}\cdot\frac{1}{S(\gam)^2}H_{Id_V}^{-\gam, \gam} \preceq H_{\gam}'\]
\end{restatable}

\subsubsection{Proof of PSD mass condition}

We introduce some notation which makes it easy to show the qualitative bounds and which also sheds light on the structure of the coefficient matrices. When we compose shapes $\sig, \sig'$, from \cref{def: tpca_coeffs}, in order for $\lda_{\sig\circ \sig'}$ to be nonzero, observe that all vertices $i$ in $\lda_{\sig \circ \sig'}$ should have $deg^{\sig \circ \sig'}(i) + deg^{U_{\sig \circ \sig'}}(i) + deg^{V_{\sig \circ \sig'}}(i)$ to be even. To partially capture this notion conveniently, we will introduce the notion of parity vectors.

\begin{definition}
	Define a parity vector $\rho$ to be a vector whose entries are in $\{0, 1\}$.
	For $U\in \calI_{mid}$, define $\calP_U$ to be the set of parity vectors $\rho$ whose coordinates are indexed by $U$.
\end{definition}

\begin{definition}
	For a left shape $\sig$, define $\rho_{\sig} \in \calP_{V_{\sig}}$, called the parity vector of $\sig$, to be the parity vector such that for each vertex $i \in V_{\sig}$, the $i$-th entry of $\rho_{\sig}$ is the parity of $deg^{U_{\sig}}(i) + deg^{\sig}(i)$, that is $(\rho_{\sig})_i \equiv deg^{U_{\sig}}(i) + deg^{\sig}(i) \pmod 2$.
	For $U \in \calI_{mid}$ and $\rho \in \calP_U$, let $\calL_{U, \rho}$ be the set of all left shapes $\sig \in \calL_U$ such that $\rho_{\sig} = \rho$, that is, the set of all left shapes with parity vector $\rho$.
\end{definition}

For a shape $\tau$, for a $\tau$ coefficient matrix $H_{\tau}$ and parity vectors $\rho \in \calP_{U_{\tau}}, \rho' \in \calP_{V_{\tau}}$, define the $\tau$-coefficient matrix $H_{\tau, \rho, \rho'}$ as $H_{\tau ,\rho, \rho'}(\sig, \sig') = H_{\tau}(\sig, \sig')$ if $\sig \in \calL_{U_{\tau}, \rho}, \sig' \in \calL_{V_{\tau}, \rho'}$ and $0$ otherwise.
The following proposition is immediate.

\begin{propn}
	For any shape $\tau$ and $\tau$-coefficient matrix $H_{\tau}$, $H_{\tau} = \sum_{\rho \in \calP_{U_{\tau}}, \rho' \in \calP_{V_{\tau}}} H_{\tau, \rho, \rho'}$
\end{propn}

\begin{propn}
	For any $U \in \calI_{mid}$, $H_{Id_U} = \sum_{\rho \in \calP_U} H_{Id_U, \rho, \rho}$
\end{propn}

\begin{proof}
	For any $\sig, \sig' \in \calL_U$, using \cref{def: tpca_coeffs}, note that in order for $H_{Id_U}(\sig, \sig')$ to be nonzero, we must have $\rho_{\sig} = \rho_{\sig'}$.
\end{proof}

We define the following quantity to capture the contribution of the vertices within $\sig$ to the Fourier coefficients.

\begin{definition}
	For a shape $\sig\in \calL$, if $deg^{\sig}(i) + deg^{U_{\sig}}(i)$ is even for all vertices $i \in V(\sig) \setminus V_{\sig}$, define
	\[T(\sig) = \Delta^{|V(\sig)| - \frac{|V_{\sig}|}{2}}\left(\frac{1}{\sqrt{\Delta n}}\right)^{deg(U_{\sig})}\prod_{e \in E(\sig)}\left(\frac{\lda}{(\Del n)^{\frac{k}{2}}}\right)^{l_e}\]
	Otherwise, define $T(\sig) = 0$.
	For $U \in \calI_{mid}$ and $\rho \in \calP_U$, define $v_{\rho}$ to be the vector indexed by $\sig \in \calL$ such that $v_{\rho}(\sig)$ is $T(\sig)$ if $\sig \in \calL_{U, \rho}$ and $0$ otherwise.
\end{definition}

With this notation, the PSD mass condition is easily shown.

\begin{proof}[Proof of the PSD mass condition \cref{lem: tpca_cond1}]
    For all $U\in \calI_{mid}, \rho \in \calP_U$, \cref{def: tpca_coeffs} implies $H_{Id_U, \rho, \rho} = \frac{1}{|Aut(U)|}v_{\rho}v_{\rho}^T$.
	Therefore, $H_{Id_U} = \sum_{\rho \in \calP_U} H_{Id_U, \rho, \rho} = \frac{1}{|Aut(U)|} \sum_{\rho \in \calP_U} v_{\rho}v_{\rho}^T \succeq 0$.
\end{proof}

\subsubsection{Qualitative middle shape bounds}

The next proposition captures the fact that when we compose shapes $\sig, \tau, \sig'^T$, in order for $\lda_{\sig \circ \tau \sig'^T}$ to be nonzero, the parities of the degrees of the merged vertices should add up correspondingly.

\begin{propn}\label{propn: tpca_coeff_2}
	For all $U \in \calI_{mid}$ and $\tau \in \calM_U$, there exist two sets of parity vectors $P_{\tau}, Q_{\tau} \subseteq \calP_{U}$ and a bijection $\pi : P_{\tau} \to Q_{\tau}$ such that $H_{\tau} = \sum_{\rho \in P_{\tau}} H_{\tau, \rho, \pi(\rho)}$.
\end{propn}

\begin{proof}
	Using \cref{def: tpca_coeffs}, in order for $H_{\tau}(\sig, \sig')$ to be nonzero, in $\sig \circ \tau \circ \sig'$, we must have that for all $i \in U_{\tau} \cup V_{\tau}$, $deg^{U_{\sig}}(i) + deg^{U_{\sig'}}(i) + deg^{\sigma \circ \tau \circ \sigma'^T}(i)$ must be even. In other words, for any $\rho \in \calP_U$, there is at most one $\rho' \in \calP_U$ such that if we take $\sig \in \calL_{U, \rho}, \sig' \in \calL_U$ with $H_{\tau}(\sig, \sig')$ nonzero, then the parity of $\sig'$ is $\rho'$. Also, observe that $\rho'$ determines $\rho$. We then take $P_{\tau}$ to be the set of $\rho$ such that $\rho'$ exists, $Q_{\tau}$ to be the set of $\rho'$ and in this case, we define $\pi(\rho) = \rho'$.
\end{proof}



A straightforward verification of the conditions of \cref{def: tpca_coeffs} implies the following proposition.

\begin{propn}
	For any $U \in \calI_{mid}$ and $\tau \in \calM_U$, suppose we take $\rho \in P_{\tau}$.  Let $\pi$ be the bijection from \cref{propn: tpca_coeff_2} so that $\pi(\rho) \in Q_{\tau}$. Then, $H_{\tau, \rho, \pi(\rho)} = \frac{1}{|Aut(U)|^2} S(\tau) v_{\rho}v_{\pi(\rho)}^T$.
\end{propn}


We can now prove the qualitative middle shape bounds.


\begin{proof}[Proof of the qualitative middle shape bounds \cref{lem: tpca_cond2_simplified}]
	Let $P_{\tau}, Q_{\tau}, \pi$ be from \cref{propn: tpca_coeff_2}. For $\rho, \rho' \in \calP_U$, let $W_{\rho, \rho'} = v_{\rho}(v_{\rho'})^T$. Then, $H_{Id_U} = \sum_{\rho \in \calP_U} H_{Id_U, \rho, \rho} = \frac{1}{|Aut(U)|} \sum_{\rho \in \calP_U}W_{\rho, \rho}$ and $H_{\tau} = \sum_{\rho \in P_{\tau}} H_{\tau, \rho, \pi(\rho)} = \frac{1}{|Aut(U)|^2}S(\tau)\sum_{\rho \in P_{\tau}} W_{\rho, \pi(\rho)}$. We have

	\begin{align*}
		\begin{bmatrix}
			\frac{S(\tau)}{|Aut(U)|}H_{Id_U} & H_{\tau}\\
			H_{\tau}^T & \frac{S(\tau)}{|Aut(U)|}H_{Id_U}
		\end{bmatrix}
		&= \frac{S(\tau)}{|Aut(U)|^2}
		\begin{bmatrix}
			\sum_{\rho \in \calP_U} W_{\rho, \rho} & \sum_{\rho \in P_{\tau}} W_{\rho, \pi(\rho)}\\
			\sum_{\rho \in P_{\tau}} W_{\rho, \pi(\rho)}^T & \sum_{\rho \in \calP_U} W_{\rho, \rho}
		\end{bmatrix}
	\end{align*}
	Since $\frac{S(\tau)}{|Aut(U)|^2} \ge 0$, it suffices to prove that $\begin{bmatrix}
		\sum_{\rho \in \calP_U} W_{\rho, \rho} & \sum_{\rho \in P_{\tau}} W_{\rho, \pi(\rho)}\\
		\sum_{\rho \in P_{\tau}} W_{\rho, \pi(\rho)}^T & \sum_{\rho \in \calP_U} W_{\rho, \rho}
	\end{bmatrix}\succeq 0$. Consider
	\begin{align*}
		\begin{bmatrix}
			\sum_{\rho \in \calP_U} W_{\rho, \rho} & \sum_{\rho \in P_{\tau}} W_{\rho, \pi(\rho)}\\
			\sum_{\rho \in P_{\tau}} W_{\rho, \pi(\rho)}^T & \sum_{\rho \in \calP_U} W_{\rho, \rho}
		\end{bmatrix} =& \begin{bmatrix}
			\sum_{\rho \in \calP_U \setminus P_{\tau}} W_{\rho, \rho} & 0\\
			0 & \sum_{\rho \in \calP_U \setminus Q_{\tau}} W_{\rho, \rho}
		\end{bmatrix}\\
		& + \begin{bmatrix}
			\sum_{\rho \in P_{\tau}} W_{\rho, \rho} & \sum_{\rho \in P_{\tau}} W_{\rho, \pi(\rho)}\\
			\sum_{\rho \in P_{\tau}} W_{\rho, \pi(\rho)}^T & \sum_{\rho \in P_{\tau}} W_{\pi(\rho), \pi(\rho)}
		\end{bmatrix}\\
	\end{align*}

	We have $\sum_{\rho \in \calP_U \setminus P_{\tau}} W_{\rho, \rho} = \sum_{\rho \in \calP_U \setminus P_{\tau}} v_{\rho}v_{\rho}^T \succeq 0$. Similarly, $\sum_{\rho \in \calP_U \setminus Q_{\tau}} W_{\rho, \rho} \succeq 0$ and so, the first term in the above expression,
	$\begin{bmatrix}
		\sum_{\rho \in \calP_U \setminus P_{\tau}} W_{\rho, \rho} & 0\\
		0 & \sum_{\rho \in \calP_U \setminus Q_{\tau}} W_{\rho, \rho}
	\end{bmatrix}$ is positive semidefinite. For the second term,
	\begin{align*}
		\begin{bmatrix}
			\sum_{\rho \in P_{\tau}} W_{\rho, \rho} & \sum_{\rho \in P_{\tau}} W_{\rho, \pi(\rho)}\\
			\sum_{\rho \in P_{\tau}} W_{\rho, \pi(\rho)}^T & \sum_{\rho \in P_{\tau}} W_{\pi(\rho), \pi(\rho)}
		\end{bmatrix} &= \sum_{\rho \in P_{\tau}}
		\begin{bmatrix}
			W_{\rho, \rho} & W_{\rho, \pi(\rho)}\\
			W_{\rho, \pi(\rho)}^T & W_{\pi(\rho), \pi(\rho)}
		\end{bmatrix}\\
		&= \sum_{\rho \in P_{\tau}}
		\begin{bmatrix}
			v_{\rho}v_{\rho}^T & v_{\rho}(v_{\pi(\rho)})^T\\
			v_{\pi(\rho)}(v_{\rho})^T & v_{\pi(\rho)}(v_{\pi(\rho)})^T
		\end{bmatrix}\\
		&= \sum_{\rho \in P_{\tau}}
		\begin{bmatrix}
			v_{\rho}\\
			v_{\pi(\rho)}
		\end{bmatrix}
		\begin{bmatrix}
			v_{\rho} &
			v_{\pi(\rho)}
		\end{bmatrix}\\
		& \succeq 0
	\end{align*}
\end{proof}

\subsubsection{Qualitative intersection term bounds}

Similar to \cref{propn: tpca_coeff_2}, the next proposition captures the fact that when we compose shapes $\sig, \gam, \gam^T, \sig'^T$, in order for $\lda_{\sig \circ \gam \circ \gam'^T \circ \sig'^T}$ to be nonzero, the parities of the degrees of the merged vertices should add up correspondingly. 

We use the following notation.
For all $U, V \in \calI_{mid}$ where $w(U) > w(V)$, for $\gam \in \Gam_{U, V}$ and parity vectors $\rho, \rho' \in \calP_U$,  define the $\gam \circ \gam^T$-coefficient matrix $H_{Id_V, \rho, \rho'}^{-\gam, \gam}$ as $H_{Id_V, \rho, \rho'}^{-\gam, \gam}(\sig, \sig') = H_{Id_V}^{-\gam, \gam}(\sig, \sig')$ if $\sig \in \calL_{U, \rho},  \sig' \in \calL_{U, \rho'}$ and $0$ otherwise.

\begin{propn}
	For all $U, V \in \calI_{mid}$ where $w(U) > w(V)$, for all $\gam \in \Gam_{U, V}$, there exists a set of parity vectors $P_{\gam} \subseteq \calP_U$ such that
	$H_{Id_V}^{-\gam, \gam} = \sum_{\rho \in P_{\gam}} H_{Id_V, \rho, \rho}^{-\gam, \gam}$.
\end{propn}

\begin{proof}
	Take any $\rho \in \calP_U$. For $\sig \in \calL_{U, \rho}, \sig' \in \calL_U$,  since $H_{Id_V}^{-\gam, \gam}(\sigma, \sigma') = \frac{\lda_{\sig \circ \gam \circ \gam^T \circ \sig'^T}}{|Aut(V)|}$, $H_{Id_V}^{-\gam, \gam}(\sig, \sig')$ is nonzero precisely when $\lda_{\sig \circ \gam \circ \gam^T \circ \sig'^T}$ is nonzero. For this quantity to be nonzero, using \cref{def: tpca_coeffs}, we get that it is necessary, but not sufficient, that the parity vector of $\sig'$ must also be $\rho$. And also observe that there exists a set $P_{\gam}$ of parity vectors $\rho$ for which $H_{Id_V, \rho, \rho}^{-\gam, \gam}$ is nonzero and their sum is precisely $H_{Id_V}^{-\gam, \gam}$.
\end{proof}

For all $U, V \in \calI_{mid}$ where $w(U) > w(V)$, for all $\gam \in \Gam_{U, V}$ and parity vector $\rho \in \calP_U$, define the matrix $H'_{\gam, \rho, \rho}$ as $H'_{\gam, \rho, \rho}(\sig, \sig') = H'_{\gam}(\sig, \sig')$ if $\sig, \sig' \in \calL_{U, \rho}$ and $0$ otherwise. The following proposition is immediate from the definition.

\begin{propn}
	For all $U, V \in \calI_{mid}$ where $w(U) > w(V)$, for $\gam \in \Gam_{U, V}$, $H_{\gam}' = \sum_{\rho \in P_{\gam}} H_{\gam, \rho, \rho}'$.
\end{propn}



\begin{propn}
	For all $U, V \in \calI_{mid}$ where $w(U) > w(V)$, for all $\gam \in \Gam_{U, V}$ and $\rho \in P_{\gam}$,
	\[H_{Id_V, \rho, \rho}^{-\gam, \gam} = \frac{|Aut(U)|}{|Aut(V)|} S(\gam)^2 H'_{\gam, \rho, \rho}\]
\end{propn}

\begin{proof}
	Fix $\sig, \sig' \in \calL_{U, \rho}$ such that $|V(\sig \circ \gam)|, |V(\sig' \circ \gam)| \le D_V$. Note that $|V(\sig)| - \frac{|V_{\sig}|}{2} + |V(\sig')| - \frac{|V_{\sig'}|}{2} + 2(|V(\gam)| - \frac{|U_{\gam}| + |V_{\gam}|}{2}) = |V(\sig \circ \gam \circ \gam^T \circ \sig'^T)|$. Using \cref{def: tpca_coeffs}, we can easily verify that $\lda_{\sig \circ \gam \circ \gam^T \circ \sig'^T} = T(\sigma)T(\sigma') S(\gam)^2$. Therefore, $H_{Id_V, \rho, \rho}^{-\gam, \gam}(\sig, \sig') = \frac{|Aut(U)|}{|Aut(V)|} S(\gam)^2 H_{Id_U, \rho, \rho}(\sig, \sig')$. Since $H'_{\gam, \rho, \rho}(\sig, \sig') = H_{Id_U, \rho, \rho}(\sig, \sig')$ whenever $|V(\sig \circ \gam)|, |V(\sig' \circ \gam)| \le D_V$, this completes the proof.
\end{proof}

With this, we can prove the qualitative intersection term bounds.

\begin{proof}[Proof of qualitative intersection term bounds \cref{lem: tpca_cond3_simplified}]
	We have
	\begin{align*}
		\frac{|Aut(V)|}{|Aut(U)|}\cdot\frac{1}{S(\gam)^2}H_{Id_V}^{-\gam, \gam} &= \sum_{\rho \in P_{\gam}} \frac{|Aut(V)|}{|Aut(U)|}\cdot\frac{1}{S(\gam)^2} H_{Id_V, \rho, \rho}^{-\gam, \gam}
		= \sum_{\rho \in P_{\gam}} H'_{\gam, \rho, \rho}
		\preceq \sum_{\rho \in \calP_U} H'_{\gam, \rho, \rho}
		= H'_{\gam}
	\end{align*}
	where we used the fact that for all $\rho \in \calP_U$, we have $H'_{\gam,\rho, \rho} \succeq 0$.
\end{proof}

%% file: sparse_pca_qual.tex
\subsection{Pseudo-calibration}

\begin{definition}[Slack parameter]
	Define the slack parameter to be $\Delta = d^{-C_{\Del}\eps}$ for a constant $C_{\Del} > 0$.
\end{definition}

We will pseudo-calibrate with respect the following pair of random and planted distributions which we denote $\nu$ and $\mu$ respectively.

\SPCAdistributions*

We will again work with the Hermite basis of polynomials. For $a \in \NN^{m \times d}$ and variables $v_{i, j}$ for $i \in [m], j \in [n]$, define $h_a(v) \defeq \prod_{i \in [m], j \in [n]} h_{a_{i, j}}(v_{i, j})$.
For a nonnegative integer $t$, define $t!!= \frac{(2t)!}{t!2^t} = 1 \times 3 \times \ldots \times t$ if $t$ is odd and $0$ otherwise.

\begin{lemma}
	Let $I \in \NN^d, a \in \NN^{m \times d}$. For $i \in [m]$, let $e_i = \sum_{j \in [d]} a_{ij}$ and for $j \in [d]$, let $f_j = I_j + \sum_{i \in [m]} a_{ij}$. Let $c_1$ (resp. $c_2$) be the number of $i$ (resp. $j$) such that $e_i > 0$ (resp. $f_j > 0$). Then, if $e_i, f_j$ are all even, we have
	\[\EE_{\mu}[u^I h_a(v)] = \left(\frac{1}{\sqrt{k}}\right)^{|I|} \left(\frac{k}{d}\right)^{c_2} \Delta^{c_1}\prod_{i \in [m]} (e_i - 1)!! \prod_{i, j} \frac{\sqrt{\lambda}^{a_{ij}}}{\sqrt{k}^{a_{ij}}}\]
	Else, $\EE_{\mu}[u^I h_a(v)] = 0$.
\end{lemma}

\begin{proof}
	$v_1, \ldots, v_m \sim \mu$ can be written as $v_i = g_i + \sqrt{\lda} b_i l_i u$ where $g_i \sim \GN(0, I_d), l_i \sim \GN(0, 1), b_i \in \{0, 1\}$ where $b_i = 1$ with probability $\Del$.
	Let's analyze when the required expectation is nonzero. We can first condition on $b_i, l_i, u$ and use the fact that for a fixed $t$, $\EE_{g \sim \GN(0, 1)}[h_k(g + t)] = t^k$ to obtain
	\[\EE_{(u, l_i, b_i, g_i) \sim \mu}[u^I h_a(v)] = \EE_{(u, l_i, b_i) \sim \mu}[u^I\prod_{i, j}(\sqrt{\lda}b_il_iu_j)^{a_{ij}}] = \EE_{(u, l_i, b_i) \sim \mu}[\prod_{i \in [m]} (b_il_i)^{e_i}\prod_{j \in [d]} u_j^{f_j}] \prod_{i, j} \sqrt{\lda}^{a_{ij}}\]
	For this to be nonzero, the set of $c_1$ indices $i$ such that $e_i > 0$, should not have been resampled otherwise $b_i = 0$, each of which happens independently with probability $\Del$. And the set of $c_2$ indices $j$ such that $f_j > 0$ should have been such that $u_j$ is nonzero, each of which happens independently with probability $\frac{k}{d}$. Since $l_i, u_j$ are have zero expectation in $\nu$, we need $e_i, f_j$ to be even. The expectation then becomes
	\[\Del^{c_1} \left(\frac{k}{d}\right)^{c_2}\EE_{(u, l_i) \sim \mu}[\prod_{i \in [m]} l_i^{e_i}\prod_{j \in [d]} u_j^{f_j}] \prod_{i, j} \sqrt{\lda}^{a_{ij}} = \left(\frac{1}{\sqrt{k}}\right)^{|I|} \left(\frac{k}{d}\right)^{c_2} \Del^{c_1}\prod_{i \in [m]} (e_i - 1)!! \prod_{i, j} \frac{\sqrt{\lambda}^{a_{ij}}}{\sqrt{k}^{a_{ij}}}\]
	The last equality follows because, for each $j$ such that $u_j$ is nonzero, we have $u_j^t = (\frac{1}{\sqrt{k}})^t$ and $\EE_{g \sim \GN(0, 1)}[g^t] = (t - 1)!!$ if $t$ is even.
\end{proof}

Define the degree of SoS to be $D_{sos} = d^{C_{sos}\eps}$ for some constant $C_{sos} > 0$ that we choose later.
Define the truncation parameters to be $D_V = d^{C_V\eps}, D_E = d^{C_E\eps}$ for some constants $C_V, C_E > 0$. Regarding the choice of parameters, although we are working with a different problem, \cref{rmk: choice_of_params2} directly applies.

The underlying graphs for the graph matrices have the following structure:
There will be two types of vertices - $d$ type $1$ vertices corresponding to the dimensions of the space and $m$ type $2$ vertices corresponding to the different input vectors. The shapes will correspond to bipartite graphs with edges going between across of different types.
For the analysis of Sparse PCA, we will use the following notation.
\begin{itemize}
	\item For a shape $\al$ and type $t \in \{1, 2\}$, let $V_t(\al)$ denote the vertices of $V(\al)$ that are of type $t$. Let $|\al|_t = |V_t(\al)|$.
	\item For an index shape $U$ and a vertex $i$, define $deg^{U}(i)$ as follows: If $i \in V(U)$, then it is the power of the unique index shape piece $A \in U$ such that $i \in V(A)$. Otherwise, it is $0$.
	\item For an index shape $U$, define $deg(U) = \sum_{i \in V(U)} deg^U(i)$. This is also the degree of the monomial $p_U$.
	\item For a shape $\alpha$ and vertex $i$ in $\alpha$, let $deg^{\alpha}(i) = \sum_{i \in e \in E(\alpha)} l_e$.
	\item For any shape $\alpha$, let $deg(\alpha) = deg(U_{\al}) + deg(V_{\al})$.
	\item For an index shape $U \in \calI_{mid}$ and type $t \in \{1, 2\}$, let $U_t \in U$ denote the index shape piece of type $t$ in $U$ if it exists, otherwise define $U_t$ to be $\emptyset$. Note that this is well defined since for each type $t$, there is at most one index shape piece of type $t$ in $U$ since $U \in \calI_{mid}$. Also, denote by $|U|_t$ the length of the tuple $U_t$.
\end{itemize}

We will now describe the decomposition of the moment matrix $\Lda$.

\begin{definition}\label{def: spca_coeffs}
	If a shape $\alpha$ satisfies the following properties:
	\begin{itemize}
		\item Both $U_{\alpha}$ and $V_{\alpha}$ only contain index shape pieces of type $1$,
		\item $deg^{\alpha}(i) + deg^{U_{\alpha}}(i) + deg^{V_{\alpha}}(i)$ is even for all $i \in V(\alpha)$,
		\item $\alpha$ is proper,
		\item $\alpha$ satisfies the truncation parameters $D_{sos}, D_V, D_E$.
	\end{itemize}
	then define \[\lambda_{\alpha} = \left(\frac{1}{\sqrt{k}}\right)^{deg(\alpha)}\left(\frac{k}{d}\right)^{|\alpha|_1}\Del^{|\alpha|_2} \prod_{j \in V_2(\alpha)} (deg^{\alpha}(j) - 1)!!\prod_{e \in E(\alpha)} \frac{\sqrt{\lambda}^{l_e}}{\sqrt{k}^{l_e}}\]
	Otherwise, define $\lambda_{\alpha} = 0$.
\end{definition}

\begin{corollary}
	$\Lambda = \sum \lda_{\al}M_{\al}$.
\end{corollary}

\subsection{Qualitative machinery bounds}

In this section, we will prove the main PSD mass condition and obtain qualitative bounds of the other two conditions, which we will reuse in the full verification.
As in prior sections, we will state the bounds first, introduce notation and then prove them all in a unified manner.

\begin{restatable}[PSD mass]{lemma}{SPCAone}\label{lem: spca_cond1}
	For all $U \in \calI_{mid}$, $H_{Id_U} \succeq 0$
\end{restatable}

We define the following quantities to capture the contribution of the vertices within $\tau, \gam$ to the Fourier coefficients.

\begin{restatable}{definition}{SPCAstau}\label{def: spca_stau}
	For $U \in \calI_{mid}$ and $\tau \in \calM_U$, if $deg^{\tau}(i)$ is even for all vertices $i \in V(\tau) \setminus U_{\tau} \setminus V_{\tau}$, define
	\[S(\tau) =
	\left(\frac{k}{d}\right)^{|\tau|_1 - |U_{\tau}|_1}\Del^{|\tau|_2 - |U_{\tau}|_2} \prod_{j \in V_2(\tau) \setminus U_{\tau} \setminus V_{\tau}} (deg^{\tau}(j) - 1)!!\prod_{e \in E(\tau)} \frac{\sqrt{\lambda}^{l_e}}{\sqrt{k}^{l_e}}\]
	Otherwise, define $S(\tau) = 0$. 	For all $U, V \in \calI_{mid}$ where $w(U) > w(V)$ and $\gam \in \Gam_{U, V}$, if $deg^{\gam}(i)$ is even for all vertices $i$ in $V(\gam) \setminus U_{\gam} \setminus V_{\gam}$, define
	\[S(\gam) =
	\left(\frac{k}{d}\right)^{|\gamma|_1 - \frac{|U_{\gamma}|_1 + |V_{\gamma}|_1}{2}}\Del^{|\gamma|_2 - \frac{|U_{\gamma}|_2 + |V_{\gamma}|_2}{2}} \prod_{j \in V_2(\gamma) \setminus U_{\gamma} \setminus V_{\gamma}} (deg^{\gamma}(j) - 1)!!\prod_{e \in E(\gamma)} \frac{\sqrt{\lambda}^{l_e}}{\sqrt{k}^{l_e}}\]
	Otherwise, define $S(\gam) = 0$.
\end{restatable}

For getting the best bounds, it will be convenient to discretize the Normal distribution. The following fact follows from standard results on Gaussian quadrature, see for e.g. \cite[Lemma 4.3]{diakonikolas2017statistical}.

\begin{fact}[Discretizing the Normal distribution]\label{fact: quadrature}
	There is an absolute constant $C_{disc}$ such that, for any positive integer $D$, there exists a distribution $\calE$ over the real numbers supported on $D$ points $p_1, \ldots, p_D$, such that $|p_i| \le C_{disc} \sqrt{D}$ for all $i \le D$ and
    $\EE_{g \sim \calE}[g^t] = \EE_{g \sim \GN(0, 1)}[g^t]$ for all $t = 0, 1, \ldots, 2D - 1$.
\end{fact}

\begin{definition} For any shape $\tau$, suppose $U' = (U_{\tau})_2, V' = (V_{\tau})_2$ are the type $2$ vertices in $U_{\tau}, V_{\tau}$ respectively. Define
$R(\tau) = (C_{disc}\sqrt{D_E})^{\sum_{j \in U' \cup V'} deg^{\tau}(j)}$.
\end{definition}

We can now state our qualitative bounds.

\begin{restatable}[Qualitative middle shape bounds]{lemma}{SPCAtwosimplified}\label{lem: spca_cond2_simplified}
	For all $U \in\calI_{mid}$ and $\tau \in \calM_U$,
	\[
	\begin{bmatrix}
	\frac{S(\tau)R(\tau)}{|Aut(U)|}H_{Id_U} & H_{\tau}\\
	H_{\tau}^T & \frac{S(\tau)R(\tau)}{|Aut(U)|}H_{Id_U}
	\end{bmatrix}
	\succeq 0\]
\end{restatable}



We again use the canonical definition of $H_{\gam}'$ from \cref{sec: hgamma_qual}.

\begin{restatable}[Qualitative intersection term bounds]{lemma}{SPCAthreesimplified}\label{lem: spca_cond3_simplified}
	For all $U, V \in \calI_{mid}$ where $w(U) > w(V)$ and all $\gam \in \Gam_{U, V}$,
	\[\frac{|Aut(V)|}{|Aut(U)|}\cdot\frac{1}{S(\gam)^2R(\gam)^2}H_{Id_V}^{-\gam, \gam} \preceq H_{\gam}'\]
\end{restatable}

\subsubsection{Proof of the PSD mass condition}

Most of the notation and analysis here are similar to the case of Tensor PCA, we just need to appropriately modify them since there are two types of vertices in the Sparse PCA application.
When we compose shapes $\sig, \sig'$, from \cref{def: spca_coeffs}, in order for $\lda_{\sig\circ \sig'}$ to be nonzero, observe that all vertices $i$ in $\lda_{\sig \circ \sig'}$ should have $deg^{\sig \circ \sig'}(i) + deg^{U_{\sig \circ \sig'}}(i) + deg^{V_{\sig \circ \sig'}}(i)$ to be even. To capture this notion conveniently, we again use the notion of parity vectors.

\begin{definition}
	Define a parity vector $\rho$ to be a vector whose entries are in $\{0, 1\}$.
	For $U\in \calI_{mid}$, define $\calP_U$ to be the set of parity vectors $\rho$ whose coordinates are indexed by $U_1$ followed by $U_2$.
\end{definition}

\begin{definition}
	For a left shape $\sig$, define $\rho_{\sig} \in \calP_{V_{\sig}}$, called the parity vector of $\sig$, to be the parity vector such that for each vertex $i \in V_{\sig}$, the $i$-th entry of $\rho_{\sig}$ is the parity of $deg^{U_{\sig}}(i) + deg^{\sig}(i)$, that is, $(\rho_{\sig})_i \equiv deg^{U_{\sig}}(i) + deg^{\sig}(i) \pmod 2$.
	For $U \in \calI_{mid}$ and $\rho \in \calP_U$, let $\calL_{U, \rho}$ be the set of all left shapes $\sig \in \calL_U$ such that $\rho_{\sig} = \rho$, that is, the set of all left shapes with parity vector $\rho$.
\end{definition}

For a shape $\tau$, for a $\tau$ coefficient matrix $H_{\tau}$ and parity vectors $\rho \in \calP_{U_{\tau}}, \rho' \in \calP_{V_{\tau}}$, define the $\tau$-coefficient matrix $H_{\tau, \rho, \rho'}$ as $H_{\tau ,\rho, \rho'}(\sig, \sig') = H_{\tau}(\sig, \sig')$ if $\sig \in \calL_{U_{\tau}, \rho}, \sig' \in \calL_{V_{\tau}, \rho'}$ and $0$ otherwise. This immediately implies the following proposition.

\begin{propn}
	For any shape $\tau$ and $\tau$-coefficient matrix $H_{\tau}$, $H_{\tau} = \sum_{\rho \in \calP_{U_{\tau}}, \rho' \in \calP_{V_{\tau}}} H_{\tau, \rho, \rho'}$
\end{propn}

\begin{propn}
	For any $U \in \calI_{mid}$, $H_{Id_U} = \sum_{\rho \in \calP_U} H_{Id_U, \rho, \rho}$
\end{propn}

\begin{proof}
	For any $\sig, \sig' \in \calL_U$, using \cref{def: spca_coeffs}, note that in order for $H_{Id_U}(\sig, \sig')$ to be nonzero, we must have $\rho_{\sig} = \rho_{\sig'}$.
\end{proof}

We now discretize the normal distribution while matching the first $2D_E - 1$ moments.

\begin{definition}\label{def: discretized_gaussian}
	Let $\calD$ be a distribution over the real numbers obtained by setting $D = D_E$ in \cref{fact: quadrature}. So, in particular, for any $x$ sampled from $\calD$, we have $|x| \le C_{disc}\sqrt{D_E}$ and for $t \le 2D_E - 1$, $\EE_{x \sim \calD}[x^t] = (t - 1)!!$.
\end{definition}

We define the following quantities to capture the contribution of the vertices within $\sig$ to the Fourier coefficients.

\begin{definition}
	For a shape $\sig\in \calL$, if $deg^{\sig}(i) + deg^{U_{\sig}}(i)$ is even for all vertices $i \in V(\sig) \setminus V_{\sig}$, define
	\[T(\sig) = \left(\frac{1}{\sqrt{k}}\right)^{deg(U_{\sig})}\left(\frac{k}{d}\right)^{|\sig|_1 - \frac{|V_{\sig}|_1}{2}}\Del^{|\sig|_2 - \frac{|V_{\sig}|_2}{2}} \prod_{j \in V_2(\sig) \setminus V_{\sig}} (deg^{\sig}(j) - 1)!!\prod_{e \in E(\sig)} \frac{\sqrt{\lambda}^{l_e}}{\sqrt{k}^{l_e}}\]
	Otherwise, define $T(\sig) = 0$.
\end{definition}

\begin{definition}
	Let $U \in \calI_{mid}$. Let $x_i$ for $i \in U_2$ be variables. Denote them collectively as $x_{U_2}$. For $\rho \in \calP_U$, define $v_{\rho, x_{U_2}}$ to be the vector indexed by left shapes $\sig \in \calL$ such that the $\sig$th entry is $T(\sig) \prod_{i \in {U_2}} x_i^{deg^{\sig}(i)}$ if $\sig \in \calL_{U, \rho}$ and $0$ otherwise.
\end{definition}

The following proposition is obvious and immediately implies the PSD mass condition.

\begin{propn}
	For any $U\in \calI_{mid}, \rho \in \calP_U$, suppose $x_i$ for $i \in U_2$ are random variables sampled from $\calD$. Then,
	$H_{Id_U, \rho, \rho} = \frac{1}{|Aut(U)|}\EE_{x}[v_{\rho, x_{U_2}}v_{\rho, x_{U_2}}^T]$.
\end{propn}

\begin{proof}
	Observe that for $\sig, \sig' \in \calL_{U, \rho}$ and $t \in \{1, 2\}$, $(|\sig|_t - \frac{|V_{\sig}|_t}{2}) + (|\sig'|_t - \frac{|V_{\sig'}|_t}{2}) = |\sig \circ \sig'|_t$. The result follows by verifying the conditions of \cref{def: spca_coeffs} and using \cref{def: discretized_gaussian}.
\end{proof}


\begin{proof}[Proof of the PSD mass condition \cref{lem: spca_cond1}]
	We have $H_{Id_U} = \sum_{\rho \in \calP_U} H_{Id_U, \rho, \rho} \succeq 0$ because of the above proposition.
\end{proof}

\subsubsection{Qualitative middle shape bounds}

The next proposition captures the fact that when we compose shapes $\sig, \tau, \sig'^T$, in order for $\lda_{\sig \circ \tau \circ \sig'^T}$ to be nonzero, the parities of the degrees of the merged vertices should add up correspondingly.

\begin{propn}\label{propn: spca_coeff_2}
	For all $U \in \calI_{mid}$ and $\tau \in \calM_U$, there exist two sets of parity vectors $P_{\tau}, Q_{\tau} \subseteq \calP_{U}$ and a bijection $\pi : P_{\tau} \to Q_{\tau}$ such that $H_{\tau} = \sum_{\rho \in P_{\tau}} H_{\tau, \rho, \pi(\rho)}$.
\end{propn}

\begin{proof}
	Using \cref{def: spca_coeffs}, in order for $H_{\tau}(\sig, \sig')$ to be nonzero, we must have that, in $\sig \circ \tau \circ \sig'$, for all $i \in U_{\tau} \cup V_{\tau}$, $deg^{U_{\sig}}(i) + deg^{U_{\sig'}}(i) + deg^{\sigma \circ \tau \circ \sigma'^T}(i)$ must be even. In other words, for any $\rho \in \calP_U$, there is at most one $\rho' \in \calP_U$ such that if we take $\sig \in \calL_{U, \rho}, \sig' \in \calL_U$ with $H_{\tau}(\sig, \sig')$ nonzero, then the parity of $\sig'$ is $\rho'$. Also, observe that $\rho'$ determines $\rho$. We then take $P_{\tau}$ to be the set of $\rho$ such that $\rho'$ exists, $Q_{\tau}$ to be the set of $\rho'$ and in this case, we define $\pi(\rho) = \rho'$.
\end{proof}



\begin{propn}
	For any $U \in \calI_{mid}$ and $\tau \in \calM_U$, suppose we take $\rho \in P_{\tau}$.  Let $\pi$ be the bijection from \cref{propn: spca_coeff_2} so that $\pi(\rho) \in Q_{\tau}$. Let $U' = (U_{\tau})_2, V' = (V_{\tau})_2$ be the type $2$ vertices in $U_{\tau}, V_{\tau}$ respectively. Let $x_i$ for $i \in U' \cup V'$ be random variables independently sampled from $\calD$. Define $x_{U'}$ (resp. $x_{V'}$) to be the subset of variables $x_i$ for $i \in U'$ (resp. $i \in V'$). Then,
	\[H_{\tau, \rho, \pi(\rho)} = \frac{1}{|Aut(U)|^2} S(\tau) \EE_x\left[v_{\rho, x_{U'}}\left(\prod_{i \in U' \cup V'} x_i^{deg^{\tau}(i)}\right)v_{\pi(\rho), x_{V'}}^T\right]\]
\end{propn}

\begin{proof}
	For $\sig \in L_{U, \rho}, \sig' \in \calL_{U, \pi(\rho)}$ and $t \in \{1, 2\}$, we have $(|\tau|_t - |U_{\tau}|_t) + (|\sig|_t - \frac{|V_{\sig}|_t}{2}) + (|\sig'|_t - \frac{|V_{\sig'}|_t}{2}) = |\sig \circ \tau\circ \sig'|_t$.
	The result then follows by a straightforward verification of the conditions of \cref{def: spca_coeffs} using \cref{def: discretized_gaussian}.
\end{proof}


We are ready to show the qualitative middle shape bounds.

\begin{proof}[Proof of the qualitative middle shape bounds \cref{lem: spca_cond2_simplified}]
	Let $P_{\tau}, Q_{\tau}, \pi$ be from \cref{propn: spca_coeff_2}. Let $U' = (U_{\tau})_2, V' = (V_{\tau})_2$ be the type $2$ vertices in $U_{\tau}, V_{\tau}$ respectively. Let $x_i$ for $i \in U' \cup V'$ be random variables independently sampled from $\calD$. Define $x_{U'}$ (resp. $x_{V'}$) to be the subset of variables $x_i$ for $i \in U'$ (resp. $i \in V'$).

	For $\rho \in \calP_U$, define $W_{\rho, \rho} = \EE_{y_{U_2} \sim \calD^{U_2}}[v_{\rho, y_{U_2}}v_{\rho, y_{U_2}}^T]$ so that $H_{Id_U, \rho, \rho} = \frac{1}{|Aut(U)|} W_{\rho, \rho}$. Observe that $W_{\rho, \rho} = \EE[v_{\rho, x_{U'}}v_{\rho, x_{U'}}^T] = \EE[v_{\rho, x_{V'}}v_{\rho, x_{V'}}^T]$ because $x_{U'}$ and $x_{V'}$ are also sets of variables sampled from $\calD$ and, $U'$, $V'$ have the same size as $U_2$ because $U_{\tau} = V_{\tau} = U$.

	For $\rho, \rho' \in \calP_U$, define $Y_{\rho, \rho'} = \EE\left[v_{\rho, x_{U'}}\left(\prod_{i \in U' \cup V'} x_i^{deg^{\tau}(i)}\right)v_{\pi(\rho), x_{V'}}^T\right]$. Then, $H_{\tau} = \sum_{\rho \in P_{\tau}} H_{\tau, \rho, \pi(\rho)} = \frac{1}{|Aut(U)|^2}S(\tau)\sum_{\rho \in P_{\tau}} Y_{\rho, \pi(\rho)}$. We have

	\begin{align*}
	\begin{bmatrix}
	\frac{S(\tau)R(\tau)}{|Aut(U)|}H_{Id_U} & H_{\tau}\\
	H_{\tau}^T & \frac{S(\tau)R(\tau)}{|Aut(U)|}H_{Id_U}
	\end{bmatrix}
	&= \frac{S(\tau)}{|Aut(U)|^2}
	\begin{bmatrix}
	R(\tau)\sum_{\rho \in \calP_U} W_{\rho, \rho} & \sum_{\rho \in P_{\tau}} Y_{\rho, \pi(\rho)}\\
	\sum_{\rho \in P_{\tau}} Y_{\rho, \pi(\rho)}^T & R(\tau)\sum_{\rho \in \calP_U} W_{\rho, \rho}
	\end{bmatrix}
	\end{align*}
	Since $\frac{S(\tau)}{|Aut(U)|^2} \ge 0$, it suffices to prove that $\begin{bmatrix}
	R(\tau)\sum_{\rho \in \calP_U} W_{\rho, \rho} & \sum_{\rho \in P_{\tau}} Y_{\rho, \pi(\rho)}\\
	\sum_{\rho \in P_{\tau}} Y_{\rho, \pi(\rho)}^T & R(\tau)\sum_{\rho \in \calP_U} W_{\rho, \rho}
	\end{bmatrix}\succeq 0$. Consider
	\begin{align*}
		\begin{bmatrix}
			R(\tau)\sum_{\rho \in \calP_U} W_{\rho, \rho} & \sum_{\rho \in P_{\tau}} Y_{\rho, \pi(\rho)}\\
			\sum_{\rho \in P_{\tau}} Y_{\rho, \pi(\rho)}^T & R(\tau)\sum_{\rho \in \calP_U} W_{\rho, \rho}
		\end{bmatrix} =& R(\tau)\begin{bmatrix}
		\sum_{\rho \in \calP_U \setminus P_{\tau}} W_{\rho, \rho} & 0\\
		0 & \sum_{\rho \in \calP_U \setminus Q_{\tau}} W_{\rho, \rho}
		\end{bmatrix}\\
		& + \begin{bmatrix}
		R(\tau)\sum_{\rho \in P_{\tau}} W_{\rho, \rho} & \sum_{\rho \in P_{\tau}} Y_{\rho, \pi(\rho)}\\
		\sum_{\rho \in P_{\tau}} Y_{\rho, \pi(\rho)}^T & R(\tau)\sum_{\rho \in P_{\tau}} W_{\pi(\rho), \pi(\rho)}
		\end{bmatrix}\\
	\end{align*}

	We have $\sum_{\rho \in \calP_U \setminus P_{\tau}} W_{\rho, \rho} = \sum_{\rho \in \calP_U \setminus P_{\tau}} \EE[v_{\rho, x_{U'}}v_{\rho, x_{U'}}^T] \succeq 0$. Similarly, $\sum_{\rho \in \calP_U \setminus Q_{\tau}} W_{\rho, \rho} \succeq 0$. Also, $R(\tau) \ge 0$ and so, the first term in the above expression,
	$R(\tau)\begin{bmatrix}
	\sum_{\rho \in \calP_U \setminus P_{\tau}} W_{\rho, \rho} & 0\\
	0 & \sum_{\rho \in \calP_U \setminus Q_{\tau}} W_{\rho, \rho}
	\end{bmatrix}$ is positive semidefinite. For the second term,
	\begin{align*}
	&\begin{bmatrix}
	R(\tau)\sum_{\rho \in P_{\tau}} W_{\rho, \rho} & \sum_{\rho \in P_{\tau}} Y_{\rho, \pi(\rho)}\\
	\sum_{\rho \in P_{\tau}} Y_{\rho, \pi(\rho)}^T & R(\tau)\sum_{\rho \in P_{\tau}} W_{\pi(\rho), \pi(\rho)}
	\end{bmatrix}\\
	&\qquad= \sum_{\rho \in P_{\tau}}
	\begin{bmatrix}
	R(\tau)\EE[v_{\rho, x_{U'}}v_{\rho, x_{U'}}^T] & \EE\left[v_{\rho, x_{U'}}\left(\prod_{i \in U' \cup V'} x_i^{deg^{\tau}(i)}\right)v_{\pi(\rho), x_{V'}}^T\right]\\
	\EE\left[v_{\rho, x_{U'}}^T\left(\prod_{i \in U' \cup V'} x_i^{deg^{\tau}(i)}\right)v_{\pi(\rho), x_{V'}}\right] & R(\tau)\EE[v_{\pi(\rho), x_{V'}}v_{\pi(\rho), x_{V'}}^T]
	\end{bmatrix}\\
	&\qquad= \sum_{\rho \in P_{\tau}}\EE
	\begin{bmatrix}
	R(\tau)v_{\rho, x_{U'}}v_{\rho, x_{U'}}^T & v_{\rho, x_{U'}}\left(\prod_{i \in U' \cup V'} x_i^{deg^{\tau}(i)}\right)v_{\pi(\rho), x_{V'}}^T\\
	v_{\rho, x_{U'}}^T\left(\prod_{i \in U' \cup V'} x_i^{deg^{\tau}(i)}\right)v_{\pi(\rho), x_{V'}} & R(\tau)v_{\pi(\rho), x_{V'}}v_{\pi(\rho), x_{V'}}^T
	\end{bmatrix}
	\end{align*}

We will prove that the term inside the expectation is positive semidefinite for each $\rho \in P_{\tau}$ and each sampling of the $x_i$ from $\calD$, which will complete the proof. Fix $\rho \in P_{\tau}$ and any sampling of the $x_i$ from $\calD$. Let $w_1 = v_{\rho, X_{U'}}, w_2 = v_{\pi(\rho), x_{V'}}$. Let $E = \prod_{i \in U' \cup V'} x_i^{deg^{\tau}(i)}$. We would like to prove that $\begin{bmatrix}
	R(\tau)w_1w_1^T & Ew_1w_2^T\\
	Ew_1^Tw_2 & R(\tau)w_2w_2^T
\end{bmatrix} \succeq 0$. For all $y$ sampled from $\calD$, $|y| \le C_{disc}\sqrt{D_E}$ and so, $|E| \le (C_{disc}\sqrt{D_E})^{\sum_{j \in U' \cup V'} deg^{\tau}(j)} = R(\tau)$.

If $E \ge 0$, then
\begin{align*}
	\begin{bmatrix}
		R(\tau)w_1w_1^T & Ew_1w_2^T\\
		Ew_1^Tw_2 & R(\tau)w_2w_2^T
	\end{bmatrix} &= (R(\tau) - E)
	\begin{bmatrix}
		w_1w_1^T & 0\\
		0 & w_2w_2^T
	\end{bmatrix}
	+ E\begin{bmatrix}
		w_1w_1^T & w_1w_2^T\\
		w_1^Tw_2 & w_2w_2^T
	\end{bmatrix}\\
	&= (R(\tau) - E)\left(
	\begin{bmatrix}
	w_1\\
	0
	\end{bmatrix}
	\begin{bmatrix}
	w_1 & 0
	\end{bmatrix} +
	\begin{bmatrix}
		0\\
		w_2
	\end{bmatrix}
	\begin{bmatrix}
		0 & w_2
	\end{bmatrix}\right) +
	E\begin{bmatrix}
	w_1\\
	w_2
	\end{bmatrix}
	\begin{bmatrix}
		w_1 & w_2
	\end{bmatrix}\\
& \succeq 0
\end{align*}
since $R(\tau) - E \ge 0$ And if $E < 0$,
\begin{align*}
	\begin{bmatrix}
		R(\tau)w_1w_1^T & Ew_1w_2^T\\
		Ew_1^Tw_2 & R(\tau)w_2w_2^T
	\end{bmatrix} &= (R(\tau) + E)
	\begin{bmatrix}
		w_1w_1^T & 0\\
		0 & w_2w_2^T
	\end{bmatrix}
	- E\begin{bmatrix}
		w_1w_1^T & -w_1w_2^T\\
		-w_1^Tw_2 & w_2w_2^T
	\end{bmatrix}\\
	&= (R(\tau) + E)\left(
	\begin{bmatrix}
		w_1\\
		0
	\end{bmatrix}
	\begin{bmatrix}
		w_1 & 0
	\end{bmatrix} +
	\begin{bmatrix}
		0\\
		w_2
	\end{bmatrix}
	\begin{bmatrix}
		0 & w_2
	\end{bmatrix}\right)
	- E\begin{bmatrix}
		w_1\\
		-w_2
	\end{bmatrix}
	\begin{bmatrix}
		w_1 & -w_2
	\end{bmatrix}\\
	& \succeq 0
\end{align*}
since $R(\tau) + E \ge 0$.
\end{proof}

\subsubsection{Qualitative intersection term bounds}

Just as in \cref{propn: spca_coeff_2}, the next proposition captures the fact that when we compose shapes $\sig, \gam, \gam^T, \sig'^T$, in order for $\lda_{\sig \circ \gam \circ \gam^T \circ \sig'^T}$ to be nonzero, the parities of the degrees of the merged vertices should add up correspondingly.
Just as in the tensor PCA application, we similarly define $H_{Id_V, \rho, \rho'}^{-\gam, \gam}$ and $H'_{\gam, \rho, \rho}$. The following propositions are simple and proved the same way.


\begin{propn}
	For all $U, V \in \calI_{mid}$ where $w(U) > w(V)$, for all $\gam \in \Gam_{U, V}$, there exists a set of parity vectors $P_{\gam} \subseteq \calP_U$ such that
	$H_{Id_V}^{-\gam, \gam} = \sum_{\rho \in P_{\gam}} H_{Id_V, \rho, \rho}^{-\gam, \gam}$.
\end{propn}



\begin{propn}
	For all $U, V \in \calI_{mid}$ where $w(U) > w(V)$, for $\gam \in \Gam_{U, V}$, $H_{\gam}' = \sum_{\rho \in P_{\gam}} H_{\gam, \rho, \rho}'$.
\end{propn}

We will now define vectors which are truncations of $v_{\rho, x_{U_2}}$. This definition and the following proposition are mostly a matter of technicality and they are essentially similar to the PSD mass condition analysis.

\begin{definition}
    Let $U, V \in \calI_{mid}$ where $w(U) > w(V)$, and let $\gam \in \Gam_{U, V}$. Let $x_i$ for $i \in U_2$ be variables. Denote them collectively as $x_{U_2}$. For $\rho \in \calP_U$, define $v_{\rho, x_{U_2}}^{-\gam}$ to be the vector indexed by left shapes $\sig \in \calL$ such that the $\sig$th entry is $v_{\rho, x_{U_2}}(\sig)$ if $|V(\sig \circ \gam)| \le D_V$ and $0$ otherwise.
\end{definition}



With this, we can decompose each slice $H_{Id_V, \rho, \rho}^{-\gam, \gam}$.

\begin{propn}\label{lem: spca_decomp}
	For any $U, V \in \calI_{mid}$ where $w(U) > w(V)$, and for any $\gam \in \Gam_{U, V}$, suppose we take $\rho \in P_{\gam}$. When we compose $\gam$ with $\gam^T$ to get $\gam \circ \gam^T$, let $U' = (U_{\gam \circ \gam^T})_2, V' = (V_{\gam \circ \gam^T})_2$ be the type $2$ vertices in $U_{\gam \circ \gam^T}, V_{\gam \circ \gam^T}$ respectively. And let $W'$ be the set of type $2$ vertices in $\gam \circ \gam^T$ that were identified in the composition when we set $V_{\gam} = U_{\gam}^T$. Let $x_i$ for $i \in U' \cup W' \cup V'$ be random variables independently sampled from $\calD$. Define $x_{U'}$ (resp. $x_{V'}, x_{W'}$) to be the subset of variables $x_i$ for $i \in U'$ (resp. $i \in V', i \in W'$). Then,
	\[H_{Id_V, \rho, \rho}^{-\gam, \gam} = \frac{1}{|Aut(V)|}S(\gam)^2 \EE_x\left[(v_{\rho, x_{U'}}^{-\gam})\left(\prod_{i \in U' \cup W' \cup V'} x_i^{deg^{\gam \circ \gam^T}(i)}\right)(v_{\rho, x_{V'}}^{-\gam})^T\right]\]
\end{propn}

\begin{proof}
	Fix $\sig, \sig' \in \calL_{U, \rho}$ such that $|V(\sig \circ \gam)|, |V(\sig' \circ \gam)| \le D_V$. Note that for $t \in \{1, 2\}$, $|\sig|_t - \frac{|V_{\sig}|_t}{2} + |\sig'|_t - \frac{|V_{\sig'}|_t}{2} + 2(|\gam|_t - \frac{|U_{\gam}|_t + |V_{\gam}|_t}{2}) = |\sig \circ \gam \circ \gam^T \circ \sig'^T|_t$. We can easily verify the equality using \cref{def: spca_coeffs} and \cref{def: discretized_gaussian}.
\end{proof}

\begin{propn}
	For any $U, V \in \calI_{mid}$ where $w(U) > w(V)$, and for any $\gam \in \Gam_{U, V}$, suppose we take $\rho \in \calP_U$. Then,
	\[H'_{\gam, \rho, \rho} = \frac{1}{|Aut(U)|}\EE_{y_{U_2}\sim \calD^{U_2}}\left[(v_{\rho, y_{U_2}}^{-\gam})(v_{\rho, y_{U_2}}^{-\gam})^T\right]\]
\end{propn}



We can finally show the qualitative intersection term bounds.

\begin{proof}[Proof of the qualitative intersection term bounds \cref{lem: spca_cond3_simplified}]
    Let $U', V', W'$ be as in \cref{lem: spca_decomp}. We have
	\begin{align*}
	\frac{|Aut(V)|}{|Aut(U)|}\cdot\frac{1}{S(\gam)^2R(\gam)^2}H_{Id_V}^{-\gam, \gam} &= \sum_{\rho \in P_{\gam}} \frac{|Aut(V)|}{|Aut(U)|}\cdot\frac{1}{S(\gam)^2R(\gam)^2} H_{Id_V, \rho, \rho}^{-\gam, \gam}\\
	&= \sum_{\rho \in P_{\gam}} \frac{1}{|Aut(U)|}\cdot\frac{1}{R(\gam)^2} \EE_x\left[(v_{\rho, x_{U'}}^{-\gam})\left(\prod_{i \in U' \cup W' \cup V'} x_i^{deg^{\gam \circ \gam^T}(i)}\right)(v_{\rho, x_{V'}}^{-\gam})^T\right]\\
\end{align*}

We will now prove that, for all $\rho \in P_{\gam}$,
\begin{align*}
    \frac{1}{|Aut(U)|}\cdot \frac{1}{R(\gam)^2} \EE_x\left[(v_{\rho, x_{U'}}^{-\gam})\left(\prod_{i \in U' \cup W' \cup V'} x_i^{deg^{\gam \circ \gam^T}(i)}\right)(v_{\rho, x_{V'}}^{-\gam})^T\right] \preceq H'_{\gam, \rho, \rho}
\end{align*}
which reduces to proving that
\begin{align*}
    \frac{2}{R(\gam)^2} \EE_x\left[(v_{\rho, x_{U'}}^{-\gam})\left(\prod_{i \in U' \cup W' \cup V'} x_i^{deg^{\gam \circ \gam^T}(i)}\right)(v_{\rho, x_{V'}}^{-\gam})^T\right] &\preceq 2\EE_{y_{U_2}\sim \calD^{U_2}}\left[(v_{\rho, y_{U_2}}^{-\gam})(v_{\rho, y_{U_2}}^{-\gam})^T\right]\\
    &= \EE_{x}\left[(v_{\rho, x_{U'}}^{-\gam})(v_{\rho, x_{U'}}^{-\gam})^T + (v_{\rho, x_{V'}}^{-\gam})(v_{\rho, x_{V'}}^{-\gam})^T\right]
\end{align*}
where the last equality followed from linearity of expectation and the fact that $U' \equiv V' \equiv U_2$.

Since $H_{Id_V, \rho, \rho}^{-\gam, \gam}$ is symmetric, we have
\[\EE_x\left[(v_{\rho, x_{U'}}^{-\gam})\left(\prod_{i \in U' \cup W' \cup V'} x_i^{deg^{\gam \circ \gam^T}(i)}\right)(v_{\rho, x_{V'}}^{-\gam})^T\right] = \EE_x\left[(v_{\rho, x_{V'}}^{-\gam})\left(\prod_{i \in U' \cup W' \cup V'} x_i^{deg^{\gam \circ \gam^T}(i)}\right)(v_{\rho, x_{U'}}^{-\gam})^T\right]\]
So, it suffices to prove
\begin{align*}
    \frac{1}{R(\gam)^2}&\EE_x\left[(v_{\rho, x_{U'}}^{-\gam})\left(\prod_{i \in U' \cup W' \cup V'} x_i^{deg^{\gam \circ \gam^T}(i)}\right)(v_{\rho, x_{V'}}^{-\gam})^T + (v_{\rho, x_{V'}}^{-\gam})\left(\prod_{i \in U' \cup W' \cup V'} x_i^{deg^{\gam \circ \gam^T}(i)}\right)(v_{\rho, x_{U'}}^{-\gam})^T\right]\\
    &\preceq \EE_{x}\left[(v_{\rho, x_{U'}}^{-\gam})(v_{\rho, x_{U'}}^{-\gam})^T + (v_{\rho, x_{V'}}^{-\gam})(v_{\rho, x_{V'}}^{-\gam})^T\right]
\end{align*}

We will prove that for every sampling of the $x_i$ from $\calD$, we have
\begin{align*}
    \frac{1}{R(\gam)^2}&\left((v_{\rho, x_{U'}}^{-\gam})\left(\prod_{i \in U' \cup W' \cup V'} x_i^{deg^{\gam \circ \gam^T}(i)}\right)(v_{\rho, x_{V'}}^{-\gam})^T + (v_{\rho, x_{V'}}^{-\gam})\left(\prod_{i \in U' \cup W' \cup V'} x_i^{deg^{\gam \circ \gam^T}(i)}\right)(v_{\rho, x_{U'}}^{-\gam})^T\right) \\
    &\preceq (v_{\rho, x_{U'}}^{-\gam})(v_{\rho, x_{U'}}^{-\gam})^T + (v_{\rho, x_{V'}}^{-\gam})(v_{\rho, x_{V'}}^{-\gam})^T
\end{align*}
Then, taking expectations will give the result. Indeed, fix a sampling of the $x_i$ from $\calD$. Let $E = \prod_{i \in U' \cup W' \cup V'} x_i^{deg^{\gam \circ \gam^T}(i)}$ and let $w_1 = v_{\rho, x_{U'}}^{-\gam}, w_2 = v_{\rho, x_{V'}}^{-\gam}$. Then, the inequality we need to show is
\[\frac{E}{R(\gam)^2}(w_1w_2^T + w_2w_1^T) \preceq w_1w_1^T + w_2w_2^T\]
Now, since $|x_i| \le C_{disc}\sqrt{D_E}$ for all $i$, we have $|E| \le \prod_{i \in U' \cup W' \cup V'} (C_{disc}\sqrt{D_E})^{deg^{\gam \circ \gam^T}(i)} =  R(\gam)^2$.
If $E \ge 0$, using $\frac{E}{R(\gam)^2}(w_1 - w_2)(w_1 - w_2)^T \succeq 0$ gives
\begin{align*}
    \frac{E}{R(\gam)^2} (w_1w_2^T + w_2w_1^T) &\preceq \frac{E}{R(\gam)^2} (w_1w_1^T + w_2w_2^T)
    \preceq w_1w_1^T + w_2w_2^T
\end{align*}
since $0 \le E \le R(\gam)^2$.
And if $E < 0$, using $\frac{-E}{R(\gam)^2}(w_1 + w_2)(w_1 + w_2)^T \succeq 0$ gives
\begin{align*}
    \frac{E}{R(\gam)^2} (w_1w_2^T + w_2w_1^T) &\preceq \frac{-E}{R(\gam)^2} (w_1w_1^T + w_2w_2^T)
    \preceq w_1w_1^T + w_2w_2^T
\end{align*}
since $0 \le -E \le R(\gam)^2$.
Finally, we use the fact that for all $\rho \in \calP_U$, we have $H'_{\gam,\rho, \rho} \succeq 0$ which can be proved the same way as the proof of \cref{lem: spca_cond1}. Therefore,
\begin{align*}
	\frac{|Aut(V)|}{|Aut(U)|}\cdot\frac{1}{S(\gam)^2R(\gam)^2}H_{Id_V}^{-\gam, \gam} &\preceq \sum_{\rho \in P_{\gam}} H'_{\gam, \rho, \rho}
	\preceq \sum_{\rho \in \calP_U} H'_{\gam, \rho, \rho}
	= H'_{\gam}
	\end{align*}
\end{proof}

\subsection{Intuition for quantitative bounds}

In this section, we will give some intuition on the bounds needed for our main theorem \cref{thm: spca_main}, which is formally proved in \cref{sec: spca_quant}. Informally, the theorem states that when $m \le \frac{d}{\lda^2}$ and $m \le \frac{k^2}{\lda^2}$, then $\Lda \succeq 0$ with high probability.

We will try and understand why the inequality $\lda_{\sig \circ \tau \circ \sig'^T}^2\norm{M_{\tau}}^2 \le \lda_{\sig \circ \sig^T}\lda_{\sig' \circ \sig'^T}$ holds. Assume for simplicity that $d < n$ and consider the shapes in \cref{fig: sparse_pca}. The assumption $d < n$ is used in this example since otherwise, if $d > n$, the decomposition differs from what's shown in the figure.

\begin{figure}[!h]
    \centering
    \includegraphics[scale=0.9, trim={2cm 5cm 0 4cm},clip]{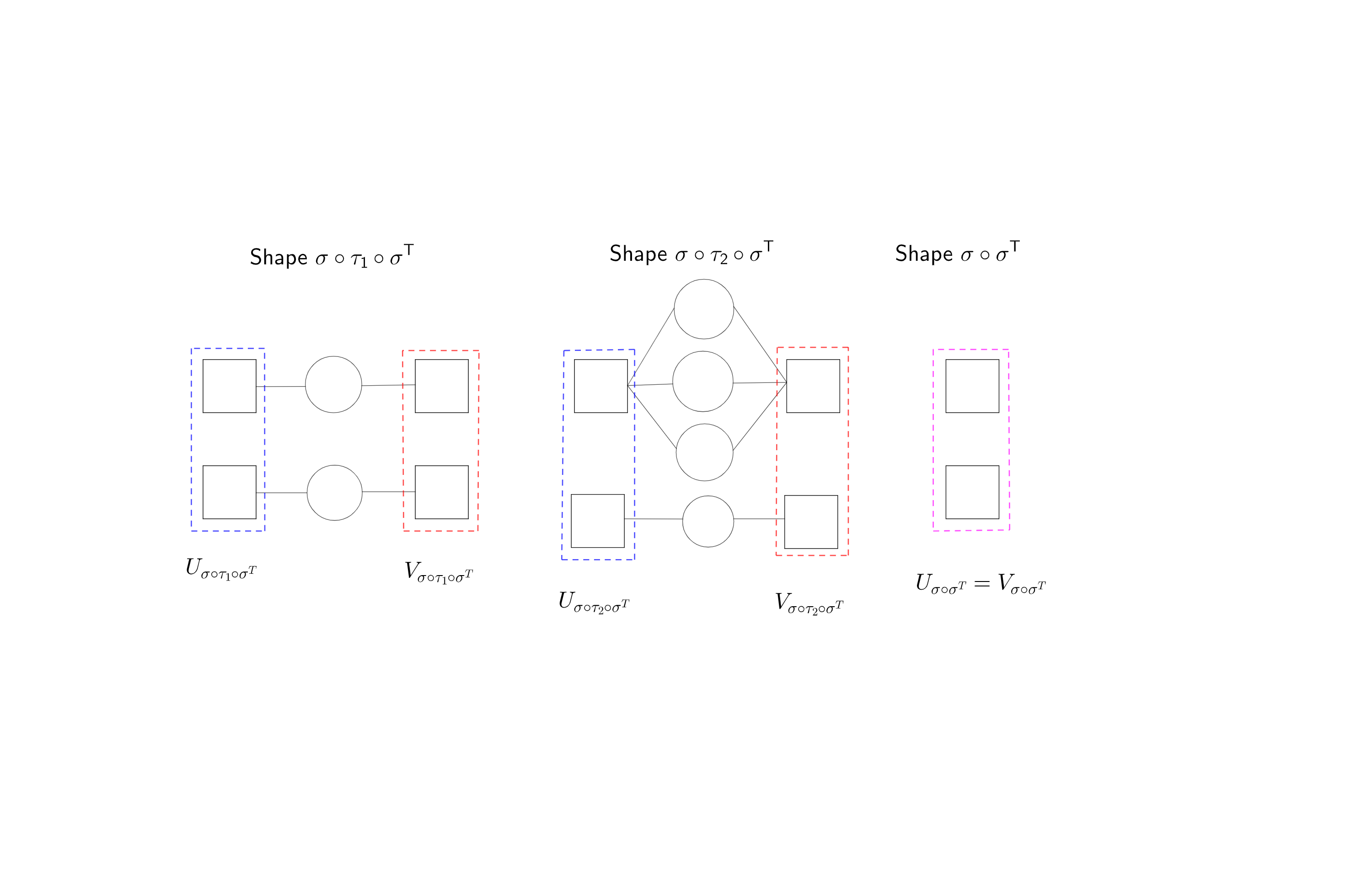}
    \caption{Shapes $\sig \circ \tau_1\circ \sig^T, \sig \circ \tau_2 \circ \sig^T$ and $\sig \circ \sig^T$. All edges have label $1$.}
    \label{fig: sparse_pca}
\end{figure}

Firstly, the shape $\sig \circ \sig^T$ has a coefficient of $\lda_{\sig \circ \sig^T} \approx \left(\frac{1}{\sqrt{k}}\right)^4\left(\frac{k}{d}\right)^2$.
The first shape $\sig \circ \tau_1 \circ \sig^T$ has a coefficient of $\lda_{\sig \circ \tau_1 \circ \sig^T} \approx \left(\frac{1}{\sqrt{k}}\right)^4\left(\frac{k}{d}\right)^4 \left(\frac{\sqrt{\lda}}{\sqrt{k}}\right)^4$ and with high probability, upto lower order terms, $\norm{M_{\tau_1}} \le md$ (these norm bounds follow from \cite{AMP20} and we will formally state them in later sections). So, the inequality $\lda_{\sig \circ \tau_1 \circ \sig^T}^2\norm{M_{\tau_1}}^2 \le \lda_{\sig \circ \sig^T}\lda_{\sig \circ \sig^T}$ rearranges to $m \le \frac{d}{\lda^2}$. But this is precisely one of the assumptions on $m$. Moreover, this also confirms that we need this assumption on $m$ in order for our strategy to go through.

The second shape $\sig \circ \tau_2 \circ \sig^T$ has a coefficient of $\lda_{\sig \circ \tau_2 \circ \sig^T} \approx \left(\frac{1}{\sqrt{k}}\right)^4\left(\frac{k}{d}\right)^4 \left(\frac{\sqrt{\lda}}{\sqrt{k}}\right)^8$ and with high probability, upto lower order terms, $\norm{M_{\tau_2}} \le m^2d$. So, the inequality $\lda_{\sig \circ \tau_2 \circ \sig^T}^2\norm{M_{\tau_2}}^2 \le \lda_{\sig \circ \sig^T}\lda_{\sig \circ \sig^T}$ rearranges to $m^2 \le \frac{k^2d}{\lda^4}$. But this is obtained simply by multiplying our assumptions on $m$, namely $m \le \frac{k^2}{\lda^2}$ and $m \le \frac{d}{\lda^2}$.

Moreover, consider a shape of the form $\sig \circ \tau_3 \circ \sig^T$ where $\tau_3$ is similar to $\tau_2$ except it has $t$ (instead of $3$) different circle vertices that are common neighbors to the top 2 square vertices. Analyzing our required inequality, we get for our strategy to go through, $m$ has to satisfy $m \le \frac{k^2}{\lda^2} \cdot \left(\frac{d}{k^2}\right)^{\frac{2}{t + 1}}$. By taking $t$ arbitrarily large, we can see that the condition $m \le \frac{k^2}{\lda^2}$ is needed.

So, we get that for our analysis to go through, the assumptions $m \le \frac{d}{\lda^2}$ and $m \le \frac{k^2}{\lda^2}$ are necessary. We will prove that in fact, these are sufficient. To do this, we use a charging argument that exploits the special structure of the shapes $\al$ that appear in our decomposition of $\Lda$ and their coefficients $\lda_{\al}$, as we obtained in \cref{def: spca_coeffs}. For details, see \cref{sec: spca_quant}.

%% file: technical_def_and_main_theorem.tex
\subsection{Section Introduction}
In this section, we make our definitions and results more precise. We also generalize our definitions and results to handle problems where one or more of the following is true:
\begin{enumerate}
	\item The input entries correspond to hyperedges rather than edges.
	\item We have different types of indices.
	\item $\Omega$ is a more complicated distribution than $\{-1,+1\}$.
	\item We have to consider matrix indices which are not multilinear.
\end{enumerate}

Throughout this section and the remainder of this manuscript, we give the reader a choice for the level of generality of this machinery. In particular, we will first recall our definition for the simpler case when our input is $\{-1,+1\}^{\binom{n}{2}}$ and we only consider multilinear indices. We will then discuss how this simpler definition generalizes. We denote these generalizations with an asterix $*$.
\subsubsection{Additional Parameters for the General Case*}
In the general case we will need a few additional parameters which we define here.
\begin{definition} \
	\begin{enumerate}
		\item We define $k$ to be the arity of the hyperedges corresponding to the input.
		\item We define $t_{max}$ to be the number of different types of indices. We define $n_i$ to be the number of possibilities for indicies of type $i$ and we define $n = \max{\{n_i: i \in [t_{max]}\}}$.
	\end{enumerate}
\end{definition}
\subsection{Indices, Input Entries, Vertices, and Edges}
Note: For this section, we use $X$ to denote the input, we use $x$ to denote entries of the input and we use $y$ to denote solution variables.
\begin{definition}[Vertices: Simplified Case]
	When the input and solution variables are indexed by one type of index which takes values in $[n]$ then we represent the index $i$ by a vertex labeled $i$.

	If we want to leave an index unspecified, we instead represent it by a vertex labeled with a variable (we will generally use $u$, $v$, or $w$ for these variables).
\end{definition}
\begin{definition}[Vertices: General Case*]
	When the input and solution variables are indexed by several types of indices where indices of type $t$ take values in $[n_t]$, we represent an index of type $t$ with value $i$ as a vertex labeled by the tuple $(t,i)$. We say that such a vertex has type $t$.

	If we want to leave an index of type $t$ unspecified, we instead represent it by a vertex labeled with a tuple $(t,?)$ where $?$ is a variable (which will generally be $u$, $v$, or $w$).
\end{definition}
\begin{definition}[Edges: Simplified Case]
	When the input is $X \in \{-1,+1\}^{\binom{n}{2}}$, we represent the entries of the input by the undirected edges $\{(i,j): i < j \in [n]\}$. Given an edge $e = (i,j)$, we take $x_e = x_{ij}$ to be the input entry corresponding to $e$.
\end{definition}
\begin{definition}[Edges: General Case*]
	In general, we represent the entries of the input by hyperedges whose form depends on nature of the input. We still take $x_e$ to be the input entry corresponding to $e$.
\end{definition}
\begin{example}
	If the input is an $n_1 \times n_2$ matrix $X$ then we will have two types of indices, one for the row and one for the column. Thus, we will have the vertices $\{(1,i): i \in [n_1]\} \cup \{(2,j): j \in [n_2]\}$. In this case, we have an edge $((1,i),(2,j))$ for each entry $x_{ij}$ of the input.
\end{example}
\begin{example}
	If the input is an $n \times n$ matrix $X$ which is not symmetric then we only need the indices $[n]$. In this case, we have a directed edge $(i,j)$ for each entry $x_{ij}$ where $i \neq j$. If the entries $x_{ii}$ are also part of the input than we also have loops $(i,i)$ for these entries.
\end{example}
\begin{example}
	If our input is a symmetric $n \times n \times n$ tensor $X$ (i.e. $x_{ijk} = x_{ikj} = x_{jik} = x_{jki} = x_{kij} = x_{kji}$) and $x_{ijk} = 0$ whenever $i,j,k$ are not distinct then we only need the indices $[n]$. In this case, we have an undirected hyperedge $e = (i,j,k)$ for each entry $x_{e} = x_{ijk}$ of the input where $i,j,k$ are distinct.
\end{example}
\begin{example}
	If the input is an $n_1 \times n_2 \times n_3$ tensor $X$ then we will have three types of indices. Thus, we will have the vertices $\{(1,i): i \in [n_1]\} \cup \{(2,j): j \in [n_2]\} \cup \{(3,k): k \in [n_3]\}$. In this case, we have a hyperedge $e = ((1,i),(2,j),(3,k))$ for each entry $x_e = x_{ijk}$ of the input.
\end{example}
\subsection{Matrix Indices and Monomials}
In this subsection, we discuss how our matrices are indexed and how we associate matrix indices with monomials. We also describe the automorphism groups of matrix indices.
\begin{definition}[Matrix Indices: Simplified Case]
	If there is only one type of index and we have the constraints $y^2_i = 1$ or $y^2_i = y_i$ on the solution variables then we define a matrix index $A$ to be a tuple of indices $(a_1,\ldots,a_{|A|})$. We make the following definitions about matrix indices:
	\begin{enumerate}
		\item We associate the monomial $\prod_{j=1}^{|A|}{y_{a_j}}$ to $A$.
		\item We define $V(A)$ to be the set of vertices $\{a_i: i \in [|A|]\}$. For brevity, we will often write $A$ instead of $V(A)$ when it is clear from context that we are referring to $A$ as a set of vertices rather than a matrix index.
		\item We take the automorphism group of $A$ to be $Aut(A) = S_{|A|}$ (the permutations of the elements of $A$)
	\end{enumerate}
\end{definition}
\begin{example}
	The matrix index $A = (4,6,1)$ represents the monomial ${y_4}{y_6}{y_1} = {y_1}{y_4}{y_6}$ and $Aut(A) = S_3$
\end{example}
\begin{remark}
	We take $A$ to be an ordered tuple rather than a set for technical reasons.
\end{remark}
In general, we need a more intricate definition for matrix indices. We start by defining matrix index pieces
\begin{definition}[Matrix Index Piece Definition*]
	We define a matrix index piece $A_i = ((a_{i1},\ldots,a_{i|A_i|}), t_i, p_i)$ to be a tuple of indices $(a_{i1},\ldots,a_{i|A_i|})$ together with a type $t_i$ and a power $p_i$. We make the following definitions about matrix index pieces:
	\begin{enumerate}
		\item We associate the monomial $p_{A_i} = \prod_{j = 1}^{|A_i|}{y^{p_i}_{{t_i}j}}$ with $A_i$.
		\item We define $V(A_i)$ to be the set of vertices $\{(t_i,a_{ij}): j \in [|A_i|]\}$.
		\item We take the automorphism group of $A_i$ to be $Aut(A_i) = S_{|A_i|}$
		\item We say that $A_i$ and $A_j$ are disjoint if $V(A_i) \cap V(A_j) = \emptyset$ (i.e. $t_i \neq t_j$ or $\{a_{i1},\ldots,a_{i|A_i|}\} \cap \{a_{j1},\ldots,a_{j|A_j|}\} = \emptyset$)
	\end{enumerate}
\end{definition}
\begin{definition}[General Matrix Index Definition*]
	We define a matrix index $A = \{A_i\}$ to be a set of disjoint matrix index pieces. We make the following definitions about matrix indices:
	\begin{enumerate}
		\item We associate the monomial $p_{A} = \prod_{A_i \in A}{p(A_i)}$ with $A$.
		\item We define $V(A)$ to be the set of vertices $\cup_{A_i \in A}{V(A_i)}$. For brevity, we will often write $A$ instead of $V(A)$ when it is clear from context that we are referring to $A$ as a set of vertices rather than a matrix index.
		\item We take the automorphism group of $A$ to be $Aut(A) = \prod_{A_i \in A}{Aut(A_i)}$
	\end{enumerate}
\end{definition}
\begin{example}[*]
	If $A_1 = ((2),1,1)$, $A_2 = ((3,1),1,2)$, and $A_3 = ((1,2,3),2,1)$ then $A = \{A_1,A_2,A_3\}$ represesents the monomial $p = {y_{12}}{y^2_{13}}{y^2_{11}}{y_{21}}{y_{22}}{y_{23}}$ and we have $Aut(A) = S_1 \times S_2 \times S_3$
\end{example}
\subsection{Fourier Characters and Ribbons}
A key idea is to analyze Fourier characters of the input.
\begin{definition}[Simplified Fourier Characters]
	If the input distribution is $\Omega = \{-1,1\}$ then given a multi-set of edges $E$, we define $\chi_{E}(X) = \prod_{e \in E}{x_e}$.
\end{definition}
\begin{example}
	If the input is a graph $G \in \{-1,1\}^{\binom{n}{2}}$ and $E$ is a set of potential edges of $G$ (with no multiple edges) then $\chi_E(G) = (-1)^{|E \setminus E(G)|}$.
\end{example}
In general, the Fourier characters are somewhat more complicated.
\begin{definition}[Orthonormal Basis for ${\Omega}$*]
	We define the polynomials $\{h_i: i \in \mathbb{Z} \cap [0,|supp(\Omega)|-1]\}$ to be the unique polynomials (which can be found through the Gram-Schmidt process) such that
	\begin{enumerate}
		\item $\forall i, E_{\Omega}[h^2_i(x)] = 1$
		\item $\forall i \neq j, E_{\Omega}[h_i(x)h_j(x)] = 0$
		\item For all $i$, the leading coefficient of $h_i(x)$ is positive.
	\end{enumerate}
\end{definition}
\begin{example}\label{example: hermite_basis}
	If $\Omega$ is the normal distribution then the polynomials $\{h_i\}$ are the Hermite polynomials with the appropriate normalization so that for all $i$, $E_{\Omega}[h^2_i(x)] = 1$. In particular, $h_0(x) = 1$, $h_1(x) = x$, $h_2(x) = \frac{x^2 - 1}{\sqrt{2!}}$, $h_3(x) = \frac{x^3 - 3x}{\sqrt{3!}}$, etc.
\end{example}
\begin{definition}[General Fourier Characters*]
	Given a multi-set of hyperedges $E$, each of which has a label $l(e) \in [|support(\Omega)|-1]$ (or $\mathbb{N}$ if $\Omega$ has infinite support), we define $\chi_E = \prod_{e \in E}{h_{l(e)}{(X_e)}}$.

	We say that such a multi-set of hyperedges $E$ is proper if it contains no duplicate hyperedges, i.e. it is a set (though the labels on the hyperedges can be arbitrary non-negative integers). Otherwise, we say that $E$ is improper.
\end{definition}
\begin{remark}
	The Fourier characters are $\{\chi_{E}: E \text{ is proper}\}$. For improper $E$, $\chi_{E}$ can be decomposed as a linear combination of $\chi_{E_j}$ where each $E_j$ is proper. We allow improper $E$ because it is sometimes more convenient to have improper $E$ in the middle of the analysis and then do this decomposition at the end.
\end{remark}
\begin{definition}[Ribbons]\label{def: ribbons}
	A ribbon $R$ is a tuple $(H_R,A_R,B_R)$ where $H_R$ is a multi-graph (*or multi-hypergraph with labeled edges in the general case) whose vertices are indices of the input and $A_R$ and $B_R$ are matrix indices such that $V(A_R) \subseteq V(H_R)$ and $V(B_R) \subseteq V(H_R)$. We make the following definitions about ribbons:
	\begin{enumerate}
		\item We define $V(R) = V(H_R)$ and $E(R) = E(H_R)$
		\item We define $\chi_R = \chi_{E(R)}$.
		\item We define $M_R$ to be the matrix such that $(M_R)_{{A_R}{B_R}} = \chi_R$ and $M_{AB} = 0$ whenever $A \neq A_R$ or $B \neq B_R$.
	\end{enumerate}
	We say that $R$ is a proper ribbon if $H_R$ contains no isolated vertices outside of $A_R \cup B_R$ and $E(R)$ is proper. If there is an isolated vertex in $(V(R) \setminus A_R) \setminus B_R$ or $E(R)$ is improper then we say that $R$ is an improper ribbon.
\end{definition}
Proper ribbons are useful because they give an orthonormal basis for the space of matrix valued functions.
\begin{definition}[Inner products of matrix functions]
	For a pair of real matrices $M_1,M_2$ of the same dimension, we write $\langle{M_1,M_2}\rangle = tr({M_1}{M_2}^T)$ (i.e. $\langle{M_1,M_2}\rangle$ is the entrywise dot product of $M_1$ and $M_2$). For a pair of matrix-valued functions $M_1, M_2$ (of the same dimensions), we define
	\[
	\langle{M_1,M_2}\rangle = E_{X}\left[\langle{M_1(X),M_2(X)}\rangle\right]
	\]
\end{definition}
\begin{proposition}
	If $R$ and $R'$ are two proper ribbons then $\langle{M_R,M_{R'}}\rangle = 1$ if $R = R'$ and is $0$ otherwise.
\end{proposition}
\subsection{Shapes}
In this subsection, we describe a basis for $S$-invariant matrix valued functions where each matrix in this basis can be described by a relatively small \emph{shape} $\alpha$. The fundamental idea behind shapes is that we keep the structure of the objects we are working with but leave the elements of the object unspecified.
\subsubsection{Simplified Index Shapes}
\begin{definition}[Simplified Index shapes]
	With our simplifying assumptions, an index shape $U$ is a tuple of unspecified indices $(u_1,\cdots,u_{|U|})$. We make the following definitions about index shapes:
	\begin{enumerate}
		\item We define $V(U)$ to be the set of vertices $\{u_i: i \in [|U|]\}$. For brevity, we will often write $U$ instead of $V(U)$ when it is clear from context that we are referring to $U$ as a set of vertices rather than an index shape.
		\item We define the weight of $U$ to be $w(U) = |U|$.
		\item We take the automorphism group of $U$ to be $Aut(U) = S_{|U|}$ (the permutations of the elements of $U$)
	\end{enumerate}
\end{definition}
\begin{definition}
	We say that a matrix index $A = (a_1,\ldots,a_{|A|})$ has index shape $U = (u_1,\ldots,u_{|U|})$ if $|U| = |A|$. Note that in this case, if we take the map $\phi: \{u_j: j \in [|U|]\} \to [n]$ where $\phi(u_j) = a_j$ then $\phi(U) = (\phi(u_1),\ldots,\phi(u_{|U|})) = (a_1,\ldots,a_{|A|}) = A$
\end{definition}
\begin{definition}
	We say that index shapes $U = (u_1,\ldots,u_{|U|})$ and $V = (v_1,\ldots,v_{|V|})$ are equivalent (which we write as $U \equiv V$) if $|U| = |V|$. If $U \equiv V$ then we can set $U = V$ by setting $v_j = u_j$ for all $j \in [|U|]$.
\end{definition}
\begin{example}
	The matrix index $A = \{4,6,1\}$ has shape $U = \{u_1,u_2,u_3\}$ which has weight $3$.
\end{example}
\subsubsection{General Index Shapes*}
In general, we define general index shapes in the same way that we defined general matrix indices (just with unspecified indices)
\begin{definition}[Index Shape Piece Definition]
	We define a index shape piece $U_i = ((u_{i1},\ldots,u_{i|U_i|}), t_i, p_i)$ to be a tuple of indices $(u_{i1},\ldots,u_{i|A_i|})$ together with a type $t_i$ and a power $p_i$. We make the following definitions about index shape pieces:
	\begin{enumerate}
		\item We define $V(U_i)$ to be the set of vertices $\{(t_i,u_{ij}): j \in [|U_i|]\}$.
		\item We define $w(U_i) = |U_i|log_{n}(n_{t_i})$
		\item We take the automorphism group of $U_i$ to be $Aut(U_i) = S_{|U_i|}$
	\end{enumerate}
\end{definition}
\begin{definition}[General Index Shape Definition]
	We define an index shape $U = \{U_i\}$ to be a set of index shape pieces such that for all $i' \neq i$, either $t_{i'} \neq t_i$ or $p_{i'} \neq p_i$. We make the following definitions about index shapes:
	\begin{enumerate}
		\item We define $V(U)$ to be the set of vertices $\cup_{U_i \in U}{V(U_i)}$. For brevity, we will often write $U$ instead of $V(U)$ when it is clear from context that we are referring to $U$ as a set of vertices rather than an index shape.
		\item We define $w(U)$ to be $w(U) = \sum_{U_i \in U}{w(U_i)}$
		\item We take the automorphism group of $U$ to be $Aut(U) = \prod_{U_i \in U}{Aut(U_i)}$
	\end{enumerate}
\end{definition}
\begin{remark}
	For technical reasons, we want to ensure that if two index shapes $U$ and $U'$ have the same weight then $U$ and $U'$ have the same number of each type of vertex. To ensure this, we add an infinitesimal perturbation to each $n_i$ if necessary.
\end{remark}
\begin{definition}
	We say that a matrix index $A$ has index shape $U$ if there is an assignment of values to the unspecified indices of $U$ which results in $A$. More precisely, we say that $A$ has index shape $U$ if there is a map $\phi: \{u_{ij}\} \to \mathbb{N}$ such that if we define $\phi(U_i)$ to be $\phi(U_i) = ((\phi(u_{i1}),\ldots,\phi(u_{i|U_i|})),t_i,p_i)$ then $\phi(U) = \{\phi(U_i)\} = \{A_i\} = A$.
\end{definition}
\begin{definition}
	If $U$ and $V$ are two index shapes, we say that $U$ is equivalent to $V$ (which we write as $U \equiv V$) if $U$ and $V$ have the same number of index shape pieces and we can order the index shape pieces of $U$ and $V$ so that writing $U = \{U_i\}$ and $V = \{V_i\}$ where $U_i = ((u_{i1},\ldots,u_{i|U_i|}), t_i, p_i)$ and $V_{i} = ((v_{i1},\ldots,v_{i|V_i|}), t'_i, p'_i)$, we have that for all $i$, $|V_i| = |U_i|$, $t'_i = t_i$, and $p'_i = p_i$. If $U \equiv V$ then we can set $U = V$ by setting $u_{ij} = v_{ij}$ for all $i$ and all $j \in [|U_i|]$.
\end{definition}
\subsubsection{Ribbon Shapes}
With these definitions, we are now ready to define shapes and the matrices associated to them.
\begin{definition}[Shapes]\label{def: shapes}
	A ribbon shape $\alpha$ (which we call a shape for brevity) is a tuple $\alpha = (H_{\alpha},U_{\alpha},V_{\alpha})$ where $H_{\alpha}$ is a multi-graph (*or multi-hypergraph with labeled edges in the general case) whose vertices are unspecified distinct indices of the input (*whose type is specified in the general case) and $U_{\alpha}$ and $V_{\alpha}$ are index shapes such that $V(U_{\alpha}) \subseteq V(H_\alpha)$ and $V(V_{\alpha}) \subseteq V(H_\alpha)$. We make the following definitions about shapes:
	\begin{enumerate}
		\item We define $V(\alpha) = V(H_{\alpha})$ (note that $V(\alpha)$ and $V_{\alpha}$ are not the same thing) and we define $E(\alpha) = E(H_{\alpha})$.
		\item We say that a shape $\alpha$ is proper if it contains no isolated vertices outside of $V(U_{\alpha}) \cup V(V_{\alpha})$, $E(\alpha)$ has no multiple edges/hyperedges and edges in $E(\al)$ do not have label $0$. If there is an isolated vertex in $V(\alpha) \setminus V(U_{\alpha}) \setminus V(V_{\alpha})$ or $E(\alpha)$ has a multiple edge/hyperedge then we say that $\alpha$ is an improper shape.
	\end{enumerate}
	Note: For brevity, we will often write $U_{\alpha}$ and $V_{\alpha}$ instead of $V(U_{\alpha})$ and $V(V_{\alpha})$ when it is clear from context that we are referring to $U_{\alpha}$ and $V_{\alpha}$ as sets of vertices rather than index shapes.
\end{definition}
\begin{definition}[Trivial shapes]
	We say that a shape $\alpha$ is trivial if $V(\alpha) = V(U_{\alpha}) = V(V_{\alpha})$ and $E(\alpha) = \emptyset$. Otherwise, we say that $\alpha$ is non-trivial.
\end{definition}
\begin{remark}
	Note that all trivial shapes can do is permute the order of the vertices in $V(U_{\alpha}) = V(V_{\alpha})$.
\end{remark}
\begin{definition}
	Informally, we say that a ribbon $R$ has shape $\alpha$ if replacing the indices in $R$ with unspecified labels results in $\alpha$. Formally, we say that $R$ has shape $\alpha$ if there is an injective mapping $\phi:V(\alpha) \to [n]$ (*or $[t_{max}] \times [n]$ in the general case) such that $\phi(\alpha) = R$, i.e. $\phi(H_{\alpha}) = H_R$, $\phi(U_{\alpha}) = A_R$, and $\phi(V_{\alpha}) = B_R$
\end{definition}
\begin{definition}
	We say that two shapes $\alpha$ and $\beta$ are equivalent (which we write as $\alpha \equiv \beta$) if they are the same up to renaming their indices. More precisely, we say that $\alpha \equiv \beta$ if there is a bijective map $\pi: V(H_\alpha) \to V(H_\beta)$ such that $\pi(H_\alpha) = H_{\beta}$, $\pi(U_{\alpha}) = U_{\beta}$, and $\pi(V_{\alpha}) = V_{\beta}$.
\end{definition}
\begin{definition}
	Given a shape $\alpha$ and matrix indices $A,B$ of shapes $U_\alpha$ and $V_\alpha$ respectively, we define $\mathcal{R}(\alpha,A,B)$ to be the set of ribbons $R$ such that $R$ \emph{has shape $\alpha$}, $A_R = A$, and $B_R = B$.
\end{definition}
\begin{definition}
	For a shape $\alpha$, we define the matrix-valued function $M_\alpha$ to have entries $M_{\alpha}(A,B)$ given by
	\[
	(M_\alpha)_{A,B}(X) = \sum_{R \in \mathcal{R}(\alpha, A,B )} \chi_R(X)
	\]
\end{definition}
For examples of $M_{\al}$, see \cite{AMP20}.
\begin{proposition}
	The $M_\alpha$'s for proper shapes $\alpha$ are an orthogonal basis for the $S$-invariant functions.\footnote{
		Because of orthogonality of the underlying Fourier characters, it is not hard to check that when $\alpha \neq \alpha'$ and $M_\alpha, M_{\alpha'}$ have the same dimensions, $\langle{M_\alpha, M_{\alpha'}}\rangle = 0$.}
\end{proposition}
\begin{remark}
	Conceptually, one may think of forming an orthonormal basis for this space with the functions $M_\alpha / \sqrt{\langle{M_\alpha, M_\alpha}\rangle}$, but for technical reasons it is easiest to work with these functions without normalizing them to $1$.
	By orthogonality and the fact that every Boolean function is a polynomial, any $S$-invariant matrix-valued function $\Lambda$ is expressible as
	\[
	\Lambda = \sum_{\alpha} \frac{\langle{\Lambda, M_\alpha}\rangle}{\langle{M_\alpha, M_\alpha}\rangle} \cdot M_\alpha
	\]
\end{remark}

In the proof of our main theorem, we encounter improper shapes. We can handle them by decomposing them into proper shapes using basic Fourier analysis. For now, we will illustrate how this can be done via an example.

\begin{example}[Reducing improper shape to proper shapes*]
Consider the case when the input distribution is Gaussian and there is only one type of vertex. Consider the improper shape $\alpha$ where $U_{\alpha} = (u_1,u_2,u_3)$, $V_{\alpha} = (v_1,v_2,v_3)$, and $V(\alpha) = U_{\alpha} \cup V_{\alpha} \cup \{w_1,w_2,w_3,w_4\}$ with edges
\begin{align*} E(\alpha) = &\{(u_1,w_1),(u_2,w_1),(u_3,w_1),(u_1,w_2),(u_2,w_2),(u_3,w_2)\} \\
&\cup \{(w_1,w_3), (w_1,w_3), (w_2,w_4), (w_2,w_4)\} \\
&\cup \{(w_3,v_1), (w_3,v_2), (w_3,v_3), (w_4,v_1), (w_4,v_2), (w_4,v_3)\}
\end{align*}
where all edges have label $1$.
$M_{\alpha}$ can be decomposed into a linear combination $M_{\alpha_1}$, $M_{\alpha_2}$, and $M_{\alpha_3}$ for the following proper shapes $\alpha_1$, $\alpha_2$, and $\alpha_3$
\begin{enumerate}
\item $\alpha_1$ is the same as $\alpha$ except that the edges $\{(w_1,w_3), (w_1,w_3), (w_2,w_4), (w_2,w_4)\}$ are replaced by $\{(w_1,w_3)_2, (w_2,w_4)_2\}$. The subscript notation means that we replaced the multiedge by a single edge with label $2$ so the edge now corresponds to the Hermite polynomial $h_2(x)$ as in \cref{example: hermite_basis}.
\item $\alpha_2$ is the same as $\alpha$ except that the edges $\{(w_1,w_3), (w_1,w_3), (w_2,w_4), (w_2,w_4)\}$ are replaced by $\{(w_1,w_3)_2\}$.
\item $\alpha_3$ is the same as $\alpha$ except that the edges $\{(w_1,w_3), (w_1,w_3), (w_2,w_4), (w_2,w_4)\}$ are deleted.
\end{enumerate}
With our definition of graph matrices
\[
M_{\alpha} = 2M_{\alpha_1} + \sqrt{2}M_{\alpha_2} + 2M_{\alpha_3}
\]
\end{example}

\subsection{Composing Ribbons and Shapes}
\begin{definition}[Composing Ribbons]
	We say that ribbons $R_1$ and $R_2$ are composable if $B_{R_1} = A_{R_2}$. Note that this definition is not symmetric so we may have that $R_1$ and $R_2$ are composable but $R_2$ and $R_1$ are not composable.

	We say that $R_1$ and $R_2$ are properly composable if we also have that $V(R_1) \cap V(R_2) = V(B_{R_1}) = V(A_{R_2})$ (there are no unexpected intersections between $R_1$ and $R_2$).

	If $R_1$ and $R_2$ are composable ribbons then we define the composition of $R_1$ and $R_2$ to be the ribbon $R_1 \circ R_2$ such that
	\begin{enumerate}
		\item $A_{R_1 \circ R_2} = A_{R_1}$ and $B_{R_1 \circ R_2} = B_{R_2}$
		\item $V(R_1 \circ R_2) = V(R_1) \cup V(R_2)$
		\item $E(R_1 \circ R_2) = E(R_1) \cup E(R_2)$ (and thus $\chi_{R_1 \circ R_2} = \chi_{R_1}\chi_{R_2}$)
	\end{enumerate}
	We say that ribbons $R_1,\ldots,R_k$ are composable/properly composable if for all $j \in [k-1]$, $R_1 \circ \ldots \circ R_j$ and $R_{j+1}$ are composable/properly composable. If $R_1,\ldots,R_k$ are composable then we define $R_1 \circ \ldots \circ R_k$ to be
	$R_1 \circ \ldots \circ R_k = (R_1 \circ \ldots \circ R_{k-1}) \circ R_k$
\end{definition}
\begin{proposition}
	Ribbon composition is associative, i.e. if $R_1,R_2,R_3$ are composable/properly composable ribbons then $R_2, R_3$ are composable/properly composable, $R_1, (R_2 \circ R_3)$ are composable/properly composable, and $R_1 \circ (R_2 \circ R_3) = (R_1 \circ R_2) \circ R_3$
\end{proposition}
\begin{proposition}
	If $R_1$ and $R_2$ are composable ribbons then $M_{R_1 \cup R_2} = M_{R_1}M_{R_2}$.
\end{proposition}
We have similar definitions for composing shapes.
\begin{definition}[Composing Shapes]
	We say that shapes $\alpha$ and $\beta$ are composable if $U_{\beta} \equiv V_{\alpha}$. Note that this definition is not symmetric so we may have that $\alpha$ and $\beta$ are composable but $\beta$ and $\alpha$ are not composable.

	If $\alpha$ and $\beta$ are composable shapes then we define the composition of $\alpha$ and $\beta$ to be the shape $\alpha \circ \beta$ such that
	\begin{enumerate}
		\item $U_{\alpha \circ \beta} = U_{\alpha}$ and $V_{\alpha \circ \beta} = V_{\beta}$
		\item After setting $U_{\beta} = V_{\alpha}$, we take $V(\alpha \circ \beta) = V(\alpha) \cup V(\beta)$
		\item $E(\alpha \circ \beta) = E(\alpha) \cup E(\beta)$
	\end{enumerate}
	We say that shapes $\alpha_1,\ldots,\alpha_k$ are composable if for all $j \in [k-1]$, $\alpha_1 \circ \ldots \circ \alpha_j$ and $\alpha_{j+1}$ are composable. If $\alpha_1,\ldots,\alpha_k$ are composable then we define the shape $\alpha_1 \circ \ldots \circ \alpha_k$ to be
	$\alpha_1 \circ \ldots \circ \alpha_k = (\alpha_1 \circ \ldots \circ \alpha_{k-1}) \circ \alpha_k$
\end{definition}
\begin{proposition}
	Shape composition is associative, i.e. if $\alpha_1,\alpha_2,\alpha_3$ are composable shapes then $\alpha_2, \alpha_3$ are composable, $\alpha_1, (\alpha_2 \circ \alpha_3)$ are composable, and $\alpha_1 \circ (\alpha_2 \circ \alpha_3) = (\alpha_1 \circ \alpha_2) \circ \alpha_3$
\end{proposition}

\begin{example}
    \cref{fig: shape_comp} illustrates an example of shape composition. We have two types of vertices that we diagrammaticaly represent by squares and circles. Observe how the shapes $\sig \circ \sig'^T$ and $\sig \circ \tau \circ \sig'^T$ are obtained from the shapes $\sig, \tau$ and $\sig'^T$.
\end{example}

\begin{figure}[!h]
    \centering
    \includegraphics[scale=0.38, trim={8cm 2cm 0 2cm},clip]{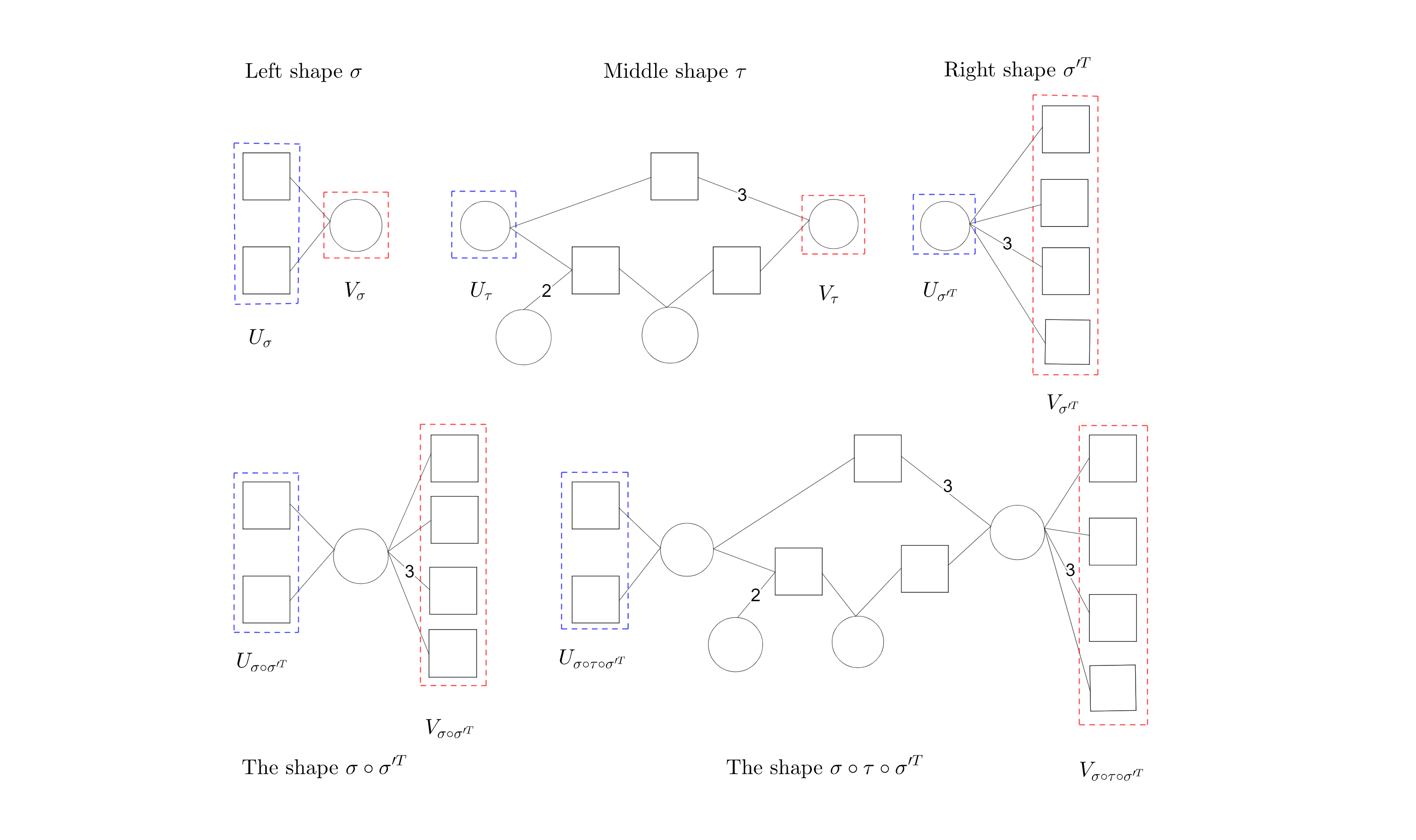}
    \caption{Illustration of shape composition and decomposition.}
    \label{fig: shape_comp}
\end{figure}

\subsection{Decomposition of Shapes into Left, Middle, and Right parts}
In this subsection, we describe how shapes can be decomposed into left, middle, and right parts based on the leftmost and rightmost \emph{minimum vertex separators}, which is a crucial idea for our analysis.
\begin{definition}[Paths]
	A \emph{path} in a shape $\alpha$ is a sequence of vertices $v_1,\ldots,v_t$ such that $v_i, v_{i+1}$ are in some edge/hyperedge together.
	A pair of paths is vertex-disjoint if the corresponding sequences of vertices are disjoint.
\end{definition}
\begin{definition}[Vertex separators]
	Let $\alpha$ be a shape and let $U$ and $V$ be sets of vertices in $\alpha$. We say that a set of vertices $S \subseteq V(\alpha)$ is a \emph{vertex separator} of $U$ and $V$ if every path in $\alpha$ from $U$ to $V$ contains at least one vertex in $S$. Note that any vertex separator $S$ of $U$ and $V$ must contain all of the vertices in $U \cap V$.

	As a special case, we say that $S$ is a vertex separator of $\alpha$ if $S$ is a vertex separator of $U_{\alpha}$ and $V_{\alpha}$
\end{definition}
We define the weight of a set of vertices $S \subseteq V(\alpha)$ in the same way that weight is defined for index shapes.
\begin{definition}[Simplified Weight]
	When there is only one type of index, the weight of a set of vertices $S \subseteq V(\alpha)$ is simply $|S|$.
\end{definition}
\begin{definition}[General Weight*]
	In general, given a set of vertices $S \subseteq V(\alpha)$, writing $S = \cup_{t}{S_t}$ where $S_t$ is the set of vertices of type $t$ in $S$, we define the weight of $S$ to be $w(S) = \sum_{t}{|S_t|log_{n}(n_t)}$
\end{definition}
\begin{remark}[*]
	Again, if necessary, we add an infinitesimal perturbation to $n_1,n_2,\ldots,n_{t_{max}}$ so that if two separators $S$ and $S'$ have the same weight then $S$ and $S'$ have the same number of each type of vertex.
\end{remark}
\begin{definition}[Leftmost and rightmost minimum vertex separators]
	The \emph{leftmost} minimum vertex separator is the vertex separator $S$ of minimum weight such that for every other minimum-weight vertex separator $S'$, $S$ is a separator of $U_\alpha$ and $S'$.
	The \emph{rightmost} minimum vertex separator is the vertex separator $T$ of minimum weight such that for every other minimum-weight vertex separator $T'$, $T$ is a separator of $T'$ and $V_{\alpha}$
\end{definition}
The work \cite{BHKKMP16} showed that under our simplifying assumptions, leftmost and rightmost minimum vertex separators are well defined. For a general proof that leftmost and rightmost minimum vertex separators are well defined, see Appendix \ref{separatorswelldefinedsection}.

We now have the following crucial idea. Every shape $\alpha$ can be decomposed into the composition of three composable shapes $\sigma,\tau,{\sigma'}^T$ based on the leftmost and rightmost minimum vertex separators $S,T$ of $\alpha$ together with orderings of $S$ and $T$.
\begin{definition}[Simplified Separators With Orderings]
	Under our simplifying assumptions, given a set of vertices $S \subseteq V(\alpha)$ and an ordering $O_S = s_1,\ldots,s_{|S|}$ of the vertices of $S$, we define the index shape $(S, O_S)$ to be $(S, O_S) = (s_1,\ldots,s_{|S|})$.
\end{definition}
\begin{definition}[General Separators With Orderings*]
	In the general case, we need to give an ordering for each type of vertex. Let $S \subseteq V(\alpha)$ be a subset of the vertices of $\alpha$ and write $S = \cup_{t}{S_t}$ where $S_t$ is the set of vertices in $S$ of type $t$. Given $O_S = \{O_t\}$ where $O_t = s_{t1},\ldots,s_{t|S_t|}$ is an ordering of the vertices of $S_t$, we define the index shape piece $(S_t, O_t)$ to be $(S_t, O_t) = ((s_{t1},\ldots,s_{t|S_t|}), t, 1)$ and we define the index shape $(S,O_S)$ to be $(S,O_S) = \{(S_t,O_t)\}$.
\end{definition}
\begin{proposition}
	The number of possible orderings $O$ for $S$ is equal to $|Aut((S,O_S))|$
\end{proposition}
\begin{definition}[Shape transposes]
	Given a shape $\alpha$, we define $\alpha^{T}$ to be the shape $\alpha$ with $U_{\alpha}$ and $V_{\alpha}$ swapped i.e. $U_{\sigma^{T}} = V_{\sigma}$ and $V_{\sigma^{T}} = U_{\sigma}$.
\end{definition}
\begin{definition}[Left, middle, and right parts]
	Let $\alpha$ be a shape. Let $S$ and $T$ be the leftmost and rightmost minimal vertex separators of $\alpha$ together with orderings $O_S,O_T$ of $S$ and $T$.
	\begin{itemize}
		\item We define the \emph{left part} $\sigma_{\alpha}$ of $\alpha$ to be the shape such that
		\begin{enumerate}
			\item $H_{\sigma_{\alpha}}$ is the induced subgraph of $H_{\alpha}$ on all of the vertices of $\alpha$ reachable from $U_{\alpha}$ without passing through $S$ (note that $H_{\sigma_{\alpha}}$ includes the vertices of $S$) except that we remove any edges/hyperedges which are contained entirely within $S$.
			\item $U_{\sigma_{\alpha}} = U_{\alpha}$ and $V_{\sigma_{\alpha}} = (S,O_S)$
		\end{enumerate}
		\item We define the \emph{right part} ${\sigma'}^T_{\alpha}$ of $\alpha$ to be the shape such that
		\begin{enumerate}
			\item $H_{{\sigma'}^T_{\alpha}}$ is the induced subgraph of $H_{\alpha}$ on all of the vertices of $\alpha$ reachable from $V_{\alpha}$ without passing through $T$ (note that $H_{{\sigma'}^T_{\alpha}}$ includes the vertices of $T$) except that we remove any edges/hyperedges which are contained entirely within $T$.
			\item $V_{{\sigma'}^T_{\alpha}} = V_{\alpha}$ and $U_{{\sigma'}^T_{\alpha}} = (T,O_T)$
		\end{enumerate}
		\item We define the \emph{middle part} $\tau_{\alpha}$ of $\alpha$ to be the shape such that
		\begin{enumerate}
			\item $H_{\tau_{\alpha}}$ is the induced subgraph of $H_{\alpha}$ on all of the vertices of $\alpha$ which are not reachable from $U_{\alpha}$ and $V_{\alpha}$ without touching $S$ and $T$ (note that $H_{\tau_{\alpha}}$ includes the vertices of $S$ and $T$). $H_{\tau_{\alpha}}$ also includes the hyperedges entirely within $S$ and the hyperedges entirely within $T$.
			\item $U_{\tau_{\alpha}} = (S,O_S)$ and $V_{\tau_{\alpha}} = (T,O_T)$
		\end{enumerate}
	\end{itemize}.
\end{definition}

\begin{example}
    \cref{fig: shape_comp} illustrates an example decomposition. We have two types of vertices that we diagrammatically represent by squares and circles. In this example, we assume that the set containing a single circle vertex has a lower weight compared to a set of two square vertices.
    \begin{enumerate}
        \item If we start with the shape $\sig \circ \sig'^T$, then it can be decomposed uniquely in to the composition of the left shape $\sig$, the right shape $\sig'^T$. In this case, the middle shape (not shown in this figure) is trivial.
        \item If we start with the shape $\sig \circ \tau \circ \sig'^T$, then it can be decomposed uniquely into the composition of the left shape $\sig$, the middle shape $\tau$ and the right shape $\sig'^T$, which are all shown in this figure.
    \end{enumerate}
\end{example}

\begin{proposition}
	If $\sigma,\tau,{\sigma'}^{T}$ are the left, middle, and rights parts for $\alpha$ for given orderings $O_S,O_T$ of $S$ and $T$ then $\alpha = \sigma \circ \tau \circ {\sigma'}^T$.
\end{proposition}
\begin{remark}
	One may ask which ordering(s) we should take of $S$ and $T$. The answer is that we will take all of the possible orderings of $S$ and $T$ simultaneously, giving equal weight to each.
\end{remark}
Based on this decomposition and the following claim, we make the following definitions for what it means for a shape to be a left, middle, or right part.
\begin{claim}[Proved in Section 6.1 in \cite{BHKKMP16}]
	\footnote{The proof in \cite{BHKKMP16} only explicitly treats the case when the shapes $\alpha$ are graphs, but the proof easily generalizes to the case when the $\alpha$ are hypergraphs.}
	\begin{itemize}
		\item Every shape $\sigma$ which is the left part of some other shape $\alpha$ has that $V_\sigma$ is its left-most and right-most minimum-weight separator.
		\item Every shape ${\sigma}^T$ which is the right part of some other shape $\alpha$ has that $U_{{\sigma}^T}$ is its left-most and right-most minimum-weight separator.
		\item Every shape $\tau$ which is the middle part of some other shape $\alpha$ has $U_\tau$ as its left-most minimum size separator and $V_{\tau}$ as its right-most minimum-weight separator.
	\end{itemize}
\end{claim}
\begin{definition}\label{leftmiddlerightshapedefinitions} \
	\begin{enumerate}
		\item We say that a shape $\sigma$ is a left shape if $\sigma$ is a proper shape, $V_{\sigma}$ is the left-most and right-most minimum-weight separator of $\sigma$, every vertex in $V(\sigma) \setminus V_{\sigma}$ is reachable from $U_{\sigma}$ without touching $V_{\sigma}$, and $\sigma$ has no hyperedges entirely within $V_{\sigma}$.
		\item We say that a shape $\tau$ is a proper middle shape if $\tau$ is a proper shape, $U_{\tau}$ is the left-most minimum-weight separator of $\tau$, and $V_{\tau}$ is the right most minimum-weight separator of $\tau$. In the analysis, we will also need to consider improper middle shapes $\tau$ which may not be proper shapes and which may have smaller separators between $U_{\tau}$ and $V_{\tau}$.
		\item We say that a shape ${\sigma}^{T}$ is a right shape if ${\sigma}^{T}$ is a proper shape, $U_{{\sigma}^{T}}$ is the left-most and right-most minimum-weight separator of ${\sigma}^{T}$, every vertex in $V({\sigma}^{T}) \setminus U_{{\sigma}^{T}}$ is reachable from $V_{{\sigma}^{T}}$ without touching $U_{{\sigma}^{T}}$, and ${\sigma}^{T}$ has no hyperedges entirely within $U_{{\sigma}^{T}}$.
	\end{enumerate}
\end{definition}
\begin{proposition}
	For all shapes $\sigma$, $\sigma$ is a left shape if and only if $\sigma^{T}$ is a right shape.
\end{proposition}
\begin{remark}
	As the reader has likely guessed, throughout this section we use $\sigma$ to denote left parts and $\tau$ to denote middle parts. Instead of having a separate letter for right parts, we express right parts as the transpose of a left part.
\end{remark}
\subsection{Coefficient matrices}\label{fullcoefficientmatrixsubsection}
We will have that $\Lambda = \sum_{\alpha}{\lambda_{\alpha}M_{\alpha}}$. To analyze $\Lambda$, it is extremely useful to express these coefficients in terms of matrices. To do this, we will need a few more definitions. We start by defining the sets of index shapes that can appear when analyzing $\Lambda$.
\begin{definition}
	Given a moment matrix $\Lambda$, we define the following sets of index shapes.
	\begin{enumerate}
		\item We define $\mathcal{I}(\Lambda) = \{U: \exists \text{ matrix index } A: A \text{ is a row index of } \Lambda, A \text{ has shape } U\}$ to be the set of index shapes which describe row and column indices of $\Lambda$.
		\item We define $w_{max}$ to be $w_{max} = \max{\{w(U):U \in \mathcal{I}(\Lambda)\}}$.
		\item With our simplifying assumptions, we define $\mathcal{I}_{mid}$ to be $\mathcal{I}_{mid} = \{U: |U| \leq w_{max}\}$
		\item[3*.] In general, we define $\mathcal{I}_{mid}$ to be $\mathcal{I}_{mid} = \{U: w(U) \leq w_{max}, \forall U_i \in U, p_i = 1\}$
	\end{enumerate}
\end{definition}
We also need to define the sets of shapes which can appear when analyzing $\Lambda$.
\begin{definition}[Truncation Parameters]
	Given a moment matrix $\Lambda = \sum_{\alpha}{\lambda_{\alpha}M_{\alpha}}$, we define $D_V,D_E$ to be the smallest natural numbers such that for all shapes $\alpha$ such that $\lambda_{\alpha} \neq 0$, decomposing $\alpha$ as $\alpha = \sigma \circ \tau \circ {\sigma'}^T$,
	\begin{enumerate}
		\item $|V(\sigma)| \leq D_V$, $|V(\tau)| \leq D_V$, and $|V(\sigma')| \leq D_V$.
		\item[2.*] For all edges $e \in E(\sigma) \cup E(\tau) \cup E(\sigma')$, $l_e \leq D_E$.
	\end{enumerate}
\end{definition}
\begin{remark}
	Under our simplifying assumptions, all edges have label $1$ so we will take $D_E = 1$ and ignore conditions involving $D_E$.
\end{remark}
\begin{definition}
	Given a moment matrix $\Lambda$, we define the following sets of shapes:
	\begin{enumerate}
		\item $\mathcal{L} = \{\sigma: \sigma \text{ is a left shape}, U_{\sigma} \in \mathcal{I}(\Lambda), V_{\sigma} \in \mathcal{I}_{mid}, |V(\sigma)| \leq D_V, \forall e \in E(\sigma), l_e \leq D_E\}$
		\item Given $V \in \mathcal{I}_{mid}$, we define $\mathcal{L}_V = \{\sigma \in \mathcal{L}: V_{\sigma} \equiv V\}$
		\item Given $U \in \mathcal{I}_{mid}$, we define $\mathcal{M}_U = \{\tau: \tau \text{ is a non-trivial proper middle shape}, U_{\tau} \equiv V_{\tau} \equiv U,
		|V(\tau)| \leq D_V, \forall e \in E(\tau), l_e \leq D_E\}$
	\end{enumerate}
\end{definition}
\begin{definition}
	Given a moment matrix $\Lambda$, we define a $\Lambda$-coefficient matrix (which we call a coefficient matrix for brevity) to be a matrix whose rows and columns are indexed by left shapes $\sigma,\sigma' \in \mathcal{L}$.

	We say that a coefficient matrix $H$ is SOS-symmetric if $H(\sigma,\sigma')$ is invariant under permuting the vertices of $U_{\sigma}$ and permuting the vertices of $U_{\sigma'}$ (*more precisely, for the general case we permute the vertices within each index shape piece of $U_{\sigma}$ and permute the vertices within each index shape piece of $U_{\sigma'}$).
\end{definition}
\begin{definition}
	Given a shape $\tau$, we say that a coefficient matrix $H$ is a $\tau$-coefficient matrix if $H(\sigma,\sigma') = 0$ whenever $V_{\sigma} \not\equiv U_{\tau}$ or $V_{\tau} \not\equiv U_{{\sigma'}^T}$.
\end{definition}
\begin{definition}
	Given an index shape $U$, we define $Id_{U}$ to be the shape with $U_{Id_{U}} = V_{Id_{U}} = U$, no other vertices, and no edges.
\end{definition}
Given a shape $\tau$ and a $\tau$-coefficient matrix $H$, we create two different matrix-valued functions, $M^{fact}_{\tau}(H)$ and $M^{orth}_{\tau}(H)$. As we will see, we can express $\Lambda$ in terms of $M^{orth}$ but to show PSDness we will need to shift to $M^{fact}$. We analyze the difference betweem $M^{fact}$ and $M^{orth}$ in subsections \ref{qualitativeintersectionpatternssection}, \ref{intersectiontermanalysissection}, and \ref{boundingdifferencesection}.
\begin{definition}
	Given a shape $\tau$ and a $\tau$-coefficient matrix $H$, define
	\[
	M^{fact}_{\tau}(H) = \sum_{\sigma \in \mathcal{L}_{U_{\tau}},\sigma' \in \mathcal{L}_{V_{\tau}}}{H(\sigma,\sigma')M_{\sigma}M_{\tau}M_{\sigma'}^T}
	\]
\end{definition}
\begin{proposition}
	For all $A$ and $B$ with shapes in $\mathcal{I}(\Lambda)$,
	\begin{align*}
		&\left(M^{fact}_{\tau}(H)\right)(A,B) = \\
		&\sum_{\sigma \in \mathcal{L}_{U_{\tau}},\sigma' \in \mathcal{L}_{V_{\tau}}}{H(\sigma,\sigma')\sum_{A',B'}{
				\sum_{R_1 \in \mathcal{R}(\sigma,A,A'), R_2 \in \mathcal{R}(\tau,A',B'), \atop R_3 \in \mathcal{R}({\sigma'}^T,B',B)}M_{R_1}(A,A')M_{R_2}(A',B')M_{R_3}(B',B)}}
	\end{align*}
\end{proposition}
If $R_1,R_2,R_3$ are properly composable then $R = R_1 \circ R_2 \circ R_3$ has the expected shape $\sigma \circ \tau \circ {\sigma'}^T$. Otherwise, $R_1 \circ R_2 \circ R_3$ will have a different shape. We define $M^{orth}_{\tau}(H)$ to be the same sum as $M^{fact}_{\tau}(H)$ except that it is restricted to properly composable ribbons $R_1,R_2,R_3$.
\begin{definition}
	We define $M^{orth}_{\tau}(H)$ so that for all $A$ and $B$ with shapes in $\mathcal{I}(\Lambda)$,
	\begin{align*}
		&\left(M^{orth}_{\tau}(H)\right)(A,B) \\
		&= \sum_{\sigma \in \mathcal{L}_{U_{\tau}},\sigma' \in \mathcal{L}_{V_{\tau}}}{H(\sigma,\sigma')\sum_{A',B'}{
				\sum_{R_1 \in \mathcal{R}(\sigma,A,A'), R_2 \in \mathcal{R}(\tau,A',B'), \atop {R_3 \in \mathcal{R}({\sigma'}^T,B',B), R_1,R_2,R_3 \text{ are properly composable}}}M_{R_1}(A,A')M_{R_2}(A',B')M_{R_3}(B',B)}} \\
		&=  \sum_{\sigma \in \mathcal{L}_{U_{\tau}},\sigma' \in \mathcal{L}_{V_{\tau}}}{H(\sigma,\sigma')\sum_{A',B'}{
				\sum_{R_1 \in \mathcal{R}(\sigma,A,A'), R_2 \in \mathcal{R}(\tau,A',B'), \atop {R_3 \in \mathcal{R}({\sigma'}^T,B',B), R_1,R_2,R_3 \text{ are properly composable}}}M_{R_1 \circ R_2 \circ R_3}(A,B)}}
	\end{align*}
\end{definition}
It would be nice if we had that $M^{orth}_{\tau}(H) = \sum_{\sigma \in \mathcal{R}_{U_{\tau}},\sigma' \in \mathcal{R}_{V_{\tau}}}{H(\sigma,\sigma')M_{\sigma \circ \tau \circ {\sigma'}^T}}$. However, this is not quite correct because there is an additional term related to automorphism groups.
\begin{definition}
	Given a shape $\alpha$, define $Aut(\alpha)$ to be the set of mappings from $\alpha$ to itself which keep $U_{\alpha}$ and $V_{\alpha}$ fixed.
\end{definition}

\begin{example}
Consider the shape $\sigma$ where $U_{\sigma} = (u_1,u_2,u_3)$, $V_{\sigma} = (v_1,v_2,v_3)$, and $V(\sigma) = U_{\sigma} \cup V_{\sigma} \cup \{w_1,w_2,w_3\}$ with edges
\begin{align*} E(\alpha) = &\{(u_1,w_1),(u_2,w_1),(u_3,w_1),(u_1,w_2),(u_2,w_2),(u_3,w_2),(u_1,w_3),(u_2,w_3),(u_3,w_3)\} \\
&\cup \{(w_1,v_1), (w_1,v_2), (w_2,v_1), (w_2,v_2),(w_3,v_1), (w_3,v_2)\}
\end{align*}
where all edges have label $1$. Then, $Aut(\sigma) = Aut(\sigma^{T}) = S_3$ and $Aut(\sigma \circ \sigma^{T}) = S_3 \times S_2 \times S_3$. Note that in this case $Aut(\sigma \circ \sigma^{T})/(Aut(\sigma) \times Aut(\sigma^{T})) = S_2$. The last computation will be useful for the definition that follows.
\end{example}

\begin{definition}
	Given composable shapes $\sigma,\tau,{\sigma'}^T$, we define
	\[
	Decomp(\sigma,\tau,\sigma') = Aut(\sigma \circ \tau \circ {\sigma'})/(Aut(\sigma) \times Aut(\tau) \times Aut({\sigma'}^T))
	\]
\end{definition}
\begin{remark}
	Each element $\pi \in Decomp(\sigma,\tau,\sigma')$ decomposes $\sigma \circ \tau \circ {\sigma'}^T$ into $\sigma$, $\tau$, and ${\sigma'}^T$ by specifying copies $\pi(\sigma)$, $\pi(\tau)$, $\pi({\sigma'}^T)$ of $\sigma$, $\tau$, and ${\sigma'}^T$ such that $\pi(\sigma) \circ \pi(\tau) \circ \pi({\sigma'}^T) = \pi(\sigma \circ \tau \circ {\sigma'}^T) = \sigma \circ \tau \circ {\sigma'}^T$. Thus, $|Decomp(\sigma,\tau,\sigma')|$ is the number of ways to decompose $\sigma \circ \tau \circ {\sigma'}^T$ into $\sigma$, $\tau$, and ${\sigma'}^T$.
\end{remark}
\begin{lemma}\label{lm:morthsimplereexpression}
	\[
	M^{orth}_{\tau}(H) = \sum_{\sigma \in \mathcal{L}_{U_{\tau}},\sigma' \in \mathcal{L}_{V_{\tau}}}{H(\sigma,\sigma')|Decomp(\sigma,\tau,{\sigma'}^T)|M_{\sigma \circ \tau \circ {\sigma'}^T}}
	\]
\end{lemma}
\begin{proof}[Proof sketch]
	Observe that there is a bijection between ribbons $R$ with shape $\sigma \circ \tau \circ {\sigma'}^T$ together with an element $\pi \in Decomp(\sigma,\tau,\sigma')$ and triples of ribbons $(R_1,R_2,R_3)$ such that
	\begin{enumerate}
		\item $R_1,R_2,R_3$ have shapes $\sigma$, $\tau$, and ${\sigma'}^T$, respectively.
		\item $V(R_1) \cap V(R_2) = A_{R_2} = B_{R_1}$,  $V(R_2) \cap V(R_3) = A_{R_3} = B_{R_2}$, and $V(R_1) \cap V(R_3) = A_{R_2} \cap B_{R_2}$
	\end{enumerate}
	To see this, note that given such ribbons $R_1,R_2,R_3$, the ribbon $R = R_1 \circ R_2 \circ R_3$ has shape $\sigma \circ \tau \circ {\sigma'}^T$ and the ribbons $R_1,R_2,R_3$ specify a decomposition of $\sigma \circ \tau \circ {\sigma'}^T$ into $\sigma$, $\tau$, and ${\sigma'}^T$.

	Conversely, given $R$ and an element $\pi \in Decomp(\sigma,\tau,\sigma')$, $\pi$ specifies how to decompose $R$ into ribbons $R_1,R_2,R_3$ of shapes $\sigma$, $\tau$, and ${\sigma'}^T$.

	For a more rigorous proof, see Appendix \ref{canonicalmapsection}.
\end{proof}
\begin{remark}
	As this lemma shows, we have to be very careful about symmetry groups in our analysis. For accuracy, it is safest to check that the coefficients for each individual ribbon match.
\end{remark}
Given a matrix-valued function $\Lambda$, we can associate coefficient matrices to $\Lambda$ as follows:
\begin{definition}
	Given a matrix-valued function $\Lambda = \sum_{\alpha: \alpha \text{ is proper}}{\lambda_{\alpha}M_{\alpha}}$,
	\begin{enumerate}
		\item For each index shape $U \in \mathcal{I}_{mid}$ and every $\sigma,\sigma' \in \mathcal{L}_{U}$, we take $H_{Id_U}(\sigma,\sigma') = \frac{1}{|Aut(U)|}\lambda_{\sigma \circ {\sigma'}^T}$
		\item For each $U \in \mathcal{I}_{mid}$, $\tau \in \mathcal{M}_U$ and $\sigma, \sigma' \in \mathcal{L}_{U}$, we take
		$H_{\tau}(\sigma,\sigma') = \frac{1}{|Aut(U_{\tau})|\cdot|Aut(V_{\tau})|}\lambda_{\sigma \circ \tau \circ {\sigma'}^T}$
	\end{enumerate}
\end{definition}
\begin{lemma}\label{lm:determiningcoefficientmatrices}
	$\Lambda = \sum_{U \in \mathcal{I}_{mid}}{M^{orth}_{Id_U}(H_{Id_U})} + \sum_{U \in \mathcal{I}_{mid}}{\sum_{\tau \in \mathcal{M}_U}{M^{orth}_{\tau}(H_{\tau})}}$
\end{lemma}
\begin{proof}
	We check that the coefficients for each individual ribbon $R$ match. There are two cases to consider.

	If $R$ has shape $\alpha$ where $\alpha$ has a unique minimum vertex separator $S$, then there is a bijection between orderings $O_S$ for $S$ and pairs of ribbons $R_1,R_2$ such that $R_1 \circ R_2 = R$ and the shapes $\sigma,{\sigma'}^T$ of $R_1,R_2$ are left and right shapes respectively.

	To see this, observe that when we concatenate $R_1$ and $R_2$, this assigns the matrix index $B_{R_1} = A_{R_2}$ to $S$, which is equivalent to specifying an ordering $O_S$ for $S$. Conversely, given an ordering $O_S$ for $S$, we take $R_1$ to be the part of $R$ between $A_R$ and $(S,O_S)$ and we take $R_2$ to be the part of $R$ between $(S,O_S)$ and $B_R$.

	From this bijection, it follows that the coefficient of $M_R$ is $\lambda_{\alpha}$ on both sides of the equation.

	Similarly, if $R$ has shape $\alpha$ where $\alpha$ does not have a unique minimal vertex separator, then there is a bijection between orderings $O_S,O_T$ for the leftmost and rightmost minimum vertex separators $S,T$ of $R$ and triples of ribbons $R_1,R_2,R_3$ such that $R_1 \circ R_2 \circ R_3 = R$ and the shapes $\sigma,\tau,{\sigma'}^T$ of $R_1,R_2,R_3$ are left, proper middle, and right shapes respectively.

	To see this, observe that when we concatenate $R_1$, $R_2$, and $R_3$, this assigns the matrix index $B_{R_1} = A_{R_2}$ to $S$ and assigns the matrix index
	$B_{R_2} = A_{R_3}$ to $T$, which is equivalent to specifying orderings $O_S,O_T$ for $S,T$. Conversely, given orderings $O_S,O_T$ for $S,T$, we take $R_1$ to be the part of $R$ between $A_R$ and $(S,O_S)$, we take $R_2$ to be the part of $R$ between $(S,O_S)$ and $(T,O_T)$, and we take $R_2$ to be the part of $R$ between $(T,O_T)$ and $B_R$.

	From this bijection, it again follows that the coefficient of $M_R$ is $\lambda_{\alpha}$ on both sides of the equation.
\end{proof}
\subsection{The $-\gamma,-\gamma$ operation and qualitative theorem statement}
In the intersection term analysis (see subsections \ref{qualitativeintersectionpatternssection}, \ref{intersectiontermanalysissection}, and \ref{boundingdifferencesection}), we will need to further decompose left shapes $\sigma$ as $\sigma = \sigma_2 \circ \gamma$ where $\sigma_2$ and $\gamma$ are themselves left shapes. Accordingly, we make the following definitions
\begin{definition} Given a moment matrix $\Lambda$, we define the following sets of left shapes:
	\begin{enumerate}
		\item $\Gamma = \{\gamma: \gamma \text{ is a non-trivial left shape}, U_{\gamma}, V_{\gamma} \in \mathcal{I}_{mid}, |V(\gamma)| \leq D_V, \forall e \in E(\gamma), l_e \leq D_E\}$
		\item Given $U,V \in \mathcal{I}_{mid}$ such that $w(U) > w(V)$, define $\Gamma_{U,V} = \{\gamma \in \Gamma: U_{\gamma} \equiv U, V_{\gamma} \equiv V\}$.
		\item Given $U \in \mathcal{I}_{mid}$, define $\Gamma_{U,*} = \{\gamma \in \Gamma: U_{\gamma} \equiv U\}$
		\item Given $V \in \mathcal{I}_{mid}$, define $\Gamma_{*,V} = \{\gamma \in \Gamma: V_{\gamma} \equiv V\}$
	\end{enumerate}
	\begin{remark}
		Under our simplifying assumptions, $\Gamma$ is the same as $\mathcal{L}$ except that $\Gamma$ excludes the trivial shapes. In general, while $\mathcal{L}$ requires that $U_{\sigma} \in \mathcal{I}(\Lambda)$, $\Gamma$ requires that $U_{\gamma} \in \mathcal{I}_{mid}$. Note that $\mathcal{I}(\Lambda)$ and $\mathcal{I}_{mid}$ may be incomparable because
		\begin{enumerate}
			\item There may be index shapes $U \in \mathcal{I}_{mid}$ such that no matrix index of $\Lambda$ has shape $U$.
			\item All index shape pieces $U_i$ for index shapes $U \in \mathcal{I}_{mid}$ must have $p_i = 1$ while this is not the case for $\mathcal{I}(\Lambda)$.
		\end{enumerate}
	\end{remark}
\end{definition}
We now state our theorem qualitatively after giving one more definition.
\begin{definition}
	Given a shape $\tau$, left shapes $\gamma \in {\Gamma}_{*,U_{\tau}}$ and $\gamma' \in {\Gamma}_{*,V_{\tau}}$, and a $\tau$-coefficient matrix $H$, define $H^{-\gamma,\gamma'}$ to be the $(\gamma \circ \tau \circ {\gamma'}^T)$-coefficient matrix with entries
	\begin{enumerate}
		\item $H^{-\gamma,\gamma'}(\sigma,\sigma') = H(\sigma \circ \gamma,\sigma' \circ \gamma')$ if $|V(\sigma \circ \gamma)| \leq D_V$ and $|V(\sigma' \circ \gamma')| \leq D_V$.
		\item $H^{-\gamma,\gamma'}(\sigma,\sigma') = 0$ if $|V(\sigma \circ \gamma)| > D_V$ or $|V(\sigma' \circ \gamma')| > D_V$.
	\end{enumerate}
\end{definition}
\begin{remark}
	For the theorem, we will only need the case when $\gamma' = \gamma$
\end{remark}
Our qualitative theorem statement is as follows:
\begin{theorem}\label{thm:mainqualitative}
	Let $\Lambda = \sum_{U \in \mathcal{I}_{mid}}{M^{orth}_{Id_U}(H_{Id_U})} + \sum_{U \in \mathcal{I}_{mid}}{\sum_{\tau \in \mathcal{M}_U}{M^{orth}_{\tau}(H_{\tau})}}$ be an SOS-symmetric matrix valued function.

	There exist functions $f(\tau)$ and $f(\gamma)$ depending on $n$ and other parameters such that if the following conditions hold:
	\begin{enumerate}
		\item For all $U \in \mathcal{I}_{mid}$,  $H_{Id_{U}} \succeq 0$
		\item For all $U \in \mathcal{I}_{mid}$ and all $\tau \in \mathcal{M}_{U}$,
		\[
		\left[ {\begin{array}{cc}
				H_{Id_{U}} & f(\tau)H_{\tau} \\
				f(\tau)H^T_{\tau} & H_{Id_{U}}
		\end{array}} \right] \succeq 0
		\]
		\item For all $U,V \in \mathcal{I}_{mid}$ where $w(U) > w(V)$ and all $\gamma \in \Gamma_{U,V}$, $H^{-\gamma,\gamma}_{Id_{V}} \preceq f(\gamma)H_{Id_{U}}$
	\end{enumerate}
	then with high probability $\Lambda \succeq 0$
\end{theorem}
\begin{remark}
	Roughly speaking, conditions 1 and 2 give us an approximate PSD decomposition for the moment matrix $M$. Condition 3 comes from the intersection term analysis, which is the most technically intensive part of the proof.
\end{remark}
\subsection{Quantitative theorem statement}\label{quantitativetheoremstatementsection}
To state our theorem quantitatively, we will need a few more things. First, the conditions of the theorem will involve functions $B_{norm}(\alpha)$, $B(\gamma)$, $N(\gamma)$, and $c(\alpha)$. Roughly speaking, these functions will be used as follows in the analysis:
\begin{enumerate}
	\item $B_{norm}(\alpha)$ will bound the norms of the matrices $M_{\alpha}$
	\item $B(\gamma)$ and $N(\gamma)$ will help us bound the intersection terms (see Section \ref{boundingdifferencesection}).
	\item $c(\alpha)$ will help us sum over the possible $\gamma$ and $\tau$.
\end{enumerate}
Second, for technical reasons it turns out that comparing $H^{-\gamma,\gamma}_{Id_{V_{\gamma}}}$ to $H_{Id_{U_{\gamma}}}$ doesn't quite work. Instead, we compare $H^{-\gamma,\gamma}_{Id_{V_{\gamma}}}$ to a matrix $H'_{\gamma}$ of our choice where $H'_{\gamma}$ is very close to $H_{Id_{U_{\gamma}}}$ ($H'_{\gamma}$ will be the same as $H_{Id_{U_{\gamma}}}$ up to truncation error).
\begin{definition}
	Given a function $B_{norm}(\alpha)$, we define the distance $d_{\tau}(H_{\tau},H'_{\tau})$ between two $\tau$-coefficient matrices $H_{\tau}$ and $H'_{\tau}$ to be
	\[
	d_{\tau}(H_{\tau},H'_{\tau}) = \sum_{\sigma \in \mathcal{L}_{U_{\tau}},\sigma' \in \mathcal{L}_{V_{\tau}}}{|H'_{\tau}(\sigma,\sigma') - H_{\tau}(\sigma,\sigma')|B_{norm}(\sigma)B_{norm}(\tau)B_{norm}(\sigma')}
	\]
\end{definition}
Third, we need an SOS-symmetric analogue of the identity matrix.
\begin{definition}
	We define $Id_{Sym}$ to be the matrix such that
	\begin{enumerate}
		\item The rows and columns of $Id_{Sym}$ are indexed by the matrix indices $A,B$ whose index shape is in $\mathcal{I}(\Lambda)$.
		\item $Id_{Sym}(A,B) = 1$ if $p_A = p_B$ and $Id_{Sym}(A, B) = 0$ if $p_A \neq p_B$.
	\end{enumerate}
\end{definition}
\begin{proposition}
	If $M$ has SOS-symmetry and the rows and columns of $Id_{Sym}$ are indexed by matrix indices $A,B$ whose index shape is in $\mathcal{I}(\Lambda)$ then $M \preceq \norm{M}Id_{Sym}$
\end{proposition}
\begin{corollary}\label{distanceboundingcorollary}
	For all $\tau$ and all SOS-symmetric $\tau$-coefficient matrices $H_{\tau}$ and $H'_{\tau}$,
	\[
	M^{fact}_{\tau}(H'_{\tau}) + M^{fact}_{{\tau}^T}(H'_{{\tau}^T}) - M^{fact}_{\tau}(H_{\tau}) - M^{fact}_{{\tau}^T}(H_{{\tau}^T}) \preceq 2d_{\tau}(H_{\tau},H'_{\tau})Id_{Sym}
	\]
	Note that if $\tau$, $H_{\tau}$ and $H'_{\tau}$ are all symmetric then
	\[
	M^{fact}_{\tau}(H'_{\tau}) - M^{fact}_{\tau}(H_{\tau}) \preceq d_{\tau}(H_{\tau},H'_{\tau})Id_{Sym}
	\]
\end{corollary}
Finally, we need a few more definitions about shapes $\alpha$.
\begin{definition}[$\mathcal{M}'$]
	We define $\mathcal{M}'$ to be the set of all shapes $\alpha$ such that
	\begin{enumerate}
		\item[1.] $|V(\alpha)| \leq 3D_V$
		\item[2.*] $\forall e \in E(\alpha), l_e \leq D_E$
		\item[3.*] All edges $e \in E(\alpha)$ have multiplicity at most $3D_V$.
	\end{enumerate}
\end{definition}
\begin{definition}[$S_{\alpha}$]
	Given a shape $\alpha$, define $S_{\alpha}$ to be the leftmost minimum vertex separator of $\alpha$
\end{definition}
\begin{definition}[$I_{\alpha}$]
	Given a shape $\alpha$, define $I_{\alpha}$ to be the set of vertices in $V(\alpha) \setminus (U_{\alpha} \cup V_{\alpha})$ which are isolated.
\end{definition}
Our main theorem will require the choice of several functions and parameters $q, B_{vertex}, B_{edge}(e), B_{norm}(\al), B(\gam), N(\gam), c(\al)$ satisfying certain conditions. $B_{edge}$ is not needed in the simplified case. For simplicity, we defer the formal conditions to the next section.

\begin{definition}[$\eps$-feasible parameters]
    For $\eps > 0$, define $q, B_{vertex}, B_{edge}(e), B_{norm}(\al), B(\gam), N(\gam), c(\al)$ to be $\eps$-feasible parameters if they satisfy the conditions in \cref{maintheoremviaproperties}.
\end{definition}

For our applications, we can work with the parameters as given by the following lemma, justified in \cref{sec: choosing_funcs}.

\begin{lemma}
    For all $\epsilon > 0$, the parameters
	\begin{enumerate}
		\item $q = 3\left\lceil{{D_V}ln(n) + \frac{ln(\frac{1}{\epsilon})}{3} + {D_V}ln(5) + 3{D^2_V}ln(2)}\right\rceil$
		\item $B_{vertex} = 6{D_V}\sqrt[4]{2eq}$
		\item $B_{norm}(\alpha) = {B_{vertex}^{|V(\alpha) \setminus U_{\alpha}| + |V(\alpha) \setminus V_{\alpha}|}}n^{\frac{w(V(\alpha)) + w(I_{\alpha}) - w(S_{\alpha})}{2}}$
		\item $B(\gamma) = B_{vertex}^{|V(\gamma) \setminus U_{\gamma}| + |V(\gamma) \setminus V_{\gamma}|}n^{\frac{w(V(\gamma) \setminus U_{\gamma})}{2}}$
		\item $N(\gamma) = (3D_V)^{2|V(\gamma) \setminus V_{\gamma}| + |V(\gamma) \setminus U_{\gamma}|}$
		\item $c(\alpha) = 100(3D_V)^{|U_{\alpha} \setminus V_{\alpha}| + |V_{\alpha} \setminus U_{\alpha}| + 2|E(\alpha)|}2^{|V(\alpha) \setminus (U_{\alpha} \cup V_{\alpha})|}$
	\end{enumerate}
	are $\eps$-feasible.
\end{lemma}

\begin{remk}\label{remk: understanding_the_parameters}
    In our applications, we show SoS lower bounds for $n^{\eps}$ degrees of SoS, where input size is $n^{O(1)}$. In this setting, we take $D_V, D_E$ to be of the order of $n^{O(\eps)}$. Therefore, for simplicity, we can interpret the parameters as
    \[q = n^{O(\eps)}, B_{vertex} = n^{O(\eps)}, B_{norm}(\alpha) =n^{O(\eps)|V(\al)|}n^{\frac{w(V(\alpha)) + w(I_{\alpha}) - w(S_{\alpha})}{2}}\]
    \[B(\gamma) = n^{O(\eps)|V(\gam)|}n^{\frac{w(V(\gamma) \setminus U_{\gamma})}{2}}, N(\gamma) = n^{O(\eps)|V(\gam)|}, c(\alpha) = n^{O(\eps)|V(\al)|}\]
\end{remk}

We can now state our main theorem.
\begin{theorem}\label{simplifiedmaintheorem}
	Given the moment matrix $\Lambda = \sum_{U \in \mathcal{I}_{mid}}{M^{orth}_{Id_U}(H_{Id_U})} + \sum_{U \in \mathcal{I}_{mid}}{\sum_{\tau \in \mathcal{M}_U}{M^{orth}_{\tau}(H_{\tau})}}$,
	for all $\epsilon > 0$, if we take $\eps$-feasible parameters, and we have SOS-symmetric coefficient matrices $\{H'_{\gamma}: \gamma \in \Gamma\}$ such that the following conditions hold:
	\begin{enumerate}
		\item \psdmass For all $U \in \mathcal{I}_{mid}$,  $H_{Id_{U}} \succeq 0$
		\item \middleshapebounds For all $U \in \mathcal{I}_{mid}$ and $\tau \in \mathcal{M}_U$,
		\[
		\left[ {\begin{array}{cc}
				\frac{1}{|Aut(U)|c(\tau)}H_{Id_{U}} & B_{norm}(\tau)H_{\tau} \\
				B_{norm}(\tau)H^T_{\tau} & \frac{1}{|Aut(U)|c(\tau)}H_{Id_{U}}
		\end{array}} \right] \succeq 0
		\]
		\item \intersectionbounds For all $U,V \in \mathcal{I}_{mid}$ where $w(U) > w(V)$ and all $\gamma \in \Gamma_{U,V}$,
		\[
		c(\gamma)^2{N(\gamma)}^2{B(\gamma)^2}H^{-\gamma,\gamma}_{Id_{V}} \preceq H'_{\gamma}
		\]
	\end{enumerate}
	then with probability at least $1 - \epsilon$,
	\[
	\Lambda \succeq \frac{1}{2}\left(\sum_{U \in \mathcal{I}_{mid}}{M^{fact}_{Id_U}{(H_{Id_U})}}\right) - 3\left(\sum_{U \in \mathcal{I}}{\sum_{\gamma \in \Gamma_{U,*}}{\frac{d_{Id_{U}}(H'_{\gamma},H_{Id_{U}})}{|Aut(U)|c(\gamma)}}}\right)Id_{sym}
	\]
    \truncationbounds If it is also true that whenever $\norm{M_{\alpha}} \leq B_{norm}(\alpha)$ for all $\alpha \in \mathcal{M}'$,
	\[
	\sum_{U \in \mathcal{I}_{mid}}{M^{fact}_{Id_U}{(H_{Id_U})}} \succeq 6\left(\sum_{U \in \mathcal{I}}{\sum_{\gamma \in \Gamma_{U,*}}{\frac{d_{Id_{U}}(H'_{\gamma},H_{Id_{U}})}{|Aut(U)|c(\gamma)}}}\right)Id_{sym}
	\]
	then with probability at least $1 - \epsilon$, $\Lambda \succeq 0$.
\end{theorem}

\subsubsection{General Main Theorem}\label{generalmaintheoremstatementsection}
Before stating the general main theorem, we need to modify a few definitions for $\alpha$ and give a few definitions for $\Omega$
\begin{definition}[$S_{\alpha,min}$ and $S_{\alpha,max}$]
	Given a shape $\alpha \in \mathcal{M}'$, define $S_{\alpha,min}$ to be the leftmost minimum vertex separator of $\alpha$ if all edges with multiplicity at least $2$ are deleted and define $S_{\alpha,max}$ to be the leftmost minimum vertex separator of $\alpha$ if all edges with multiplicity at least $2$ are present.
\end{definition}
\begin{definition}[General $I_{\alpha}$]
	Given a shape $\alpha$, define $I_{\alpha}$ to be the set of vertices in $V(\alpha) \setminus (U_{\alpha} \cup V_{\alpha})$ such that all edges incident with that vertex have multplicity at least $2$.
\end{definition}
\begin{definition}[$B_{\Omega}$]
	We take $B_{\Omega}(j)$ to be a non-decreasing function such that for all $j \in \mathbb{N}$, $E_{\Omega}[x^{j}] \leq B_{\Omega}(j)^{j}$
\end{definition}
\begin{definition}[$h^{+}_j$]
	For all $j$, we define $h^{+}_j$ to be the polynomial $h_j$ where we make all of the coefficients have positive sign.
\end{definition}
\begin{lemma}
If $\Omega = N(0,1)$ then we can take $B_{\Omega}(j) = \sqrt{j}$ and we have that
	\[
	h^{+}_j(x) \leq \frac{1}{\sqrt{j!}}(x^2 + j)^{\frac{j}{2}} \leq \left(\frac{e}{j}(x^2 + j)\right)^{\frac{j}{2}}
	\]
\end{lemma}

For a proof, see \cite[Lemma 8.15]{AMP20}.
We again give a choice of $\eps$-feasible parameters used in our applications, justified in \cref{sec: choosing_funcs}.

\begin{lemma}
    For all $\epsilon > 0$, the parameters
	\begin{enumerate}
		\item $q = \left\lceil{3{D_V}ln(n) + ln(\frac{1}{\epsilon}) + {(3D_V)^k}ln(D_E + 1) + 3{D_V}ln(5)}\right\rceil$
		\item $B_{vertex} = 6q{D_V}$
		\item $B_{edge}(e) = 2h^{+}_{l_e}(B_{\Omega}(6{D_V}D_E))
		\max_{j \in [0,3{D_V}D_E]}{\left\{\left(h^{+}_{j}(B_{\Omega}(2qj))\right)^{\frac{l_e}{\max{\{j,l_e\}}}}\right\}}$

		As a special case, if $\Omega = N(0,1)$ then we can take $B_{edge}(e) = \left(400{D^2_V}{D^2_E}q\right)^{l_e}$
		\item $B_{norm}(\alpha) =
		2e{B_{vertex}^{|V(\alpha) \setminus U_{\alpha}| + |V(\alpha) \setminus V_{\alpha}|}}\left(\prod_{e \in E(\alpha)}{B_{edge}(e)}\right)n^{\frac{w(V(\alpha)) + w(I_{\alpha}) - w(S_{\alpha,min})}{2}}$
		\item $B(\gamma) = B_{vertex}^{|V(\gamma) \setminus U_{\gamma}| + |V(\gamma) \setminus V_{\gamma}|}\left(\prod_{e \in E(\gamma)}{B_{edge}(e)}\right)n^{\frac{w(V(\gamma) \setminus U_{\gamma})}{2}}$
		\item $N(\gamma) = (3D_V)^{2|V(\gamma) \setminus V_{\gamma}| + |V(\gamma) \setminus U_{\gamma}|}$
		\item $c(\alpha) = 100(3{t_{max}}D_V)^{|U_{\alpha} \setminus V_{\alpha}| + |V_{\alpha} \setminus U_{\alpha}| + k|E(\alpha)|}(2t_{max})^{|V(\alpha) \setminus (U_{\alpha} \cup V_{\alpha})|}$
	\end{enumerate}
	are $\eps$-feasible.
\end{lemma}

Similar to \cref{remk: understanding_the_parameters}, in our applications, we can interpret the above parameters in a much simpler manner. Just as in all our applications, assume we work with the Gaussian distribution $\Omega = N(0, 1)$, $k$ is a constant and we work with SoS degree $n^{\eps}$. Then, we think of each vertex or edge of the shape $\al$ or $\gam$ essentially contributing a factor of $n^{\eps}$. Therefore, we can interpret
\[q = n^{O(\eps)}, B_{vertex} = n^{O(\eps)}, B_{edge} = n^{O(\eps)|E(\al)|}\]
\[B_{norm}(\al) = n^{O(\eps)(|V(\al)| + |E(\al)|)}n^{\frac{w(V(\alpha)) + w(I_{\alpha}) - w(S_{\alpha,min})}{2}}, B(\gamma) = n^{O(\eps)(|V(\gam)| + |E(\gam)|)}n^{\frac{w(V(\gamma) \setminus U_{\gamma})}{2}}\]
\[N(\gamma) = n^{O(\eps)|V(\gam)|}, c(\alpha) = n^{O(\eps)(|V(\al)| + |E(\al)|)}\]

\begin{theorem}\label{generalmaintheorem}
	Given the moment matrix $\Lambda = \sum_{U \in \mathcal{I}_{mid}}{M^{orth}_{Id_U}(H_{Id_U})} + \sum_{U \in \mathcal{I}_{mid}}{\sum_{\tau \in \mathcal{M}_U}{M^{orth}_{\tau}(H_{\tau})}}$,
	for all $\epsilon > 0$, if we take $\eps$-feasible parameters
	and we have SOS-symmetric coefficient matrices $\{H'_{\gamma}: \gamma \in \Gamma\}$ such that the following conditions hold:
	\begin{enumerate}
		\item \psdmass For all $U \in \mathcal{I}_{mid}$,  $H_{Id_{U}} \succeq 0$
		\item \middleshapebounds For all $U \in \mathcal{I}_{mid}$ and $\tau \in \mathcal{M}_U$,
		\[
		\left[ {\begin{array}{cc}
				\frac{1}{|Aut(U)|c(\tau)}H_{Id_{U}} & B_{norm}(\tau)H_{\tau} \\
				B_{norm}(\tau)H^T_{\tau} & \frac{1}{|Aut(U)|c(\tau)}H_{Id_{U}}
		\end{array}} \right] \succeq 0
		\]
		\item \intersectionbounds For all $U,V \in \mathcal{I}_{mid}$ where $w(U) > w(V)$ and all $\gamma \in \Gamma_{U,V}$,
		\[
		c(\gamma)^2{N(\gamma)}^2{B(\gamma)^2}H^{-\gamma,\gamma}_{Id_{V}} \preceq H'_{\gamma}
		\]
	\end{enumerate}
	then with probability at least $1 - \epsilon$,
	\[
	\Lambda \succeq \frac{1}{2}\left(\sum_{U \in \mathcal{I}_{mid}}{M^{fact}_{Id_U}{(H_{Id_U})}}\right) - 3\left(\sum_{U \in \mathcal{I}}{\sum_{\gamma \in \Gamma_{U,*}}{\frac{d_{Id_{U}}(H'_{\gamma},H_{Id_{U}})}{|Aut(U)|c(\gamma)}}}\right)Id_{sym}
	\]
	\truncationbounds If it is also true that whenever $\norm{M_{\alpha}} \leq B_{norm}(\alpha)$ for all $\alpha \in \mathcal{M}'$,
	\[
	\sum_{U \in \mathcal{I}_{mid}}{M^{fact}_{Id_U}{(H_{Id_U})}} \succeq 6\left(\sum_{U \in \calI_{mid}}{\sum_{\gamma \in \Gamma_{U,*}}{\frac{d_{Id_{U}}(H'_{\gamma},H_{Id_{U}})}{|Aut(U)|c(\gamma)}}}\right)Id_{sym}
	\]
	then with probability at least $1 - \epsilon$, $\Lambda \succeq 0$.
\end{theorem}
\subsection{Choosing $H'_{\gamma}$ and Truncation Error}\label{sec: choosing_hgamma}
A canonical choice for $H'_{\gamma}$ is to take
\begin{enumerate}
	\item $H'_{\gamma}(\sigma,\sigma') = H_{Id_U}(\sigma, \sigma')$ whenever $|V(\sigma \circ \gamma)| \leq D_V$ and $|V(\sigma' \circ \gamma)| \leq D_V$.
	\item $H'_{\gamma}(\sigma,\sigma') = 0$ whenever $|V(\sigma \circ \gamma)| > D_V$ or $|V(\sigma' \circ \gamma)| > D_V$.
\end{enumerate}
With this choice, the truncation error is
\[
d_{Id_{U_{\gamma}}}(H_{Id_{U_{\gamma}}},H'_{\gamma}) = \sum_{\sigma,\sigma' \in \mathcal{L}_{U_{\gamma}}: V(\sigma) \leq D_V, V(\sigma') \leq D_V,
	\atop |V(\sigma \circ \gamma)| > D_V \text{ or } |V(\sigma' \circ \gamma)| > D_V}{B_{norm}(\sigma)B_{norm}(\sigma')H_{Id_{U_{\gamma}}}(\sigma,\sigma')}
\]

%% file: proof_of_main.tex
In this section, we prove the main theorem under the assumption that the functions $B_{norm}(\alpha)$, $B(\gamma)$, $N(\gamma)$, and $c(\alpha)$ have certain properties. More precisely, we prove the following theorem. 
\begin{theorem}\label{maintheoremviaproperties}
	For all $\epsilon > 0$ and all $\epsilon' \in (0,\frac{1}{20}]$, for any moment matrix 
	\[
	\Lambda = \sum_{U \in \mathcal{I}_{mid}}{M^{orth}_{Id_U}(H_{Id_U})} + \sum_{U \in \mathcal{I}_{mid}}{\sum_{\tau \in \mathcal{M}_U}{M^{orth}_{\tau}(H_{\tau})}},
	\]
	if $B_{norm}(\alpha)$, $B(\gamma)$, $N(\gamma)$, and $c(\alpha)$ are functions such that 
	\begin{enumerate}
		\item With probability at least $(1-\epsilon)$, for all shapes $\alpha \in \mathcal{M}', ||M_{\alpha}|| \leq B_{norm}(\alpha)$.
		\item For all $\tau \in \mathcal{M}'$, $\gamma \in \Gamma_{*,U_{\tau}}$, $\gamma' \in \Gamma_{*,V_{\tau}}$, and all intersection patterns $P \in \mathcal{P}_{\gamma,\tau,\gamma'}$, 
		\[
		B_{norm}(\tau_{P}) \leq B(\gamma)B(\gamma')B_{norm}(\tau)
		\]
		Note: Intersection patterns and $\mathcal{P}_{\gamma,\tau,\gamma'}$ will be defined later, see Definitions \ref{intersectionpatternroughdef} and \ref{setPdefinition}.
		\item For all composable $\gamma_1,\gamma_2$, $B(\gamma_1)B(\gamma_2) = B(\gamma_1 \circ \gamma_2)$.
		\item $\forall U \in \mathcal{I}_{mid}, \sum_{\gamma \in \Gamma_{U,*}}{\frac{1}{|Aut(U)|c(\gamma)}} < \epsilon'$ 
		\item $\forall V \in \mathcal{I}_{mid}, \sum_{\gamma \in \Gamma_{*,V}}{\frac{1}{|Aut(U_{\gamma})|c(\gamma)}} < \epsilon'$ 
		\item $\forall U \in \mathcal{I}_{mid}, \sum_{\tau \in \mathcal{M}_{U}}{\frac{1}{|Aut(U)|c(\tau)}} < \epsilon'$
		\item For all $\tau \in \mathcal{M}'$, $\gamma \in \Gamma_{*,U_{\tau}} \cup \{Id_{U_{\tau}}\}$, and $\gamma' \in \Gamma_{*,V_{\tau}} \cup \{Id_{V_{\tau}}\}$,
		\begin{align*}
			&\sum_{j>0}{\sum_{\gamma_1,\gamma'_1,\cdots,\gamma_j,\gamma'_j \in \Gamma_{\gamma,\gamma',j}}{\prod_{i:\gamma_i \text{ is non-trivial}}{\frac{1}{|Aut(U_{\gamma_i})|}}
					\prod_{i:\gamma'_i \text{ is non-trivial}}{\frac{1}{|Aut(U_{\gamma'_i})|}}}}\sum_{P_1,\cdots,P_j:P_i \in \mathcal{P}_{\gamma_i,\tau_{P_{i-1}},{\gamma'_i}^T}}{\left(\prod_{i=1}^{j}{N(P_i)}\right)} \\
			&\leq \frac{N(\gamma)N(\gamma')}
			{(|Aut(U_{\gamma})|)^{1_{\gamma \text{ is non-trivial}}}(|Aut(U_{\gamma'})|)^{1_{\gamma' \text{ is non-trivial}}}}
		\end{align*}
		Note: $\Gamma_{\gamma,\gamma',j}$ will be defined later, see Definition \ref{multigammadefinition}.
	\end{enumerate}
	and we have SOS-symmetric coefficient matrices $\{H'_{\gamma}: \gamma \in \Gamma\}$ such that the following conditions hold:
	\begin{enumerate}
		\item For all $U \in \mathcal{I}_{mid}$,  $H_{Id_{U}} \succeq 0$
		\item For all $U \in \mathcal{I}_{mid}$ and $\tau \in \mathcal{M}_U$,
		\[
		\left[ {\begin{array}{cc}
				\frac{1}{|Aut(U)|c(\tau)}H_{Id_{U}} & B_{norm}(\tau)H_{\tau} \\
				B_{norm}(\tau)H^T_{\tau} & \frac{1}{|Aut(U)|c(\tau)}H_{Id_{U}}
		\end{array}} \right] \succeq 0
		\]
		\item For all $U,V \in \mathcal{I}_{mid}$ where $w(U) > w(V)$ and all $\gamma \in \Gamma_{U,V}$, 
		\[
		c(\gamma)^2{N(\gamma)}^2{B(\gamma)^2}H^{-\gamma,\gamma}_{Id_{V}} \preceq H'_{\gamma}
		\]
	\end{enumerate}
	then with probability at least $1 - \epsilon$, 
	\[
	\Lambda \succeq \frac{1}{2}\left(\sum_{U \in \mathcal{I}_{mid}}{M^{fact}_{Id_U}{(H_{Id_U})}}\right) - 3\left(\sum_{U \in \mathcal{I}_{mid}}{\sum_{\gamma \in \Gamma_{U,*}}{\frac{d_{Id_{U}}(H'_{\gamma},H_{Id_{U}})}{|Aut(U)|c(\gamma)}}}\right)Id_{sym}
	\]
	If it is also true that whenever $||M_{\alpha}|| \leq B_{norm}(\alpha)$ for all $\alpha \in \mathcal{M}'$, 
	\[
	\sum_{U \in \mathcal{I}_{mid}}{M^{fact}_{Id_U}{(H_{Id_U})}} \succeq 6\left(\sum_{U \in \mathcal{I}_{mid}}{\sum_{\gamma \in \Gamma_{U,*}}{\frac{d_{Id_{U}}(H'_{\gamma},H_{Id_{U}})}{|Aut(U)|c(\gamma)}}}\right)Id_{sym}
	\]
	then with probability at least $1 - \epsilon$, $\Lambda \succeq 0$.
\end{theorem}
Throughout this section, we assume that we have functions $B_{norm}(\alpha)$, $B(\gamma)$, $N(\gamma)$, and $c(\alpha)$. If $\forall \alpha \in \mathcal{M}', ||M_{\alpha}|| \leq B_{norm}(\alpha)$ then we say that the norm bounds hold. For the other properties of these functions, we will either restate these properties in our intermediate results to highlight where these properties are needed or just state that the conditions on these functions are satisfied for brevity.
\subsection{Warm-up: Analysis with no intersection terms}
In this subsection, we show how the analysis works if we ignore the difference between $M^{fact}$ and $M^{orth}$
\begin{theorem}\label{thm:nointersectionanalysis}
	For all $\epsilon' \in (0,\frac{1}{2}]$, if the norm bounds hold and the following conditions hold
	\begin{enumerate}
		\item For all $U \in \mathcal{I}_{mid}$,  $H_{Id_{U}} \succeq 0$
		\item For all $U \in \mathcal{I}_{mid}$ and all $\tau \in \mathcal{M}_U$
		\[
		\left[ {\begin{array}{cc}
				\frac{1}{|Aut(U)|c(\tau)}H_{Id_{U}} & B_{norm}(\tau)H_{\tau} \\
				B_{norm}(\tau)H^T_{\tau} & \frac{1}{|Aut(U)|c(\tau)}H_{Id_{U}}
		\end{array}} \right] \succeq 0
		\]
		\item $\forall U \in \mathcal{I}_{mid}, \sum_{\tau \in \mathcal{M}_U}{\frac{1}{|Aut(U)|c(\tau)}} \leq \epsilon'$.
	\end{enumerate}
	then
	\[
	\sum_{U \in \mathcal{I}_{mid}}{M^{fact}_{Id_U}(H_{Id_U})} + \sum_{U \in \mathcal{I}_{mid}}{\sum_{\tau \in \mathcal{M}_U}{M^{fact}_{\tau}(H_{\tau})} \succeq (1-2\epsilon')\sum_{U \in \mathcal{I}_{mid}}{M^{fact}_{Id_U}(H_{Id_U})}} \succeq 0
	\]
\end{theorem}
\begin{proof}
	We first show how a single term $M_{\sigma}M_{\tau}M_{{\sigma'}^T}$ plus its transpose $M_{\sigma'}M_{{\tau}^T}M_{{\sigma}^T}$ can be bounded.
	\begin{lemma}
		If the norm bounds hold then for all $\tau \in \mathcal{M}'$ and shapes $\sigma,\sigma'$ such that $\sigma,\tau,{\sigma'}^T$ are composable, for all $a,b$ such that $a > 0$, $b > 0$, and $ab = B_{norm}(\tau)^2$,
		\[
		M_{\sigma}M_{\tau}M_{{\sigma'}^T} + M_{\sigma'}M_{{\tau}^T}M_{{\sigma}^T} \preceq aM_{\sigma}M_{\sigma^T} + bM_{\sigma'}M_{{\sigma'}^T}
		\]
	\end{lemma}
	\begin{proof}
		Observe that 
		\begin{align*}
			0 \preceq &\left(\sqrt{a}M_{\sigma} - \frac{\sqrt{b}}{B_{norm}(\tau)}M_{{\sigma'}}M_{\tau^T}\right)\left(\sqrt{a}M_{\sigma} - \frac{\sqrt{b}}{B_{norm}(\tau)}M_{{\sigma'}}M_{\tau^T}\right)^T = \\
			&\left(\sqrt{a}M_{\sigma} - \frac{\sqrt{b}}{B_{norm}(\tau)}M_{{\sigma'}}M_{\tau}\right)\left(\sqrt{a}M_{{\sigma}^T} - \frac{\sqrt{b}}{B_{norm}(\tau)}M_{{\tau}}M_{{\sigma'}^T}\right) = \\
			&aM_{\sigma}M_{\sigma^T} - M_{\sigma}M_{\tau}M_{{\sigma'}^T} - M_{\sigma'}M_{{\tau}^T}M_{{\sigma}^T} + \frac{b}{B_{norm}(\tau)^2}M_{\sigma'}M_{\tau^{T}}M_{\tau}M_{{\sigma'}^T} \preceq \\
			&aM_{\sigma}M_{\sigma^T} - M_{\sigma}M_{\tau}M_{{\sigma'}^T} - M_{\sigma'}M_{{\tau}^T}M_{{\sigma}^T} + \frac{b}{B_{norm}(\tau)^2}M_{\sigma'}(B_{norm}(\tau)^2{Id})M_{{\sigma'}^T}
		\end{align*}
		Thus, $M_{\sigma}M_{\tau}M_{{\sigma'}^T} + M_{\sigma'}M_{{\tau}^T}M_{{\sigma}^T} \preceq aM_{\sigma}M_{\sigma^T} + bM_{\sigma'}M_{{\sigma'}^T}$, as needed.
	\end{proof}
	Unfortunately, if we try to bound everything term by term, there may be too many terms to bound. Instead, we generalize this argument for vectors and coefficient matrices.
	\begin{definition}
		Let $\tau$ be a shape. We say that a vector $v$ is a left $\tau$-vector if the coordinates of $v$ are indexed by left shapes $\sigma \in \mathcal{L}_{U_{\tau}}$. We say that a vector $w$ is a right $\tau$-vector if the coordinates of $w$ are indexed by left shapes $\sigma' \in \mathcal{L}_{V_{\tau}}$.
	\end{definition}
	\begin{lemma}\label{lm:rankonetosquares}
		For all $\tau \in \mathcal{M}'$, if the norm bounds hold, $v$ is a left $\tau$-vector, and $w$ is a right $\tau$-vector then 
		\[
		M^{fact}_{\tau}(vw^T) + M^{fact}_{{\tau}^T}(wv^T) \preceq B_{norm}(\tau)\left(M^{fact}_{Id_{U_{\tau}}}(vv^T) + M^{fact}_{Id_{V_{\tau}}}(ww^T)\right)
		\]
		and 
		\[
		-M^{fact}_{\tau}(vw^T) - M^{fact}_{{\tau}^T}(wv^T) \preceq B_{norm}(\tau)\left(M^{fact}_{Id_{U_{\tau}}}(vv^T) + M^{fact}_{Id_{V_{\tau}}}(ww^T)\right)
		\]
	\end{lemma}
	\begin{proof}
		Observe that 
		\begin{align*}
			0 \preceq &\left(\sum_{\sigma}{v_{\sigma}M_{\sigma} \mp \frac{w_{\sigma}M_{\sigma}M_{{\tau}^T}}{B_{norm}(\tau)}}\right)
			\left(\sum_{\sigma'}{v_{\sigma'}M_{\sigma'} \mp \frac{w_{\sigma'}M_{\sigma'}M_{{\tau}^T}}{B_{norm}(\tau)}}\right)^T = \\
			&\left(\sum_{\sigma}{v_{\sigma}M_{\sigma} \mp \frac{w_{\sigma}M_{\sigma}M_{{\tau}^T}}{B_{norm}(\tau)}}\right)\left(\sum_{\sigma'}{v_{\sigma'}M_{{\sigma'}^T} \mp \frac{w_{{\sigma'}}M_{\tau}M_{{\sigma'}^T}}{B_{norm}(\tau)}}\right) = \\
			&\sum_{\sigma,\sigma'}{\left(v_{\sigma}v_{\sigma'}\right)M_{\sigma}M_{{\sigma'}^T}} \mp \sum_{\sigma,\sigma'}{\frac{\left(v_{\sigma}w_{\sigma'}\right)}{B_{norm}(\tau)}M_{\sigma}M_{\tau}M_{\sigma'}} \\
			&\mp \sum_{\sigma,\sigma'}{\frac{\left(w_{\sigma}v_{\sigma'}\right)}{B_{norm}(\tau)}M_{\sigma}M_{{\tau}^T}M_{\sigma'}} + 
			\frac{1}{B_{norm}(\tau)^2}\sum_{\sigma,\sigma'}{\left(v_{\sigma}v_{\sigma'}\right)M_{\sigma}M_{\tau}M_{{\tau}^T}M_{{\sigma'}^T}}
		\end{align*}
		Further observe that 
		\begin{enumerate}
			\item $\sum_{\sigma,\sigma'}{\left(v_{\sigma}v_{\sigma'}\right)M_{\sigma}M_{{\sigma'}^T}} = M^{fact}_{Id_{U_{\tau}}}(vv^T)$
			\item $\sum_{\sigma,\sigma'}{\left(v_{\sigma}w_{\sigma'}\right)M_{\sigma}M_{\tau}M_{{\sigma'}^T}} = M^{fact}_{\tau}(vw^T)$
			\item $\sum_{\sigma,\sigma'}{\left(w_{\sigma}v_{\sigma'}\right)M_{\sigma}M_{{\tau}^T}M_{{\sigma'}^T}} = M^{fact}_{{\tau}^T}(wv^T)$
			\item 
			\begin{align*}
				\sum_{\sigma,\sigma'}{\left(w_{\sigma}w_{\sigma'}\right)M_{\sigma}M_{\tau}M_{{\tau}^T}M_{{\sigma'}^T}} 
				&= \left(\sum_{\sigma}{w_{\sigma}M_{\sigma}}\right)M_{\tau}M_{{\tau}^T}\left(\sum_{\sigma}{w_{\sigma}M_{\sigma}}\right)^T \\
				&\preceq \left(\sum_{\sigma}{w_{\sigma}M_{\sigma}}\right)B_{norm}(\tau)^2{Id}\left(\sum_{\sigma}{w_{\sigma}M_{\sigma}}\right)^T \\
				&= B_{norm}(\tau)^2\sum_{\sigma,\sigma'}{\left(w_{\sigma}w_{\sigma'}\right)M_{\sigma}M_{{\sigma'}^T}} \\
				&= {B_{norm}(\tau)^2}M^{fact}_{Id_{V_{\tau}}}(ww^T)
			\end{align*}
		\end{enumerate}
		Putting everything together, 
		\[
		\frac{M^{fact}_{\tau}(vw^T) + M^{fact}_{{\tau}^T}(wv^T)}{B_{norm}(\tau)} \preceq M^{fact}_{Id_{U_{\tau}}}(vv^T) + M^{fact}_{Id_{V_{\tau}}}(ww^T)
		\] 
		and 
		\[
		-\frac{M^{fact}_{\tau}(vw^T) + M^{fact}_{{\tau}^T}(wv^T)}{B_{norm}(\tau)} \preceq M^{fact}_{Id_{U_{\tau}}}(vv^T) + M^{fact}_{Id_{V_{\tau}}}(ww^T)
		\] 
		as needed.
	\end{proof}
	\begin{corollary}\label{cor:factorizedmatrixbound}
		For all $\tau \in \mathcal{M}'$, if the norm bounds hold and $H_U$ and $H_V$ are matrices such that 
		\[
		\left[ {\begin{array}{cc}
				H_{U} & B_{norm}(\tau)H_{\tau} \\
				B_{norm}(\tau)H^T_{\tau} & H_{V}
		\end{array}} \right] \succeq 0
		\]
		then $M^{fact}_{\tau}(H_{\tau}) + M^{fact}_{{\tau}^T}(H_{\tau^T}) \preceq M^{fact}_{Id_{U_{\tau}}}(H_{U}) + M^{fact}_{Id_{V_{\tau}}}(H_{V})$
	\end{corollary}
	\begin{proof}
		If $            \left[ {\begin{array}{cc}
				H_{U} & B_{norm}(\tau)H_{\tau} \\
				B_{norm}(\tau)H^T_{\tau} & H_{V}
		\end{array}} \right] \succeq 0$ then we can write 
		\[            \left[ {\begin{array}{cc}
				H_{U} & B_{norm}(\tau)H_{\tau} \\
				B_{norm}(\tau)H^T_{\tau} & H_{V}
		\end{array}} \right] = \sum_{i}{(v_i,w_i)(v_i,w_i)^T}
		\]
		Since the $M^{fact}$ operations are linear, the result now follows by summing the equation
		\[
		M^{fact}_{\tau}({v_i}w_i^T) + M^{fact}_{{\tau}^T}({w_i}v_i^T) \preceq B_{norm}(\tau)\left(M^{fact}_{Id_{U_{\tau}}}({v_i}v_i^T) + M^{fact}_{Id_{V_{\tau}}}({w_i}w_i^T)\right)
		\]
		over all $i$.
	\end{proof}
	Theorem \ref{thm:nointersectionanalysis} now follows directly. For all $U \in \mathcal{I}_{mid}$ and all $\tau \in \mathcal{M}_U$, using Corollary \ref{cor:factorizedmatrixbound} with $H_U = H_V = \frac{1}{|Aut(U)|c(\tau)}H_{Id_{U}}$,
	\[
	M^{fact}_{\tau}(H_{\tau}) + M^{fact}_{{\tau}^T}(H_{\tau^T}) \preceq \frac{1}{|Aut(U)|c(\tau)}M^{fact}_{Id_{U}}(H_{Id_{U}}) + \frac{1}{|Aut(U)|c(\tau)}M^{fact}_{Id_{U}}(H_{Id_{U}})
	\]
	Summing this equation over all $U \in \mathcal{I}_{mid}$ and all $\tau \in \mathcal{M}_U$, we obtain that 
	\[
	\sum_{U \in \mathcal{I}_{mid}}{\sum_{\tau \in \mathcal{M}_U}{M^{fact}_{\tau}(H_{\tau})}} \preceq 2\epsilon'\sum_{U \in \mathcal{I}_{mid}}{M^{fact}_{Id_U}(H_{Id_U})}
	\] as needed.
\end{proof}
\subsection{Intersection Term Analysis Strategy}\label{qualitativeintersectionpatternssection}
As we saw in the previous subsection, the analysis works out nicely if we work with $M^{fact}$. Unfortunately, our matrices are expressed in terms of $M^{orth}$. In this subsection, we describe our strategy for analyzing the difference between $M^{fact}$ and $M^{orth}$.

Recall the following expressions for $\left(M^{fact}_{\tau}(H)\right)(A,B)$ and $\left(M^{orth}_{\tau}(H)\right)(A,B)$ where $A$ has shape $U_{\tau}$ and $B$ has shape $V_{\tau}$:
\[
\left(M^{fact}_{\tau}(H)\right)(A,B) = \sum_{\sigma \in \mathcal{L}_{U_{\tau}}, \sigma' \in \mathcal{L}_{V_{\tau}}}{H(\sigma,\sigma')\sum_{A',B'}{
		\sum_{R_1 \in \mathcal{R}(\sigma,A,A'), R_2 \in \mathcal{R}(\tau,A',B'), \atop R_3 \in \mathcal{R}({\sigma'}^T,B',B)}M_{R_1}(A,A')M_{R_2}(A',B')M_{R_3}(B',B)}}
\]
\begin{align*}
	&\left(M^{orth}_{\tau}(H)\right)(A,B) \\
	&= \sum_{\sigma \in \mathcal{L}_{U_{\tau}}, \sigma' \in \mathcal{L}_{V_{\tau}}}{H(\sigma,\sigma')\sum_{A',B'}{
			\sum_{R_1 \in \mathcal{R}(\sigma,A,A'), R_2 \in \mathcal{R}(\tau,A',B'), \atop {R_3 \in \mathcal{R}({\sigma'}^T,B',B), R_1,R_2,R_3 \text{ are properly composable}}}M_{R_1}(A,A')M_{R_2}(A',B')M_{R_3}(B',B)}} 
\end{align*}
This implies that $\left(M^{fact}_{\tau}(H)\right)(A,B) - \left(M^{orth}_{\tau}(H)\right)(A,B)$ is equal to
\[
\sum_{\sigma \in \mathcal{L}_{U_{\tau}}, \sigma' \in \mathcal{L}_{V_{\tau}}}{H(\sigma,\sigma')\sum_{A',B'}{
		\sum_{R_1 \in \mathcal{R}(\sigma,A,A'), R_2 \in \mathcal{R}(\tau,A',B'), \text{ and } R_3 \in \mathcal{R}({\sigma'}^T,B',B) \atop R_1,R_2,R_3 \text{ are not properly composable}}M_{R_1}(A,A')M_{R_2}(A',B')M_{R_3}(B',B)}} 
\]
Thus, to understand the difference between $M^{fact}$ and $M^{orth}$, we need to analyze the terms $\chi_{R_1}\chi_{R_2}\chi_{R_3} = \chi_{R_1 \circ R_2 \circ R_3}$ for ribbons $R_1,R_2,R_3$ which are composable but not properly composable. These terms, which we call intersection terms, are not negligible and must be analyzed carefully. In particular, we decompose each resulting ribbon $R = R_1 \circ R_2 \circ R_3$ into new left, middle, and right parts. We do this as follows:
\begin{enumerate}
	\item Let $V_{*}$ be the set of vertices which appear more than once in $V(R_1 \circ R_2 \circ R_3)$. In other words, $V_{*}$ is the set of vertices involved in the intersections between $R_1$, $R_2$, and $R_3$ (not counting the facts that $B_{R_1} = A_{R_2}$ and $B_{R_2} = A_{R_3}$ because we expect these intersections).
	\item Let $A'$ be the leftmost minimum vertex separator of $A_{R_1}$ and $B_{R_1} \cup V_{*}$ in $R_1$. We turn $A'$ into a matrix index by specifying an ordering $O_{A'}$ for the vertices in $A'$.
	\item Let $B'$ be the leftmost minimum vertex separator of $A_{R_3} \cup V_{*}$ and $B_{R_3}$ in $R_2$. We turn $B'$ into a matrix index by specifying an ordering $O_{B'}$ for the vertices in $B'$.
	\item Decompose $R_1$ as $R_1 = {R'}_1 \cup R_4$ where ${R'}_1$ is the part of $R_1$ between $A_{R_1}$ and $A'$ and $R_4$ is the part of $R_1$ between $B'$ and $B_{R_1} = A_{R_2}$. Similarly, decompose 
	$R_3$ as $R_3 = R_5 \cup {R'}_3 $ where $R_5$ is the part of $R_3$ between $B_{R_1} = A_{R_2}$ and $B'$ and ${R'}_3$ is the part of $R_3$ between $B'$ and $B_{R_3}$.
	\item Take $R'_2 = R_4 \circ R_2 \circ R_5$ and note that $R'_1 \circ R'_2 \circ R'_3 = R_1 \circ R_2 \circ R_3$. We view $R'_1,R'_2,R'_3$ as the left, middle, and right parts of $R = R_1 \circ R_2 \circ R_3$
\end{enumerate}
While we will verify our analysis by checking the coefficients of the ribbons, we want to express everything in terms of shapes. We use the following conventions for the names of the shapes:
\begin{enumerate}
	\item As usual, we let $\sigma$, $\tau$, and ${\sigma'}^T$ be the shapes of $R_1$, $R_2$, and $R_3$.
	\item We let $\gamma$ and ${\gamma'}^T$ be the shapes of $R_4$ and $R_5$.
	\item We let $\sigma_2$, $\tau_{P}$, and ${\sigma'_2}^T$ be the shapes of $R'_1$, $R'_2$, and $R'_3$. Here $P$ is the intersection pattern induced by $R_4$, $R_2$, and $R_5$ which we define in the next subsection.
\end{enumerate}
\begin{remark}
	A key feature of our analysis is that it will work the same way regardless of the shapes $\sigma_2,{\sigma'_2}^T$ of $R'_1$ and $R'_3$. In other words, if we replace $\sigma_2$ by $\sigma_{2a}$ and $\sigma'_2$ by $\sigma'_{2a}$ for a given intersection term, this just replaces $\sigma = \sigma_2 \cup \gamma$ with  $\sigma_{a} = \sigma_{2a} \cup \gamma$ and $\sigma' = \sigma'_2 \cup \gamma'$ with  $\sigma'_{a} = \sigma'_{2a} \cup \gamma'$. This allows us to focus on the shapes $\gamma$, $\tau$, and ${\gamma'}^T$ and is the reason why the $-\gamma,\gamma$ operation appears in our results.
\end{remark}
\subsection{Intersection Term Analysis}\label{intersectiontermanalysissection}
In this section, we implement our strategy for analyzing intersection terms. For simplicity, we only give rough definitions and proof sketches here. For a more rigorous treatment, see Apendix \ref{canonicalmapsection}.

We begin by defining intersection patterns which describe how the ribbons $R_1$, $R_2$, and $R_3$ intersect. 
\begin{definition}[Rough Definition of Intersection Patterns]\label{intersectionpatternroughdef}
	Given $\tau \in \mathcal{M}'$, $\gamma \in \Gamma_{*,U_{\tau}} \cup \{Id_{U_{\tau}}\}$, $\gamma' \in \Gamma_{*,V_{\tau}} \cup \{Id_{V_{\tau}}\}$, and ribbons $R_1$, $R_2$, and $R_3$ of shapes $\gamma$, $\tau$, and ${\gamma'}^T$ which are composable but not properly composable, we define the intersection pattern $P$ induced by $R_1$, $R_2$, and $R_3$ and the resulting shape $\tau_P$ as follows:
	\begin{enumerate}
		\item We take $V(P) = V(\gamma \circ \tau \circ {\gamma'}^T)$.
		\item We take $E(P)$ to be the set of edges $(u,v)$ such that $u,v$ are distinct vertices in $V(\sigma \circ \tau \circ {\sigma'}^T)$ but $u$ and $v$ correspond to the same vertex in $R_1 \circ R_2 \circ R_3$
		\item We define $\tau_{P}$ to be the shape of the ribbon $R = R_1 \circ R_2 \circ R_3$
	\end{enumerate}
\end{definition}
\begin{definition}\label{setPdefinition}
	Given $\tau \in \mathcal{M}'$, $\gamma \in \Gamma_{*,U_{\tau}} \cup \{Id_{U_{\tau}}\}$, and $\gamma' \in \Gamma_{*,V_{\tau}} \cup \{Id_{V_{\tau}}\}$, we define $\mathcal{P}_{\gamma,\tau,{\gamma'}^T}$ to be the set of all possible intersection patterns $P$ which can be induced by ribbons $R_1$, $R_2$, and $R_3$ of shapes $\gamma$, $\tau$, and ${\gamma'}^T$.
\end{definition}
\begin{remark}
	Note that if $\gamma = Id_{U_{\tau}}$ and $\gamma' = Id_{V_{\tau}}$ then $\mathcal{P}_{\gamma,\tau,{\gamma'}^T} = \emptyset$ as every intersection pattern must have an unexpected intersection so either $\gamma$ or $\gamma'$ must be non-trivial.
\end{remark}
It would be nice if the intersection pattern $P$ together with the ribbon $R$ allowed us to recover the original ribbons $R_1$, $R_2$, and $R_3$. Unfortunately, it is possible for different triples of ribbons to result in the same intersection pattern $P$ and ribbon $R$. That said, the number of such triples cannot be too large, and this is sufficient for our purposes.
\begin{definition}
	Given an intersection pattern $P \in \mathcal{P}_{\gamma,\tau,{\gamma'}^T}$, let $R$ be a ribbon of shape $\tau_{P}$. We define $N(P)$ to be the number of different triples of ribbons $R_1,R_2,R_3$ such that $R_1 \circ R_2 \circ R_3 = R$ and $R_1,R_2,R_3$ induce the intersection pattern $P$.
\end{definition}
\begin{lemma}
	For all intersection patterns $P \in \mathcal{P}_{\gamma,\tau,{\gamma'}^T}$, $N(P) \leq |V(\tau_{P})|^{|V(\gamma) \setminus U_{\gamma}| + |V(\gamma') \setminus U_{\gamma'}|}$
\end{lemma}
\begin{proof}[Proof sketch]
	This can be proved by making the following observations:
	\begin{enumerate}
		\item $A_{R_1} = A_{R}$ and $B_{R_3} = B_{R}$.
		\item All of the remining vertices in $V(R_1)$ and $V(R_3)$ must be equal to some vertex in $V(R)$.
		\item Once $R_1$ and $R_3$ are determined, there is at most one ribbon $R_2$ such that $R_1,R_2,R_3$ are composable, $R = R_1 \circ R_2 \circ R_3$, and $R_1,R_2,R_3$ induce the intersection pattern $P$.
	\end{enumerate}
\end{proof}
With these definitions, we can now analyze the intersection terms.
\begin{definition}
	Given a left shape $\sigma$, define $e_{\sigma}$ to be the vector which has a $1$ in coordinate $\sigma$ and has a $0$ in all other coordinates.
\end{definition}
\begin{lemma}\label{lm:singleshapeintersections}
	For all $\tau \in \mathcal{M}'$, $\sigma \in \mathcal{L}_{U_{\tau}}$, and $\sigma' \in \mathcal{L}_{V_{\tau}}$, 
	\begin{align*}
		&M^{fact}_{\tau}(e_{\sigma}e^T_{\sigma'}) - M^{orth}_{\tau}(e_{\sigma}e^T_{\sigma'}) = \sum_{\sigma_2 \in \mathcal{L}, \gamma \in \Gamma: \sigma_2 \circ \gamma = \sigma}{\frac{1}{|Aut(U_{\gamma})|}\sum_{P \in \mathcal{P}_{\gamma,\tau,Id_{V_{\tau}}}}N(P)M^{orth}_{\tau_P}(e_{\sigma_2}e^T_{\sigma'})} \\
		&+ \sum_{\sigma'_2 \in \mathcal{L}, \gamma' \in \Gamma: \sigma'_2 \circ \gamma' = \sigma'}{\frac{1}{|Aut(U_{\gamma'})|}\sum_{P \in \mathcal{P}_{Id_{U_{\tau}},\tau,{\gamma'}^T}}N(P)M^{orth}_{\tau_P}(e_{\sigma}e^T_{\sigma'_2})} \\
		&+ \sum_{\sigma_2 \in \mathcal{L}, \gamma \in \Gamma: \sigma_2 \circ \gamma = \sigma}{\sum_{\sigma'_2 \in \mathcal{L}, \gamma' \in \Gamma: \sigma'_2 \circ \gamma' = \sigma'}{
				\frac{1}{|Aut(U_{\gamma})|\cdot|Aut(U_{\gamma'})|}\sum_{P \in \mathcal{P}_{\gamma,\tau,{\gamma'}^T}}N(P)M^{orth}_{\tau_P}(e_{\sigma_2}e^T_{\sigma'_2})}}
	\end{align*}
\end{lemma}
\begin{proof}[Proof sketch]
	This lemma follows from the following bijection. Consider the third term
	\[
	\sum_{\sigma_2 \in \mathcal{L}, \gamma \in \Gamma: \sigma_2 \circ \gamma = \sigma}{\sum_{\sigma'_2 \in \mathcal{L}, \gamma' \in \Gamma: \sigma'_2 \circ \gamma' = \sigma'}{
			\frac{1}{|Aut(U_{\gamma})|\cdot|Aut(U_{\gamma'})|}\sum_{P \in \mathcal{P}_{\gamma,\tau,{\gamma'}^T}}N(P)M^{orth}_{\tau_P}(e_{\sigma_2}e^T_{\sigma'_2})}}
	\]
	On one side, we have the following data:
	\begin{enumerate}
		\item Ribbons $R_1$, $R_2$, and $R_3$ of shapes $\gamma,\tau,{\gamma'}^T$ such that $R_1,R_2,R_3$ are composable but $R_1$ and $R_2 \circ R_3$ are not properly composable (i.e. $R_1$ has an unexpected intersection with $R_2$ and/or $R_3$) and $R_1 \circ R_2$ and $R_3$ are not properly composable (i.e. $R_3$ has an unexpected intersection with $R_1$ and/or $R_2$).
		\item An ordering $O_{A'}$ on the leftmost minimum vertex separator $A'$ of $A_{R_1}$ and $V_{*} \cup B_{R_1}$ (recall that $V_{*}$ is the set of vertices which appear more than once in $V(R_1 \circ R_2 \circ R_3)$).
		\item An ordering $O_{B'}$ on the rightmost minimum vertex separator $B'$ of $V_{*} \cup A_{R_3}$ and $B_{R_3}$.
	\end{enumerate}
	On the other side, we have the following data
	\begin{enumerate}
		\item An intersection pattern $P \in \mathcal{P}_{\gamma,\tau,{\gamma'}^T}$ where $\gamma$ and ${\gamma'}^T$ are non-trivial.
		\item Ribbons $R'_1$, $R'_2$, $R'_3$ of shapes $\sigma_2$, $\tau_P$, ${\sigma'_2}^T$ which are properly composable
		\item A number in $[N(P)]$ describing which possible triple of ribbons resulted in the intersection pattern $P$ and the ribbon $R'_2$.
	\end{enumerate}
	To see this bijection, note that given the data on the first side, we can recover the ribbons $R'_1$, $R'_2$, and $R'_3$ as follows:
	\begin{enumerate}
		\item We decompose $R_1$ as $R_1 = R'_1 \circ R_4$ where $B_{R'_1} = A_{R_4} = A'$ with the ordering $O_{A'}$.
		\item We decompose $R_3$ as $R_3 = R_5 \circ R'_3$ where  where $B_{R_5} = A_{R'_3} = B'$ with the ordering $O_{B'}$.
		\item We take $R'_2 = R_4 \circ R_2 \circ R_5$.
	\end{enumerate}
	The intersection pattern $P$ and the number in $[N(P)]$ can be obtained from $R_1$, $R_2$, and $R_3$.
	
	Conversely, with the data on the other side, we can recover the data on the first side as follows:
	\begin{enumerate}
		\item $R'_2$ gives an ordering $O_{A'}$ for $A' = A_{R'_2}$ and an ordering $O_{B'}$ for $B' = B_{R'_2}$.
		\item The ribbon $R'_2$, intersection pattern $P$, and number in $[N(P)]$ allow us to recover $R_4$, $R_2$, and $R_5$.
		\item We take $R_1 = R'_1 \circ R_4$ and $R_3 = R_5 \circ R'_3$.
	\end{enumerate}
	Thus, both sides have the same coefficient for each ribbon.
	
	The analysis for the the first term is the same except that when $\gamma'$ is trivial, we always take $\gamma' = Id_{V_{\tau}}$. Thus, we always have that $B' = B_{R'_2} = B_{R_2}$ (with the same ordering) and $R'_3 = R_3 = Id_{B'}$. Because of this, there is no need to specify $R_3$, $R'_3$, $R_5$, or an ordering on $B'$.
	
	Similarly, the analysis for the the second term is the same except that when $\gamma$ is trivial, we always take $\gamma = Id_{U_{\tau}}$. Thus, we always have that $A' = A_{R'_2} = A_{R_2}$ (with the same ordering) and $R'_1 = R_1 = Id_{A'}$. Because of this, there is no need to specify $R_1$, $R'_1$, $R_4$, or an ordering on $A'$.
\end{proof}
Applying Lemma \ref{lm:singleshapeintersections} for all $\sigma$ and $\sigma'$ simultaneously, we obtain the following corollary.
\begin{definition}
	For all $U,V \in \mathcal{I}_{mid}$, given a $\gamma \in \Gamma_{U,V}$ and a vector $v$ indexed by left shapes $\sigma \in \mathcal{L}_V$, define $v^{-\gamma}$ to be the vector indexed by left shapes $\sigma_2 \in \mathcal{L}_{U}$ such that $v^{-\gamma}(\sigma_2) = v(\sigma_2 \circ \gamma)$ if $\sigma_2 \circ \gamma \in \mathcal{L}_V$ and $v^{-\gamma}(\sigma_2) = 0$ otherwise.
\end{definition}
\begin{proposition}
	For all composable $\gamma_2,\gamma_1 \in \Gamma$ and all vectors $v$ indexed by left shapes in $\mathcal{L}_{V_{\gamma_1}}$, $(v^{-\gamma_1})^{-\gamma_2} = v^{-\gamma_2 \circ \gamma_1}$
\end{proposition}
\begin{corollary}\label{cor:singlestepintersections}
	For all $\tau \in \mathcal{M}'$, for all left $\tau$-vectors $v$ and all right $\tau$-vectors $w$, 
	\begin{align*}
		&M^{orth}_{\tau}(vw^T) = M^{fact}_{\tau}(vw^T) - \sum_{\gamma \in \Gamma_{*,U_{\tau}}}{\frac{1}{|Aut(U_{\gamma})|}\sum_{P \in \mathcal{P}_{\gamma,\tau,Id_{V_{\tau}}}}N(P)M^{orth}_{\tau_P}(v^{-\gamma}w^T)} \\
		&- \sum_{\gamma' \in \Gamma_{*,V_{\tau}}}{\frac{1}{|Aut(U_{\gamma'})|}\sum_{P \in \mathcal{P}_{Id_{U_{\tau}},\tau,{\gamma'}^T}}N(P)M^{orth}_{\tau_P}(v(w^{-\gamma})^T)} \\
		&- \sum_{\gamma \in \Gamma_{*,U_{\tau}}}{\sum_{\gamma' \in \Gamma_{*,V_{\tau}}}{
				\frac{1}{|Aut(U_{\gamma})|\cdot|Aut(U_{\gamma'})|}\sum_{P \in \mathcal{P}_{\gamma,\tau,{\gamma'}^T}}N(P)M^{orth}_{\tau_P}(v^{-\gamma}(w^{-\gamma'})^T)}}
	\end{align*}
\end{corollary}
Applying Corollary \ref{cor:singlestepintersections} iteratively, we obtain the following theorem:
\begin{definition}\label{multigammadefinition}
	Given $\gamma,\gamma' \in \Gamma \cup \{Id_U:U \in \mathcal{I}_{mid}\}$ and $j > 0$, let $\Gamma_{\gamma,\gamma',j}$ be the set of all $\gamma_1,\gamma'_1,\cdots,\gamma_j,\gamma'_j \in \Gamma \cup \{Id_U:U \in \mathcal{I}_{mid}\}$ such that:
	\begin{enumerate}
		\item $\gamma_j,\ldots,\gamma_1$ are composable and $\gamma_j \circ \ldots \circ \gamma_1 = \gamma$
		\item $\gamma'_j,\ldots,\gamma'_1$ are composable and $\gamma'_j \circ \ldots \circ \gamma'_1 = \gamma'$
		\item For all $i \in [1,j]$, $\gamma_i$ or $\gamma'_i$ is non-trivial (i.e. $\gamma_i \neq Id_{U_{\gamma_i}}$ or $\gamma'_i \neq Id_{U_{\gamma'_i}}$).
	\end{enumerate}
\end{definition}
\begin{remark}
	Note that if $\gamma = Id_{U}$ and $\gamma' = Id_{V}$ then for all $j > 0$, $\Gamma_{\gamma,\gamma',j} = \emptyset$.
\end{remark}
\begin{theorem}\label{thm:mfactmorthdifference}
	For all $\tau \in \mathcal{M}'$, left $\tau$-vectors $v$, and right $\tau$-vectors $w$, 
	\begin{align*}
		&M^{orth}_{\tau}(v{w^T}) = M^{fact}_{\tau}(v{w^T}) + \\
		&\sum_{\gamma \in \Gamma_{*,U_{\tau}} \cup \{Id_{U_{\tau}}\},\gamma' \in \Gamma_{*,V_{\tau}} \cup \{Id_{V_{\tau}}\}: \atop \gamma \text{ or } \gamma' \text{ is non-trivial }}\sum_{j>0}{(-1)^{j}\sum_{\gamma_1,\gamma'_1,\cdots,\gamma_j,\gamma'_j \in \Gamma_{\gamma,\gamma',j}}{\prod_{i:\gamma_i \text{ is non-trivial}}{\frac{1}{|Aut(U_{\gamma_i})|}}
				\prod_{i:\gamma'_i \text{ is non-trivial}}{\frac{1}{|Aut(U_{\gamma'_i})|}}}} \\
		&\sum_{P_1,\cdots,P_j:P_i \in \mathcal{P}_{\gamma_i,\tau_{P_{i-1}},{\gamma'_i}^T}}{\left(\prod_{i=1}^{j}{N(P_i)}\right)M^{fact}_{\tau_{P_j}}(v^{-\gamma}{(w^{-\gamma'})^T})}
	\end{align*}
	where we take $\tau_{P_0} = \tau$.
\end{theorem}
\subsection{Bounding the difference between $M^{fact}$ and $M^{orth}$}\label{boundingdifferencesection}
In this subsection, we bound the difference between $M^{fact}_{\tau}(H_{\tau})$ and $M^{orth}_{\tau}(H_{\tau})$. We recall the following conditions on $B(\gamma)$, $N(\gamma)$, and $c(\gamma)$:
\begin{enumerate}
	\item For all $\tau \in \mathcal{M}'$, $\gamma \in \Gamma_{*,U_{\tau}} \cup \{Id_{U_{\tau}}\}$, and $\gamma' \in \Gamma_{*,V_{\tau}} \cup \{Id_{V_{\tau}}\}$,
	\begin{align*}
		&\sum_{j>0}{\sum_{\gamma_1,\gamma'_1,\cdots,\gamma_j,\gamma'_j \in \Gamma_{\gamma,\gamma',j}}{\left(\prod_{i:\gamma_i \text{ is non-trivial}}{\frac{1}{|Aut(U_{\gamma_i})|}}\right)
				\left(\prod_{i:\gamma'_i \text{ is non-trivial}}{\frac{1}{|Aut(U_{\gamma'_i})|}}\right)}}\\
		&\sum_{P_1,\cdots,P_j:P_i \in \mathcal{P}_{\gamma_i,\tau_{P_{i-1}},{\gamma'_i}^T}}{\left(\prod_{i=1}^{j}{N(P_i)}\right)} \leq \frac{N(\gamma)N(\gamma')}{(|Aut(U_{\gamma})|)^{1_{\gamma \text{ is non-trivial}}}(|Aut(U_{\gamma'})|)^{1_{\gamma' \text{ is non-trivial}}}}
	\end{align*}
	\item For all $\tau \in \mathcal{M}'$, $\gamma \in \Gamma_{*,U_{\tau}}$, and $\gamma' \in \Gamma_{*,V_{\tau}}$, for all $P \in \mathcal{P}_{\gamma,\tau,{\gamma'}^T}$, 
	$B_{norm}(\tau_{P}) \leq B(\gamma)B(\gamma')B_{norm}(\tau)$
	\item $\forall V \in \mathcal{I}_{mid}, \sum_{\gamma \in \Gamma_{*,V}}{\frac{1}{|Aut(U_{\gamma})|c(\gamma)}} \leq \epsilon' \leq \frac{1}{20}$
\end{enumerate}
With these conditions, we can now bound the difference between $M^{fact}$ and $M^{orth}$.
\begin{lemma}\label{keyboundinglemma}
	If the norm bounds and the conditions on $B(\gamma)$, $N(\gamma)$, and $c(\gamma)$ hold then for all $\tau \in \mathcal{M}'$, left $\tau$-vectors $v$, and right $\tau$-vectors $w$, 
	\begin{align*}
		&\left(M^{fact}_{\tau}(v{w^T}) + M^{fact}_{{\tau}^T}(w{v^T})\right) - \left(M^{orth}_{\tau}(v{w^T}) + M^{orth}_{{\tau}^T}(w{v^T})\right) \preceq \\
		&\\
		&{\epsilon'}B_{norm}(\tau)M^{fact}_{Id_{U_{\tau}}}(vv^{T}) + 
		2\sum_{\gamma \in \Gamma_{*,U_{\tau}}}{\frac{B(\gamma)^{2}N(\gamma)^{2}B_{norm}(\tau){c(\gamma)}}{|Aut(U_{\gamma})|}M^{fact}_{Id_{U_{\gamma}}}(v^{-\gamma}(v^{-\gamma})^T)} + \\
		&{\epsilon'}B_{norm}(\tau)M^{fact}_{Id_{V_{\tau}}}(ww^{T}) + 
		2\sum_{\gamma' \in \Gamma_{*,V_{\tau}}}{\frac{B(\gamma')^{2}N(\gamma')^{2}B_{norm}(\tau){c(\gamma')}}{|Aut(U_{\gamma'})|}M^{fact}_{Id_{U_{\gamma'}}}(w^{-\gamma'}(w^{-\gamma'})^T)}
	\end{align*}
\end{lemma}
\begin{proof}
	By Theorem \ref{thm:mfactmorthdifference}, taking $\tau_{P_0} = \tau$, 
	\begin{align*}
		&M^{orth}_{\tau}(v{w^T}) = M^{fact}_{\tau}(v{w^T}) + \\
		&\sum_{\gamma \in \Gamma_{*,U_{\tau}} \cup \{Id_{U_{\tau}}\},\gamma' \in \Gamma_{*,V_{\tau}} \cup \{Id_{V_{\tau}}\}: \atop \gamma \text{ or } \gamma' \text{ is non-trivial }}\sum_{j>0}{(-1)^{j}\sum_{\gamma_1,\gamma'_1,\cdots,\gamma_j,\gamma'_j \in \Gamma_{\gamma,\gamma',j}}{\prod_{i:\gamma_i \text{ is non-trivial}}{\frac{1}{|Aut(U_{\gamma_i})|}}
				\prod_{i:\gamma'_i \text{ is non-trivial}}{\frac{1}{|Aut(U_{\gamma'_i})|}}}} \\
		&\sum_{P_1,\cdots,P_j:P_i \in \mathcal{P}_{\gamma_i,\tau_{P_{i-1}},{\gamma'_i}^T}}{\left(\prod_{i=1}^{j}{N(P_i)}\right)M^{fact}_{\tau_{P_j}}(v^{-\gamma}{(w^{-\gamma'})^T})}
	\end{align*}
	Taking the transpose of this equation gives
	\begin{align*}
		&M^{orth}_{{\tau}^T}(w{v^T}) = M^{fact}_{{\tau}^T}(w{v^T}) + \\
		&\sum_{\gamma \in \Gamma_{*,U_{\tau}} \cup \{Id_{U_{\tau}}\},\gamma' \in \Gamma_{*,V_{\tau}} \cup \{Id_{V_{\tau}}\}: \atop \gamma \text{ or } \gamma' \text{ is non-trivial }}\sum_{j>0}{(-1)^{j}\sum_{\gamma_1,\gamma'_1,\cdots,\gamma_j,\gamma'_j \in \Gamma_{\gamma,\gamma',j}}{\prod_{i:\gamma_i \text{ is non-trivial}}{\frac{1}{|Aut(U_{\gamma_i})|}}
				\prod_{i:\gamma'_i \text{ is non-trivial}}{\frac{1}{|Aut(U_{\gamma'_i})|}}}} \\
		&\sum_{P_1,\cdots,P_j:P_i \in \mathcal{P}_{\gamma_i,\tau_{P_{i-1}},{\gamma'_i}^T}}{\left(\prod_{i=1}^{j}{N(P_i)}\right)M^{fact}_{\tau^{T}_{P_j}}({w^{-\gamma'}(v^{-\gamma})^T})}
	\end{align*}
	Now observe that by Lemma \ref{lm:rankonetosquares}, if the norm bounds hold, 
	\begin{align*}
		&{\pm}\left(M^{fact}_{\tau_{P_j}}(v^{-\gamma}{(w^{-\gamma'})^T}) + M^{fact}_{\tau^{T}_{P_j}}({w^{-\gamma'}(v^{-\gamma})^T})\right) = \\
		&{\pm}M^{fact}_{\tau_{P_j}}\left(\left(\sqrt{\frac{N(\gamma)B(\gamma){c(\gamma)}}{N(\gamma')B(\gamma'){c(\gamma')}}}v^{-\gamma}\right)
		\left(\sqrt{\frac{N(\gamma')B(\gamma'){c(\gamma')}}{N(\gamma)B(\gamma){c(\gamma)}}}{(w^{-\gamma'})^T}\right)\right) \pm \\
		&M^{fact}_{\tau^{T}_{P_j}}\left(\left(\sqrt{\frac{N(\gamma')B(\gamma'){c(\gamma')}}{N(\gamma)B(\gamma){c(\gamma)}}}w^{-\gamma'}\right)
		\left(\sqrt{\frac{N(\gamma)B(\gamma){c(\gamma)}}{N(\gamma')B(\gamma'){c(\gamma')}}}{(v^{-\gamma})^T}\right)\right) \preceq \\
		&B_{norm}(\tau_{P_j})\left(\frac{N(\gamma)B(\gamma){c(\gamma)}}{N(\gamma')B(\gamma'){c(\gamma')}}M^{fact}_{Id_{U_{\gamma}}}(v^{-\gamma}(v^{-\gamma})^T) + 
		\frac{N(\gamma')B(\gamma'){c(\gamma')}}{N(\gamma)B(\gamma){c(\gamma)}}M^{fact}_{Id_{U_{\gamma'}}}(w^{-\gamma'}(w^{-\gamma'})^T)\right)
	\end{align*}
	Combining these equations,
	\begin{align*}
		&\left(M^{fact}_{\tau}(v{w^T}) + M^{fact}_{{\tau}^T}(w{v^T})\right) - \left(M^{orth}_{\tau}(v{w^T}) + M^{orth}_{{\tau}^T}(w{v^T})\right) \preceq \\
		&\sum_{\gamma \in \Gamma_{*,U_{\tau}} \cup \{Id_{U_{\tau}}\},\gamma' \in \Gamma_{*,V_{\tau}} \cup \{Id_{V_{\tau}}\}: \atop \gamma \text{ or } \gamma' \text{ is non-trivial }}\sum_{j>0}{\sum_{\gamma_1,\gamma'_1,\cdots,\gamma_j,\gamma'_j \in \Gamma_{\gamma,\gamma',j}}{\prod_{i:\gamma_i \text{ is non-trivial}}{\frac{1}{|Aut(U_{\gamma_i})|}}
				\prod_{i:\gamma'_i \text{ is non-trivial}}{\frac{1}{|Aut(U_{\gamma'_i})|}}}} \\
		&\sum_{P_1,\cdots,P_j:P_i \in \mathcal{P}_{\gamma_i,\tau_{P_{i-1}},{\gamma'_i}^T}}{\left(\prod_{i=1}^{j}{N(P_i)}\right)
			B_{norm}(\tau_{P_j})} \\
		&\left(\frac{N(\gamma)B(\gamma){c(\gamma)}}{N(\gamma')B(\gamma'){c(\gamma')}}M^{fact}_{Id_{U_{\gamma}}}(v^{-\gamma}(v^{-\gamma})^T) + 
		\frac{N(\gamma')B(\gamma'){c(\gamma')}}{N(\gamma)B(\gamma){c(\gamma)}}M^{fact}_{Id_{U_{\gamma'}}}(w^{-\gamma'}(w^{-\gamma'})^T)\right)
	\end{align*}
	From the conditions on $B(\gamma)$ and $N(\gamma)$, 
	\begin{enumerate}
		\item $B_{norm}(\tau_{P_j}) \leq B(\gamma)B(\gamma')B_{norm}(\tau)$
		\item
		\begin{align*}
			&\sum_{j>0}{\sum_{\gamma_1,\gamma'_1,\cdots,\gamma_j,\gamma'_j \in \Gamma_{\gamma,\gamma',j}}{\left(\prod_{i:\gamma_i \text{ is non-trivial}}{\frac{1}{|Aut(U_{\gamma_i})|}}\right)
					\left(\prod_{i:\gamma'_i \text{ is non-trivial}}{\frac{1}{|Aut(U_{\gamma'_i})|}}\right)}}\\
			&\sum_{P_1,\cdots,P_j:P_i \in \mathcal{P}_{\gamma_i,\tau_{P_{i-1}},{\gamma'_i}^T}}{\left(\prod_{i=1}^{j}{N(P_i)}\right)} \leq \frac{N(\gamma)N(\gamma')}{(|Aut(U_{\gamma})|)^{1_{\gamma \text{ is non-trivial}}}(|Aut(U_{\gamma'})|)^{1_{\gamma' \text{ is non-trivial}}}}
		\end{align*}
	\end{enumerate}
	Putting these equations together, 
	\begin{align*}
		&\left(M^{fact}_{\tau}(v{w^T}) + M^{fact}_{{\tau}^T}(w{v^T})\right) - \left(M^{orth}_{\tau}(v{w^T}) + M^{orth}_{{\tau}^T}(w{v^T})\right) \preceq \\
		&\sum_{\gamma \in \Gamma_{*,U_{\tau}} \cup \{Id_{U_{\tau}}\},\gamma' \in \Gamma_{*,V_{\tau}} \cup \{Id_{V_{\tau}}\}: \atop \gamma \text{ or } \gamma' \text{ is non-trivial }}{\frac{B(\gamma)^{2}N(\gamma)^{2}B_{norm}(\tau){c(\gamma)}}{(|Aut(U_{\gamma})|)^{1_{\gamma \text{ is non-trivial}}}(|Aut(U_{\gamma'})|)^{1_{\gamma' \text{ is non-trivial}}}{c(\gamma')}}M^{fact}_{Id_{U_{\gamma}}}(v^{-\gamma}(v^{-\gamma})^T)} + \\
		&\sum_{\gamma \in \Gamma_{*,U_{\tau}} \cup \{Id_{U_{\tau}}\},\gamma' \in \Gamma_{*,V_{\tau}} \cup \{Id_{V_{\tau}}\}: \atop \gamma \text{ or } \gamma' \text{ is non-trivial }}{\frac{B(\gamma')^{2}N(\gamma')^{2}B_{norm}(\tau){c(\gamma')}}{(|Aut(U_{\gamma})|)^{1_{\gamma \text{ is non-trivial}}}(|Aut(U_{\gamma'})|)^{1_{\gamma' \text{ is non-trivial}}}{c(\gamma)}}M^{fact}_{Id_{U_{\gamma'}}}(w^{-\gamma'}(w^{-\gamma'})^T)}
	\end{align*}
	Now observe that
	\begin{align*}
		&\sum_{\gamma \in \Gamma_{*,U_{\tau}} \cup \{Id_{U_{\tau}}\},\gamma' \in \Gamma_{*,V_{\tau}} \cup \{Id_{V_{\tau}}\}: \atop \gamma \text{ or } \gamma' \text{ is non-trivial }}{\frac{B(\gamma)^{2}N(\gamma)^{2}B_{norm}(\tau){c(\gamma)}}{(|Aut(U_{\gamma})|)^{1_{\gamma \text{ is non-trivial}}}(|Aut(U_{\gamma'})|)^{1_{\gamma' \text{ is non-trivial}}}{c(\gamma')}}M^{fact}_{Id_{U_{\gamma}}}(v^{-\gamma}(v^{-\gamma})^T)} \preceq \\
		&\left(\sum_{\gamma' \in \Gamma_{*,V_{\tau}}}{\frac{1}{|Aut(U_{\gamma'})|{c(\gamma')}}}\right)B_{norm}(\tau)M^{fact}_{Id_{U_{\tau}}}(vv^{T}) + \\
		&\sum_{\gamma \in \Gamma_{*,U_{\tau}}}{\left(\sum_{\gamma' \in \Gamma_{*,V_{\tau}} \cup \{Id_{V_{\tau}}\}}{\frac{1}{(|Aut(U_{\gamma'})|)^{1_{\gamma' \text{ is non-trivial}}}{c(\gamma')}}}\right)\frac{B(\gamma)^{2}N(\gamma)^{2}B_{norm}(\tau){c(\gamma)}}{(|Aut(U_{\gamma})|)^{1_{\gamma \text{ is non-trivial}}}}M^{fact}_{Id_{U_{\gamma}}}(v^{-\gamma}(v^{-\gamma})^T)}  \preceq \\
		&{\epsilon'}B_{norm}(\tau)M^{fact}_{Id_{U_{\tau}}}(vv^{T}) + 
		2\sum_{\gamma \in \Gamma_{*,U_{\tau}}}{\frac{B(\gamma)^{2}N(\gamma)^{2}B_{norm}(\tau){c(\gamma)}}{|Aut(U_{\gamma})|}M^{fact}_{Id_{U_{\gamma}}}(v^{-\gamma}(v^{-\gamma})^T)}
	\end{align*}
	Following similar logic, 
	\begin{align*}
		&\sum_{\gamma \in \Gamma_{*,U_{\tau}} \cup \{Id_{U_{\tau}}\},\gamma' \in \Gamma_{*,V_{\tau}} \cup \{Id_{V_{\tau}}\}: \atop \gamma \text{ or } \gamma' \text{ is non-trivial }}{\frac{B(\gamma')^{2}N(\gamma')^{2}B_{norm}(\tau){c(\gamma')}}{(|Aut(U_{\gamma})|)^{1_{\gamma \text{ is non-trivial}}}(|Aut(U_{\gamma'})|)^{1_{\gamma' \text{ is non-trivial}}}{c(\gamma)}}M^{fact}_{Id_{U_{\gamma'}}}(w^{-\gamma'}(w^{-\gamma'})^T)} \preceq \\
		&{\epsilon'}B_{norm}(\tau)M^{fact}_{Id_{V_{\tau}}}(ww^{T}) + 
		2\sum_{\gamma' \in \Gamma_{*,V_{\tau}}}{\frac{B(\gamma')^{2}N(\gamma')^{2}B_{norm}(\tau){c(\gamma')}}{|Aut(U_{\gamma'})|}M^{fact}_{Id_{U_{\gamma'}}}(w^{-\gamma'}(w^{-\gamma'})^T)}
	\end{align*}
	Putting everything together,
	\begin{align*}
		&\left(M^{fact}_{\tau}(v{w^T}) + M^{fact}_{{\tau}^T}(w{v^T})\right) - \left(M^{orth}_{\tau}(v{w^T}) + M^{orth}_{{\tau}^T}(w{v^T})\right) \preceq \\
		&\\
		&{\epsilon'}B_{norm}(\tau)M^{fact}_{Id_{U_{\tau}}}(vv^{T}) + 
		2\sum_{\gamma \in \Gamma_{*,U_{\tau}}}{\frac{B(\gamma)^{2}N(\gamma)^{2}B_{norm}(\tau){c(\gamma)}}{|Aut(U_{\gamma})|}M^{fact}_{Id_{U_{\gamma}}}(v^{-\gamma}(v^{-\gamma})^T)} + \\
		&{\epsilon'}B_{norm}(\tau)M^{fact}_{Id_{V_{\tau}}}(ww^{T}) + 
		2\sum_{\gamma' \in \Gamma_{*,V_{\tau}}}{\frac{B(\gamma')^{2}N(\gamma')^{2}B_{norm}(\tau){c(\gamma')}}{|Aut(U_{\gamma'})|}M^{fact}_{Id_{U_{\gamma'}}}(w^{-\gamma'}(w^{-\gamma'})^T)}
	\end{align*}
	as needed.
\end{proof}
Using Lemma \ref{keyboundinglemma} we have the following corollaries:
\begin{corollary}\label{notauintersectiontermcorollary}
	For all $U \in \mathcal{I}_{mid}$, if the norm bounds and the conditions on $B(\gamma)$, $N(\gamma)$, and $c(\gamma)$ hold and $H_{Id_U} \succeq 0$ then  
	\[
	M^{fact}_{Id_{U}}(H_{Id_U}) - M^{orth}_{Id_{U}}(H_{Id_U}) \preceq {\epsilon'}M^{fact}_{Id_{U}}(H_{Id_U}) + 
	2\sum_{\gamma \in \Gamma_{*,U}}{\frac{B(\gamma)^{2}N(\gamma)^{2}{c(\gamma)}}{|Aut(U_{\gamma})|}M^{fact}_{Id_{U_{\gamma}}}(H^{-\gamma,\gamma}_{Id_{U}})}
	\]
\end{corollary}
\begin{corollary}\label{yestauintersectiontermcorollary}
	For all $U \in \mathcal{I}_{mid}$ and all $\tau \in \mathcal{M}_{U}$, if the norm bounds and the conditions on $B(\gamma)$, $N(\gamma)$, and $c(\gamma)$ hold and 
	\[
	\left[ {\begin{array}{cc}
			\frac{1}{|Aut(U)|c(\tau)}H_{Id_U} & B_{norm}(\tau)H_{\tau} \\
			B_{norm}(\tau)H^T_{\tau} & \frac{1}{|Aut(U)|c(\tau)}H_{Id_U}
	\end{array}} \right] \succeq 0
	\]
	then
	\begin{align*}
		&\left(M^{fact}_{\tau}(H_{\tau}) + M^{fact}_{{\tau}^T}(H^{T}_{\tau})\right) - \left(M^{orth}_{\tau}(H_{\tau}) + M^{orth}_{{\tau}^T}(H^{T}_{\tau})\right) \preceq \\
		&\\
		&2{\epsilon'}\frac{1}{|Aut(U)|c(\tau)}M^{fact}_{Id_{U}}(H_{Id_U}) + 
		4\sum_{\gamma \in \Gamma_{*,U}}{\frac{B(\gamma)^{2}N(\gamma)^{2}{c(\gamma)}}{|Aut(U_{\gamma})|\cdot|Aut(U)|c(\tau)}M^{fact}_{Id_{U_{\gamma}}}(H_{Id_{U}}^{-\gamma,\gamma}})
	\end{align*}
\end{corollary}
\subsection{Proof of the Main Theorem}
We now prove the following theorem which is a slight modification of Theorem \ref{maintheoremviaproperties} and which implies Theorem \ref{maintheoremviaproperties}.
\begin{theorem}\label{maintheoremviapropertiescopy}
	For all $\epsilon > 0$ and all $\epsilon' \in (0,\frac{1}{20}]$, for any moment matrix 
	\[
	\Lambda = \sum_{U \in \mathcal{I}_{mid}}{M^{orth}_{Id_U}(H_{Id_U})} + \sum_{U \in \mathcal{I}_{mid}}{\sum_{\tau \in \mathcal{M}_U}{M^{orth}_{\tau}(H_{\tau})}},
	\]
	if we have that for all $\alpha \in \mathcal{M}', ||M_{\alpha}|| \leq B_{norm}(\alpha)$ and $B(\gamma)$, $N(\gamma)$, and $c(\alpha)$ are functions such that 
	\begin{enumerate}
		\item For all $\tau \in \mathcal{M}'$, $\gamma \in \Gamma_{*,U_{\tau}}$, $\gamma' \in \Gamma_{*,V_{\tau}}$, and all intersection patterns $P \in \mathcal{P}_{\gamma,\tau,\gamma'}$, 
		\[
		B_{norm}(\tau_{P}) \leq B(\gamma)B(\gamma')B_{norm}(\tau)
		\]
		\item For all composable $\gamma_1,\gamma_2$, $B(\gamma_1)B(\gamma_2) = B(\gamma_1 \circ \gamma_2)$.
		\item $\forall U \in \mathcal{I}_{mid}, \sum_{\gamma \in \Gamma_{U,*}}{\frac{1}{|Aut(U)|c(\gamma)}} < \epsilon'$ 
		\item $\forall V \in \mathcal{I}_{mid}, \sum_{\gamma \in \Gamma_{*,V}}{\frac{1}{|Aut(U_{\gamma})|c(\gamma)}} < \epsilon'$ 
		\item $\forall U \in \mathcal{I}_{mid}, \sum_{\tau \in \mathcal{M}_{U}}{\frac{1}{|Aut(U)|c(\tau)}} < \epsilon'$
		\item For all $\tau \in \mathcal{M}'$, $\gamma \in \Gamma_{*,U_{\tau}} \cup \{Id_{U_{\tau}}\}$, and $\gamma' \in \Gamma_{*,V_{\tau}} \cup \{Id_{V_{\tau}}\}$,
		\begin{align*}
			&\sum_{j>0}{\sum_{\gamma_1,\gamma'_1,\cdots,\gamma_j,\gamma'_j \in \Gamma_{\gamma,\gamma',j}}{\prod_{i:\gamma_i \text{ is non-trivial}}{\frac{1}{|Aut(U_{\gamma_i})|}}
					\prod_{i:\gamma'_i \text{ is non-trivial}}{\frac{1}{|Aut(U_{\gamma'_i})|}}}}\sum_{P_1,\cdots,P_j:P_i \in \mathcal{P}_{\gamma_i,\tau_{P_{i-1}},{\gamma'_i}^T}}{\left(\prod_{i=1}^{j}{N(P_i)}\right)} \\
			&\leq \frac{N(\gamma)N(\gamma')}
			{(|Aut(U_{\gamma})|)^{1_{\gamma \text{ is non-trivial}}}(|Aut(U_{\gamma'})|)^{1_{\gamma' \text{ is non-trivial}}}}
		\end{align*}
	\end{enumerate}
	and we have SOS-symmetric coefficient matrices $\{H'_{\gamma}: \gamma \in \Gamma\}$ such that the following conditions hold:
	\begin{enumerate}
		\item For all $U \in \mathcal{I}_{mid}$,  $H_{Id_{U}} \succeq 0$
		\item For all $U \in \mathcal{I}_{mid}$ and $\tau \in \mathcal{M}_U$,
		\[
		\left[ {\begin{array}{cc}
				\frac{1}{|Aut(U)|c(\tau)}H_{Id_{U}} & B_{norm}(\tau)H_{\tau} \\
				B_{norm}(\tau)H^T_{\tau} & \frac{1}{|Aut(U)|c(\tau)}H_{Id_{U}}
		\end{array}} \right] \succeq 0
		\]
		\item For all $U,V \in \mathcal{I}_{mid}$ where $w(U) > w(V)$ and all $\gamma \in \Gamma_{U,V}$, 
		\[
		c(\gamma)^2{N(\gamma)}^2{B(\gamma)^2}H^{-\gamma,\gamma}_{Id_{V}} \preceq H'_{\gamma}
		\]
	\end{enumerate}
	then 
	\[
	\Lambda \succeq \frac{1}{2}\left(\sum_{U \in \mathcal{I}_{mid}}{M^{fact}_{Id_U}{(H_{Id_U})}}\right) - 3\left(\sum_{U \in \mathcal{I}_{mid}}{\sum_{\gamma \in \Gamma_{U,*}}{\frac{d_{Id_{U}}(H'_{\gamma},H_{Id_{U}})}{|Aut(U)|c(\gamma)}}}\right)Id_{sym}
	\]
	If it is also true that 
	\[
	\sum_{U \in \mathcal{I}_{mid}}{M^{fact}_{Id_U}{(H_{Id_U})}} \succeq 6\left(\sum_{U \in \mathcal{I}_{mid}}{\sum_{\gamma \in \Gamma_{U,*}}{\frac{d_{Id_{U}}(H'_{\gamma},H_{Id_{U}})}{|Aut(U)|c(\gamma)}}}\right)Id_{sym}
	\]
	then $\Lambda \succeq 0$.
\end{theorem}
\begin{proof}
	We make the following observations:
	\begin{enumerate}
		\item By Theorem \ref{thm:nointersectionanalysis},
		\[
		\sum_{U \in \mathcal{I}_{mid}}{M^{fact}_{Id_U}(H_{Id_U})} + \sum_{U \in \mathcal{I}_{mid}}{\sum_{\tau \in \mathcal{M}_U}{M^{fact}_{\tau}(H_{\tau})} \succeq (1-2\epsilon')\sum_{U \in \mathcal{I}_{mid}}{M^{fact}_{Id_U}(H_{Id_U})}}
		\]
		\item By Corollary \ref{notauintersectiontermcorollary}, 
		\[
		\sum_{U \in \mathcal{I}_{mid}}{\left(M^{fact}_{Id_{U}}(H_{Id_U}) - M^{orth}_{Id_{U}}(H_{Id_U})\right)} \preceq {\epsilon'}\sum_{U \in \mathcal{I}_{mid}}{M^{fact}_{Id_{U}}(H_{Id_U})} + 
		2\sum_{U \in \mathcal{I}_{mid}}{\sum_{\gamma \in \Gamma_{*,U}}{\frac{M^{fact}_{Id_{U_{\gamma}}}(H'_{\gamma})}{c(\gamma)|Aut(U_{\gamma})|}}}
		\]
		\item By Corollary \ref{yestauintersectiontermcorollary}, 
		\begin{align*}
			&\sum_{U \in \mathcal{I}_{mid}}{\sum_{\tau \in \mathcal{M}_U}{\left(M^{fact}_{\tau}(H_{\tau}) - M^{orth}_{\tau}(H_{\tau})\right)}} \preceq \\
			&\\
			&\sum_{U \in \mathcal{I}_{mid}}{\sum_{\tau \in \mathcal{M}_U}{\left(\frac{2{\epsilon'}}{|Aut(U)|c(\tau)}M^{fact}_{Id_{U}}(H_{Id_U}) + 
					4\sum_{\gamma \in \Gamma_{*,U}}{\frac{B(\gamma)^{2}N(\gamma)^{2}{c(\gamma)}}{|Aut(U_{\gamma})|\cdot|Aut(U)|c(\tau)}M^{fact}_{Id_{U_{\gamma}}}(H_{Id_{U}}^{-\gamma,\gamma})}\right)}} \preceq \\
			&2{\epsilon'}^2\sum_{U \in \mathcal{I}_{mid}}{M^{fact}_{Id_{U}}(H_{Id_U})} + 
			4{\epsilon'}\sum_{U \in \mathcal{I}_{mid}}{\sum_{\gamma \in \Gamma_{*,U}}{\frac{M^{fact}_{Id_{U_{\gamma}}}(H'_{\gamma})}{c(\gamma)|Aut(U_{\gamma})|}}}
		\end{align*}
		\item
		\begin{align*}
			&\sum_{U \in \mathcal{I}_{mid}}{\sum_{\gamma \in \Gamma_{*,U}}{\frac{M^{fact}_{Id_{U_{\gamma}}}(H'_{\gamma})}{c(\gamma)|Aut(U_{\gamma})|}}} = 
			\sum_{U \in \mathcal{I}_{mid}}{\sum_{\gamma \in \Gamma_{*,U}}{\frac{M^{fact}_{Id_{U_{\gamma}}}(H_{Id_{U_\gamma}}) + \left(M^{fact}_{Id_{U_{\gamma}}}(H'_{\gamma}) - M^{fact}_{Id_{U_{\gamma}}}(H_{Id_{U_\gamma}})\right)}{c(\gamma)|Aut(U_{\gamma})|}}} \preceq \\
			&\sum_{U \in \mathcal{I}_{mid}}{\sum_{\gamma \in \Gamma_{*,U}}{\frac{M^{fact}_{Id_{U_{\gamma}}}(H_{Id_{U_{\gamma}}})}{c(\gamma)|Aut(U_{\gamma})|}}} + 
			\left(\sum_{U \in \mathcal{I}_{mid}}{\sum_{\gamma \in \Gamma_{U,*}}{\frac{d_{Id_{U_{\gamma}}}(H'_{\gamma},H_{Id_{U_{\gamma}}})}{|Aut(U_{\gamma})|c(\gamma)}}}\right)Id_{sym} \preceq \\
			&{\epsilon'}\sum_{U \in \mathcal{I}_{mid}}{M^{fact}_{U}(H_{Id_{U}})} + 
			\left(\sum_{U \in \mathcal{I}_{mid}}{\sum_{\gamma \in \Gamma_{U,*}}{\frac{d_{Id_{U_{\gamma}}}(H'_{\gamma},H_{Id_{U_{\gamma}}})}{|Aut(U_{\gamma})|c(\gamma)}}}\right)Id_{sym}
		\end{align*}
	\end{enumerate}
	Putting everything together, 
	\begin{align*}
		&\Lambda = \sum_{U \in \mathcal{I}_{mid}}{M^{orth}_{Id_U}(H_{Id_U})} + \sum_{U \in \mathcal{I}_{mid}}{\sum_{\tau \in \mathcal{M}_U}{M^{orth}_{\tau}(H_{\tau})}} = \\
		&\sum_{U \in \mathcal{I}_{mid}}{M^{fact}_{Id_U}(H_{Id_U})} + \sum_{U \in \mathcal{I}_{mid}}{\sum_{\tau \in \mathcal{M}_U}{M^{fact}_{\tau}(H_{\tau})}} + 
		\sum_{U \in \mathcal{I}_{mid}}{\left(M^{fact}_{Id_{U}}(H_{Id_U}) - M^{orth}_{Id_{U}}(H_{Id_U})\right)} + \\
		&\sum_{U \in \mathcal{I}_{mid}}{\sum_{\tau \in \mathcal{M}_U}{\left(M^{fact}_{\tau}(H_{\tau}) - M^{orth}_{\tau}(H_{\tau}) \right)}} \succeq \\
		&(1 - 3{\epsilon'} - 2{\epsilon'}^2)\sum_{U \in \mathcal{I}_{mid}}{M^{fact}_{Id_U}(H_{Id_U})} - 
		(2 + 4\epsilon')\sum_{U \in \mathcal{I}_{mid}}{\sum_{\gamma \in \Gamma_{*,U}}{\frac{M^{fact}_{Id_{U_{\gamma}}}(H'_{\gamma})}{c(\gamma)|Aut(U_{\gamma})|}}} \succeq \\
		&(1 - 5{\epsilon'} - 6{\epsilon'}^2)\sum_{U \in \mathcal{I}_{mid}}{M^{fact}_{Id_U}(H_{Id_U})} - 
		(2 + 4\epsilon')\left(\sum_{U \in \mathcal{I}_{mid}}{\sum_{\gamma \in \Gamma_{U,*}}{\frac{d_{Id_{U_{\gamma}}}(H'_{\gamma},H_{Id_{U_{\gamma}}})}{|Aut(U_{\gamma})|c(\gamma)}}}\right)Id_{sym} \succeq \\
		&\frac{1}{2}\sum_{U \in \mathcal{I}_{mid}}{M^{fact}_{Id_U}(H_{Id_U})} - 
		3\left(\sum_{U \in \mathcal{I}_{mid}}{\sum_{\gamma \in \Gamma_{U,*}}{\frac{d_{Id_{U_{\gamma}}}(H'_{\gamma},H_{Id_{U_{\gamma}}})}{|Aut(U_{\gamma})|c(\gamma)}}}\right)Id_{sym}
	\end{align*}
\end{proof}

%% file: choosing_funcs.tex
In this subsection, we give functions $B_{norm}(\alpha)$, $B(\gamma)$, $N(\gamma)$, and $c(\alpha)$ which satisfy the conditions needed for our machinery.
\subsection{Theorem Statements}
Recall the following definitions from Section \ref{quantitativetheoremstatementsection}.
\begin{definition}
	We define $S_{\alpha}$ to be the leftmost minimum vertex separator of $\alpha$
\end{definition}
\begin{definition}[Simplified Isolated Vertices]
	Under our simplifying assumptions, we define 
	\[
	I_{\alpha} = \{v \in W_{\alpha}: v \text{ is not incident to any edges in } E(\alpha)\}
	\]
\end{definition}
\begin{theorem}[Simplified $B_{norm}(\alpha)$, $B(\gamma)$, $N(\gamma)$, and $c(\alpha)$]\label{simplifiedfunctionstheorem}
	Under our simplifying assumptions, for all $\epsilon, \epsilon' > 0$ and all $D_V \in \mathbb{N}$, if we take
	\begin{enumerate}
		\item $q = 3\left\lceil{{D_V}ln(n) + \frac{ln(\frac{1}{\epsilon})}{3} + {D_V}ln(5) + 3{D^2_V}ln(2)}\right\rceil$
		\item $B_{vertex} = 6{D_V}\sqrt[4]{2eq}$
		\item $B_{norm}(\alpha) = {B_{vertex}^{|V(\alpha) \setminus U_{\alpha}| + |V(\alpha) \setminus V_{\alpha}|}}n^{\frac{w(V(\alpha)) + w(I_{\alpha}) - w(S_{\alpha})}{2}}$
		\item $B(\gamma) = B_{vertex}^{|V(\gamma) \setminus U_{\gamma}| + |V(\gamma) \setminus V_{\gamma}|}n^{\frac{w(V(\gamma) \setminus U_{\gamma})}{2}}$
		\item $N(\gamma) = (3D_V)^{2|V(\gamma) \setminus V_{\gamma}| + |V(\gamma) \setminus U_{\gamma}|}$
		\item $c(\alpha) = \frac{5(3D_V)^{|U_{\alpha} \setminus V_{\alpha}| + |V_{\alpha} \setminus U_{\alpha}| + 2|E(\alpha)|}2^{|V(\alpha) \setminus (U_{\alpha} \cup V_{\alpha})|}}{\epsilon'}$
	\end{enumerate}
	then the following conditions hold:
	\begin{enumerate}
		\item With probability at least $(1-\epsilon)$, $\forall \alpha \in \mathcal{M}'$, $||M_{\alpha}|| \leq B_{norm}(\alpha)$
		\item For all $\tau \in \mathcal{M}'$, $\gamma \in \Gamma_{*,U_{\tau}} \cup \{Id_{U_{\tau}}\}$, $\gamma' \in \Gamma_{*,V_{\tau}} \cup \{Id_{V_{\tau}}\}$, and intersection patterns $P \in \mathcal{P}_{\gamma,\tau,\gamma'}$,  
		\[
		B_{norm}(\tau_{P}) \leq B(\gamma)B(\gamma')B_{norm}(\tau)
		\]
		\item For all composable $\gamma_1,\gamma_2$, $B(\gamma_1)B(\gamma_2) = B(\gamma_1 \circ \gamma_2)$.
		\item $\forall U \in \mathcal{I}_{mid}, \sum_{\gamma \in \Gamma_{U,*}}{\frac{1}{|Aut(U)|c(\gamma)}} < \epsilon'$ 
		\item $\forall V \in \mathcal{I}_{mid}, \sum_{\gamma \in \Gamma_{*,V}}{\frac{1}{|Aut(U_{\gamma})|c(\gamma)}} < \epsilon'$ 
		\item $\forall U \in \mathcal{I}_{mid}, \sum_{\tau \in \mathcal{M}_{U}}{\frac{1}{|Aut(U)|c(\tau)}} < \epsilon'$
		\item For all $\tau \in \mathcal{M}'$, $\gamma \in \Gamma_{*,U_{\tau}} \cup \{Id_{U_{\tau}}\}$, and $\gamma' \in \Gamma_{*,V_{\tau}} \cup \{Id_{V_{\tau}}\}$,
		\begin{align*}
			&\sum_{j>0}{\sum_{\gamma_1,\gamma'_1,\cdots,\gamma_j,\gamma'_j \in \Gamma_{\gamma,\gamma',j}}{\prod_{i:\gamma_i \text{ is non-trivial}}{\frac{1}{|Aut(U_{\gamma_i})|}}
					\prod_{i:\gamma'_i \text{ is non-trivial}}{\frac{1}{|Aut(U_{\gamma'_i})|}}}}\sum_{P_1,\cdots,P_j:P_i \in \mathcal{P}_{\gamma_i,\tau_{P_{i-1}},{\gamma'_i}^T}}{\left(\prod_{i=1}^{j}{N(P_i)}\right)} \\
			&\leq \frac{N(\gamma)N(\gamma')}
			{(|Aut(U_{\gamma})|)^{1_{\gamma \text{ is non-trivial}}}(|Aut(U_{\gamma'})|)^{1_{\gamma' \text{ is non-trivial}}}}
		\end{align*}
	\end{enumerate}
\end{theorem}
\subsubsection{General functions $B_{norm}(\alpha)$, $B(\gamma)$, $N(\gamma)$, and $c(\alpha)$*}
Recall the following definitions from Section \ref{generalmaintheoremstatementsection}.
\begin{definition}[$S_{\alpha,min}$ and $S_{\alpha,max}$]
	Given a shape $\alpha \in \mathcal{M}'$, define $S_{\alpha,min}$ to be the leftmost minimum vertex separator of $\alpha$ if all edges with multiplicity at least $2$ are deleted and define $S_{\alpha,max}$ to be the leftmost minimum vertex separator of $\alpha$ if all edges with multiplicity at least $2$ are present.
\end{definition}
\begin{definition}[General $I_{\alpha}$]
	Given a shape $\alpha$, define $I_{\alpha}$ to be the set of vertices in $V(\alpha) \setminus (U_{\alpha} \cup V_{\alpha})$ such that all edges incident with that vertex have multplicity at least $2$.
\end{definition}
\begin{definition}[$B_{\Omega}$]
	We take $B_{\Omega}(j)$ to be a non-decreasing function such that for all $j \in \mathbb{N}$, $E_{\Omega}[x^{j}] \leq B_{\Omega}(j)^{j}$ 
\end{definition}
\begin{definition}
	For all $i$, we define $h^{+}_i$ to be the polynomial $h_i$ where we make all of the coefficients have positive sign.
\end{definition}
\begin{lemma}
	If $\Omega = N(0,1)$ then we can take $B_{\Omega}(j) = \sqrt{j}$ and we have that 
\end{lemma}
\begin{theorem}[General $B_{norm}(\alpha)$, $B(\gamma)$, $N(\gamma)$, and $c(\alpha)$]\label{generalfunctionstheorem}
	For all $\epsilon, \epsilon' > 0$ and all $D_V,D_E \in \mathbb{N}$, if we take
	\begin{enumerate}
		\item $q = \left\lceil{3{D_V}ln(n) + ln(\frac{1}{\epsilon}) + {(3D_V)^k}ln(D_E + 1) + 3{D_V}ln(5)}\right\rceil$
		\item $B_{vertex} = 6q{D_V}$
		\item $B_{edge}(e) = 2h^{+}_{l_e}(B_{\Omega}(6{D_V}D_E))
		\max_{j \in [0,3{D_V}D_E]}{\left\{\left(h^{+}_{j}(B_{\Omega}(2qj))\right)^{\frac{l_e}{\max{\{j,l_e\}}}}\right\}}$
		\item $B_{norm}(\alpha) = 
		2e{B_{vertex}^{|V(\alpha) \setminus U_{\alpha}| + |V(\alpha) \setminus V_{\alpha}|}}\left(\prod_{e \in E(\alpha)}{B_{edge}(e)}\right)n^{\frac{w(V(\alpha)) + w(I_{\alpha}) - w(S_{\alpha})}{2}}$
		\item $B(\gamma) = B_{vertex}^{|V(\gamma) \setminus U_{\gamma}| + |V(\gamma) \setminus V_{\gamma}|}\left(\prod_{e \in E(\gamma)}{B_{edge}(e)}\right)n^{\frac{w(V(\gamma) \setminus U_{\gamma})}{2}}$
		\item $N(\gamma) = (3D_V)^{2|V(\gamma) \setminus V_{\gamma}| + |V(\gamma) \setminus U_{\gamma}|}$
		\item $c(\alpha) = \frac{5(3{t_{max}}D_V)^{|U_{\alpha} \setminus V_{\alpha}| + |V_{\alpha} \setminus U_{\alpha}| + k|E(\alpha)|}(2t_{max})^{|V(\alpha) \setminus (U_{\alpha} \cup V_{\alpha})|}}{\epsilon'}$
	\end{enumerate}
	then the following conditions hold:
	\begin{enumerate}
		\item With probability at least $(1-\epsilon)$, $\forall \alpha \in \mathcal{M}'$, $||M_{\alpha}|| \leq B_{norm}(\alpha)$
		\item For all $\tau \in \mathcal{M}'$, $\gamma \in \Gamma_{*,U_{\tau}} \cup \{Id_{U_{\tau}}\}$, $\gamma' \in \Gamma_{*,V_{\tau}} \cup \{Id_{V_{\tau}}\}$, and intersection patterns $P \in \mathcal{P}_{\gamma,\tau,\gamma'}$,  
		\[
		B_{norm}(\tau_{P}) \leq B(\gamma)B(\gamma')B_{norm}(\tau)
		\]
		\item For all composable $\gamma_1,\gamma_2$, $B(\gamma_1)B(\gamma_2) = B(\gamma_1 \circ \gamma_2)$.
		\item $\forall U \in \mathcal{I}_{mid}, \sum_{\gamma \in \Gamma_{U,*}}{\frac{1}{|Aut(U)|c(\gamma)}} < \epsilon'$ 
		\item $\forall V \in \mathcal{I}_{mid}, \sum_{\gamma \in \Gamma_{*,V}}{\frac{1}{|Aut(U_{\gamma})|c(\gamma)}} < \epsilon'$ 
		\item $\forall U \in \mathcal{I}_{mid}, \sum_{\tau \in \mathcal{M}_{U}}{\frac{1}{|Aut(U)|c(\tau)}} < \epsilon'$
		\item For all $\tau \in \mathcal{M}'$, $\gamma \in \Gamma_{*,U_{\tau}} \cup \{Id_{U_{\tau}}\}$, and $\gamma' \in \Gamma_{*,V_{\tau}} \cup \{Id_{V_{\tau}}\}$,
		\begin{align*}
			&\sum_{j>0}{\sum_{\gamma_1,\gamma'_1,\cdots,\gamma_j,\gamma'_j \in \Gamma_{\gamma,\gamma',j}}{\prod_{i:\gamma_i \text{ is non-trivial}}{\frac{1}{|Aut(U_{\gamma_i})|}}
					\prod_{i:\gamma'_i \text{ is non-trivial}}{\frac{1}{|Aut(U_{\gamma'_i})|}}}}\sum_{P_1,\cdots,P_j:P_i \in \mathcal{P}_{\gamma_i,\tau_{P_{i-1}},{\gamma'_i}^T}}{\left(\prod_{i=1}^{j}{N(P_i)}\right)} \\
			&\leq \frac{N(\gamma)N(\gamma')}
			{(|Aut(U_{\gamma})|)^{1_{\gamma \text{ is non-trivial}}}(|Aut(U_{\gamma'})|)^{1_{\gamma' \text{ is non-trivial}}}}
		\end{align*}
	\end{enumerate}
\end{theorem}
\begin{remark}
	Recall that if $\Omega = N(0,1)$ then we may take $B_{\Omega}(j) = \sqrt{j}$ and we have that 
	\[
	h^{+}_j(x) \leq \frac{1}{\sqrt{j!}}(x^2 + j)^{\frac{j}{2}} \leq \left(\frac{e}{j}(x^2 + j)\right)^{\frac{j}{2}}
	\]
	Thus, when $\Omega = N(0,1)$ we can take 
	\[
	B_{edge}(e) = 2\left(\frac{e}{l_e}(6{D_V}D_E + l_e)\right)^{l_e}\left(e(6{D_V}{D_E}q + 1)\right)^{l_e} \leq \left(400{D^2_V}{D^2_E}q\right)^{l_e}
	\]
\end{remark}
\subsection{Choosing $B_{norm}(\alpha)$}
We need matrix norm bounds which hold for all $\alpha \in \mathcal{M}'$. For convenience, we recall the definition of $\mathcal{M}'$ below.
\begin{definition}[$\mathcal{M}'$]
	We define $\mathcal{M}'$ to be the set of all shapes $\alpha$ such that
	\begin{enumerate}
		\item[1.] $|V(\alpha)| \leq 3D_V$
		\item[2.*] $\forall e \in E(\alpha), l_e \leq D_E$
		\item[3.*] All edges $e \in E(\alpha)$ have multiplicity at most $3D_V$.
	\end{enumerate}
\end{definition}
To obtain such norm bounds, we start with the norm bounds in the graph matrix norm bound paper. We then modify these bounds as follows:
\begin{enumerate}
	\item We make the bounds more compatible with the conditions of our machinery. To do this, we upper bound many of the terms in the norm bound by $B_{vertex}^{|V(\alpha) \setminus U_{\alpha}| + |V(\alpha) \setminus V_{\alpha}|}$ where $B_{vertex}$ is a function of our parameters. In general, we will also need to upper bound some of the terms by $\prod_{e \in E(\alpha)}(B_{edge}(e))$ where $B_{edge}(e)$ is a function of $l_e$, $\Omega$, and our parameters.
	\item We generalize the bounds so that they apply to improper shapes as well as proper shapes. Under our simplifying assumptions, all we need to do here is to take isolated vertices into account. In general, we also need to handle multi-edges. 
\end{enumerate}
\subsubsection{Simplified $B_{norm}(\alpha)$}
Under our simplifying assumptions, we start with the following norm bound from the updated graph matrix norm bound paper \cite{AMP20}:
\begin{theorem}[Simplified Graph Matrix Norm Bounds]\label{originalsimplifiednormbounds}
	Under our simplifying assumptions, for all $\epsilon > 0$ and all proper shapes $\alpha$, taking $c_{\alpha} = |V(\alpha) \setminus (U_{\alpha} \cup V_{\alpha})| + |S_{\alpha} \setminus (U_{\alpha} \cap V_{\alpha})|$,
	\[
	Pr\left(||M_{\alpha}|| > (2|V_{\alpha} \setminus (U_{\alpha} \cap V_{\alpha})|)^{|V(\alpha) \setminus (U_{\alpha} \cap V_{\alpha})|}(2eq)^{\frac{c_{\alpha}}{2}}n^{\frac{w(V(\alpha)) - w(S_{\alpha})}{2}}\right) < \epsilon
	\]
	where $q = 3\left\lceil\frac{ln(\frac{n^{w(S_{\alpha})}}{\epsilon})}{3c_{\alpha}}\right\rceil$
\end{theorem}
\begin{corollary}\label{firsttweakedsimplifiednormbound}
	For all shapes $\alpha$ and all $\epsilon > 0$, 
	\[
	Pr\left(||M_{\alpha}|| > \left(2|V_{\alpha}|\sqrt[4]{2eq}\right)^{|V(\alpha) \setminus U_{\alpha}| + |V(\alpha) \setminus V_{\alpha}|}n^{\frac{w(V(\alpha)) + w(I_{\alpha}) - w(S_{\alpha})}{2}}\right) < \epsilon
	\]
	where $q = 3\left\lceil\frac{ln(\frac{n^{w(S_{\alpha})}}{\epsilon})}{3c_{\alpha}}\right\rceil$.
\end{corollary}

\begin{proof}
	Observe that adding an isolated vertex to $\alpha$ is equivalent to multiplying $M_{\alpha}$ by $n - |V(\alpha)|$. Thus, if the bound holds for all proper $\alpha$ then it will hold for improper $\alpha$ as well. 
	
	We now make the following observations:
	\begin{enumerate}
		\item $|S_{\alpha} \setminus (U_{\alpha} \cap V_{\alpha})| \leq |U_{\alpha} \setminus V_{\alpha}|$, so $c_{\alpha} = |W_{\alpha}| + |S_{\alpha} \setminus (U_{\alpha} \cap V_{\alpha})| \leq |V(\alpha) \setminus V_{\alpha}|$. Similarly, $|S_{\alpha} \setminus (U_{\alpha} \cap V_{\alpha})| \leq |V_{\alpha} \setminus U_{\alpha}|$, so $c_{\alpha} \leq |V(\alpha) \setminus U_{\alpha}|$. Thus, $c_{\alpha} \leq \frac{|V(\alpha) \setminus U_{\alpha}| + |V(\alpha) \setminus V_{\alpha}|}{2}$.
		\item $|V(\alpha) \setminus (U_{\alpha} \cap V_{\alpha})| \leq |V(\alpha) \setminus U_{\alpha}| + |V(\alpha) \setminus V_{\alpha}|$
	\end{enumerate}
	Thus, by Theorem \ref{originalsimplifiednormbounds}, for all proper shapes $\alpha$ and all $\epsilon > 0$, 
	\[
	Pr\left(||M_{\alpha}|| > \left(2|V_{\alpha}|\sqrt[4]{2eq}\right)^{|V(\alpha) \setminus U_{\alpha}| + |V(\alpha) \setminus V_{\alpha}|}n^{\frac{w(V(\alpha)) + w(I_{\alpha}) - w(S_{\alpha})}{2}}\right) < \epsilon''
	\]
	where $q = 3\left\lceil\frac{ln(\frac{n^{w(S_{\alpha})}}{\epsilon})}{3c_{\alpha}}\right\rceil$.
\end{proof}
\begin{corollary}\label{tweakedsimplifiednormbound}
	For all $z \in \mathbb{N}$ and all $\epsilon > 0$, taking $\epsilon'' = \frac{\epsilon}{5^{z}2^{z^2}}$, with probability at least $1-\epsilon$ we have that for all shapes $\alpha$ such that $|V(\alpha)| \leq z$, 
	\[
	||M_{\alpha}|| \leq \left(2|V_{\alpha}|\sqrt[4]{2eq}\right)^{|V(\alpha) \setminus U_{\alpha}| + |V(\alpha) \setminus V_{\alpha}|}n^{\frac{w(V(\alpha)) + w(I_{\alpha}) - w(S_{\alpha})}{2}}
	\]
	where $q = 3\left\lceil\frac{ln(\frac{n^{w(S_{\alpha})}}{\epsilon''})}{3c_{\alpha}}\right\rceil$.
\end{corollary}
\begin{proof}
	This result can be proved from Corollary \ref{firsttweakedsimplifiednormbound} using a union bound and the following proposition:
	\begin{proposition}\label{simplifiedcountingalpha}
		Under our simplifying assumptions, for all $z \in \mathbb{N}$, there are at most $5^{z}2^{z^2}$ proper shapes $\alpha$ such that $V(\alpha) \leq z$.
	\end{proposition}
	\begin{proof}
		Observe that we can construct any proper shape $\alpha$ with at most $m$ vertices as follows:
		\begin{enumerate}
			\item Start with $z$ vertices $v_1,\ldots,v_z$.
			\item For each vertex $v_i$, choose whether $v_i \in V(\alpha) \setminus U_{\alpha} \setminus V_{\alpha}$, $v_i \in U_{\alpha} \setminus V_{\alpha}$, $v_i \in V_{\alpha} \setminus U_{\alpha}$, $v_i \in U_{\alpha} \cap V_{\alpha}$, or $v_i \notin V(\alpha)$.
			\item For each pair of vertices $v_i,v_j \in V(\alpha)$, choose whether or not $(v_i,v_j) \in E(\alpha)$
		\end{enumerate}
	\end{proof}
\end{proof}
\begin{corollary}\label{finalsimplifiednormbound}
	For all $D_V \in \mathbb{N}$ and all $\epsilon > 0$, taking 
	\[
	q = 3\left\lceil\frac{ln(\frac{5^{3D_V}2^{9D^2_V}n^{3D_V}}{\epsilon})}{3}\right\rceil = 
	3\left\lceil{{D_V}ln(n) + \frac{ln(\frac{1}{\epsilon})}{3} + {D_V}ln(5) + 3{D^2_V}ln(2)}\right\rceil,
	\] 
	$B_{vertex} = 6{D_V}\sqrt[4]{2eq}$, and
	\[
	B_{norm}(\alpha) = {B_{vertex}^{|V(\alpha) \setminus U_{\alpha}| + |V(\alpha) \setminus V_{\alpha}|}}n^{\frac{w(V(\alpha)) + w(I_{\alpha}) - w(S_{\alpha})}{2}},
	\] 
	with probability at least $(1-\epsilon)$ we have that for all shapes $\alpha \in \mathcal{M}'$, $||M_{\alpha}|| \leq B_{norm}(\alpha)$
\end{corollary}
\begin{proof}
	This follows from Corollary \ref{tweakedsimplifiednormbound} and the fact that for all $\alpha \in \mathcal{M}'$, $w(S_{\alpha}) \leq |V(\alpha)| \leq3D_V$
\end{proof}
\subsubsection{General $B_{norm}(\alpha)$}
In general, we start with the following norm bound from the updated graph matrix norm bound paper \cite{AMP20}:
\begin{theorem}[General Graph Matrix Norm Bounds]\label{originalgeneralnormbounds}
	For all $\epsilon > 0$ and all proper shapes $\alpha$, taking $q = \lceil{ln(\frac{n^{w(S_{\alpha})}}{\epsilon})}\rceil$
	\[
	P\left(||M_{\alpha}|| > 2e(2q|V(\alpha)|)^{|V(\alpha) \setminus (U_{\alpha} \cap V_{\alpha})|}\left(\prod_{e \in E(\alpha)}{h^{+}_{l_{e}}(B_{\Omega}(2q{l_e}))}\right)
	n^{\frac{(w(V(\alpha)) - w(S_{\alpha}))}{2}}\right) < \epsilon
	\]
\end{theorem}
\begin{corollary}
	For all $\epsilon > 0$, for all $z,l_{max},m \in \mathbb{N}$, taking $\epsilon'' = \frac{\epsilon}{5^{z}(l_{max} + 1)^{z^k}}$, with probability at least $1-\epsilon$, for all shapes $\alpha$ such that 
	\begin{enumerate}
		\item $|V(\alpha)| \leq z$.
		\item All edges in $E(\alpha)$ have label at most $l_{max}$.
		\item All edges in $E(\alpha)$ have multiplicity at most $m$.
	\end{enumerate}, 
	\begin{align*}
		||M_{\alpha}|| \leq &2e(2q|V(\alpha)|)^{|V(\alpha) \setminus U_{\alpha}| + |V(\alpha) \setminus V_{\alpha}|}\left(\prod_{e \in E(\alpha)}{2h^{+}_{l_e}(B_{\Omega}(2ml_{max}))
			\max_{j \in [0,ml_{max}]}{\left\{\left(h^{+}_{j}(B_{\Omega}(2qj))\right)^{\frac{l_e}{\max{\{j,l_e\}}}}\right\}}}\right)\\
		&n^{\frac{w(V(\alpha)) + w(I_{\alpha}) - w(S_{\alpha,min})}{2}}
	\end{align*}
	where $q = \left\lceil{ln\left(\frac{n^{w(S_{\alpha,max})}}{\epsilon''}\right)}\right\rceil$
\end{corollary}
\begin{proof}
	Observe that for each $\alpha$ which has multi-edges, we can write $M_{\alpha} = \sum_{i}{{c_i}M_{\alpha_i}}$ where each $\alpha_i$ has no multiple edges. We first upper bound $\sum_{i}{|c_i|}$.
	\begin{lemma}
		For any $a_1,\ldots,a_m \in \mathbb{N} \cup \{0\}$, taking $p_{max} = \sum_{i=1}^{m}{a_i}$ and writing $\prod_{i=1}^{m}{h_{a_i}} = \sum_{k=0}^{p_{max}}{{c_k}h_k}$,
		\[
		\sum_{k=0}^{p_{max}}{|c_k|} \leq (p_{max}+1)\prod_{i=1}^{m}{h^{+}_{a_i}(B_{\Omega}(2p_{max}))} \leq \prod_{i=1}^{m}{2h^{+}_{a_i}(B_{\Omega}(2p_{max}))}
		\]
	\end{lemma}
	\begin{proof}
		Suppose $\prod_{i = 1}^m (h_{a_i}(x))^2 = \sum_{k = 0}^{2p_{max}} u_kx^k$ and $\prod_{i = 1}^m (h^+_{a_i}(x))^2 = \sum_{k = 0}^{p_{max}} v_kx^k$. Then, note that $|u_k| \le v_k$ and so,	\[E_{\Omega}[\prod_{i = 1}^{m}(h_{a_i}(x))^2] = \sum_{k = 0}^{2p_{max}} u_k E_{\Omega}[x^k] \le \sum_{k = 0}^{2p_{max}} v_k|E_{\Omega}[x^k]| \le \sum_{k = 0}^{2p_{max}} v_k(B_{\Omega}(2p_{max}))^k = \prod_{i = 1}^m (h_{a_i}^+(B_{\Omega}(2p_{max}))^2\]
		Therefore, using the fact that $h_k$ form an orthonormal basis,
		\[\sum_{k = 0}^{p_{max}} c_k^2 = E_{\Omega}[(\sum_{k = 0}^{p_{max}} c_kh_k(x))^2] = E_{\Omega}[\prod_{i = 1}^m (h_{a_i}(x))^2] \le  \prod_{i = 1}^m (h_{a_i}^+(B_{\Omega}(2p_{max}))^2\]
		
		This implies
		\[(\sum_{k = 0}^{p_{max}}|c_k|)^2 \le (p_{max} + 1)(\sum_{k = 0}^{p_{max}} c_k^2) \le  (p_{max} + 1)  \prod_{i = 1}^m (h_{a_i}^+(B_{\Omega}(2p_{max}))^2\]
		Taking square roots gives the inequality.
	\end{proof}
	\begin{corollary}
		For any shape $\alpha$ such that every edge of $\alpha$ has multiplicity at most $m$ and label at most $l_{max}$, if we write 
		$M_{\alpha} = \sum_{i}{{c_i}M_{\alpha_i}}$ where each $\alpha_i$ has no multi-edges then $\sum_{i}{|c_i|} \leq \prod_{e \in E(\alpha)}{2h^{+}_{l_e}(B_{\Omega}(2ml_{max}))}$
	\end{corollary}
	The result now follows from Theorem \ref{originalgeneralnormbounds} and the following observations:
	\begin{enumerate}
		\item $|V(\alpha) \setminus (U_{\alpha} \cap V_{\alpha})| \leq |V(\alpha) \setminus U_{\alpha}| + |V(\alpha) \setminus V_{\alpha}|$.
		\item For any $\alpha$, writing $M_{\alpha} = \sum_{i}{{c_i}M_{\alpha_i}}$ where each $\alpha_i$ has no multi-edges, for all $\alpha_i$,
		\[
		w(V(\alpha_i)) + w(I_{\alpha_i}) - w(S_{\alpha_i}) \leq w(V(\alpha)) + w(I_{\alpha}) - w(S_{\alpha,min})
		\]
		\item For any $a_1,\ldots,a_m \in \mathbb{N} \cup \{0\}$ such that $\forall i' \in [m], a_{i'} \leq l_{max}$, for all $j \in [0,ml_{max}]$
		\[
		h^{+}_{j}(B_{\Omega}(2qj)) \leq \prod_{i'=1}^{m}{\left(h^{+}_{j}(B_{\Omega}(2qj))\right)^{\frac{a_{i'}}{\max{\{j,a_{i'}\}}}}} \leq 
		\prod_{i'=1}^{m}{\max_{j' \in [0,ml_{max}]}{\left\{\left(h^{+}_{j'}(B_{\Omega}(2qj'))\right)^{\frac{a_{i'}}{\max{\{j',a_{i'}\}}}}\right\}}}
		\]
	\end{enumerate}
	\begin{proposition}
		For all $z,l_{max} \in \mathbb{N}$, there are at most $5^{z}(l_{max} + 1)^{z^k}$ proper shapes $\alpha$ such that $|V(\alpha)| \leq z$ and every edge in $E(\alpha)$.
	\end{proposition}
	\begin{proof}
		This can be proved in the same way as before. Observe that we can construct any proper shape $\alpha$ with at most $z$ vertices as follows:
		\begin{enumerate}
			\item Start with $z$ vertices $v_1,\ldots,v_z$.
			\item For each vertex $v_i$, choose whether $v_i \in V(\alpha) \setminus U_{\alpha} \setminus V_{\alpha}$, $v_i \in U_{\alpha} \setminus V_{\alpha}$, $v_i \in V_{\alpha} \setminus U_{\alpha}$, $v_i \in U_{\alpha} \cap V_{\alpha}$, or $v_i \notin V(\alpha)$.
			\item For each $k$ tuple of vertices in $V(\alpha)$, choose the label of the hyperedge between these vertices (or $0$ if the hyperedge is not in $E(\alpha)$).
		\end{enumerate}
	\end{proof}
\end{proof}
\begin{corollary}\label{finalgeneralnormbounds}
	For all $D_V,D_E \in \mathbb{N}$ and all $\epsilon > 0$, taking 
	\[
	B_{norm}(\alpha) = 2e{B_{vertex}^{|V(\alpha) \setminus U_{\alpha}| + |V(\alpha) \setminus V_{\alpha}|}}\left(\prod_{e \in E(\alpha)}{B_{edge}(e)}\right)n^{\frac{w(V(\alpha)) + w(I_{\alpha}) - w(S_{\alpha})}{2}}
	\] where 
	\begin{enumerate}
		\item $q = \left\lceil{ln\left(\frac{n^{3D_V}}{\epsilon''}\right)}\right\rceil = \left\lceil{3{D_V}ln(n) + ln(\frac{1}{\epsilon}) + {(3D_V)^k}ln(D_E + 1) + 3{D_V}ln(5)}\right\rceil$
		\item $B_{vertex} = 6q{D_V}$
		\item $B_{edge}(e) = 2h^{+}_{l_e}(B_{\Omega}(6{D_V}D_E))
		\max_{j \in [0,3{D_V}D_E]}{\left\{\left(h^{+}_{j}(B_{\Omega}(2qj))\right)^{\frac{l_e}{\max{\{j,l_e\}}}}\right\}}$
	\end{enumerate}
	with probability at least $(1-\epsilon)$, for all shapes $\alpha \in \mathcal{M}'$, $||M_{\alpha}|| \leq B_{norm}(\alpha)$.
\end{corollary}
\subsection{Choosing $B(\gamma)$}
We now describe how to choose the function $B(\gamma)$. Recall that we want the following conditions to hold:
\begin{enumerate}
	\item For all $\gamma, \tau, \gamma'$ and all intersection patterns $P \in \mathcal{P}_{\gamma,\tau,\gamma'}$, 
	\[
	B_{norm}(\tau_{P}) \leq B(\gamma)B(\gamma')B_{norm}(\tau)
	\]
	\item For all composable $\gamma_1,\gamma_2$, $B(\gamma_1)B(\gamma_2) = B(\gamma_1 \circ \gamma_2)$.
\end{enumerate}
The most important part of choosing $B(\gamma)$ is to make sure that the factors of $n$ are controlled. For this, we use the following intersection tradeoff lemma. Under our simplifying assumptions, this lemma follows from \cite[Lemma 7.12]{BHKKMP16}. We defer the general proof of this lemma to the end of this section.
\begin{lemma}[Intersection Tradeoff Lemma]
	For all $\gamma, \tau, \gamma'$ and all intersection patterns $P \in \mathcal{P}_{\gamma,\tau,\gamma'}$, 
	\[
	w(V(\tau_P)) + w(I_{\tau_P}) - w(S_{\tau_P,min}) \leq w(V(\tau)) + w(I_{\tau}) - w(S_{\tau,min})+ w(V(\gamma) \setminus U_{\gamma}) + w(V(\gamma') \setminus U_{\gamma'})
	\]
\end{lemma}
Based on this intersection tradeoff lemma, we can choose the function $B(\gamma)$ as follows.
\begin{corollary}
	If we take 
	\[
	B_{norm}(\alpha) = C \cdot B_{vertex}^{|V(\alpha) \setminus U_{\alpha}| + |V(\alpha) \setminus V_{\alpha}|}\left(\prod_{e \in E(\alpha)}{B_{edge}(e)}\right)n^{\frac{w(V(\alpha)) + w(I_{\alpha}) - w(S_{\alpha})}{2}}
	\] 
	for some constant $C > 0$ and take
	\[
	B(\gamma) = B_{vertex}^{|V(\gamma) \setminus U_{\gamma}| + |V(\gamma) \setminus V_{\gamma}|}\left(\prod_{e \in E(\gamma)}{B_{edge}(e)}\right)n^{\frac{w(V(\gamma) \setminus U_{\gamma})}{2}}
	\]
	then the following conditions hold:
	\begin{enumerate}
		\item For all $\gamma, \tau, \gamma'$ and all intersection patterns $P \in \mathcal{P}_{\gamma,\tau,\gamma'}$, 
		\[
		B_{norm}(\tau_{P}) \leq B(\gamma)B(\gamma')B_{norm}(\tau)
		\]
		\item For all composable $\gamma_1,\gamma_2$, $B(\gamma_1)B(\gamma_2) = B(\gamma_1 \circ \gamma_2)$.
	\end{enumerate}
\end{corollary}
\begin{proof}
	We have that 
	\[
	B_{norm}(\tau_{P}) = B_{vertex}^{|V(\tau_P) \setminus U_{\tau_P}| + |V(\tau_P) \setminus V_{\tau_P}|}\left(\prod_{e \in E(\tau_P)}{B_{edge}(e)}\right)n^{\frac{w(V(\tau_P)) + w(I_{\tau_P}) - w(S_{\tau_P})}{2}}
	\]
	and 
	\begin{align*}
		B(\gamma)B(\gamma')B_{norm}(\tau) &= B_{vertex}^{|V(\gamma) \setminus U_{\gamma}| + |V(\gamma) \setminus V_{\gamma}| + |V(\gamma') \setminus U_{\gamma'}| + |V(\gamma') \setminus V_{\gamma'}| + |V(\tau) \setminus U_{\tau}| + |V(\tau) \setminus V_{\tau}|} \\
		&\left(\prod_{e \in E(\gamma) \cup E(\gamma') \cup E(\tau)}{B_{edge}(e)}\right)n^{\frac{w(V(\gamma) \setminus U_{\gamma}) + w(V(\gamma') \setminus U_{\gamma'}) + 
				w(V(\tau)) + w(I_{\tau}) - w(S_{\tau})}{2}}
	\end{align*}
	The first condition now follows immediately from the following observations:
	\begin{enumerate}
		\item \begin{align*}
			&|V(\gamma) \setminus U_{\gamma}| + |V(\gamma) \setminus V_{\gamma}| + |V(\gamma') \setminus U_{\gamma'}| + |V(\gamma') \setminus V_{\gamma'}| + |V(\tau) \setminus U_{\tau}| + |V(\tau) \setminus V_{\tau}| \\
			&= |V(\gamma \circ \tau \circ {\gamma'}^T) \setminus U_{\gamma \circ \tau \circ {\gamma'}^T}| + |V(\gamma \circ \tau \circ {\gamma'}^T) \setminus V_{\gamma \circ \tau \circ {\gamma'}^T}| 
			\geq |V(\tau_P) \setminus U_{\tau_P}| + |V(\tau_P) \setminus V_{\tau_P}|
		\end{align*}
		\item $E(\tau_P) = E(\gamma) \cup E(\tau) \cup E({\gamma'}^T)$ so 
		$\prod_{e \in E(\tau_P)}{B_{edge}(e)} = \prod_{e \in E(\gamma) \cup E(\gamma') \cup E(\tau)}{B_{edge}(e)}$.
		\item By the intersection tradeoff lemma, 
		\[
		w(V(\tau_P)) + w(I_{\tau_P}) - w(S_{\tau_P}) \leq w(V(\tau)) + w(I_{\tau}) - w(S_{\tau}) + w(V(\gamma) \setminus U_{\gamma}) + w(V(\gamma') \setminus U_{\gamma'})
		\]
	\end{enumerate}
	The second condition follows  from the form of $B(\gamma)$.
\end{proof}
\subsection{Choosing $N(\gamma)$}
To choose $N(\gamma)$, we use the following lemma:
\begin{lemma}
	For all $D_V \in \mathbb{N}$, for all composable $\gamma,\tau,{\gamma'}^T$ such that $|V(\gamma)| \leq D_V$, $|V(\tau)| \leq D_V$, and $|V(\gamma')| \leq D_V$, 
	\begin{align*}
		&\sum_{j>0}{\sum_{\gamma_1,\gamma'_1,\cdots,\gamma_j,\gamma'_j \in \Gamma_{\gamma,\gamma',j}}{\prod_{i:\gamma_i \text{ is non-trivial}}{\frac{1}{|Aut(U_{\gamma_i})|}}
				\prod_{i:\gamma'_i \text{ is non-trivial}}{\frac{1}{|Aut(U_{\gamma'_i})|}}}}\sum_{P_1,\cdots,P_j:P_i \in \mathcal{P}_{\gamma_i,\tau_{P_{i-1}},{\gamma'_i}^T}}{\left(\prod_{i=1}^{j}{N(P_i)}\right)} \\
		&\leq \frac{{(3D_V)}^{2(|V(\gamma) \setminus V_{\gamma}| + |V(\gamma') \setminus V_{\gamma'}|) + (|V(\gamma) \setminus U_{\gamma}| + |V(\gamma') \setminus U_{\gamma'}|)}}
		{(|Aut(U_{\gamma})|)^{1_{\gamma \text{ is non-trivial}}}(|Aut(U_{\gamma'})|)^{1_{\gamma' \text{ is non-trivial}}}}
	\end{align*}
\end{lemma}
\begin{proof}[Proof sketch]
	Observe that aside from the orderings (which are canceled out by the $|Aut(U_{\gamma_i})|$ and $|Aut(U_{\gamma'_i})|$ factors), the intersection patterns $\{P_i: i \in [j]\}$ are determined by the following data on each vertex $v \in (V(\gamma) \setminus V_{\gamma}) \cup (V({\gamma'}^T) \setminus V_{{\gamma'}^T})$:
	\begin{enumerate}
		\item The first $i \in [j]$ such that $v \in (V(\gamma_i) \setminus V_{\gamma_i}) \cup (V({\gamma'_i}^T) \setminus V_{{\gamma'_i}^T})$. There are at most $j$ possibilities for this.
		\item A vertex $u$ (if one exists) in $V(\gamma_{i-1} \circ \ldots \circ \gamma_1 \circ \tau \circ {\gamma'_1}^T \ldots \circ {\gamma'_{i-1}}^T)$ such that $u$ and $v$ are equal. There are at most $3D_V$ possibilities for this.
	\end{enumerate}
	Using these observations and taking $j_{max} = |V(\gamma) \setminus V_{\gamma}| + |V(\gamma') \setminus V_{\gamma'}|$,
	\begin{align*}
		&\sum_{j > 0}{\sum_{\gamma_1,\gamma'_1,\cdots,\gamma_j,\gamma'_j \in \Gamma_{\gamma,\gamma',j}}{\prod_{i:\gamma_i \text{ is non-trivial}}{\frac{1}{|Aut(U_{\gamma_i})|}}
				\prod_{i:\gamma'_i \text{ is non-trivial}}{\frac{1}{|Aut(U_{\gamma'_i})|}}}}\sum_{P_1,\cdots,P_j:P_i \in \mathcal{P}_{\gamma_i,\tau_{P_{i-1}},{\gamma'_i}^T}}{1} \\
		&\leq \sum_{j =1}^{j_{max}}{\frac{{(3jD_V)}^{|V(\gamma) \setminus V_{\gamma}| + |V(\gamma') \setminus V_{\gamma'}|}}
			{(|Aut(U_{\gamma})|)^{1_{\gamma \text{ is non-trivial}}}(|Aut(U_{\gamma'})|)^{1_{\gamma' \text{ is non-trivial}}}}} \\
		&\leq j_{max}\left(\frac{2}{3}\right)^{j_{max}}\frac{{(3D_V)}^{2(|V(\gamma) \setminus V_{\gamma}| + |V(\gamma') \setminus V_{\gamma'}|)}}
		{(|Aut(U_{\gamma})|)^{1_{\gamma \text{ is non-trivial}}}(|Aut(U_{\gamma'})|)^{1_{\gamma' \text{ is non-trivial}}}} \\
		&< \frac{{(3D_V)}^{2(|V(\gamma) \setminus V_{\gamma}| + |V(\gamma') \setminus V_{\gamma'}|)}}
		{(|Aut(U_{\gamma})|)^{1_{\gamma \text{ is non-trivial}}}(|Aut(U_{\gamma'})|)^{1_{\gamma' \text{ is non-trivial}}}}
	\end{align*}
	Now recall that by Lemma ,
	for any $\gamma_i,\tau_{P_{i-1}},{\gamma'_i}^{T}$ and any intersection pattern $P_i \in \mathcal{P}_{\gamma_i,\tau_{P_{i-1}},{\gamma'_i}^T}$, 
	\[
	N(P_i) \leq |V(\tau_{P_i})|^{|V(\gamma_i) \setminus U_{\gamma_i}| + |V(\gamma'_i) \setminus U_{\gamma'_i}|}
	\]
	Thus, for any $P_1,\cdots,P_j:P_i \in \mathcal{P}_{\gamma_i,\tau_{P_{i-1}},{\gamma'_i}^T}$, $\prod_{i=1}^{j}{N(P_i)} \leq (3D_V)^{|V(\gamma) \setminus U_{\gamma}| + |V(\gamma') \setminus U_{\gamma'}|}$.
	Putting everything together, the result follows.
\end{proof}
\begin{corollary}
	For all $D_V \in \mathbb{N}$, if we take $N(\gamma) = (3D_V)^{2|V(\gamma) \setminus V_{\gamma}| + |V(\gamma) \setminus U_{\gamma}|}$ then for all composable $\gamma, \tau, {\gamma'}^T$ such that $|V(\gamma)| \leq D_V$, $|V(\tau)| \leq D_V$, and $|V(\gamma')| \leq D_V$, 
	\begin{align*}
		&\sum_{j>0}{\sum_{\gamma_1,\gamma'_1,\cdots,\gamma_j,\gamma'_j \in \Gamma_{\gamma,\gamma',j}}{\prod_{i:\gamma_i \text{ is non-trivial}}{\frac{1}{|Aut(U_{\gamma_i})|}}
				\prod_{i:\gamma'_i \text{ is non-trivial}}{\frac{1}{|Aut(U_{\gamma'_i})|}}}}\sum_{P_1,\cdots,P_j:P_i \in \mathcal{P}_{\gamma_i,\tau_{P_{i-1}},{\gamma'_i}^T}}{\left(\prod_{i=1}^{j}{N(P_i)}\right)} \\
		&\leq \frac{N(\gamma)N(\gamma')}
		{(|Aut(U_{\gamma})|)^{1_{\gamma \text{ is non-trivial}}}(|Aut(U_{\gamma'})|)^{1_{\gamma' \text{ is non-trivial}}}}
	\end{align*}
\end{corollary}
\subsection{Choosing $c(\alpha)$}
In this section, we describe how to choose $c(\alpha)$. For simplicity, we first describe how to choose $c(\alpha)$ under our simplifying assumptions. We then describe the minor adjustments that are needed when we have hyperedges and multiple types of vertices.
\begin{lemma}\label{calphalemma}
	Under our simplifying assumptons, for all $U \in \mathcal{I}_{mid}$,
	\[
	\sum_{\alpha: U_{\alpha} \equiv U, \alpha \text{ is proper and non-trivial}}{\frac{1}{|Aut(U_{\alpha} \cap V_{\alpha})|(3D_V)^{|U_{\alpha} \setminus V_{\alpha}| + |V_{\alpha} \setminus U_{\alpha}| + 2|E(\alpha)|}2^{|V(\alpha) \setminus (U_{\alpha} \cup V_{\alpha})|}}} < 5
	\]
\end{lemma} 
\begin{proof}
	In order to choose $\alpha$, it is sufficient to choose the following:
	\begin{enumerate}
		\item The number $j_1$ of vertices in $U_{\alpha} \setminus V_{\alpha}$, the number $j_2$ of vertices in $V_{\alpha} \setminus U_{\alpha}$, and the number $j_3$ of vertices in $V(\alpha) \setminus (U_{\alpha} \cup V_{\alpha})$.
		\item A mapping in $Aut(U_{\alpha} \cap V_{\alpha})$ determining how the vertices in $U_{\alpha} \cap V_{\alpha}$ match up with each other.
		\item The position of each vertex $u \in U_{\alpha} \setminus V_{\alpha}$ within $U_{\alpha}$ (there are at most $|U_{\alpha}| \leq D_V$ choices for this).
		\item The position of each vertex $v \in V_{\alpha} \setminus U_{\alpha}$ within $V_{\alpha}$ (there are at most $|U_{\alpha}| \leq D_V$ choices for this).
		\item The number $j_4$ of edges in $E(\alpha)$.
		\item The endpoints of each edge in $E(\alpha)$.
	\end{enumerate}
	This implies that for all $j_1,j_2,j_3,j_4 \geq 0$
	\[
	\sum_{\alpha: U_{\alpha} \equiv U, |U_{\alpha} \setminus V_{\alpha}| = j_1, |V_{\alpha} \setminus U_{\alpha}| = j_2 \atop
		|V(\alpha) \setminus (U_{\alpha} \cup V_{\alpha})| = j_3, |E(\alpha)| = j_4}{\frac{1}{|Aut(U_{\alpha} \cap V_{\alpha})|(D_V)^{j_1 + j_2}
			(D_V)^{2j_4}}} \leq 1
	\]
	Using this, we have that
	\begin{align*}
		&\sum_{\alpha: U_{\alpha} \equiv U, \alpha \text{ is proper and non-trivial}}{\frac{1}{|Aut(U_{\alpha} \cap V_{\alpha})|(3D_V)^{|U_{\alpha} \setminus V_{\alpha}| + |V_{\alpha} \setminus U_{\alpha}| + 2|E(\alpha)|}2^{|V(\alpha) \setminus (U_{\alpha} \cup V_{\alpha})|}}} \\
		&\leq \sum_{j_1,j_2,j_3,j_4 \in \mathbb{N} \cup \{0\}: j_1 + j_2 + j_3 + j_4 \geq 1}{\frac{1}{3^{j_1 + j_2}9^{j_4}2^{j_3}}} \leq 2\left(\frac{3}{2}\right)^2\frac{9}{8} - 1 < 5
	\end{align*}
\end{proof}
\begin{corollary}
	For all $\epsilon' > 0$, if we take 
	\[
	c(\alpha) = \frac{5(3D_V)^{|U_{\alpha} \setminus V_{\alpha}| + |V_{\alpha} \setminus U_{\alpha}| + 2|E(\alpha)|}2^{|V(\alpha) \setminus (U_{\alpha} \cup V_{\alpha})|}}{\epsilon'}
	\]
	then
	\begin{enumerate}
		\item $\forall U \in \mathcal{I}_{mid}, \sum_{\gamma \in \Gamma_{U,*}}{\frac{1}{|Aut(U)|c(\gamma)}} < \epsilon'$ 
		\item $\forall V \in \mathcal{I}_{mid}, \sum_{\gamma \in \Gamma_{*,V}}{\frac{1}{|Aut(U_{\gamma})|c(\gamma)}} < \epsilon'$ 
		\item $\forall U \in \mathcal{I}_{mid}, \sum_{\tau \in \mathcal{M}_{U}}{\frac{1}{|Aut(U)|c(\tau)}} < \epsilon'$
	\end{enumerate}
\end{corollary}
\subsubsection{Choosing $c(\alpha)$ in general*}
When we have multiple types of vertices and hyperedges of arity $k$, Lemma \ref{calphalemma} can be generalized as follows:
\begin{lemma}
	Under our simplifying assumptons, for all $U \in \mathcal{I}_{mid}$,
	\[
	\sum_{\alpha: U_{\alpha} \equiv U, \alpha \text{ is proper and non-trivial}}{\frac{1}{|Aut(U_{\alpha} \cap V_{\alpha})|(3D_V{t_{max}})^{|U_{\alpha} \setminus V_{\alpha}| + |V_{\alpha} \setminus U_{\alpha}| + k|E(\alpha)|}(2t_{max})^{|V(\alpha) \setminus (U_{\alpha} \cup V_{\alpha})|}}} < 5
	\]
\end{lemma}
\begin{proof}[Proof sketch]
	This can be proved in the same way as Lemma \ref{calphalemma} with the following modifications:
	\begin{enumerate}
		\item In addition to choosing the number of vertices in $U_{\alpha} \setminus V_{\alpha}$, $V_{\alpha} \setminus U_{\alpha}$, and $V(\alpha) \setminus (U_{\alpha} \cap V_{\alpha})$, we also have to choose the types of these vertices.
		\item For each hyperedge, we have to choose $k$ endpoints rather than $2$ endpoints.
	\end{enumerate}
\end{proof}

\begin{corollary}\label{cor: calphachoice}
	For all $\eps' > 0$, if we take
	\[c(\alpha) = \frac{5(3{t_{max}}D_V)^{|U_{\alpha} \setminus V_{\alpha}| + |V_{\alpha} \setminus U_{\alpha}| + k|E(\alpha)|}(2t_{max})^{|V(\alpha) \setminus (U_{\alpha} \cup V_{\alpha})|}}{\epsilon'}\]
	then
	\begin{enumerate}
		\item $\forall U \in \mathcal{I}_{mid}, \sum_{\gamma \in \Gamma_{U,*}}{\frac{1}{|Aut(U)|c(\gamma)}} < \epsilon'$ 
		\item $\forall V \in \mathcal{I}_{mid}, \sum_{\gamma \in \Gamma_{*,V}}{\frac{1}{|Aut(U_{\gamma})|c(\gamma)}} < \epsilon'$ 
		\item $\forall U \in \mathcal{I}_{mid}, \sum_{\tau \in \mathcal{M}_{U}}{\frac{1}{|Aut(U)|c(\tau)}} < \epsilon'$
	\end{enumerate}
\end{corollary}

For technical reasons, we will need a more refined bound when the sum is over all shapes $\gamma$ of at least a prescribed size.

\begin{lemma}\label{lem: tpca_refined_bound}
	For all $\eps' > 0$, for the same choice of $c(\al)$ as in $\cref{cor: calphachoice}$, for any $U \in \calI_{mid}$ and integer $m \ge 1$, we have \[\sum_{\gam \in \Gam_{U, *}: |V(\gam)| \ge |U| + m} \frac{1}{|Aut(U)|c(\gam)} \le \frac{\eps'}{5\cdot 2^{m - 1}}\]
\end{lemma}

\begin{proof}[Proof sketch]
	The proof is similar to the proof of \cref{cor: calphachoice}, but we now have the extra condition $j_2 + j_3 \ge m$ in the proof of \cref{calphalemma}. Then,
	\[\sum_{j_1, j_2, j_3, j_4 \in \NN\cup \{0\}: j_2 + j_3 \ge m} \frac{1}{3^{j_1 + j_2}9^{j_4}2^{j_3}} \le \sum_{j_1, j_4 \in \NN\cup \{0\}} \frac{1}{2^m3^{j_1}9^{j_4}} = \frac{27}{16\cdot 2^m}\le \frac{1}{2^{m - 1}}\]
\end{proof}

\subsection{Proof of the Generalized Intersection Tradeoff Lemma}
We now prove the generalized intersection tradeoff lemma, which in particular generalizes \cite[Lemma 7.12]{BHKKMP16}.
\begin{lemma}
	For all $\gamma, \tau, \gamma'$ and all intersection patterns $P \in \mathcal{P}_{\gamma,\tau,\gamma'}$, 
	\[
	w(V(\tau_P)) + w(I_{\tau_P}) - w(S_{\tau_P,min}) \leq w(V(\tau)) + w(I_{\tau}) - w(S_{\tau,min})+ w(V(\gamma) \setminus U_{\gamma}) + w(V(\gamma') \setminus U_{\gamma'})
	\]
\end{lemma}
\begin{proof}
	\begin{definition} \ 
		\begin{enumerate}
			\item We define $I_{LM}$ to be the set of vertices which, after intersections, touch $\gamma$ and $\tau$ but not ${\gamma'}^T$. In particular, $I_{LM}$ consists of the vertices which result from intersecting a pair of vertices in $V(\gamma) \setminus V_{\gamma}$ and $V(\tau) \setminus U_{\tau} \setminus V_{\tau}$ and the vertices which are in $U_{\tau} \setminus V_{\tau}$ and are not intersected with any other vertex.
			\item We define $I_{MR}$ to be the set of vertices which, after intersections, touch $\tau$ and ${\gamma'}^T$ but not $\gamma$. In particular, $I_{MR}$ consists of the vertices which result from intersecting a pair of vertices in $V(\tau) \setminus U_{\tau} \setminus V_{\tau}$ and $V({\gamma'}^T) \setminus U_{{\gamma'}^T}$ and the vertices which are in $V_{\tau} \setminus U_{\tau}$ and are not intersected with any other vertex.
			\item We define $I_{LR}$ to be the set of vertices which, after intersections, touch $\gamma$ and ${\gamma'}^T$ but not $\tau$. In particular, $I_{LR}$ consists of the vertices which result from intersecting a pair of vertices in $V(\gamma) \setminus V_{\gamma}$ and $V({\gamma'}^T) \setminus U_{{\gamma'}^T}$.
			\item We define $I_{LMR}$ to be the set of vertices which, after intersections, touch $\gamma$, $\tau$, and ${\gamma'}^T$. In particular, $I_{LMR}$ consists of the vertices which result from intersecting a triple of vertices in $V(\gamma) \setminus V_{\gamma}$,  $V(\tau) \setminus U_{\tau} \setminus V_{\tau}$, and $V({\gamma'}^T) \setminus U_{{\gamma'}^T}$, intersecting a pair of vertices in $V(\gamma) \setminus V_{\gamma}$ and $V_{\tau} \setminus U_{\tau}$, intersecting a pair of vertices in $U_{\tau} \setminus V_{\tau}$ and $V({\gamma'}^T) \setminus U_{{\gamma'}^T}$, and single vertices in $U_{\tau} \cap V_{\tau}$. 
		\end{enumerate}
	\end{definition}
	The main idea is as follows. A priori, any of the vertices in $I_{LM} \cup I_{MR} \cup I_{LR} \cup I_{LMR}$ could become isolated. We handle this by keeping track of the following types of flows:
	\begin{enumerate}
		\item Flows from $U_{\gamma}$ to $I_{LM} \cup I_{LR} \cup I_{LMR}$
		\item Flows from $I_{LR} \cup I_{MR} \cup I_{LMR}$ to $V_{{\gamma'}^T}$
		\item Flows from $I_{LM}$ to $I_{MR}$. For technical reasons, we also view vertices in $I_{LMR}$ as having flow to themselves.
	\end{enumerate}
	We then observe that flows to and from these vertices prevent these vertices from being isolated and can provide flow from $U_{\gamma}$ to $V_{{\gamma'}^T}$, which gives a lower bound on $w(S_{\tau_P})$.
	
	We now implement this idea.
	\begin{definition}[Flow Graph]
		Given a shape $\alpha$, we define the directed graph $H_{\alpha}$ as follows:
		\begin{enumerate}
			\item For each vertex $v \in V(\alpha)$, we create two vertices $v_{in}$ and $v_{out}$. We then create a directed edge from $v_{in}$ to $v_{out}$ with capacity $w(v)$
			\item For each pair of vertices $(v,w)$ which is an edge of multiplicity $1$ in $E(\alpha)$ (or part of a hyperedge of multiplicity $1$ in $E(\alpha)$), we create a directed edge with infinite capacity from $v_{out}$ to $w_{in}$ and we create a directed edge with infinite capacity from $w_{out}$ to $v_{in}$.
			\item We define $U_{H_{\alpha}}$ to be $U_{H_{\alpha}} = \{u_{in}: u \in U_{\alpha}\}$ and we define $V_{H_{\alpha}}$ to be $V_{H_{\alpha}} = \{v_{out}: v \in V_{\alpha}\}$
		\end{enumerate}
	\end{definition}
	\begin{lemma}
		The maximum flow from $U_{H_{\alpha}}$ to $V_{H_{\alpha}}$ is equal to the minimum weight of a separator between $U_{\alpha}$ and $V_{\alpha}$.
	\end{lemma}
	\begin{proof}
		This can be proved using the max flow min cut theorem.
	\end{proof}
	\begin{definition}[Modified Flow Graph]
		Given a shape $\alpha$ together with a set $I_L \subseteq V(\alpha)$ of vertices in $\alpha$ (which will be the vertices in $\alpha$ which are intersected with a vertex to the left of $\alpha$) and a set $I_R \subseteq V(\alpha)$ of vertices in $\alpha$ (which will be the vertices in $\alpha$ which are intersected with a vertex to the right of $\alpha$), we define the modified flow graph $H^{I_L,I_R}_{\alpha}$ as follows:
		\begin{enumerate}
			\item We start with the flow graph $H_{\alpha}$
			\item For each vertex $u \in I_L$, we delete all of the edges into $u_{in}$ and add $u_{in}$ to $U_{H_{\alpha}}$
			\item For each vertex $v \in I_R$, we delete all of the edges out of $v_{out}$ and add $v_{out}$ to $V_{H_{\alpha}}$
			\item We call the resulting graph $H^{I_L,I_R}_{\alpha}$ and the resulting sets $U_{H^{I_L,I_R}_{\alpha}}$ and $V_{H^{I_L,I_R}_{\alpha}}$
		\end{enumerate}
	\end{definition}
	\begin{lemma}\label{increasedflowlemma}
		The maximum flow from $U_{H^{I_L,I_R}_{\alpha}}$ to $V_{H^{I_L,I_R}_{\alpha}}$ in $H^{I_L,I_R}_{\alpha}$ is at least as large as the maximum flow from $U_{H_{\alpha}}$ to $V_{H_{\alpha}}$ in $H_{\alpha}$
	\end{lemma}
	\begin{proof}[Proof sketch]
		Observe that if we have a cut $C$ in $H^{I_L,I_R}_{\alpha}$ which separates $U_{H^{I_L,I_R}_{\alpha}}$ and $V_{H^{I_L,I_R}_{\alpha}}$ then $C$ separates $U_{H_{\alpha}}$ and $V_{H_{\alpha}}$ in $H_{\alpha}$
	\end{proof}
	Before the intersections, we have the following flows.
	\begin{enumerate}
		\item We take $F_1$ to be the maximum flow from $U_{\gamma}$ to $V_{\gamma}$ in $\gamma$. Note that $F_1$ has value $w(V_{\gamma})$
		\item We take $F_2$ to be the maximum flow from $U_{\tau}$ to $V_{\tau}$ in $\tau$. Note that $F_2$ has value $w(S_{\tau,min})$
		\item We take $F_3$ to be the maximum flow from $U_{{\gamma'}^T}$ to $V_{{\gamma'}^T}$ in ${\gamma'}^T$. Note that $F_1$ has value $w(U_{{\gamma'}^T})$
	\end{enumerate}
	After the intersections, we take the following flows:
	\begin{enumerate}
		\item We take $F'_1$ to be the maximum flow from $U_{H^{\emptyset,I_{LM} \cup I_{LR} \cup I_{LMR}}_{\gamma}}$ to $V_{H^{\emptyset,I_{LM} \cup I_{LR} \cup I_{LMR}}_{\gamma}}$ in \\
		$H^{\emptyset,I_{LM} \cup I_{LR} \cup I_{LMR}}_{\gamma}$.
		\item We take $F'_2$ to be the maximum flow from $U_{H^{I_{LM} \cup I_{LMR},I_{MR} \cup I_{LMR}}_{\tau}}$ to $V_{H^{I_{LM} \cup I_{LMR},I_{MR} \cup I_{LMR}}_{\tau}}$ in \\
		$H^{I_{LM} \cup I_{LMR},I_{MR} \cup I_{LMR}}_{\tau}$
		\item We take $F'_3$ to be the maximum flow from $U_{H^{I_{MR} \cup I_{LR} \cup I_{LMR},\emptyset}_{{\gamma'}^T}}$ to $V_{H^{I_{MR} \cup I_{LR} \cup I_{LMR},\emptyset}_{{\gamma'}^T}}$ in \\
		$H^{I_{MR} \cup I_{LR} \cup I_{LMR},\emptyset}_{{\gamma'}^T}$.
	\end{enumerate}
	Observe that because of how intersection patterns are defined, $val(F'_1) = w(U_{\gamma})$ and $val(F'_3) = w(V_{{\gamma'}^T})$. By Lemma \ref{increasedflowlemma}, the value of $F'_2$ is at least as large as the value of $F_2$, so $val(F'_2) \geq w(S_{\tau,min})$.
	
	We now consider $F'_1 + F'_2 + F'_3$. As is, this is not a flow, but we can fix this.
	\begin{definition}
		For each vertex $v \in V(\tau_P)$,
		\begin{enumerate}
			\item We define $f_{in}(v)$ to be the flow into $v_{in}$ in $F'_1 + F'_2 + F'_3$.
			\item We define $f_{out}(v)$ to be the flow out of $v_{out}$ in $F'_1 + F'_2 + F'_3$.
			\item We define $f_{through}(v)$ to be the flow from $v_{in}$ to $v_{out}$ in $F'_1 + F'_2 + F'_3$
			\item We define $f_{imbalance}(v)$ to be $f_{imbalance}(v) = |f_{in}(v) - f_{out}(v)|$
			\item We define $f_{excess}(v)$ to be $f_{excess}(v) = f_{through}(v) - max{\{f_{in}(v),f_{out}(v)\}}$
		\end{enumerate}
		With this information, we fix the flow $F'_1 + F'_2 + F'_3$ as follows. For each vertex $v \in V(\tau_P)$, 
		\begin{enumerate}
			\item If $f_{in}(v) > f_{out}(v)$ then we create a vertex $v_{supplemental,out}$ and an edge from $v_{out}$ to $v_{supplemental,out}$ with capacity $f_{imbalance}(v)$ and we route $f_{imbalance}(v)$ of flow along this edge. We then add $v_{supplemental,out}$ to a set of vertices $V_{supplemental}$.
			\item If $f_{in}(v) < f_{out}(v)$ then we create a vertex $v_{supplemental,in}$ and an edge from $v_{supplemental,in}$ to $v_{in}$ with capacity $f_{imbalance}(v)$ and we route $f_{imbalance}(v)$ of flow along this edge. We then add $v_{supplemental,out}$ to a set of vertices $V_{supplemental}$.
			\item We reduce the flow on the edge from $v_{in}$ to $v_{out}$ by $f_{excess}(v)$
		\end{enumerate}
		We call the resulting flow $F'$
	\end{definition}
	\begin{proposition}
		$F'$ is a flow from  $U_{H^{\emptyset,I_{LM} \cup I_{LR} \cup I_{LMR}}_{\gamma}} \cup U_{supplemental}$ to $V_{H^{I_{MR} \cup I_{LR} \cup I_{LMR},\emptyset}_{{\gamma'}^T}} \cup V_{supplemental}$ with value $val(F') = val(F'_1) + val(F'_2) + val(F'_3) - \sum_{v \in V(\tau)}{f_{excess}(v)}$
	\end{proposition}
	\begin{corollary}\label{fixedflowcorollary}
		There exists a flow $F''$ from $U_{H^{\emptyset,I_{LM} \cup I_{LR} \cup I_{LMR}}_{\gamma}}$ to $V_{H^{I_{MR} \cup I_{LR} \cup I_{LMR},\emptyset}_{{\gamma'}^T}}$ with value $val(F'') \geq val(F'_1) + val(F'_2) + val(F'_3) - \sum_{v \in V(\tau)}{(f_{excess}(v) + f_{imbalance}(v))}$
	\end{corollary}
	\begin{proof}
		Consider the minimum cut $C$ between $U_{H^{\emptyset,I_{LM} \cup I_{LR} \cup I_{LMR}}_{\gamma}}$ and $V_{H^{I_{MR} \cup I_{LR} \cup I_{LMR},\emptyset}_{{\gamma'}^T}}$. If we add all of the supplemental edges to $C$ then this gives a cut $C'$ between $U_{H^{\emptyset,I_{LM} \cup I_{LR} \cup I_{LMR}}_{\gamma}}$ and $V_{H^{I_{MR} \cup I_{LR} \cup I_{LMR},\emptyset}_{{\gamma'}^T}}$ with capacity 
		\[
		capacity(C') = capacity(C) + \sum_{v \in V(\tau)}{f_{imbalance}(v)} \geq val(F')
		\]
		Thus, $capacity(C) \geq val(F') - \sum_{v \in V(\tau)}{f_{imbalance}(v)}$ so there exists a flow $F''$ from $U_{H^{\emptyset,I_{LM} \cup I_{LR} \cup I_{LMR}}_{\gamma}}$ to $V_{H^{I_{MR} \cup I_{LR} \cup I_{LMR},\emptyset}_{{\gamma'}^T}}$ with value 
		\[
		val(F'') = capacity(C) \geq val(F'_1) + val(F'_2) + val(F'_3) - \sum_{v \in V(\tau)}{(f_{excess}(v) + f_{imbalance}(v))}
		\]
	\end{proof}
	We now make the following observations:
	\begin{lemma} \ 
		\begin{enumerate}
			\item For all vertices $v \notin I_{LM} \cup I_{MR} \cup I_{LR} \cup I_{LMR}$, $f_{excess}(v) = f_{imbalance}(v) = 0$ (and these vertices can never be isolated).
			\item For all vertices $v \in I_{LM}$, $f_{excess}(v) + f_{imbalance}(v) \leq w(v)$. Moreover, for all vertices $v \in I_{LM}$ which are isolated, $f_{excess}(v) = f_{imbalance}(v) = 0$.
			\item For all vertices $v \in I_{MR}$, $f_{excess}(v) + f_{imbalance}(v) \leq w(v)$. Moreover, for all vertices $v \in I_{LM}$ which are isolated, $f_{excess}(v) = f_{imbalance}(v) = 0$.
			\item For all vertices $v \in I_{LR}$, $f_{excess}(v) + f_{imbalance}(v) \leq w(v)$. Moreover, for all vertices $v \in I_{LM}$ which are isolated, $f_{excess}(v) = f_{imbalance}(v) = 0$.
			\item For all vertices $v \in I_{LMR}$, $f_{excess}(v) + f_{imbalance}(v) \leq 2w(v)$. Moreover, for all vertices $v \in I_{LMR}$ which are isolated, $f_{excess}(v) = w(v)$ and $f_{imbalance}(v) = 0$.
		\end{enumerate}
	\end{lemma}
	\begin{proof}
		For the first statement, observe that for vertices $v \notin I_{LM} \cup I_{MR} \cup I_{LR} \cup I_{LMR}$, neither $v_{in}$ nor $v_{out}$ is ever a sink or source so the flow into these vertices must equal the flow out of these vertices and thus $f_{in}(v) = f_{out}(v) = f_{through}(v)$.
		
		For the second statement, observe that for a vertex $v \in I_{LM}$, 
		\begin{enumerate}
			\item $F'_1$ will have a flow of $f_{in}(v)$ into $v_{in}$ and along the edge from $v_{in}$ to $v_{out}$ 
			\item $F'_2$ will have a flow of $f_{out}(v)$ along the edge from $v_{in}$ to $v_{out}$ and out of $v_{out}$.
		\end{enumerate}
		Thus, $f_{excess}(v) = f_{in}(v) + f_{out}(v) - \max\{f_{in}(v),f_{out}(v)\}$. Since $f_{imbalance}(v) = |f_{in}(v) - f_{out}(v)|$, 
		$f_{excess}(v) + f_{imbalance}(v) = f_{in}(v) + f_{out}(v) - \min\{f_{in}(v),f_{out}(v)\} \leq w(v)$.
		
		If $v$ is isolated then neither $F'_1$ nor $F'_2$ can have any flow to $v_{in}$ or out of $v_{out}$ so $f_{in}(v) = f_{through}(v) = f_{out}(v) = 0$
		
		The third and fourth statements can be proved in the same way as the second statement.
		
		For the fifth statement, observe that for a vertex $v \in I_{LMR}$, 
		\begin{enumerate}
			\item $F'_1$ will have a flow of $f_{in}(v)$ into $v_{in}$ and along the edge from $v_{in}$ to $v_{out}$.
			\item $F'_2$ will have a flow of $w(v)$ along the edge from $v_{in}$ to $v_{out}$
			\item $F'_3$ will have a flow of $f_{out}(v)$ along the edge from $v_{in}$ to $v_{out}$ and out of $v_{out}$.
		\end{enumerate}
		Thus, $f_{excess}(v) = w(v) + f_{in}(v) + f_{out}(v) - \max\{f_{in}(v),f_{out}(v)\}$. Since $f_{imbalance}(v) = |f_{in}(v) - f_{out}(v)|$, 
		$f_{excess}(v) + f_{imbalance}(v) = w(v) + f_{in}(v) + f_{out}(v) - \min\{f_{in}(v),f_{out}(v)\} \leq 2w(v)$.
		
		If $v$ is isolated then neither $F'_1$ nor $F'_3$ can have any flow to $v_{in}$ or out of $v_{out}$ so $f_{in}(v) = f_{out}(v) = 0$ and $f_{through}(v) = w(v)$.
	\end{proof}
	Putting everything together, we have the following corollary:
	\begin{corollary} \ 
		\[
		\sum_{v \in V(\tau_P)}{(f_{excess}(v) + f_{imbalance}(v))} \leq w(I_{LM}) + w(I_{LR}) + w(I_{MR}) + 2w(I_{LMR}) - (w(I_{\tau_P}) - w(I_{\tau}))
		\]
	\end{corollary}
	Combining this with Corollary \ref{fixedflowcorollary},
	\begin{align*}
		w(S_{\tau_P,min}) &\geq val(F'_1) + val(F'_2) + val(F'_3) - \sum_{v \in V(\tau_P)}{(f_{excess}(v) + f_{imbalance}(v))}\\
		&\geq w(U_{\gamma}) + w(S_{\tau,min}) + w(V_{{\gamma'}^T}) - w(I_{LM}) - w(I_{LR}) - w(I_{MR}) - 2w(I_{LMR}) + (w(I_{\tau_P}) - w(I_{\tau}))
	\end{align*}
	Since $w(V(\tau_P)) = w(V(\tau)) + w(V(\gamma)) + w(V(\gamma')) - w(I_{LM}) - w(I_{LR}) - w(I_{MR}) - 2w(I_{LMR})$, 
	\[
	w(S_{\tau_P,min}) \geq w(U_{\gamma}) + w(S_{\tau,min}) + w(V_{{\gamma'}^T}) + w(V(\tau_P)) - w(V(\tau)) - w(V(\gamma)) - w(V(\gamma')) + (w(I_{\tau_P}) - w(I_{\tau}))
	\]
	Rearranging this gives 
	\[
	w(V(\tau_P)) - w(S_{\tau_P,min}) + w(I_{\tau_P}) \leq w(V(\tau)) - w(S_{\tau,min}) + w(I_{\tau}) + w(V(\gamma) \setminus U_{\gamma}) + w(V(\gamma') \setminus U_{\gamma'})
	\]
	which is the generalized intersection tradeoff lemma.
\end{proof}

%% file: showing_positivity.tex
Now, we illustrate one way to show truncation error bounds when we apply the machinery. Assume $\norm{M_{\alpha}} \le B_{norm}(\alpha)$ for all $\alpha \in \mathcal{M}'$. We want to show
\[
\sum_{U \in \mathcal{I}_{mid}}{M^{fact}_{Id_U}{(H_{Id_U})}} \succeq 6\left(\sum_{U \in \mathcal{I}_{mid}}{\sum_{\gamma \in \Gamma_{U,*}}{\frac{d_{Id_{U}}(H'_{\gamma},H_{Id_{U}})}{|Aut(U)|c(\gamma)}}}\right)Id_{sym}
\]

To do this, we simply sandwich a factor of $Id_{sym}$ between the two terms. Let $D_{sos}$ be the degree of the SoS program. We will describe in \cref{subsec: positivity_lower_bound_strategy} how to show $\sum_{U \in \mathcal{I}_{mid}}{M^{fact}_{Id_U}{(H_{Id_U})}} \succeq \frac{1}{n^{K_1D_{sos}^2}} Id_{sym}$
for a constant $K_1 > 0$. We also show
$\sum_{U\in \calI_{mid}} \sum_{\gam \in \Gam_{U, *}} \frac{d_{Id_{U}}(H_{Id_{U}}, H'_{\gam})}{|Aut(U)|c(\gam)} \le \frac{n^{K_2D_{sos}}}{2^{D_V}}$
for a constant $K_2 > 0$.
Along with the fact that $Id_{Sym} \succeq 0$, we can choose $D_{sos}$ small enough so that $\frac{1}{n^{K_1D_{sos}^2}} > \frac{n^{K_2D_{sos}}}{2^{D_V}}$, completing the proof.

We will need the following simple bound that says that if we have sufficient decay for each vertex, then, the sum of this decay, over all shapes $\sig \circ \sig'$ for $\sig, \sig' \in \calL_U'$, is bounded.

\begin{definition}
    For $U \in \calI_{mid}$, let $\calL_{U}' \subset \calL_U$ be the set of non-trivial shapes in $\calL_U$.
\end{definition}

\begin{lemma}\label{lem: gp_sum}
    Suppose $D_V = n^{C_V\eps}, D_E = n^{C_E\eps}$ for constants $C_V, C_E > 0$, are the truncation parameters for our shapes. For any $U \in \calI_{mid}$,
    \begin{align*}
        \sum_{U\in \calI_{mid}} \sum_{\sig, \sig' \in \calL'_U} \frac{1}{D_{sos}^{D_{sos}}n^{F \eps|V(\sig \circ \sig')|}} \le 1
    \end{align*}
    for a constant $F > 0$ that depends only on $C_V, C_E$. In particular, by setting $C_V, C_E$ small enough, we can make this constant arbitrarily small.
\end{lemma}

\begin{proof}
    For a given $j = |U|$, the number of ways to choose $U$ is at most $t_{max}^j$. For a given $U \in \calI_{mid}$, we will bound the number of ways to choose $\sig, \sig' \in \calL_U'$. To choose $\sig, \sig'\in \calL'_U$, it is sufficient to choose
    \begin{itemize}
        \item The number of vertices $j_1 \ge 1$ (resp. $j_1' \ge 1$) in $U_{\sig} \setminus V_{\sig}$ (resp. $U_{\sig'} \setminus V_{\sig'}$), their types of which there are at most $t_{max}$, and their powers which have at most $D_{sos}$ choices.
        \item The number of vertices $j_2$ (resp. $j_2'$) in $V(\sig) \setminus (U_{\sig} \cup V_{\sig})$ (resp. $V(\sig') \setminus (U_{\sig'} \cup V_{\sig'})$) and also their types, of which there are at most $t_{max}$.
        \item The position of each vertex $i$ in $U_{\sig} \setminus V_{\sig}$ (resp. $U_{\sig'} \setminus V_{\sig'}$) within $U_{\sig}$ (resp. $U_{\sig'}$). There are at most $D_V$ choices for each vertex.
        \item The subset of $U_{\sig}$ (resp. $U_{\sig'}$) that is in $V_{\sig}$ (resp. $V_{\sig'}$) and a mapping in $Aut(U_{\sig} \cap V_{\sig})$ (resp. $Aut(U_{\sig'} \cap V_{\sig'})$) that determines the matching between the vertices in $U_{\sig} \cap V_{\sig}$ (resp. $U_{\sig'} \cap V_{\sig'}$).
        \item The number $j_3$ (resp. $j_3'$) of edges in $E(\sig)$ (resp. $E(\sig')$). and the $k$ endpoints of each edge. Each endpoint has at most $D_V$ choices.
    \end{itemize}
    Therefore, for all $j \ge 0, j_1, j_1' \ge 1, j_2, j_2', j_3, j_3' \ge 0$, we have
    {\footnotesize
        \begin{align*}
            \sum_{U \in \calI_{mid}}\sum_{\substack{\sig, \sig' \in \calL'_U \\ |U_{\sig} \setminus V_{\sig}| = j_1, |U_{\sig'} \setminus V_{\sig'}| = j_1' \\ |V(\sig) \setminus (U_{\sig} \cup V_{\sig})| = j_2, |V(\sig') \setminus (U_{\sig'} \cup V_{\sig'})| = j_2'\\ |E(\sig)| = j_3, |E(\sig')| = j_3'}} \frac{1}{|Aut(U_{\sig'} \cap V_{\sig'})||Aut(U_{\sig} \cap V_{\sig})|(2t_{max})^{j + j_2 + j_2'}(D_Vt_{max}D_{sos})^{j_1 + j_1'}(D_V)^{kj_3}} \le 1
        \end{align*}
    }
    This implies that
    \begin{align*}
        \sum_{U\in \calI_{mid}} \sum_{\sig, \sig' \in \calL'_U} \frac{1}{D_{sos}^{D_{sos}}n^{F \eps|V(\sig \circ \sig')|}} \le 1
    \end{align*}
    for a constant $F > 0$ that only depends on $C_V, C_E$.
\end{proof}

\subsection{General strategy to lower bound $\sum_{V \in \mathcal{I}_{mid}}{M^{fact}(H_{Id_V})}$}\label{subsec: positivity_lower_bound_strategy}
In this section, we describe how to show that $\sum_{V \in \mathcal{I}_{mid}}{M^{fact}(H_{Id_V})} \succeq {\delta}Id_{Sym}$ for some $\delta > 0$ where $\delta$ will depend on $n$ and other parameters. For this, we use a similar strategy as \cite{jones2021sum}. For each $V \in \mathcal{I}_{mid}$, we choose a weight $w_V \in (0,1]$. We then observe that since each coefficient matrix $H_{Id_{V}}$ is PSD,
\[
\sum_{V \in \mathcal{I}_{mid}}{M^{fact}(H_{Id_V})} \succeq \sum_{V \in \mathcal{I}_{mid}}{{w_V}M^{fact}(H_{Id_V})}
\]
By choosing the weights $w_V$ appropriately, we can bound the off-diagonal parts by the diagonal parts, giving us ${\delta}Id_{Sym}$.
\begin{definition}
	For all $V \in \mathcal{I}_{mid}$ we define $Id_{Sym,V}$ to be the matrix such that
	\begin{enumerate}
		\item $Id_{Sym,V}(A,B) = 1$ if $A$ and $B$ both have index shape $V$.
		\item Otherwise, $Id_{Sym,V}(A,B) = 0$.
	\end{enumerate}
\end{definition}
\begin{proposition}
	$Id_{Sym} = \sum_{V \in \mathcal{I}_{mid}}{Id_{Sym,V}}$
\end{proposition}
\begin{definition}
    For each $V \in \mathcal{I}_{mid}$, we define $\lambda_V = |Aut(V)|H_{Id_V}(Id_V,Id_V)$.
\end{definition}
\begin{theorem}\label{thm: main_positivity}
    If $\{w_V: V \in \mathcal{I}_{mid}\}$ are weights such that for all $V \in \mathcal{I}_{mid}$ and all left shapes $\sigma \in \mathcal{L}_{V}$, $w_{V} \leq \frac{w_{U_{\sigma}}\lambda_{U_{\sigma}}}{|\mathcal{I}_{mid}|B_{norm}(\sigma)^2{c(\sigma)^2}{H_{Id_V}(\sigma,\sigma)}}$
    then
    \[
        \sum_{V \in \mathcal{I}_{mid}}{M^{fact}(H_{Id_V})}  \succeq \frac{1}{2}\sum_{V \in \mathcal{I}_{mid}}{{w_V}{\lambda_V}{Id_{Sym,V}}}
\succeq \frac{1}{2}\min_{V \in \mathcal{I}_{mid}}{\{{w_V}{\lambda_V}\}}Id_{Sym}
    \]
\end{theorem}
\begin{proof}
Observe that for each $V \in \mathcal{I}_{mid}$,
\[
{w_V}\sum_{\sigma,\sigma' \in \mathcal{L}_V}{H_{Id_V}(\sigma,\sigma')M_{\sigma}M_{\sigma'}^T} = {w_V}{\lambda_V}{Id_{Sym,V}} + {w_V}\sum_{\sigma,\sigma' \in \mathcal{L}_V: \sigma \neq Id_V \text{ or } \sigma' \neq Id_V}{H_{Id_V}(\sigma,\sigma')\left(\frac{M_{\sigma}M_{\sigma'}^T + M_{\sigma'}M_{\sigma}^T}{2}\right)}
\]
The first part of the right hand side is a diagonal part that we want to extract. We now show that we can bound the second part in terms of the diagonal parts.
\begin{proposition}\label{prop:boundingoffdiagonal}
For all $V \in \mathcal{I}_{mid}$ and all shapes $\sigma,\sigma' \in \mathcal{L}_V$, for all $a,b > 0$ such that $ab \geq B_{norm}(\sigma)^{2}B_{norm}(\sigma')^{2}$, if $\norm{M_{\sigma}} \leq B_{norm}(\sigma)$ and $\norm{M_{\sigma'}} \leq B_{norm}(\sigma')$ then
\[
M_{\sigma}M_{\sigma'}^T + M_{\sigma'}M_{\sigma}^T \succeq -{a}Id_{Sym,U_{\sigma}} -{b}Id_{Sym,U_{\sigma'}}
\]
\end{proposition}
\begin{corollary}
    If $H_{Id_V} \succeq 0$ then For any shapes $\sigma,\sigma' \in \mathcal{L}_V$,
    \begin{align*}
    w_{V}H_{Id_V}(\sigma,\sigma')\left(M_{\sigma}M_{\sigma'}^T + M_{\sigma'}M_{\sigma}^T\right) &\succeq -\frac{c(\sigma)}{c(\sigma')}{w_V}H_{Id_V}(\sigma,\sigma)B_{norm}(\sigma)^2{Id_{Sym,U_{\sigma}}} \\
    &- \frac{c(\sigma')}{c(\sigma)}{w_V}H_{Id_V}(\sigma',\sigma')B_{norm}(\sigma')^2{Id_{Sym,U_{\sigma'}}}
    \end{align*}
\end{corollary}
\begin{proof}
    This follows from Proposition \ref{prop:boundingoffdiagonal} and the observation that since $H_{Id_V} \succeq 0$, for all $\sigma,\sigma' \in \mathcal{L}_V$,  $H_{Id_V}(\sigma,\sigma')^2 \leq H_{Id_V}(\sigma,\sigma)H_{Id_V}(\sigma',\sigma')$
\end{proof}
Since $w_{V} \leq \frac{w_{U_{\sigma}}\lambda_{U_{\sigma}}}{|\mathcal{I}_{mid}|B_{norm}(\sigma)^2{c(\sigma)^2}{H_{Id_V}(\sigma,\sigma)}}$ and $w_{V} \leq \frac{w_{U_{\sigma'}}\lambda_{U_{\sigma'}}}{|\mathcal{I}_{mid}|B_{norm}(\sigma')^2{c(\sigma')^2}{H_{Id_V}(\sigma',\sigma')}}$, we have that
\begin{align*}
\sum_{\sigma,\sigma' \in \mathcal{L}_V: \sigma \neq Id_V \text{ or } \sigma' \neq Id_V}{w_{V}H_{Id_V}(\sigma,\sigma')\left(\frac{M_{\sigma}M_{\sigma'}^T + M_{\sigma'}M_{\sigma}^T}{2}\right)}
&\succeq -2\sum_{\sigma \in \mathcal{L}_V}{\frac{w_{U_{\sigma}}\lambda_{U_{\sigma}}Id_{Sym,U_{\sigma}}}{|\mathcal{I}_{mid}|c(\sigma)}\left(\sum_{\sigma' \in \mathcal{L}_V: \sigma' \neq Id_V}{\frac{1}{c(\sigma')}}\right)}\\
&\succeq -\frac{1}{4|\mathcal{I}_{mid}|}\sum_{U \in \mathcal{I}_{mid}}{w_{U}\lambda_{U}Id_{Sym,U}\left(\sum_{\sigma \in \mathcal{L}_{U,V}}{\frac{1}{c(\sigma)}}\right)}\\
&\succeq -\frac{1}{2|\mathcal{I}_{mid}|}\sum_{U \in \mathcal{I}_{mid}}{w_{U}\lambda_{U}Id_{Sym,U}}
\end{align*}
Thus, for each $V \in \mathcal{I}_{mid}$,
\[
{w_V}M_{fact}(H_{Id_V}) \succeq {w_V}{\lambda_V}{Id_{Sym,V}} - \frac{1}{2|\mathcal{I}_{mid}|}\sum_{U \in \mathcal{I}_{mid}}{w_{U}\lambda_{U}Id_{Sym,U}}
\]
Summing this equation over all $V \in \mathcal{V}$, we have that
\[
\sum_{V \in \mathcal{I}_{mid}}{M^{fact}(H_{Id_V})} \succeq \sum_{V \in \mathcal{I}_{mid}}{{w_V}M^{fact}(H_{Id_V})} \succeq \frac{1}{2}\sum_{V \in \mathcal{I}_{mid}}{{w_V}{\lambda_V}{Id_{Sym,V}}}
\succeq \frac{1}{2}\min_{V \in \mathcal{I}_{mid}}{\{{w_V}{\lambda_V}\}}Id_{Sym}
\]
as needed.
\end{proof}
\subsubsection{Handling Non-multilinear Matrix Indices*}
If there are multilinear matrix indices, then Theorem \ref{thm: main_positivity} still holds and it can be shown in a similar way, but we need to make a few adjustments.
\begin{enumerate}
    \item We modify the definition of $Id_{Sym,V}$ as follows. For all $V \in \mathcal{I}_{mid}$ we define $Id_{Sym,V}$ to be the matrix such that
	\begin{enumerate}
		\item $Id_{Sym,V}(A,B) = 1$ if $A$ and $B$ have the same index shape $U$ and $U$ has the same number of each type of vertex as $V$. Note that $B$ may be a permutation of $A$ and $U$ may have different powers than $V$.
		\item Otherwise, $Id_{Sym,V}(A,B) = 0$.
	\end{enumerate}
    Observe that with this modified definition, we will still have $Id_{Sym} = \sum_{V \in \mathcal{I}_{mid}}{Id_{Sym,V}}$.
    \item Instead of taking $\lambda_V = |Aut(V)|H_{Id_V}(Id_V,Id_V)$, we define $\lambda_V$ as follows. Letting $H_{Id_V,\text{no expansion}}$ be the diagonal submatrix of $H_{Id_V}$ indexed by left shapes $\sigma$ such that $U_{\sigma}$ has the same number of each type of vertex as $V$ (though the powers may be different), we take
    \[
    \lambda_V = |Aut(V)|min{\{\lambda: H_{Id_V,\text{no expansion}} \succeq {\lambda}Id_{Sym,V}\}}
    \]
    \item We similarly extend the definition of $c$ to left shapes $\sigma$ with multilinear indices in $U_{\sigma}$ so that we still have $\sum_{\sigma \in \mathcal{L}_V: (U_{\sigma})_{reduced} \neq v}{\frac{1}{c(\sigma)}} \leq \frac{1}{10}$
\end{enumerate}

%% file: planted_ds_quant.tex
In this section, we will prove our main theorem on Planted slightly denser subgraph, \cref{thm: plds_main}.

\PLDSmain*

We will apply the simplified machinery, in particular \cref{simplifiedmaintheorem}. Here, we choose $\eps$ in the theorem, not to be confused with the $\eps$ in \cref{thm: plds_main}, to be an arbitrarily small constant.
We build on the qualitative bounds (and use the same notation) from \cref{sec: plds_qual}.
The result will follow once we verify the main conditions and apply the machinery.

\subsection{\middleshapeboundstwo}

\begin{lemma}\label{lem: plds_charging}
    Suppose $k \le n^{1/2 - \eps}$. For all $U \in \calI_{mid}$ and $\tau \in \calM_U$,
	\[\sqrt{n}^{|V(\tau)| - |U_{\tau}|}S(\tau) \le \frac{1}{n^{C_p\eps|E(\tau)|}}\]
\end{lemma}

\begin{proof}
    This result follows by plugging in the value of $S(\tau)$. Using $k \le n^{1/2 - \eps}$,
	\begin{align*}
	\sqrt{n}^{|V(\tau)| - |U_{\tau}|}S(\tau) &= \sqrt{n}^{|V(\tau)| - |U_{\tau}|} \left(\frac{k}{n}\right)^{|V(\tau)| - |U_{\tau}|}(2(\frac{1}{2} + \frac{1}{2n^{C_p\eps}}) -1 )^{|E(\tau)|}
	\le \frac{1}{n^{C_p\eps|E(\tau)|}}
	\end{align*}
\end{proof}

\begin{corollary}\label{cor: plds_norm_decay}
	For all $U \in \calI_{mid}$ and $\tau \in \calM_U$, we have \[c(\tau)B_{norm}(\tau)S(\tau) \le 1\]
\end{corollary}

\begin{proof}
	Since $\tau$ is a proper middle shape, we have $w(I_{\tau}) = 0$ and $w(S_{\tau}) = w(U_{\tau})$. This implies
	$n^{\frac{w(V(\tau)) + w(I_{\tau}) - w(S_{\tau})}{2}} = \sqrt{n}^{|V(\tau)| - |U_{\tau}|}$.
	Since $\tau$ is proper, every vertex $i \in V(\tau) \setminus U_{\tau}$ or $i \in V(\tau) \setminus V_{\tau}$ has $deg^{\tau}(i) \ge 1$ and hence, $|V(\tau)\setminus U_{\tau}| + |V(\tau)\setminus V_{\tau}| \le 4|E(\tau)|$. Also, $q = n^{O(1) \cdot \eps C_V}$. We can set $C_V$ sufficiently small so that, using \cref{lem: plds_charging},
	{\footnotesize
	\begin{align*}
	c(\tau)B_{norm}(\tau)S(\tau)
	&= 100(3D_V)^{|U_{\tau}\setminus V_{\tau}| + |V_{\tau}\setminus U_{\tau}| + 2|E(\tau)|}2^{|V(\tau)\setminus (U_{\tau}\cup V_{\tau})|}
	\cdot (6D_V\sqrt[4]{2eq})^{|V(\tau)\setminus U_{\tau}| + |V(\tau)\setminus V_{\tau}|}\sqrt{n}^{|V(\tau)| - |U_{\tau}|}S(\tau)\\
	&\le n^{O(1) \cdot \eps C_V \cdot |E(\tau)|} \cdot \sqrt{n}^{|V(\tau)| - |U_{\tau}|}S(\tau)\\
	&\le n^{O(1) \cdot \eps C_V \cdot |E(\tau)|} \cdot \frac{1}{n^{C_p\eps|E(\tau)|}}\\
	&\le 1
	\end{align*}
	}
\end{proof}

We can now obtain middle shape bounds.

\begin{lemma}
    For all $U \in \calI_{mid}$ and $\tau \in \calM_U$,
    \[
\begin{bmatrix}
    \frac{1}{|Aut(U)|c(\tau)}H_{Id_U} & B_{norm}(\tau) H_{\tau}\\
    B_{norm}(\tau) H_{\tau}^T & \frac{1}{|Aut(U)|c(\tau)}H_{Id_U}
\end{bmatrix}
\succeq 0
\]
\end{lemma}

\begin{proof}
	We have
	\begin{align*}
	\begin{bmatrix}
	\frac{1}{|Aut(U)|c(\tau)}H_{Id_U} & B_{norm}(\tau)H_{\tau}\\
	B_{norm}(\tau)H_{\tau}^T & \frac{1}{|Aut(U)|c(\tau)}H_{Id_U}
	\end{bmatrix}
	=& \begin{bmatrix}
	\left(\frac{1}{|Aut(U)|c(\tau)} - \frac{S(\tau)B_{norm}(\tau)}{|Aut(U)|}\right)H_{Id_U} & 0\\
	0 & \left(\frac{1}{|Aut(U)|c(\tau)} - \frac{S(\tau)B_{norm}(\tau)}{|Aut(U)|}\right)H_{Id_U}
	\end{bmatrix}\\
	&+ B_{norm}(\tau)\begin{bmatrix}
	\frac{S(\tau)}{|Aut(U)|}H_{Id_U} & H_{\tau}\\
	H_{\tau}^T & \frac{S(\tau)}{|Aut(U)|}H_{Id_U}
	\end{bmatrix}
	\end{align*}
	By \cref{lem: plds_cond2_simplified}, $\begin{bmatrix}
	\frac{S(\tau)}{|Aut(U)|}H_{Id_U} & H_{\tau}\\
	H_{\tau}^T & \frac{S(\tau)}{|Aut(U)|}H_{Id_U}
	\end{bmatrix}
	\succeq 0$, so the second term above is positive semidefinite. For the first term, by \cref{lem: plds_cond1}, $H_{Id_U} \succeq 0$ and by \cref{cor: plds_norm_decay}, $\frac{1}{|Aut(U)|c(\tau)} - \frac{S(\tau)B_{norm}(\tau)}{|Aut(U)|} \ge 0$, which proves that the first term is also positive semidefinite.
\end{proof}

\subsection{\intersectionboundstwo}

\begin{lemma}\label{lem: plds_charging2}
	Suppose $k \le n^{1/2 - \eps}$. For all $U, V \in \calI_{mid}$ where $w(U) > w(V)$ and for all $\gam \in \Gam_{U, V}$,
	\[n^{w(V(\gam)\setminus U_{\gam})} S(\gam)^2 \le \frac{1}{n^{B\eps (|V(\gam) \setminus (U_{\gam} \cap V_{\gam})| + |E(\gam)|)}}\]
	for some constant $B$ that depends only on $C_p$. In particular, it is independent of $C_V$.
\end{lemma}

\begin{proof}
	Since $\gam$ is a left shape, we have $|U_{\gam}| \ge |V_{\gam}|$ as $V_{\gam}$ is the unique minimum vertex separator of $\gam$ and so, $n^{w(V(\gam) \setminus U_{\gam})} = n^{|V(\gam)| - |U_{\gam}|} \le n^{|V(\gam)| - \frac{|U_{\gam}| + |V_{\gam}|}{2}}$. Also, note that $2|V(\gam)| - |U_{\gam}| - |V_{\gam}| = |U_{\gam} \setminus V_{\gam}| + |V_{\gam} \setminus U_{\gam}| + 2|V(\gam) \setminus U_{\gam} \setminus V_{\gam}| \ge |V(\gam) \setminus (U_{\gam} \cap V_{\gam})|$. Therefore,
	\begin{align*}
	n^{w(V(\gam)\setminus U_{\gam})} S(\gam)^2 &= n^{|V(\gam)\setminus U_{\gam})|} \left(\frac{k}{n}\right)^{2|V(\gam)| - |U_{\gam}| - |V_{\gam}|} (2(\frac{1}{2} + \frac{1}{2n^{C_p\eps}}) - 1)^{2|E(\gam)|}\\
	&\le n^{|V(\gam)| - \frac{|U_{\gam}| + |V_{\gam}|}{2}}\left(\frac{1}{n^{1/2 + \eps}}\right)^{2|V(\gam)| - |U_{\gam}| - |V_{\gam}|}\left(\frac{1}{n^{2C_p\eps}}\right)^{|E(\gam)|}\\
	&\le  \left(\frac{1}{n^{\eps}}\right)^{2|V(\gam)| - |U_{\gam}| - |V_{\gam}|}\left(\frac{1}{n^{2C_p\eps}}\right)^{|E(\gam)|}\\
	&\le \frac{1}{n^{B\eps (|V(\gam) \setminus (U_{\gam} \cap V_{\gam})| + \sum_{e \in E(\gam)} l_e)}}
	\end{align*}
for a constant $B$ that depends only on $C_p$.
\end{proof}

We obtain intersection term bounds.

\begin{lemma}
    For all $U, V \in \calI_{mid}$ where $w(U) > w(V)$ and all $\gam \in \Gam_{U, V}$, \[c(\gam)^2N(\gam)^2B(\gam)^2H_{Id_V}^{-\gam, \gam} \preceq H_{\gam}'\]
\end{lemma}

\begin{proof}
	By \cref{lem: plds_cond3_simplified}, we have
	\begin{align*}
	c(\gam)^2N(\gam)^2B(\gam)^2H_{Id_V}^{-\gam, \gam} &= c(\gam)^2N(\gam)^2B(\gam)^2 S(\gam)^2 \frac{|Aut(U)|}{|Aut(V)|} H'_{\gam}
	\end{align*}
	Using the same proof as in \cref{lem: plds_cond1}, we can see that $H'_{\gam} \succeq 0$. Therefore, it suffices to prove that $c(\gam)^2N(\gam)^2B(\gam)^2 S(\gam)^2 \frac{|Aut(U)|}{|Aut(V)|} \le 1$.
	Since $U, V \in \calI_{mid}$, $|Aut(U)| = |U|!,|Aut(V)| = |V|!$. Therefore, $\frac{|Aut(U)|}{|Aut(V)|} = \frac{|U|!}{|V|!} \le D_V^{|U_{\gam} \setminus V_{\gam}|}$. Also, $q = n^{O(1) \cdot \eps C_V}$. Let $B$ be the constant from \cref{lem: plds_charging2}. We can set $C_V$ sufficiently small so that, using \cref{lem: plds_charging2},
	\begin{align*}
	c(\gam)^2N(\gam)^2B(\gam)^2S(\gam)^2 \frac{|Aut(U)|}{|Aut(V)|} &\le 100^2 (3D_V)^{2|U_{\gam}\setminus V_{\gam}| + 2|V_{\gam}\setminus U_{\gam}| + 4|E(\al)|}4^{|V(\gam) \setminus (U_{\gam} \cup V_{\gam})|}\\
	&\quad\cdot (3D_V)^{4|V(\gam)\setminus V_{\gam}| + 2|V(\gam)\setminus U_{\gam}|} (6D_V\sqrt[4]{2eq})^{2|V(\gam)\setminus U_{\gam}| + 2|V(\gam)\setminus V_{\gam}|}\\
	&\quad\cdot n^{w(V(\gam)\setminus U_{\gam})} S(\gam)^2 \cdot D_V^{|U_\gam \setminus V_{\gam}|} \\
	&\le n^{O(1) \cdot \eps C_V \cdot (|V(\gam) \setminus (U_{\gam} \cap V_{\gam})| + \sum_{e \in E(\gam)} l_e)} \cdot n^{w(V(\gam)\setminus U_{\gam})} S(\gam)^2\\
	&\le n^{O(1) \cdot \eps C_V \cdot (|V(\gam) \setminus (U_{\gam} \cap V_{\gam})| + \sum_{e \in E(\gam)} l_e)} \cdot \frac{1}{n^{B\eps (|V(\gam) \setminus (U_{\gam} \cap V_{\gam})| + \sum_{e \in E(\gam)} l_e)}}\\
	&\le 1
	\end{align*}
\end{proof}

\subsection{\truncationboundstwo}

In this section, we will prove truncation error bounds.
We use the strategy and notation from \cref{sec: showing_positivity}.
First, we will need a bound on $B_{norm}(\sig) B_{norm}(\sig') H_{Id_U}(\sig, \sig')$ that is obtained below.

\begin{lemma}\label{lem: plds_charging3}
	Suppose $k \le n^{1/2 - \eps}$. For all $U \in \calI_{mid}$ and $\sig, \sig' \in \calL_U$,
	\[B_{norm}(\sig) B_{norm}(\sig') H_{Id_U}(\sig, \sig') \le \frac{1}{n^{0.5\eps|V(\al)| + C_p\eps|E(\al)|}} \left(\frac{k}{n}\right)^{|U|}
	\]
\end{lemma}

\begin{proof}
	Let $\al = \sig \circ \sig'$. Observe that $|V(\sig)| + |V(\sig')| = |V(\al)| + |U|$. By choosing $C_V$ sufficiently small,
	\begin{align*}
	B_{norm}(\sig) B_{norm}(\sig') H_{Id_U}(\sig, \sig') &= (6D_V\sqrt[4]{2eq})^{|V(\sig)\setminus U_{\sig}| + |V(\sig)\setminus V_{\sig}|} n^{\frac{w(V(\sig)) - w(U)}{2}}\\
	&\quad\cdot (6D_V\sqrt[4]{2eq})^{|V(\sig')\setminus U_{\sig'}| + |V(\sig')\setminus V_{\sig'}|} n^{\frac{w(V(\sig')) - w(U)}{2}}\\
	&\quad\cdot \frac{1}{|Aut(U)|} \left(\frac{k}{n}\right)^{|V(\al)|} (2(\frac{1}{2} + \frac{1}{2n^{C_p\eps}}) - 1)^{|E(\al)|}\\
	&\le n^{O(1) \cdot \eps C_V \cdot |V(\al)|} \sqrt{n}^{|V(\sig)| - |U|}\sqrt{n}^{|V(\sig')| - |U|} \left(\frac{k}{n}\right)^{|V(\al)|}\frac{1}{n^{C_p\eps|E(\al)|}}\\
	&\le \frac{1}{n^{0.5\eps|V(\al)| + C_p\eps|E(\al)|}} \left(\frac{k}{n}\right)^{|U|}
	\end{align*}
\end{proof}

Now, we are ready to apply the strategy.

\begin{restatable}{lemma}{PLDSfive}\label{lem: plds_cond5}
	Whenever $\norm{M_{\al}} \le B_{norm}(\al)$ for all $\al \in \calM'$,
	\[
	\sum_{U \in \mathcal{I}_{mid}}{M^{fact}_{Id_U}{(H_{Id_U})}} \succeq \frac{1}{n^{K_1D_{sos}^2}} Id_{sym}
	\]
	for a constant $K_1 > 0$.
\end{restatable}

\begin{proof}
    For $V \in \calI_{mid}$, we have $\lda_V = \left(\frac{k}{n}\right)^{|V|}$. Now, we choose $w_V = \left(\frac{k}{n}\right)^{D_{sos} - |V|}$. Then, for all $\sig \in \calL_{V}$, we have $w_{V} \leq \frac{w_{U_{\sigma}}\lambda_{U_{\sigma}}}{|\mathcal{I}_{mid}|B_{norm}(\sigma)^2{c(\sigma)^2}{H_{Id_V}(\sigma,\sigma)}}$ which is easily verified using \cref{lem: plds_charging3}. The result follows from \cref{thm: main_positivity}.
\end{proof}

\begin{restatable}{lemma}{PLDSsix}\label{lem: plds_cond6}
	\[\sum_{U\in \calI_{mid}} \sum_{\gam \in \Gam_{U, *}} \frac{d_{Id_{U}}(H_{Id_{U}}, H'_{\gam})}{|Aut(U)|c(\gam)} \le \frac{n^{K_2 D_{sos}}}{2^{D_V}}\]
	for a constant $K_2 > 0$.
\end{restatable}

\begin{proof}
	We have
	\begin{align*}
	\sum_{U\in \calI_{mid}} \sum_{\gam \in \Gam_{U, *}} \frac{d_{Id_{U}}(H_{Id_{U}}, H'_{\gam})}{|Aut(U)|c(\gam)} = \sum_{U\in \calI_{mid}} \sum_{\gam \in \Gam_{U, *}} \frac{1}{|Aut(U)|c(\gam)}\sum_{\sigma,\sigma' \in \mathcal{L}_{U_{\gamma}}: |V(\sigma)| \leq D_V, |V(\sigma')| \leq D_V,
		\atop |V(\sigma \circ \gamma)| > D_V \text{ or } |V(\sigma' \circ \gamma)| > D_V}{B_{norm}(\sigma)B_{norm}(\sigma')H_{Id_{U_{\gamma}}}(\sigma,\sigma')}\\
	\end{align*}
	The set of $\sig, \sig'$ that could appear in the above sum must necessarily be non-trivial and hence, $\sig, \sig' \in \calL_U'$. Then,
	\begin{align*}
	\sum_{U\in \calI_{mid}} &\sum_{\gam \in \Gam_{U, *}} \frac{d_{Id_{U}}(H_{Id_{U}}, H'_{\gam})}{|Aut(U)|c(\gam)}\\
	&= \sum_{U\in \calI_{mid}} \sum_{\sigma,\sigma' \in \mathcal{L}'_{U}} {B_{norm}(\sigma)B_{norm}(\sigma')H_{Id_{U}}(\sigma,\sigma')}\sum_{\gam \in \Gam_{U, *}: |V(\sigma \circ \gamma)| > D_V \text{ or } |V(\sigma' \circ \gamma)| > D_V} \frac{1}{|Aut(U)|c(\gam)}
	\end{align*}
	For $\sig \in \calL'_{U}$, define $m_{\sig} = D_V + 1 - |V(\sig)| \ge 1$. This is precisely set so that for all $\gam \in \Gam_{U, *}$, we have $|V(\sigma \circ \gamma)| > D_V$ if and only if $|V(\gam)| \ge |U| + m_{\sig}$. So, for $\sig, \sig' \in \calL'_U$, using \cref{lem: tpca_refined_bound}
	\begin{align*}
	\sum_{\gam \in \Gam_{U, *}: |V(\sigma \circ \gamma)| > D_V \text{ or } |V(\sigma' \circ \gamma)| > D_V} \frac{1}{|Aut(U)|c(\gam)} &=
	\sum_{\gam \in \Gam_{U, *}: |V(\gam)| \ge |U| + \min(m_{\sig}, m_{\sig'})} \frac{1}{|Aut(U)|c(\gam)}\\
	&\le \frac{1}{2^{\min(m_{\sig}, m_{\sig'}) - 1}}
	\end{align*}
	Also, for $\sig, \sig' \in \calL_U'$, we have $|V(\sig \circ \sig')| + min(m_{\sig}, m_{\sig'}) - 1 \ge D_V$.
	Therefore,
	\begin{align*}
		\sum_{U\in \calI_{mid}} \sum_{\gam \in \Gam_{U, *}} \frac{d_{Id_{U}}(H_{Id_{U}}, H'_{\gam})}{|Aut(U)|c(\gam)} &\le \sum_{U\in \calI_{mid}} \sum_{\sigma,\sigma' \in \mathcal{L}'_{U}} {B_{norm}(\sigma)B_{norm}(\sigma')H_{Id_{U}}(\sigma,\sigma')\frac{1}{2^{\min(m_{\sig}, m_{\sig'}) - 1}}}\\
		&\le \sum_{U\in \calI_{mid}} \sum_{\sigma,\sigma' \in \mathcal{L}'_{U}}\frac{n^{O(1) D_{sos}}}{n^{0.5\eps|V(\sig \circ \sig')|}2^{\min(m_{\sig}, m_{\sig'}) - 1}}
	\end{align*}
	where we used \cref{lem: plds_charging3}. Using $n^{0.5\eps |V(\sig \circ \sig')|} \ge n^{0.1\eps |V(\sig \circ \sig')|}2^{|V(\sig \circ \sig')|}$,
	\begin{align*}
		\sum_{U\in \calI_{mid}} \sum_{\gam \in \Gam_{U, *}} \frac{d_{Id_{U}}(H_{Id_{U}}, H'_{\gam})}{|Aut(U)|c(\gam)} &\le \sum_{U\in \calI_{mid}} \sum_{\sigma,\sigma' \in \mathcal{L}'_{U}}\frac{n^{O(1) D_{sos}}}{n^{0.1\eps|V(\sig \circ \sig')|} 2^{|V(\sig \circ \sig')|}2^{\min(m_{\sig}, m_{\sig'}) - 1}}\\
		&\le \sum_{U\in \calI_{mid}} \sum_{\sigma,\sigma' \in \mathcal{L}'_{U}}\frac{n^{O(1) D_{sos}}}{n^{0.1\eps|V(\sig \circ \sig')|} 2^{D_V}}\\
		&\le \sum_{U\in \calI_{mid}} \sum_{\sigma,\sigma' \in \mathcal{L}'_{U}}\frac{n^{O(1) D_{sos}}}{D_{sos}^{D_{sos}}n^{0.1\eps|V(\sig \circ \sig')|} 2^{D_V}}
	\end{align*}
	The final step will be to argue that $\sum_{U\in \calI_{mid}} \sum_{\sigma,\sigma' \in \mathcal{L}'_{U}}\frac{1}{D_{sos}^{D_{sos}}n^{0.1 \eps|V(\sig \circ \sig')|}} \le 1$ which will complete the proof. But this will follow from \cref{lem: gp_sum} if we set $C_V$ small enough.
\end{proof}

We conclude the following.

\begin{lemma}
    Whenever $\norm{M_{\alpha}} \le B_{norm}(\alpha)$ for all $\alpha \in \mathcal{M}'$,
    \[
    \sum_{U \in \mathcal{I}_{mid}}{M^{fact}_{Id_U}{(H_{Id_U})}} \succeq 6\left(\sum_{U \in \mathcal{I}_{mid}}{\sum_{\gamma \in \Gamma_{U,*}}{\frac{d_{Id_{U}}(H'_{\gamma},H_{Id_{U}})}{|Aut(U)|c(\gamma)}}}\right)Id_{sym}
    \]
\end{lemma}

\begin{proof}
	Choose $C_{sos}$ sufficiently small so that $\frac{1}{n^{K_1D_{sos}^2}} \ge 6\frac{n^{K_2D_{sos}}}{2^{D_V}}$ which can be satisfied by setting $C_{sos} < K_3 C_V$ for a sufficiently small constant $K_3 > 0$. Then, since $Id_{Sym} \succeq 0$, using \cref{lem: plds_cond5} and \cref{lem: plds_cond6},
	\begin{align*}
		\sum_{U \in \mathcal{I}_{mid}}{M^{fact}_{Id_U}{(H_{Id_U})}} &\succeq \frac{1}{n^{K_1D_{sos}^2}} Id_{sym}\\
		&\succeq 6\frac{n^{K_2D_{sos}}}{2^{D_V}} Id_{sym}\\
		&\succeq 6\left(\sum_{U \in \mathcal{I}_{mid}}{\sum_{\gamma \in \Gamma_{U,*}}{\frac{d_{Id_{U}}(H'_{\gamma},H_{Id_{U}})}{|Aut(U)|c(\gamma)}}}\right)Id_{sym}
	\end{align*}
\end{proof}

%% file: tensor_pca_quant.tex
In this section, we will prove all the bounds required to prove \cref{thm: tpca_main}.

\TPCAmain*

We reuse the notation and qualitative bounds from \cref{sec: tpca_qual}.
Once we verify the conditions, this theorem will simply follow from the machinery, \cref{generalmaintheorem}.

\subsection{\middleshapeboundstwo}

\begin{lemma}\label{lem: tpca_charging}
	Suppose $\lda \le n^{\frac{k}{4} - \eps}$. For all $U \in \calI_{mid}$ and $\tau \in \calM_U$, suppose $deg^{\tau}(i)$ is even for all $i \in V(\tau) \setminus U_{\tau} \setminus V_{\tau}$, then
	\[\sqrt{n}^{|V(\tau)| - |U_{\tau}|}S(\tau) \le \frac{1}{n^{0.5\eps\sum_{e \in E(\tau)} l_e}}\]
\end{lemma}

\begin{proof}
	Firstly, we claim that $\sum_{e \in E(\tau)} kl_e \ge 2(|V(\tau)| - |U_{\tau}|)$. For any vertex $i \in V(\tau) \setminus U_{\tau} \setminus V_{\tau}$, $deg^{\tau}(i)$ is even and is not $0$, hence, $deg^{\tau}(i) \ge 2$. Any vertex $i \in U_{\tau} \setminus V_{\tau}$ cannot have $deg^{\tau}(i) = 0$ otherwise $U_{\tau} \setminus\{i\}$ is a vertex separator of strictly smaller weight than $U_{\tau}$, which is not possible, hence, $deg^{\tau}(i) \ge 1$. Therefore,
	\begin{align*}
	\sum_{e \in E(\tau)}kl_e &= \sum_{i \in V(\tau)} deg^{\tau}(i)\\
	&\ge \sum_{i \in V(\tau) \setminus U_{\tau} \setminus V_{\tau}} deg^{\tau}(i) + \sum_{i \in U_{\tau} \setminus V_{\tau}} deg^{\tau}(i) + \sum_{i \in V_{\tau} \setminus U_{\tau}} deg^{\tau}(i)\\
	&\ge 2|V(\tau) \setminus U_{\tau} \setminus V_{\tau}| + |U_{\tau} \setminus V_{\tau}| + |V_{\tau} \setminus U_{\tau}|\\
	&= 2(|V(\tau)| - |U_{\tau}|)
	\end{align*}
	By choosing $C_{\Del}$ sufficiently small, we have
	\begin{align*}
	\sqrt{n}^{|V(\tau)| - |U_{\tau}|}S(\tau) &= \sqrt{n}^{|V(\tau)| - |U_{\tau}|} \Delta^{|V(\tau)| - |U_{\tau}|}\prod_{e \in E(\tau)}\left(\frac{\lda}{(\Del n)^{\frac{k}{2}}}\right)^{l_e}\\
	&\le \sqrt{n}^{|V(\tau)| - |U_{\tau}|}\Delta^{|V(\tau)| - |U_{\tau}|}\prod_{e \in E(\tau)}n^{(-\frac{k}{4} - 0.5\eps)l_e}\\
	&= \sqrt{n}^{|V(\tau)| - |U_{\tau}| - \frac{\sum_{e \in E(\tau)}kl_e}{2}}\Delta^{|V(\tau)| - |U_{\tau}|}\prod_{e \in E(\tau)}n^{-0.5\eps l_e}\\
	&= \Delta^{|V(\tau)| - |U_{\tau}|}\prod_{e \in E(\tau)}n^{-0.5 \eps l_e}\\
	&\le\frac{1}{n^{0.5\eps\sum_{e \in E(\tau)} l_e}}
	\end{align*}
\end{proof}

\begin{corollary}\label{cor: tpca_norm_decay}
	For all $U \in \calI_{mid}$ and $\tau \in \calM_U$, we have \[c(\tau)B_{norm}(\tau)S(\tau) \le 1\]
\end{corollary}

\begin{proof}
	Since $\tau$ is a proper middle shape, we have $w(I_{\tau}) = 0$ and $w(S_{\tau, min}) = w(U_{\tau})$. This implies
	$n^{\frac{w(V(\tau)) + w(I_{\tau}) - w(S_{\tau, min})}{2}} = \sqrt{n}^{|V(\tau)| - |U_{\tau}|}$.
	If $deg^{\tau}(i)$ is odd for any vertex $i \in V(\tau) \setminus U_{\tau} \setminus V_{\tau}$, then $S(\tau) = 0$ and the inequality is true. So, assume $deg^{\tau}(i)$ is even for all $i \in V(\tau) \setminus U_{\tau} \setminus V_{\tau}$.	As was observed in the proof of \cref{lem: tpca_charging}, every vertex $i \in V(\tau) \setminus U_{\tau}$ or $i \in V(\tau) \setminus V_{\tau}$ has $deg^{\tau}(i) \ge 1$ and hence, $|V(\tau)\setminus U_{\tau}| + |V(\tau)\setminus V_{\tau}| \le 4 \sum_{e \in E(\tau)} l_e$. Also, $|E(\tau)| \le \sum_{e \in E(\tau)} l_e$ and $q = n^{O(1) \cdot \eps (C_V + C_E)}$. We can set $C_V, C_E$ sufficiently small so that, using \cref{lem: tpca_charging},
	\begin{align*}
	c(\tau)B_{norm}(\tau)S(\tau)
	&= 100(3D_V)^{|U_{\tau}\setminus V_{\tau}| + |V_{\tau}\setminus U_{\tau}| + k|E(\tau)|}2^{|V(\tau)\setminus (U_{\tau}\cup V_{\tau})|}\\
	&\quad\cdot 2e(6qD_V)^{|V(\tau)\setminus U_{\tau}| + |V(\tau)\setminus V_{\tau}|}\prod_{e \in E(\tau)} (400D_V^2D_E^2q)^{l_e}\sqrt{n}^{|V(\tau)| - |U_{\tau}|}S(\tau)\\
	&\le n^{O(1) \cdot \eps(C_V + C_E) \cdot \sum_{e \in E(\tau)} l_e} \cdot \sqrt{n}^{|V(\tau)| - |U_{\tau}|}S(\tau)\\
	&\le n^{O(1) \cdot \eps(C_V + C_E) \cdot \sum_{e \in E(\tau)} l_e} \cdot \frac{1}{n^{0.5\eps\sum_{e \in E(\tau)} l_e}}\\
	&\le 1
	\end{align*}
\end{proof}

We can now show middle shape bounds.

\begin{lemma}\label{lem: tpca_cond2}
    For all $U \in \calI_{mid}$ and $\tau \in \calM_U$,
    \[
    \begin{bmatrix}
        \frac{1}{|Aut(U)|c(\tau)}H_{Id_U} & B_{norm}(\tau) H_{\tau}\\
        B_{norm}(\tau) H_{\tau}^T & \frac{1}{|Aut(U)|c(\tau)}H_{Id_U}
    \end{bmatrix}
    \succeq 0
    \]
\end{lemma}

\begin{proof}
	We have
	\begin{align*}
	\begin{bmatrix}
	\frac{1}{|Aut(U)|c(\tau)}H_{Id_U} & B_{norm}(\tau)H_{\tau}\\
	B_{norm}(\tau)H_{\tau}^T & \frac{1}{|Aut(U)|c(\tau)}H_{Id_U}
	\end{bmatrix}
	=& \begin{bmatrix}
	\left(\frac{1}{|Aut(U)|c(\tau)} - \frac{S(\tau)B_{norm}(\tau)}{|Aut(U)|}\right)H_{Id_U} & 0\\
	0 & \left(\frac{1}{|Aut(U)|c(\tau)} - \frac{S(\tau)B_{norm}(\tau)}{|Aut(U)|}\right)H_{Id_U}
	\end{bmatrix}\\
	&+ B_{norm}(\tau)\begin{bmatrix}
	\frac{S(\tau)}{|Aut(U)|}H_{Id_U} & H_{\tau}\\
	H_{\tau}^T & \frac{S(\tau)}{|Aut(U)|}H_{Id_U}
	\end{bmatrix}
	\end{align*}
	By \cref{lem: tpca_cond2_simplified}, $\begin{bmatrix}
	\frac{S(\tau)}{|Aut(U)|}H_{Id_U} & H_{\tau}\\
	H_{\tau}^T & \frac{S(\tau)}{|Aut(U)|}H_{Id_U}
	\end{bmatrix}
	\succeq 0$, so the second term above is positive semidefinite. For the first term, by \cref{lem: tpca_cond1}, $H_{Id_U} \succeq 0$ and by \cref{cor: tpca_norm_decay}, $\frac{1}{|Aut(U)|c(\tau)} - \frac{S(\tau)B_{norm}(\tau)}{|Aut(U)|} \ge 0$, which proves that the first term is also positive semidefinite.
\end{proof}

\subsection{\intersectionboundstwo}

\begin{lemma}\label{lem: tpca_charging2}
	Suppose $\lda \le n^{\frac{k}{4} - \eps}$. For all $U, V \in \calI_{mid}$ where $w(U) > w(V)$ and for all $\gam \in \Gam_{U, V}$,
	\[n^{w(V(\gam)\setminus U_{\gam})} S(\gam)^2 \le \frac{1}{n^{B\eps (|V(\gam) \setminus (U_{\gam} \cap V_{\gam})| + \sum_{e \in E(\gam)} l_e)}}\]
	for some constant $B$ that depends only on $C_{\Del}$. In particular, it is independent of $C_V$ and $C_E$.
\end{lemma}

\begin{proof}
	Suppose there is a vertex $i \in V(\gam) \setminus U_{\gam} \setminus V_{\gam}$ such that $deg^{\gam}(i)$ is odd, then $S(\gam) = 0$ and the inequality is true. So, assume $deg^{\gam}(i)$ is even for all vertices $i \in V(\gam) \setminus U_{\gam} \setminus V_{\gam}$.
	We first claim that $k\sum_{e \in E(\gam)} l_e \ge 2|V(\gam) \setminus U_{\gam}|$. Since $\gam$ is a left shape, all vertices $i$ in $V(\gam) \setminus U_{\gam}$ have $deg^{\gam}(i) \ge 1$. In particular, all vertices $i \in V_{\gam} \setminus U_{\gam}$ have $deg^{\gam}(i) \ge 1$.
	Moreover, if $i \in V(\gam) \setminus U_{\gam} \setminus V_{\gam}$, since $deg^{\gam}(i)$ is even, we must have $deg^{\gam}(i) \ge 2$.

	Let $S'$ be the set of vertices $i \in U_{\gam} \setminus V_{\gam}$ that have $deg^{\gam}(i) \ge 1$. Then, note that $|S'| + |U_{\gam} \cap V_{\gam}| \ge |V_{\gam}| \Longrightarrow |S'| \ge |V_{\gam} \setminus U_{\gam}|$ since otherwise $S' \cup (U_{\gam} \cap V_{\gam})$ will be a vertex separator of $\gam$ of weight strictly less than $V_{\gam}$, which is not possible. Then,
	\begin{align*}
	\sum_{e \in E(\gam)}kl_e &= \sum_{i \in V(\gam)} deg^{\gam}(i)\\
	&\ge \sum_{i \in V(\gam) \setminus U_{\gam} \setminus V_{\gam}} deg^{\gam}(i) + \sum_{i \in U_{\gam} \setminus V_{\gam}} deg^{\gam}(i) + \sum_{i \in V_{\gam} \setminus U_{\gam}} deg^{\gam}(i)\\
	&\ge 2|V(\gam) \setminus U_{\gam} \setminus V_{\gam}| + |S'| + |V_{\gam} \setminus U_{\gam}|\\
	&\ge 2|V(\gam) \setminus U_{\gam} \setminus V_{\gam}| + 2|V_{\gam} \setminus U_{\gam}|\\
	&= 2|V(\gam) \setminus U_{\gam}|
	\end{align*}

	Finally, note that $2|V(\gam)| - |U_{\gam}| - |V_{\gam}| = |U_{\gam} \setminus V_{\gam}| + |V_{\gam} \setminus U_{\gam}| + 2|V(\gam) \setminus U_{\gam} \setminus V_{\gam}| \ge |V(\gam) \setminus (U_{\gam} \cap V_{\gam})|$. By choosing $C_{\Del}$ sufficiently small, we have
	\begin{align*}
	n^{w(V(\gam)\setminus U_{\gam})} S(\gam)^2 &= n^{|V(\gam)\setminus U_{\gam})|} \Delta^{2|V(\gam)| - |U_{\gam}| - |V_{\gam}|} \prod_{e \in E(\gam)} \left(\frac{\lda^2}{(\Del n)^k}\right)^{l_e}\\
	&\le n^{|V(\gam)\setminus U_{\gam})|} \Delta^{2|V(\gam)| - |U_{\gam}| - |V_{\gam}|} \prod_{e \in E(\gam)} n^{-(\frac{k}{2} + \eps)l_e}\\
	&\le \Delta^{2|V(\gam)| - |U_{\gam}| - |V_{\gam}|} \prod_{e \in E(\gam)} n^{-\eps l_e}\\
	&\le \frac{1}{n^{B\eps (|V(\gam) \setminus (U_{\gam} \cap V_{\gam})| + \sum_{e \in E(\gam)} l_e)}}
	\end{align*}
for a constant $B$ that depends only on $C_{\Del}$.
\end{proof}

\begin{remk}
	In the above bounds, note that there is a decay of $n^{B\eps}$ for each vertex in $V(\gam) \setminus (U_{\gam} \cap V_{\gam})$.	One of the main technical reasons for introducing the slack parameter $C_{\Del}$ in the planted distribution was to introduce this decay, which is needed in the current machinery.
\end{remk}

We can now obtain the intersection term bounds.

\begin{lemma}\label{lem: tpca_cond3}
    For all $U, V \in \calI_{mid}$ where $w(U) > w(V)$ and all $\gam \in \Gam_{U, V}$, \[c(\gam)^2N(\gam)^2B(\gam)^2H_{Id_V}^{-\gam, \gam} \preceq H_{\gam}'\]
\end{lemma}

\begin{proof}
	By \cref{lem: tpca_cond3_simplified}, we have
	\begin{align*}
	c(\gam)^2N(\gam)^2B(\gam)^2H_{Id_V}^{-\gam, \gam} &\preceq c(\gam)^2N(\gam)^2B(\gam)^2 S(\gam)^2 \frac{|Aut(U)|}{|Aut(V)|} H'_{\gam}
	\end{align*}
	Using the same proof as in \cref{lem: tpca_cond1}, we can see that $H'_{\gam} \succeq 0$. Therefore, it suffices to prove that $c(\gam)^2N(\gam)^2B(\gam)^2 S(\gam)^2 \frac{|Aut(U)|}{|Aut(V)|} \le 1$.
	Since $U, V \in \calI_{mid}$, $|Aut(U)| = |U|!,|Aut(V)| = |V|!$. Therefore, $\frac{|Aut(U)|}{|Aut(V)|} = \frac{|U|!}{|V|!} \le D_V^{|U_{\gam} \setminus V_{\gam}|}$. Also, $|E(\gam)| \le \sum_{e \in E(\gam)} l_e$ and $q = n^{O(1) \cdot \eps (C_V + C_E)}$. Let $B$ be the constant from \cref{lem: tpca_charging2}. We can set $C_V, C_E$ sufficiently small so that, using \cref{lem: tpca_charging2},
	\begin{align*}
	c(\gam)^2N(\gam)^2B(\gam)^2S(\gam)^2 \frac{|Aut(U)|}{|Aut(V)|} &\le 100^2 (3D_V)^{2|U_{\gam}\setminus V_{\gam}| + 2|V_{\gam}\setminus U_{\gam}| + 2k|E(\al)|}4^{|V(\gam) \setminus (U_{\gam} \cup V_{\gam})|}\\
	&\quad\cdot (3D_V)^{4|V(\gam)\setminus V_{\gam}| + 2|V(\gam)\setminus U_{\gam}|} (6qD_V)^{2|V(\gam)\setminus U_{\gam}| + 2|V(\gam)\setminus V_{\gam}|} \prod_{e \in E(\gam)} (400D_V^2D_E^2q)^{2l_e}\\
	&\quad\cdot  n^{w(V(\gam)\setminus U_{\gam})} S(\gam)^2 \cdot D_V^{|U_\gam \setminus V_{\gam}|} \\
	&\le n^{O(1) \cdot \eps(C_V + C_E) \cdot (|V(\gam) \setminus (U_{\gam} \cap V_{\gam})| + \sum_{e \in E(\gam)} l_e)} \cdot n^{w(V(\gam)\setminus U_{\gam})} S(\gam)^2\\
	&\le n^{O(1) \cdot \eps(C_V + C_E) \cdot (|V(\gam) \setminus (U_{\gam} \cap V_{\gam})| + \sum_{e \in E(\gam)} l_e)} \cdot \frac{1}{n^{B\eps (|V(\gam) \setminus (U_{\gam} \cap V_{\gam})| + \sum_{e \in E(\gam)} l_e)}}\\
	&\le 1
	\end{align*}
\end{proof}

\subsection{\truncationboundstwo}

In this section, we will obtain the truncation error bounds using the strategy sketched in \cref{sec: showing_positivity}. We also reuse the notation. First, we need the following bound on $B_{norm}(\sig) B_{norm}(\sig') H_{Id_U}(\sig, \sig')$.

\begin{lemma}\label{lem: tpca_charging3}
	Suppose $\lda = n^{\frac{k}{4} - \eps}$. For all $U \in \calI_{mid}$ and $\sig, \sig' \in \calL_U$,
	\[B_{norm}(\sig) B_{norm}(\sig') H_{Id_U}(\sig, \sig') \le \frac{1}{n^{0.5\eps C_{\Del}|V(\sig \circ \sig')|}\Del^{D_{sos}}n^{|U|}}\]
\end{lemma}

\begin{proof}
	Suppose there is a vertex $i \in V(\sig) \setminus V_{\sig}$ such that $deg^{\sig}(i) + deg^{U_{\sig}}(i)$ is odd, then $H_{Id_U}(\sig, \sig') = 0$ and the inequality is true. So, assume that $deg^{\sig}(i) + deg^{U_{\sig}}(i)$ is even for all $i \in V(\sig) \setminus V_{\sig}$. Similarly, assume that $deg^{\sig'}(i) + deg^{U_{\sig'}}(i)$ is even for all $i \in V(\sig') \setminus V_{\sig'}$. Also, if $\rho_{\sig} \neq \rho_{\sig'}$, we will have $H_{Id_U}(\sig, \sig') = 0$ and we'd be done. So, assume $\rho_{\sig} = \rho_{\sig'}$.

	Let $\al = \sig \circ \sig'$. We will first prove that $\sum_{e \in E(\al)} kl_e + 2deg(\al) \ge 2|V(\al)| + 2|U|$. Firstly, note that all vertices $i \in V(\al) \setminus (U_{\al} \cup V_{\al})$ have $deg^{\al}(i)$ to be even and nonzero, and hence at least $2$. Moreover, in both the sets $U_{\al} \setminus (U_{\al} \cap V_{\al})$ and $V_{\al} \setminus (U_{\al} \cap V_{\al})$, there are at least $|U| - |U_{\al} \cap V_{\al}|$ vertices of degree at least $1$, because $U$ is a minimum vertex separator. Also, note that $deg(\al) \ge |U_{\al}| + |V_{\al}|$. This implies that
	\begin{align*}
	\sum_{e \in E(\al)} kl_e + 2deg(\al) &\ge 2 |V(\al) \setminus (U_{\al} \cup V_{\al})| + 2(|U| - |U_{\al} \cap V_{\al}|) + 2(|U_{\al}| + |V_{\al}|)\\
	&= 2 (|V(\al)| - |U_{\al} \cup V_{\al}|) + 2(|U| - |U_{\al} \cap V_{\al}|) + 2(|U_{\al} \cup V_{\al}| + |U_{\al} \cap V_{\al}|)\\
	&= 2|V(\al)| + 2|U|
	\end{align*}
	where we used the fact that $U_{\al} \cap V_{\al} \subseteq U$. Finally, by choosing $C_V, C_E$ sufficiently small,
	\begin{align*}
	B_{norm}(\sig) B_{norm}(\sig') H_{Id_U}(\sig, \sig') &= 2e(6qD_V)^{|V(\sig)\setminus U_{\sig}| + |V(\sig)\setminus V_{\sig}|}\prod_{e \in E(\sig)} (400D_V^2D_E^2q)^{l_e} n^{\frac{w(V(\sig)) - w(U)}{2}}\\
	&\quad\cdot 2e(6qD_V)^{|V(\sig')\setminus U_{\sig'}| + |V(\sig')\setminus V_{\sig'}|}\prod_{e \in E(\sig')} (400D_V^2D_E^2q)^{l_e} n^{\frac{w(V(\sig')) - w(U)}{2}}\\
	&\quad\cdot \frac{1}{|Aut(U)|} \Delta^{|V(\al)|} \left(\frac{1}{\sqrt{\Delta n}}\right)^{deg(\al)} \prod_{e \in E(\al)} \left(\frac{\lda}{(\Del n)^{\frac{k}{2}}}\right)^{l_e}\\
	&\le n^{O(1) \cdot \eps (C_V + C_E) \cdot (|V(\al)| + \sum_{e \in E(\al)} l_e)} \Delta^{|V(\al)|}\left(\frac{1}{\sqrt{\Del}}\right)^{deg(\al)}\\
	&\quad\cdot \sqrt{n}^{|V(\al)| - |U|} \left(\frac{1}{\sqrt{n}}\right)^{deg(\al)}\prod_{e \in E(\al)}n^{(-\frac{k}{4} - 0.5\eps)l_e}\\
	&\le \frac{n^{O(1) \cdot \eps (C_V + C_E) \cdot (|V(\al)| + \sum_{e \in E(\al)} l_e)}}{n^{\eps C_{\Delta}|V(\al)|}n^{0.5\eps\sum_{e \in E(\al)} l_e}} \cdot \frac{1}{\Del^{D_{sos}}n^{|U|}}\sqrt{n}^{|V(\al)| + |U| - deg(\al) - \frac{1}{2}\sum_{e \in E(\al)} kl_e}\\
	&\le \frac{1}{n^{0.5\eps C_{\Del}|V(\al)|}\Del^{D_{sos}}n^{|U|}}
	\end{align*}
where we used the facts $\Del \le 1, deg(\al) \le 2D_{sos}$.
\end{proof}

We now apply the strategy by showing the following bounds.

\begin{restatable}{lemma}{TPCAfive}\label{lem: tpca_cond5}
	Whenever $\norm{M_{\al}} \le B_{norm}(\al)$ for all $\al \in \calM'$,
	\[
	\sum_{U \in \mathcal{I}_{mid}}{M^{fact}_{Id_U}{(H_{Id_U})}} \succeq \frac{\Del^{2D_{sos}^2}}{n^{D_{sos}}} Id_{sym}
	\]
\end{restatable}

\begin{proof}
    For $V \in \calI_{mid}$, $\lda_V = \frac{1}{n^{|V|}}$. We then choose $w_V = \left(\frac{1}{n}\right)^{D_{sos} - |V|}$. For all left shapes $\sig \in \calL_V$, it's easy to verify $w_{V} \leq \frac{w_{U_{\sigma}}\lambda_{U_{\sigma}}}{|\mathcal{I}_{mid}|B_{norm}(\sigma)^2{c(\sigma)^2}{H_{Id_V}(\sigma,\sigma)}}$ using \cref{lem: tpca_charging3}. \cref{thm: main_positivity} completes the proof.
\end{proof}

\begin{restatable}{lemma}{TPCAsix}\label{lem: tpca_cond6}
	\[\sum_{U\in \calI_{mid}} \sum_{\gam \in \Gam_{U, *}} \frac{d_{Id_{U}}(H_{Id_{U}}, H'_{\gam})}{|Aut(U)|c(\gam)} \le \frac{1}{\Delta^{2D_{sos}}2^{D_V}}\]
\end{restatable}

\begin{proof}
    We use the same argument and notation as in \cref{lem: plds_cond6}. When we plug in the bounds, we get
	\begin{align*}
	\sum_{U\in \calI_{mid}} \sum_{\gam \in \Gam_{U, *}} \frac{d_{Id_{U}}(H_{Id_{U}}, H'_{\gam})}{|Aut(U)|c(\gam)} &\le \sum_{U\in \calI_{mid}} \sum_{\sigma,\sigma' \in \mathcal{L}'_{U}} {B_{norm}(\sigma)B_{norm}(\sigma')H_{Id_{U}}(\sigma,\sigma')\frac{1}{2^{\min(m_{\sig}, m_{\sig'}) - 1}}}\\
	&\le \sum_{U\in \calI_{mid}} \sum_{\sigma,\sigma' \in \mathcal{L}'_{U}}\frac{1}{n^{0.5\eps C_{\Del}|V(\sig \circ \sig')|}\Del^{D_{sos}}n^{|U|}2^{\min(m_{\sig}, m_{\sig'}) - 1}}\\
	&\le \sum_{U\in \calI_{mid}} \sum_{\sigma,\sigma' \in \mathcal{L}'_{U}}\frac{1}{n^{0.5\eps C_{\Del}|V(\sig \circ \sig')|}\Del^{D_{sos}}2^{\min(m_{\sig}, m_{\sig'}) - 1}}
	\end{align*}
	where we used \cref{lem: tpca_charging3}. Using $n^{0.5 C_{\Del} |V(\sig \circ \sig')|} \ge n^{0.1\eps C_{\Del} |V(\sig \circ \sig')|}2^{|V(\sig \circ \sig')|}$,
	\begin{align*}
	\sum_{U\in \calI_{mid}} \sum_{\gam \in \Gam_{U, *}} \frac{d_{Id_{U}}(H_{Id_{U}}, H'_{\gam})}{|Aut(U)|c(\gam)} &\le \sum_{U\in \calI_{mid}} \sum_{\sigma,\sigma' \in \mathcal{L}'_{U}}\frac{1}{n^{0.1 \eps  C_{\Delta}|V(\sig \circ \sig')|} \Delta^{D_{sos}} 2^{|V(\sig \circ \sig')|}2^{\min(m_{\sig}, m_{\sig'}) - 1}}\\
	&\le \sum_{U\in \calI_{mid}} \sum_{\sigma,\sigma' \in \mathcal{L}'_{U}}\frac{1}{n^{0.1 \eps C_{\Delta}|V(\sig \circ \sig')|}\Delta^{D_{sos}} 2^{D_V}}\\
	&\le \sum_{U\in \calI_{mid}} \sum_{\sigma,\sigma' \in \mathcal{L}'_{U}}\frac{1}{D_{sos}^{D_{sos}}n^{0.1 \eps C_{\Delta}|V(\sig \circ \sig')|}\Delta^{2D_{sos}} 2^{D_V}}
	\end{align*}
	where we set $C_{sos}$ small enough so that $D_{sos} = n^{\eps C_{sos}} \le n^{c\eps C_{\Del}} = \frac{1}{\Del}$. The final step will be to argue that $\sum_{U\in \calI_{mid}} \sum_{\sigma,\sigma' \in \mathcal{L}'_{U}}\frac{1}{D_{sos}^{D_{sos}}n^{0.1 C_{\Delta}\eps|V(\sig \circ \sig')|}} \le 1$ which will complete the proof. But this will follow from \cref{lem: gp_sum} if we set $C_V, C_E$ small enough.
\end{proof}

We can finally complete the analysis of the truncation error.

\begin{lemma}\label{lem: tpca_cond4}
    Whenever $\norm{M_{\alpha}} \le B_{norm}(\alpha)$ for all $\alpha \in \mathcal{M}'$,
    \[
    \sum_{U \in \mathcal{I}_{mid}}{M^{fact}_{Id_U}{(H_{Id_U})}} \succeq 6\left(\sum_{U \in \mathcal{I}_{mid}}{\sum_{\gamma \in \Gamma_{U,*}}{\frac{d_{Id_{U}}(H'_{\gamma},H_{Id_{U}})}{|Aut(U)|c(\gamma)}}}\right)Id_{sym}
    \]
\end{lemma}

\begin{proof}
	Choose $C_{sos}$ sufficiently small so that $\frac{\Del^{2D_{sos}^2}}{n^{D_{sos}}} \ge \frac{6}{\Delta^{2D_{sos}}2^{D_V}}$ which is satisfied by setting $C_{sos} < 0.5 C_V$. Then, since $Id_{Sym} \succeq 0$, using \cref{lem: tpca_cond5} and \cref{lem: tpca_cond6},
	\begin{align*}
	\sum_{U \in \mathcal{I}_{mid}}{M^{fact}_{Id_U}{(H_{Id_U})}} &\succeq \frac{\Del^{2D_{sos}^2}}{n^{D_{sos}}} Id_{sym}\\
	&\succeq \frac{6}{\Delta^{2D_{sos}}2^{D_V}} Id_{sym}\\
	&\succeq 6\left(\sum_{U \in \mathcal{I}_{mid}}{\sum_{\gamma \in \Gamma_{U,*}}{\frac{d_{Id_{U}}(H'_{\gamma},H_{Id_{U}})}{|Aut(U)|c(\gamma)}}}\right)Id_{sym}
	\end{align*}
\end{proof}

%% file: sparse_pca_quant.tex
In this section, we will full prove \cref{thm: spca_main}.

\SPCAmain*

We already showed the relevant qualitative bounds in \cref{sec: spca_qual}. We use the bounds and also the notation from that section.
We will apply \cref{generalmaintheorem}.

\begin{definition}
    Define $n = \max(d, m)$.
\end{definition}

The above definition conforms with the notation used in \cref{generalmaintheorem}. So, we can use the bounds as stated there.
Once we verify the conditions, the theorem will immediately follow from the machinery, \cref{generalmaintheorem}.

\subsection{\middleshapeboundstwo}

\begin{lemma}\label{lem: spca_charging}
	Suppose $0 < A < \frac{1}{4}$ is a constant such that $\frac{\sqrt{\lda}}{\sqrt{k}} \le d^{-A\eps}$ and $\frac{1}{\sqrt{k}} \le d^{-2A}$. For all $m$ such that $m \le \frac{d^{1 - \eps}}{\lda^2}, m \le \frac{k^{2 - \eps}}{\lda^2}$, for all $U \in \calI_{mid}$ and $\tau \in \calM_U$, suppose $deg^{\tau}(i)$ is even for all $i \in V(\tau) \setminus U_{\tau} \setminus V_{\tau}$, then
	\[\sqrt{d}^{|\tau|_1 - |U_{\tau}|_1}\sqrt{m}^{|\tau|_2 - |U_{\tau}|_2}S(\tau) \le \prod_{j \in V_2(\tau) \setminus U_{\tau} \setminus V_{\tau}} (deg^{\tau}(j) - 1)!!\cdot \frac{1}{d^{A\eps\sum_{e \in E(\tau)} l_e}}\]
\end{lemma}

\begin{proof}
	Let $r_1 = |\tau|_1 - |U_{\tau}|_1, r_2 = |\tau|_2 - |U_{\tau}|_2$. Since $\Delta \le 1$, it suffices to prove
	\[E := \sqrt{d}^{r_1}\sqrt{m}^{r_2}\left(\frac{k}{d}\right)^{r_1} \left(\frac{\sqrt{\lda}}{\sqrt{k}}\right)^{\sum_{e \in E(\tau)} l_e} \le  \frac{1}{d^{A\eps\sum_{e \in E(\tau)} l_e}}\]

	We will need the following claim.
	\begin{claim}
		$\sum_{e \in E(\tau)} l_e \ge 2 \max(r_1, r_2)$.
	\end{claim}

	\begin{proof}
		We will first prove $\sum_{e \in E(\tau)} l_e \ge 2 r_1$. For any vertex $i \in V_1(\tau) \setminus U_{\tau} \setminus V_{\tau}$, $deg^{\tau}(i)$ is even and is not $0$, hence, $deg^{\tau}(i) \ge 2$. Any vertex $i \in U_{\tau} \setminus V_{\tau}$ cannot have $deg^{\tau}(i) = 0$ otherwise $U_{\tau} \setminus\{i\}$ is a vertex separator of strictly smaller weight than $U_{\tau}$, which is not possible, hence, $deg^{\tau}(i) \ge 1$. Similarly, for $i \in  V_{\tau} \setminus U_{\tau}$, $deg^{\tau}(i) \ge 1$. Also, since $H_{\tau}$ is bipartite, we have $\sum_{i \in V_1(\tau)} deg^{\tau}(i) = \sum_{j \in V_2(\tau)} deg^{\tau}(j)= \sum_{e \in E(\tau)} l_e$. Consider
		\begin{align*}
		\sum_{e \in E(\tau)} l_e &= \sum_{i \in V_1(\tau)} deg^{\tau}(i)\\
		&\ge \sum_{i \in V_1(\tau) \setminus U_{\tau} \setminus V_{\tau}} deg^{\tau}(i) + \sum_{i \in (U_{\tau})_1 \setminus V_{\tau}} deg^{\tau}(i) + \sum_{i \in (V_{\tau})_1 \setminus U_{\tau}} deg^{\tau}(i)\\
		&\ge 2|V_1(\tau) \setminus U_{\tau} \setminus V_{\tau}| + |(U_{\tau})_1 \setminus V_{\tau}| + |(V_{\tau})_1 \setminus U_{\tau}|\\
		&= 2r_1
		\end{align*}
		We can similarly prove $\sum_{e \in E(\tau)} l_e \ge 2 r_2$
	\end{proof}

	To illustrate the main idea, we will start by proving the weaker bound $E \le 1$. Observe that our assumptions imply $m \le \frac{d}{\lda^2}, m \le \frac{k^2}{\lda^2}$ and also, $E \le \sqrt{d}^{r_1}\sqrt{m}^{r_2}\left(\frac{k}{d}\right)^{r_1} \left(\frac{\sqrt{\lda}}{\sqrt{k}}\right)^{2\max(r_1, r_2)}$ where we used the fact that $\frac{\sqrt{\lda}}{\sqrt{k}} \le d^{-A\eps} \le 1$.

	\begin{claim}\label{claim: spca_decay}
		For integers $r_1, r_2 \ge 0$, if $m \le \frac{d}{\lda^2}$ and $m \le \frac{k^2}{\lda^2}$, then,
		\[\sqrt{d}^{r_1}\sqrt{m}^{r_2}\left(\frac{k}{d}\right)^{r_1} \left(\frac{\sqrt{\lda}}{\sqrt{k}}\right)^{2\max(r_1, r_2)} \le  1\]
	\end{claim}

	\begin{proof}
	We will consider the cases $r_1 \ge r_2$ and $r_1 < r_2$ separately. If $r_1 \ge r_2$, we have
	\begin{align*}
		\sqrt{d}^{r_1}\sqrt{m}^{r_2}\left(\frac{k}{d}\right)^{r_1} \left(\frac{\sqrt{\lda}}{\sqrt{k}}\right)^{2r_1} &\le \sqrt{d}^{r_1}\left(\frac{\sqrt{d}}{\lda}\right)^{r_2}\left(\frac{k}{d}\right)^{r_1} \left(\frac{\sqrt{\lda}}{\sqrt{k}}\right)^{2r_1}
		= \left(\frac{\lda}{\sqrt{d}}\right)^{r_1 - r_2}
		\le \left(\frac{1}{\sqrt{m}}\right)^{r_1 - r_2}
		\le 1
	\end{align*}
	And if $r_1 < r_2$, we have
	\begin{align*}
	\sqrt{d}^{r_1}\sqrt{m}^{r_2}\left(\frac{k}{d}\right)^{r_1} \left(\frac{\sqrt{\lda}}{\sqrt{k}}\right)^{2r_2} &= \sqrt{d}^{r_1}\sqrt{m}^{r_2 - r_1}\sqrt{m}^{r_1}\left(\frac{k}{d}\right)^{r_1} \left(\frac{\sqrt{\lda}}{\sqrt{k}}\right)^{2r_2}\\
	&\le \sqrt{d}^{r_1}\left(\frac{k}{\lda}\right)^{r_2 - r_1}\left(\frac{\sqrt{d}}{\lda}\right)^{r_1}\left(\frac{k}{d}\right)^{r_1} \left(\frac{\sqrt{\lda}}{\sqrt{k}}\right)^{2r_2}\\
		&= 1
	\end{align*}
	\end{proof}

	For the desired bounds, we mimic this argument while carefully keeping track of factors of $d^{\eps}$.

	\begin{claim}\label{claim: spca_decay2}
		For integers $r_1, r_2 \ge 0$ and an integer $r \ge 2\max(r_1, r_2)$, if $m \le \frac{d^{1 - \eps}}{\lda^2}$ and $m \le \frac{k^{2 - \eps}}{\lda^2}$, then,
		\[\sqrt{d}^{r_1}\sqrt{m}^{r_2}\left(\frac{k}{d}\right)^{r_1} \left(\frac{\sqrt{\lda}}{\sqrt{k}}\right)^r \le  \left(\frac{1}{d^{A\eps}}\right)^r\]
	\end{claim}
	\begin{proof}
	If $r_1 \ge r_2$,
	\begin{align*}
		E &= \sqrt{d}^{r_1}\sqrt{m}^{r_2} \left(\frac{k}{d}\right)^{r_1}\left(\frac{\sqrt{\lda}}{\sqrt{k}}\right)^{2r_1}\left(\frac{\sqrt{\lda}}{\sqrt{k}}\right)^{r - 2r_1}\\
		&\le \sqrt{d}^{r_1}\left(\frac{\sqrt{d}^{1 - \eps}}{\lda}\right)^{r_2} \left(\frac{k}{d}\right)^{r_1}\left(\frac{\sqrt{\lda}}{\sqrt{k}}\right)^{2r_1}\left(\frac{\sqrt{\lda}}{\sqrt{k}}\right)^{r - 2r_1}\\
		& = \left(\frac{\lda}{\sqrt{d}^{1 - \eps}}\right)^{r_1 - r_2} \left(\frac{1}{\sqrt{d}}\right)^{\eps r_1}\left(\frac{\sqrt{\lda}}{\sqrt{k}}\right)^{r - 2r_1}\\
		& \le \left(\frac{1}{\sqrt{m}}\right)^{r_1 - r_2} \left(\frac{1}{\sqrt{d}}\right)^{\eps r_1}\left(\frac{1}{d^{A\eps}}\right)^{r - 2r_1}\\
		&\le \left(\frac{1}{d^{2A}}\right)^{\eps r_1}\left(\frac{1}{d^{A\eps}}\right)^{r - 2r_1}\\
		&= \left(\frac{1}{d^{A\eps}}\right)^r
	\end{align*}
	And if $r_1 < r_2$,
	\begin{align*}
	E &= \sqrt{d}^{r_1}\sqrt{m}^{r_2 - r_1} \sqrt{m}^{r_1} \left(\frac{k}{d}\right)^{r_1}\left(\frac{\sqrt{\lda}}{\sqrt{k}}\right)^{2r_2}\left(\frac{\sqrt{\lda}}{\sqrt{k}}\right)^{r - 2r_2}\\
	&\le \sqrt{d}^{r_1}\left(\frac{\sqrt{k}^{2 - \eps}}{\lda}\right)^{r_2 - r_1}\left(\frac{\sqrt{d}^{1 - \eps}}{\lda}\right)^{r_1} \left(\frac{k}{d}\right)^{r_1}\left(\frac{\sqrt{\lda}}{\sqrt{k}}\right)^{2r_2}\left(\frac{\sqrt{\lda}}{\sqrt{k}}\right)^{r - 2r_2}\\
	&= \left(\frac{\sqrt{k}}{\sqrt{d}}\right)^{\eps r_1}\left(\frac{1}{\sqrt{k}}\right)^{\eps r_2}\left(\frac{\sqrt{\lda}}{\sqrt{k}}\right)^{r - 2r_2}\\
	&\le \left(\frac{1}{\sqrt{k}}\right)^{\eps r_2}\left(\frac{\sqrt{\lda}}{\sqrt{k}}\right)^{r - 2r_2}\\
	&\le \left(\frac{1}{d^{2A}}\right)^{\eps r_2}\left(\frac{1}{d^{A\eps}}\right)^{r - 2r_2}\\
	&\le \left(\frac{1}{d^{A\eps}}\right)^{\sum_{e \in E(\tau)} l_e}
	\end{align*}
\end{proof}
The result follows by setting $r = \sum_{e \in E(\tau)} l_e$ in the above claim.
\end{proof}

\begin{corollary}\label{cor: spca_norm_decay}
	For all $U \in \calI_{mid}$ and $\tau \in \calM_U$, we have
	\[c(\tau) B_{norm}(\tau)S(\tau)R(\tau) \le 1\]
\end{corollary}

\begin{proof}
	First, note that if $deg^{\tau}(i)$ is odd for any vertex $i \in V(\tau) \setminus U_{\tau} \setminus V_{\tau}$, then $S(\tau) = 0$ and the inequality is true. So, assume that $deg^{\tau}(i)$ is even for all $i \in V(\tau) \setminus U_{\tau} \setminus V_{\tau}$.
	Since $\tau$ is a proper middle shape, we have $w(I_{\tau}) = 0$ and $w(S_{\tau, min}) = w(U_{\tau})$. This implies
	$n^{\frac{w(V(\tau)) + w(I_{\tau}) - w(S_{\tau, min})}{2}} = \sqrt{d}^{|\tau|_1 - |U_{\tau}|_1}\sqrt{m}^{|\tau|_2 - |U_{\tau}|_2}$.
	As was observed in the proof of \cref{lem: spca_charging}, every vertex $i \in V(\tau) \setminus U_{\tau}$ or $i \in V(\tau) \setminus V_{\tau}$ has $deg^{\tau}(i) \ge 1$ and hence, $|V(\tau)\setminus U_{\tau}| + |V(\tau)\setminus V_{\tau}| \le 4 \sum_{e \in E(\tau)} l_e$. Also, $q = d^{O(1)\cdot \eps(C_V + C_E)}$. We can set $C_V, C_E$ sufficiently small so that
	\begin{align*}
	c(\tau)B_{norm}(\tau)S(\tau)R(\tau)	&=100(6D_V)^{|U_{\tau}\setminus V_{\tau}| + |V_{\tau}\setminus U_{\tau}| + 2|E(\tau)|}4^{|V(\tau)\setminus (U_{\tau}\cup V_{\tau})|}\\
	&\cdot 2e(6qD_V)^{|V(\tau)\setminus U_{\tau}| + |V(\tau)\setminus V_{\tau}|}\prod_{e \in E(\tau)} (400D_V^2D_E^2q)^{l_e}\\
	&\cdot \sqrt{d}^{|\tau|_1 - |U_{\tau}|_1}\sqrt{m}^{|\tau|_2 - |U_{\tau}|_2} S(\tau) (C_{disc}\sqrt{D_E})^{\sum_{j \in (U_{\tau})_2 \cup (V_{\tau})_2} deg^{\tau}(j)}\\
	&\le d^{O(1) \cdot (C_V + C_E) \cdot \eps\sum_{e \in E(\tau)} l_e}\cdot \prod_{j \in V_2(\tau) \setminus V_2(U_{\tau}) \setminus V_2(V_{\tau})} (deg^{\tau}(j) - 1)!!\cdot \frac{1}{d^{A\eps\sum_{e \in E(\tau)} l_e}}\\
	&\le d^{O(1) \cdot (C_V + C_E) \cdot \eps\sum_{e \in E(\tau)} l_e}\cdot (D_VD_E)^{\sum_{e \in E(\tau)} l_e}\cdot \frac{1}{d^{A\eps\sum_{e \in E(\tau)} l_e}}\\
	&\le 1
	\end{align*}
\end{proof}

We can now obtain our desired middle shape bounds.

\begin{lemma}\label{lem: spca_cond2}
    For all $U \in \calI_{mid}$ and $\tau \in \calM_U$,
    \[
    \begin{bmatrix}
        \frac{1}{|Aut(U)|c(\tau)}H_{Id_U} & B_{norm}(\tau) H_{\tau}\\
        B_{norm}(\tau) H_{\tau}^T & \frac{1}{|Aut(U)|c(\tau)}H_{Id_U}
    \end{bmatrix}
    \succeq 0
    \]
\end{lemma}

\begin{proof}
	We have
	\begin{align*}
		&\begin{bmatrix}
			\frac{1}{|Aut(U)|c(\tau)}H_{Id_U} & B_{norm}(\tau)H_{\tau}\\
			B_{norm}(\tau)H_{\tau}^T & \frac{1}{|Aut(U)|c(\tau)}H_{Id_U}
		\end{bmatrix}\\
		&\qquad= \begin{bmatrix}
			\left(\frac{1}{|Aut(U)|c(\tau)} - \frac{S(\tau)R(\tau)B_{norm}(\tau)}{|Aut(U)|}\right)H_{Id_U} & 0\\
			0 & \left(\frac{1}{|Aut(U)|c(\tau)} - \frac{S(\tau)R(\tau)B_{norm}(\tau)}{|Aut(U)|}\right)H_{Id_U}
		\end{bmatrix}\\
		&\qquad\quad+ B_{norm}(\tau)\begin{bmatrix}
			\frac{S(\tau)R(\tau)}{|Aut(U)|}H_{Id_U} & H_{\tau}\\
			H_{\tau}^T & \frac{S(\tau)R(\tau)}{|Aut(U)|}H_{Id_U}
		\end{bmatrix}
	\end{align*}
	By \cref{lem: spca_cond2_simplified}, $\begin{bmatrix}
		\frac{S(\tau)R(\tau)}{|Aut(U)|}H_{Id_U} & H_{\tau}\\
		H_{\tau}^T & \frac{S(\tau)R(\tau)}{|Aut(U)|}H_{Id_U}
	\end{bmatrix}
	\succeq 0$, so the second term above is positive semidefinite. For the first term, by \cref{lem: spca_cond1}, $H_{Id_U} \succeq 0$ and by \cref{cor: spca_norm_decay}, $\frac{1}{|Aut(U)|c(\tau)} - \frac{S(\tau)R(\tau)B_{norm}(\tau)}{|Aut(U)|} \ge 0$, which proves that the first term is also positive semidefinite.
\end{proof}

\subsection{\intersectionboundstwo}

\begin{lemma}\label{lem: spca_charging2}
	Suppose $0 < A < \frac{1}{4}$ is a constant such that $\frac{\sqrt{\lda}}{\sqrt{k}} \le d^{-A\eps}, \frac{1}{\sqrt{k}} \le d^{-2A}$ and $\frac{k}{d} \le d^{-A\eps}$. For all $m$ such that $m \le \frac{d^{1 - \eps}}{\lda^2}, m \le \frac{k^{2 - \eps}}{\lda^2}$, for all $U, V \in \calI_{mid}$ where $w(U) > w(V)$ and for all $\gam \in \Gamma_{U, V}$,
	\[n^{w(V(\gam)\setminus U_{\gam})} S(\gam)^2 \le \left(\prod_{j \in V_2(\gam) \setminus U_{\gam} \setminus V_{\gam}}(deg^{\gam}(j)- 1)!!\right)^2\frac{1}{d^{B\eps (|V(\gam) \setminus (U_{\gam} \cap V_{\gam})| + \sum_{e \in E(\gam)} l_e)}}\]
	for some constant $B > 0$ that depends only on $C_{\Del}$. In particular, it is independent of $C_V$ and $C_E$.
\end{lemma}

\begin{proof}
	Suppose there is a vertex $i \in V(\gam) \setminus U_{\gam} \setminus V_{\gam}$ such that $deg^{\gam}(i)$ is odd, then $S(\gam) = 0$ and the inequality is true. So, assume $deg^{\gam}(i)$ is even for all vertices $i \in V(\gam) \setminus U_{\gam} \setminus V_{\gam}$.
	We have $n^{w(V(\gam) \setminus U_{\gam})} = d^{|\gam|_1 - |U_{\gam}|_1}m^{|\gam|_2 - |U_{\gam}|_2}$. Plugging in $S(\gamma)$, we get that we have to prove
	\begin{align*}
		E := d^{|\gam|_1 - |U_{\gam}|_1}m^{|\gam|_2 - |U_{\gam}|_2} \left(\frac{k}{d}\right)^{2|\gamma|_1 - |U_{\gamma}|_1 - |V_{\gamma}|_1}\Del^{2|\gamma|_2 - |U_{\gamma}|_2 - |V_{\gamma}|_2} \prod_{e \in E(\gamma)} \frac{\lambda^{l_e}}{k^{l_e}} \le \frac{1}{d^{B\eps (|V(\gam) \setminus (U_{\gam} \cap V_{\gam})| + \sum_{e \in E(\gam)} l_e)}}
	\end{align*}

	Let $S'$ be the set of vertices $i \in U_{\gam} \setminus V_{\gam}$ that have $deg^{\gam}(i) \ge 1$. Let $e, f$ be the number of type $1$ vertices and the number of type $2$ vertices in $S'$ respectively. Observe that $S' \cup (U_{\gam} \cap V_{\gam})$ is a vertex separator of $\gam$.
	Let $g = |V_{\gam} \setminus U_{\gam}|_1$ (resp. $h = |V_{\gam} \setminus U_{\gam}|_2$) be the number of type $1$ vertices (resp. type $2$ vertices) in $V_{\gam} \setminus U_{\gam}$.
	We first claim that $d^em^f \ge d^gm^h$. To see this, note that the vertex separator $S' \cup (U_{\gam} \cap V_{\gam})$ has weight $\sqrt{d}^{e + |U_{\gam} \cap V_{\gam}|_1}\sqrt{m}^{f + |U_{\gam} \cap V_{\gam}|_2}$. On the other hand, $V_{\gam}$ has weight $\sqrt{d}^{g + |U_{\gam} \cap V_{\gam}|_1}\sqrt{m}^{h + |U_{\gam} \cap V_{\gam}|_2}$. Since $\gam$ is a left shape, $V_{\gam}$ is the unique minimum vertex separator and hence, $\sqrt{d}^{e + |U_{\gam} \cap V_{\gam}|_1}\sqrt{m}^{f + |U_{\gam} \cap V_{\gam}|_2} \ge \sqrt{d}^{g + |U_{\gam} \cap V_{\gam}|_1}\sqrt{m}^{h + |U_{\gam} \cap V_{\gam}|_2}$ which implies $d^em^f \ge d^gm^h$.
	Let $p = |V(\gam) \setminus (U_{\gam} \cup V_{\gam})|_1$ (resp. $q = |V(\gam) \setminus (U_{\gam} \cup V_{\gam})|_2$) be the number of type $1$ vertices (resp. type $2$ vertices) in $V(\gam) \setminus (U_{\gam} \cup V_{\gam})$.
	To illustrate the main idea, we will first prove the weaker inequality $E \le 1$. Since $\Del \le 1$, it suffices to prove
	\begin{align*}
		d^{|\gam|_1 - |U_{\gam}|_1}m^{|\gam|_2 - |U_{\gam}|_2} \left(\frac{k}{d}\right)^{2|\gamma|_1 - |U_{\gamma}|_1 - |V_{\gamma}|_1} \prod_{e \in E(\gamma)} \frac{\lambda^{l_e}}{k^{l_e}} \le 1
	\end{align*}
	We have
	$d^{|\gam|_1 - |U_{\gam}|_1}m^{|\gam|_2 - |U_{\gam}|_2} = d^{p + g}m^{q + h} \le n^{p + \frac{e + g}{2}}m^{q + \frac{f + h}{2}}$
	since $d^em^f \ge d^gm^h$. Also, $2|\gamma|_1 - |U_{\gamma}|_1 - |V_{\gamma}|_1 = 2p + e + g$. So, it suffices to prove
	\begin{align*}
		n^{p + \frac{e + g}{2}}m^{q + \frac{f + h}{2}}\left(\frac{k}{d}\right)^{2p + e + g} \prod_{e \in E(\gam)} \left(\frac{\lda}{k}\right)^{l_e} \le 1
	\end{align*}

	We will need the following claim.
	\begin{claim}
		$\sum_{e \in E(\gam)} l_e \ge \max(2p + e + g, 2q + f + h)$
	\end{claim}
	\begin{proof}
		Since $H_{\gam}$ is bipartite, we have $\sum_{e \in E(\gam)}l_e = \sum_{i \in V_1(\gam)} deg^{\gam}(i) = \sum_{i \in V_2(\gam)} deg^{\gam}(i)$. Observe that all vertices $i \in V(\gam) \setminus U_{\gam} \setminus V_{\gam}$ have $deg^{\gam}(i)$ nonzero and even, and hence, $deg^{\gam}(i) \ge 2$. Then,
	\begin{align*}
		\sum_{e \in E(\gam)}l_e &= \sum_{i \in V_1(\gam)} deg^{\gam}(i)\\
		&\ge \sum_{i \in V_1(\gam) \setminus U_{\gam} \setminus V_{\gam}} deg^{\gam}(i) + \sum_{i \in (U_{\gam})_1 \setminus V_{\gam}} deg^{\gam}(i) + \sum_{i \in (V_{\gam})_1 \setminus U_{\gam}} deg^{\gam}(i)\\
		&\ge 2p + e + g
	\end{align*}
	Similarly,
	\begin{align*}
	\sum_{e \in E(\gam)}l_e &= \sum_{i \in V_2(\gam)} deg^{\gam}(i)\\
	&\ge \sum_{i \in V_2(\gam) \setminus U_{\gam} \setminus V_{\gam}} deg^{\gam}(i) + \sum_{i \in (U_{\gam})_2 \setminus V_{\gam}} deg^{\gam}(i) + \sum_{i \in (V_{\gam})_2 \setminus U_{\gam}} deg^{\gam}(i)\\
	&\ge 2q + f + h
\end{align*}
Therefore, $\sum_{e \in E(\gam)} l_e \ge \max(2p + e + g, 2q + f + h)$.
\end{proof}

Now, let $r_1 = p + \frac{e + g}{2}, r_2 = q + \frac{f + h}{2}$. Then, $\sum_{e \in E(\gam)} l_e \ge 2\max(r_1, r_2)$ and we wish to prove
	$d^{r_1}m^{r_2} \left(\frac{k}{d}\right)^{2r_1} \left(\frac{\lda}{k}\right)^{2\max(r_1, r_2)} \le 1$
This expression simply follows by squaring \cref{claim: spca_decay}.

Now, to prove that $E \le \frac{1}{d^{B\eps (|V(\gam) \setminus (U_{\gam} \cap V_{\gam})| + \sum_{e \in E(\gam)} l_e)}}$, we mimic this argument while carefully keeping track of factors of $d^{\eps}$. Again, using $d^em^f \ge d^gm^h$, it suffices to prove that
\begin{align*}
	d^{p + \frac{e + g}{2}}m^{q + \frac{f + h}{2}} \left(\frac{k}{d}\right)^{2|\gamma|_1 - |U_{\gamma}|_1 - |V_{\gamma}|_1}\Del^{2|\gamma|_2 - |U_{\gamma}|_2 - |V_{\gamma}|_2} \prod_{e \in E(\gamma)} \frac{\lambda^{l_e}}{k^{l_e}} \le \frac{1}{d^{B\eps (|V(\gam) \setminus (U_{\gam} \cap V_{\gam})| + \sum_{e \in E(\gam)} l_e)}}
\end{align*}

The idea is that the $d^{B\eps}$ decay for the edges are obtained from the stronger assumption on $m$, namely $m \le \frac{d^{1 - \eps}}{\lda^2}, m \le \frac{k^{2 - \eps}}{\lda^2}$. And the $d^{B\eps}$ decay for the type $1$ vertices of $V(\gam) \setminus(U_{\gam} \cap V_{\gam})$ are obtained both from the stronger assumption on $m$ as well as the factors of $\frac{k}{d}$, the latter especially useful for the degree $0$ vertices. Finally, the $d^{B\eps}$ decay for the type $2$ vertices of $V(\gam) \setminus (U_{\gam} \cap V_{\gam})$ are obtained from the factors of $\Delta$.
Indeed, note that for a constant $B$ that depends on $C_{\Del}$, $\Del^{2|\gamma|_2 - |U_{\gamma}|_2 - |V_{\gamma}|_2} \le d^{-B\eps|V(\gam) \setminus (U_{\gam} \cap V_{\gam})|_2}$. So, we would be done if we prove
\begin{align*}
	d^{p + \frac{e + g}{2}}m^{q + \frac{f + h}{2}} \left(\frac{k}{d}\right)^{2|\gamma|_1 - |U_{\gamma}|_1 - |V_{\gamma}|_1}\left(\frac{\lambda}{k}\right)^{\sum_{e \in E(\gamma)} l_e} \le \frac{1}{d^{B\eps (|V(\gam) \setminus (U_{\gam} \cap V_{\gam})|_1 + \sum_{e \in E(\gam)} l_e)}}
\end{align*}

Let $c_0$ be the number of type $1$ vertices $i$ in $V(\gam) \setminus (U_{\gam} \cap V_{\gam})$ such that $deg^{\gam}(i) = 0$. Since they have degree $0$, they must be in $(U_{\gam})_1 \setminus V_{\gam}$. Also, we have $2|\gamma|_1 - |U_{\gamma}|_1 - |V_{\gamma}|_1 = 2p + e + g + c_0$ and hence, $\left(\frac{k}{d}\right)^{2|\gamma|_1 - |U_{\gamma}|_1 - |V_{\gamma}|_1} = \left(\frac{k}{d}\right)^{2p + e + g + c_0}$. For these degree $0$ vertices, we have that the factors of $\frac{k}{d} \le d^{-A\eps}$ offer a decay of $\frac{1}{d^{B\eps}}$. Therefore, it suffices to prove
\begin{align*}
	d^{p + \frac{e + g}{2}}m^{q + \frac{f + h}{2}} \left(\frac{k}{d}\right)^{2p + e + g}\left(\frac{\lambda}{k}\right)^{\sum_{e \in E(\gamma)} l_e} \le \frac{1}{d^{B\eps (p + q + e + f + g + h) + \sum_{e \in E(\gam)} l_e)}}
\end{align*}
for a constant $B > 0$. Observe that $p + q + e + f + g + h \le 2(\sum_{e \in E(\gam)} l_e)$. Therefore, using the notation $r_1 = p + \frac{e + g}{2}, r_2 = q + \frac{f + h}{2}$, it suffices to prove
\begin{align*}
	d^{r_1}m^{r_2} \left(\frac{k}{d}\right)^{2r_1}\left(\frac{\lambda}{k}\right)^{\sum_{e \in E(\gamma)} l_e} \le \frac{1}{d^{B\eps \sum_{e \in E(\gam)} l_e}}
\end{align*}
for a constant $B > 0$. But this follows by squaring \cref{claim: spca_decay2} where we set $r = \sum_{e \in E(\gam)} l_e$.
\end{proof}

\begin{remk}
	In the above bounds, note that there is a decay of $d^{B\eps}$ for each vertex in $V(\gam) \setminus (U_{\gam} \cap V_{\gam})$.	One of the main technical reasons for introducing the slack parameter $C_{\Del}$ in the planted distribution was to introduce this decay, which is needed in the current machinery.
\end{remk}

With this, we obtain intersection term bounds.

\begin{lemma}\label{lem: spca_cond3}
    For all $U, V \in \calI_{mid}$ where $w(U) > w(V)$ and all $\gam \in \Gam_{U, V}$, \[c(\gam)^2N(\gam)^2B(\gam)^2H_{Id_V}^{-\gam, \gam} \preceq H_{\gam}'\]
\end{lemma}

\begin{proof}
	By \cref{lem: spca_cond3_simplified}, we have
	\begin{align*}
		c(\gam)^2N(\gam)^2B(\gam)^2H_{Id_V}^{-\gam, \gam} &\preceq c(\gam)^2N(\gam)^2B(\gam)^2 S(\gam)^2R(\gam)^2 \frac{|Aut(U)|}{|Aut(V)|} H'_{\gam}
	\end{align*}
	Using the same proof as in \cref{lem: spca_cond1}, we can see that $H'_{\gam} \succeq 0$. Therefore, it suffices to prove that $c(\gam)^2N(\gam)^2B(\gam)^2 S(\gam)^2R(\gam)^2 \frac{|Aut(U)|}{|Aut(V)|} \le 1$.
	Since $U, V \in \calI_{mid}$, $Aut(U) = |U|_1!|U|_2!, Aut(V) = |V|_1!|V|_2!$. Therefore, $\frac{|Aut(U)|}{|Aut(V)|} = \frac{|U|_1!|U|_2!}{|V|_1!|V|_2!} \le D_V^{|U_{\gam} \setminus V_{\gam}|}$. Also, $|E(\gam)| \le \sum_{e \in E(\gam)} l_e$ and $q = d^{O(1) \cdot \eps (C_V + C_E)}$. Note that $R(\gam)^2 = (C_{disc}\sqrt{D_E})^{2\sum_{j \in (U_{\gam})_2 \cup (V_{\gam})_2} deg^{\gam}(j)} \le d^{O(1)\cdot \eps C_E \cdot \sum_{e \in E(\gam)} l_e}$ and $\left(\prod_{j \in V_2(\gam) \setminus U_{\gam} \setminus V_{\gam}}(deg^{\gam}(j)- 1)!!\right)^2 \le (D_VD_E)^{2\sum_{e \in E(\tau)} l_e} \le d^{O(1)\cdot \eps (C_V + C_E) \cdot \sum_{e \in E(\gam)} l_e}$.

	Let $B$ be the constant from \cref{lem: spca_charging2}. We can set $C_V, C_E$ sufficiently small so that, using \cref{lem: spca_charging2},
	\begin{align*}
		c(\gam)^2N(\gam)^2B(\gam)^2&S(\gam)^2R(\gam)^2 \frac{|Aut(U)|}{|Aut(V)|} \\
		&\le 100^2 (6D_V)^{2|U_{\gam}\setminus V_{\gam}| + 2|V_{\gam}\setminus U_{\gam}| + |E(\al)|}16^{|V(\gam) \setminus (U_{\gam} \cup V_{\gam})|}\\
		&\quad\cdot (3D_V)^{4|V(\gam)\setminus V_{\gam}| + 2|V(\gam)\setminus U_{\gam}|} (6qD_V)^{2|V(\gam)\setminus U_{\gam}| + 2|V(\gam)\setminus V_{\gam}|} \prod_{e \in E(\gam)} (400D_V^2D_E^2q)^{2l_e}\\
		&\quad\cdot  n^{w(V(\gam)\setminus U_{\gam})} S(\gam)^2 d^{O(1)\cdot \eps C_E \cdot \sum_{e \in E(\gam)} l_e}\cdot D_V^{|U_\gam \setminus V_{\gam}|} \\
		&\le d^{O(1) \cdot \eps(C_V + C_E) \cdot (|V(\gam) \setminus (U_{\gam} \cap V_{\gam})| + \sum_{e \in E(\gam)} l_e)} \cdot n^{w(V(\gam)\setminus U_{\gam})} S(\gam)^2\\
		&\le d^{O(1) \cdot \eps(C_V + C_E) \cdot (|V(\gam) \setminus (U_{\gam} \cap V_{\gam})| + \sum_{e \in E(\gam)} l_e)}\cdot \frac{1}{d^{B\eps (|V(\gam) \setminus (U_{\gam} \cap V_{\gam})| + \sum_{e \in E(\gam)} l_e)}}\\
		&\le 1
	\end{align*}
\end{proof}

\subsection{\truncationboundstwo}

In this section, we will obtain truncation error bounds using the strategy sketched in \cref{sec: showing_positivity}. We also reuse the notation. To start with, we obtain a bound on $B_{norm}(\sig) B_{norm}(\sig') H_{Id_U}(\sig, \sig')$.

\begin{lemma}\label{lem: spca_charging3}
	Suppose $0 < A < \frac{1}{4}$ is a constant such that $\frac{\sqrt{\lda}}{\sqrt{k}} \le d^{-A\eps}$ and $\frac{1}{\sqrt{k}} \le d^{-2A}$. Suppose $m$ is such that $m \le \frac{d^{1 - \eps}}{\lda^2}, m \le \frac{k^{2 - \eps}}{\lda^2}$. For all $U \in \calI_{mid}$ and $\sig, \sig' \in \calL_U$,
	\[B_{norm}(\sig) B_{norm}(\sig') H_{Id_U}(\sig, \sig') \le \frac{1}{d^{0.5A\eps(|V(\sig \circ \sig')| + \sum_{e \in E(\al) l_e}}} \cdot \frac{1}{d^{|U_{\sig}|_1 + |U_{\sig'}|_1}m^{|U_{\sig'}|_2 + |U_{\sig'}|_2}}\]
\end{lemma}

\begin{proof}
	Suppose there is a vertex $i \in V(\sig) \setminus V_{\sig}$ such that $deg^{\sig}(i) + deg^{U_{\sig}}(i)$ is odd, then $H_{Id_U}(\sig, \sig') = 0$ and the inequality is true. So, assume that $deg^{\sig}(i) + deg^{U_{\sig}}(i)$ is even for all $i \in V(\sig) \setminus V_{\sig}$. Similarly, assume that $deg^{\sig'}(i) + deg^{U_{\sig'}}(i)$ is even for all $i \in V(\sig') \setminus V_{\sig'}$. Also, if $\rho_{\sig} \neq \rho_{\sig'}$, we will have $H_{Id_U}(\sig, \sig') = 0$ and we would be done. So, assume $\rho_{\sig} = \rho_{\sig'}$.

	Let there be $e$ (resp. $f$) vertices of type $1$ (resp. type $2$) in $V(\sig) \setminus U_{\sig} \setminus V_{\sig}$. Then, $n^{\frac{w(V(\sig)) - w(U)}{2}} = \sqrt{d}^{|V(\sig)|_1 - |U|_1}\sqrt{m}^{|V(\sig)|_2 - |U|_2} = \sqrt{d}^{|U_{\sig}|_1}\sqrt{m}^{|U_{\sig}|_2} \sqrt{d}^e\sqrt{m}^f$. Let there be $g$ (resp. $h$) vertices of type $1$ (resp. type $2$) in $V(\sig') \setminus U_{\sig'} \setminus V_{\sig'}$. Then, similarly, $n^{\frac{w(V(\sig')) - w(U)}{2}} \le \sqrt{d}^{|U_{\sig'}|_1}\sqrt{m}^{|U_{\sig'}|_2}\sqrt{d}^g\sqrt{m}^h$.

	Let $\al = \sig \circ \sig'$. Since all vertices in $V(\al) \setminus U_{\al} \setminus V_{\al}$ have degree at least $2$, we have $\sum_{e \in E(\al)} l_e \ge \sum_{i \in V_1(\al) \setminus U_{\al} \setminus V_{\al}} deg^{\al}(i) \ge 2(e + g) + |U_{\sig}|_1 + |U_{\sig}|_2$. Similarly, $\sum_{e \in E(\al)} l_e \ge 2(f + h) + |U_{\sig'}|_1 + |U_{\sig'}|_2$. Therefore, by setting $r_1 = e + g, r_2 = f + h$ in \cref{claim: spca_decay2}, we have
	\[\sqrt{d}^{e + g}\sqrt{m}^{f + h} \left(\frac{k}{d}\right)^{e + g}\prod_{e \in E(\alpha)} \frac{\sqrt{\lambda}^{l_e}}{\sqrt{k}^{l_e}} \le \frac{1}{d^{A\eps \sum_{e \in E(\al)} l_e}}\]
	Also, $\left(\frac{k}{d}\right)^{|\al|_1} \le \left(\frac{k}{d}\right)^{e + g + |U_{\sig}|_1 + |U_{\sig'}|_1}$ and $\prod_{j \in V_2(\alpha)} (deg^{\alpha}(j) - 1)!! \le d^{\eps C_V \sum_{e \in E(\al)} l_e}$. Therefore,
    {\footnotesize
	\begin{align*}
		n^{\frac{w(V(\sig)) - w(U)}{2}}&n^{\frac{w(V(\sig')) - w(U)}{2}}H_{Id_U}(\sig, \sig')\\
		&\le d^{O(1)D_{sos}}\sqrt{d}^e\sqrt{m}^f d^{O(1)D_{sos}}\sqrt{d}^g\sqrt{m}^h \cdot\frac{1}{|Aut(U)|}\left(\frac{1}{\sqrt{k}}\right)^{deg(\alpha)}\left(\frac{k}{d}\right)^{|\alpha|_1}\Del^{|\alpha|_2} \prod_{j \in V_2(\alpha)} (deg^{\alpha}(j) - 1)!!\prod_{e \in E(\alpha)} \frac{\sqrt{\lambda}^{l_e}}{\sqrt{k}^{l_e}}\\
		&\le d^{O(1)D_{sos}} d^{\eps C_V \sum_{e \in E(\al)} l_e} \sqrt{d}^{e + g}\sqrt{m}^{f + h} \left(\frac{k}{d}\right)^{e + g}\prod_{e \in E(\alpha)} \frac{\sqrt{\lambda}^{l_e}}{\sqrt{k}^{l_e}} \cdot \frac{1}{d^{|U_{\sig}|_1 + |U_{\sig'}|_1}m^{|U_{\sig'}|_2 + |U_{\sig'}|_2}}\\
		&\le \frac{d^{\eps C_V \sum_{e \in E(\al)} l_e}}{d^{A\eps \sum_{e \in E(\al)} l_e}} \cdot \frac{1}{d^{|U_{\sig}|_1 + |U_{\sig'}|_1}m^{|U_{\sig'}|_2 + |U_{\sig'}|_2}}
	\end{align*}
}
	By setting $C_V, C_E$ sufficiently small and plugging in the expressions for $B_{norm}(\sig), B_{norm}(\sig')$, we obtain the result.
\end{proof}

We can apply the the strategy now.

\begin{restatable}{lemma}{SPCAfive}\label{lem: spca_cond5}
	Whenever $\norm{M_{\al}} \le B_{norm}(\al)$ for all $\al \in \calM'$,
	\[
	\sum_{U \in \mathcal{I}_{mid}}{M^{fact}_{Id_U}{(H_{Id_U})}} \succeq \frac{1}{d^{K_1D_{sos}^2}} Id_{sym}
	\]
	for a constant $K_1 > 0$ that can depend on $C_{\Del}$.
\end{restatable}

\begin{proof}
    We will use \cref{thm: main_positivity}.
    For $V \in \calI_{mid}$, $\lda_V = \frac{\Del^{|V|_2}}{d^{|V|_1}k^{|V|_2}}$. Let the minimum value of this quantity over all $V$ be $N$. We then choose $w_V = N / \lda_V$ so that for all left shapes $\sig \in \calL_V$, \cref{lem: spca_charging3} implies $w_{V} \leq \frac{w_{U_{\sigma}}\lambda_{U_{\sigma}}}{|\mathcal{I}_{mid}|B_{norm}(\sigma)^2{c(\sigma)^2}{H_{Id_V}(\sigma,\sigma)}}$, completing the proof.
\end{proof}

\begin{restatable}{lemma}{SPCAsix}\label{lem: spca_cond6}
	\[\sum_{U\in \calI_{mid}} \sum_{\gam \in \Gam_{U, *}} \frac{d_{Id_{U}}(H_{Id_{U}}, H'_{\gam})}{|Aut(U)|c(\gam)} \le \frac{d^{K_2 D_{sos}}}{2^{D_V}}\]
	for a constant $K_2 > 0$ that can depend on $C_{\Del}$.
\end{restatable}

\begin{proof}
    We do the same calculations as in the proof of \cref{lem: plds_cond6}, until
	\begin{align*}
		\sum_{U\in \calI_{mid}} \sum_{\gam \in \Gam_{U, *}} \frac{d_{Id_{U}}(H_{Id_{U}}, H'_{\gam})}{|Aut(U)|c(\gam)} &\le \sum_{U\in \calI_{mid}} \sum_{\sigma,\sigma' \in \mathcal{L}'_{U}} {B_{norm}(\sigma)B_{norm}(\sigma')H_{Id_{U}}(\sigma,\sigma')\frac{1}{2^{\min(m_{\sig}, m_{\sig'}) - 1}}}\\
		&\le \sum_{U\in \calI_{mid}} \sum_{\sigma,\sigma' \in \mathcal{L}'_{U}}\frac{d^{O(1) D_{sos}}}{d^{0.5A\eps|V(\sig \circ \sig')|}2^{\min(m_{\sig}, m_{\sig'}) - 1}}
	\end{align*}
	where we used \cref{lem: spca_charging3}. Using $d^{0.5A\eps |V(\sig \circ \sig')|} \ge d^{0.1A\eps |V(\sig \circ \sig')|}2^{|V(\sig \circ \sig')|}$,
	\begin{align*}
		\sum_{U\in \calI_{mid}} \sum_{\gam \in \Gam_{U, *}} \frac{d_{Id_{U}}(H_{Id_{U}}, H'_{\gam})}{|Aut(U)|c(\gam)} &\le \sum_{U\in \calI_{mid}} \sum_{\sigma,\sigma' \in \mathcal{L}'_{U}}\frac{d^{O(1) D_{sos}}}{d^{0.1A\eps|V(\sig \circ \sig')|}  2^{|V(\sig \circ \sig')|}2^{\min(m_{\sig}, m_{\sig'}) - 1}}\\
		&\le \sum_{U\in \calI_{mid}} \sum_{\sigma,\sigma' \in \mathcal{L}'_{U}}\frac{d^{O(1) D_{sos}}}{d^{0.1A\eps|V(\sig \circ \sig')|} 2^{D_V}}\\
		&\le \sum_{U\in \calI_{mid}} \sum_{\sigma,\sigma' \in \mathcal{L}'_{U}}\frac{d^{O(1) D_{sos}}}{D_{sos}^{D_{sos}}d^{0.1A\eps|V(\sig \circ \sig')|} 2^{D_V}}
	\end{align*}
	The final step will be to argue that $\sum_{U\in \calI_{mid}} \sum_{\sigma,\sigma' \in \mathcal{L}'_{U}}\frac{1}{D_{sos}^{D_{sos}}d^{0.1 A\eps|V(\sig \circ \sig')|}} \le 1$ which will complete the proof. But this will follow from \cref{lem: gp_sum} if we set $C_V, C_E$ small enough.
\end{proof}

We can finally show that truncation errors can be handled.

\begin{restatable}{lemma}{SPCAfour}\label{lem: spca_cond4}
    Whenever $\norm{M_{\alpha}} \le B_{norm}(\alpha)$ for all $\alpha \in \mathcal{M}'$,
    \[
    \sum_{U \in \mathcal{I}_{mid}}{M^{fact}_{Id_U}{(H_{Id_U})}} \succeq 6\left(\sum_{U \in \mathcal{I}_{mid}}{\sum_{\gamma \in \Gamma_{U,*}}{\frac{d_{Id_{U}}(H'_{\gamma},H_{Id_{U}})}{|Aut(U)|c(\gamma)}}}\right)Id_{sym}
    \]
\end{restatable}

\begin{proof}
	Choose $C_{sos}$ sufficiently small so that $\frac{1}{d^{K_1D_{sos}^2}} \ge 6\frac{d^{K_2D_{sos}}}{2^{D_V}}$ which can be satisfied by setting $C_{sos} < K_3 C_V$ for a sufficiently small constant $K_3 > 0$. Then, since $Id_{Sym} \succeq 0$, using \cref{lem: spca_cond5} and \cref{lem: spca_cond6},
	\begin{align*}
		\sum_{U \in \mathcal{I}_{mid}}{M^{fact}_{Id_U}{(H_{Id_U})}} &\succeq \frac{1}{d^{K_1D_{sos}^2}} Id_{sym}\\
		&\succeq 6\frac{d^{K_2D_{sos}}}{2^{D_V}} Id_{sym}\\
		&\succeq 6\left(\sum_{U \in \mathcal{I}_{mid}}{\sum_{\gamma \in \Gamma_{U,*}}{\frac{d_{Id_{U}}(H'_{\gamma},H_{Id_{U}})}{|Aut(U)|c(\gamma)}}}\right)Id_{sym}
	\end{align*}
\end{proof}

%% file: conclusion.tex
In this paper, we developed general machinery for proving Sum of Squares (SoS) lower bounds on certification problems. While proving SoS lower bounds is notoriously hard, our machinery reduces the task of proving SoS lower bounds to verifying conditions on the coefficient matrices. For this, the three main conditions which need to be verified are PSD mass, middle shape bounds, and intersection term bounds. Once this is done, the proof can be completed by showing positivity of the PSD mass and bounding the truncation error.

Using our machinery, we proved SoS lower bounds for three problems - planted slightly denser subgraph, tensor PCA, and the Wishart model of sparse PCA. As discussed in \cref{sec: prior_work}, our lower bounds match the best known algorithmic guarantees, thereby giving strong evidence of the correct computational threshold behavior of these problems.

That said, there is considerable room for further work. One direction is to try and use our machinery to prove a variant of the low-degree conjecture, which is a major open problem. Another direction is to push the boundaries of SoS lower bounds techniques and further extend our machinery. Our machinery currently works best when the input is dense and we have both vertex and edge decay, i.e. the coefficients for the shapes decay exponentially with the number of vertices and edges in the shape. However, several important problems such as Densest $k$-Subgraph are not captured by this setting. Thus, to prove SoS lower bounds for these problems, new techniques need to be developed as was done for the Sherrington-Kirkpatrick and sparse independent set problems by the recent works \cite{sklowerbounds,jones2021sum}. Developing these techniques will give further insight into SoS lower bounds and how our machinery can potentially be improved.

%% file: appendix.tex
\begin{appendix}
	\section{Proof that the Leftmost and Rightmost Minimum Vertex Separators are Well-defined}\label{separatorswelldefinedsection}
	In this section, we give a general proof that the leftmost and rightmost minimum vertex separators are well-defined.
	\begin{lemma}\label{leftrightseparatorlemma}
		For any two distinct vertex separators $S_1$ and $S_2$ of $\alpha$, there exist vertex separators $S_L$ and $S_R$ of $\alpha$ such that:
		\begin{enumerate}
			\item $S_L$ is a vertex separator of $U_{\alpha}$ and $S_1$ and a vertex separator of $U_{\alpha}$ and $S_2$.
			\item $S_R$ is a vertex separator of $S_1$ and $V_{\alpha}$ and a vertex separator of $S_2$ and $V_{\alpha}$.
			\item $w(S_L) + w(S_R) \leq w(S_1) + w(S_2)$
		\end{enumerate}
	\end{lemma}
	\begin{proof}
		Take $S_L$ to be the set of vertices $v \in V(\alpha) \cap (S_1 \cup S_2)$ such that there is a path from $U_{\alpha}$ to $v$ which doesn't intersect $S_1 \cup S_2$ before reaching $v$. Similarly, take $S_R$ to be the set of vertices $v \in V(\alpha) \cap (S_1 \cup S_2)$ such that there is a path from $V_{\alpha}$ to $v$ which doesn't intersect $S_1 \cup S_2$ before reaching $v$.
		
		Now observe that $S_L$ is a vertex separator between $U_{\alpha}$ and $S_1$. To see this, note that for any path $P$ from $U_{\alpha}$ to a vertex $v \in S_1$, either $P$ intersects $S_L$ before reaching $v$ or $P$ does not intersect $S_L$ before reaching $v$. In the latter case, $v \in S_L$. Thus, in either case, $P$ intersects $S_L$. Following similar logic, $S_L$ is also a vertex separator between $U_{\alpha}$ and $S_2$, $S_R$ is a vertex separator between $S_1$ and $V_{\alpha}$, and $S_R$ is also a vertex separator between $S_2$ and $V_{\alpha}$.
		
		To show that $w(S_L) + w(S_R) \leq w(S_1) + w(S_2)$, observe that $w(S_L) + w(S_R) = w(S_R \cup S_R) + w(S_L \cap S_R)$ and $w(S_1) + w(S_2) = w(S_1 \cup S_2) + w(S_1 \cap S_2)$. Thus, to show that $w(S_L) + w(S_R) \leq w(S_1) + w(S_2)$, it is sufficient to show that
		\begin{enumerate}
			\item $S_L \cup S_R \subseteq S_1 \cup S_2$
			\item $S_L \cap S_R \subseteq S_1 \cap S_2$
		\end{enumerate}
		For the first statement, note that by definition any vertex in $S_L \cup S_R$ must be in $S_1 \cup S_2$. For the second statement, note that if $v \in S_L \cap S_R$ then there is a path from $U_{\alpha}$ to $v$ which does not intersect any other vertices in $S_1 \cup S_2$ and there is a path from $v$ to $V_{\alpha}$ which does not intersect any other vertices in $S_1 \cup S_2$. Combining these paths, we obtain a path $P$ from $U_{\alpha}$ to $V_{\alpha}$ such that $v$ is the only vertex in $P$ which is in $S_1 \cup S_2$. This implies that $v \in S_1 \cap S_2$ as otherwise either $S_1$ or $S_2$ would not be a vertex separator between $U_{\alpha}$ and $V_{\alpha}$.
	\end{proof}
	\begin{corollary}
		The leftmost and rightmost minimum vertex separators between $U_{\alpha}$ and $V_{\alpha}$ are well-defined.
	\end{corollary}
	\begin{proof}
		Assume that there is no minimum leftmost vertex separator. If so, then there exists a minimum vertex separator $S_1$ between $U_{\alpha}$ and $V_{\alpha}$ such that 
		\begin{enumerate}
			\item There does not exist a minimum vertex separator $S'$ of $\alpha$ such that $S'$ is also a minimum vertex separator of $U_{\alpha}$ and $S_1$ (otherwise we would take $S'$ rather than $S$)
			\item There exists a minimum vertex separator $S_2$ of $\alpha$ such that $S'$ is not a minimum vertex separator of $U_{\alpha}$ and $S_2$ (as otherwise $S_1$ would be the leftmost minimum vertex separator)
		\end{enumerate}
		Now let $S_L$ and $S_R$ be the vertex separators of $\alpha$ obtained by applying Lemma \ref{leftrightseparatorlemma} to $S_1$ and $S_2$. Since $S_1$ and $S_2$ are minimum vertex separators of $\alpha$, we must have that $w(S_L) = w(S_R) = w(S_1) = w(S_2)$. Since $S_L$ is a vertex separator of $U_{\alpha}$ and $S_2$, $S_L \neq S_1$. However, $S_L$ is a vertex separator of $U_{\alpha}$ and $S_1$, which contradicts our choice of $S_1$.
		
		Thus, there must be a leftmost minimum vertex separator of $\alpha$. Following similar logic, there must be a rightmost minimum vertex separator of $\alpha$ as well.
	\end{proof}
	\section{Proofs with Canonical Maps}\label{canonicalmapsection}
	In this section, we give alternative proofs of Lemmas \ref{lm:morthsimplereexpression} and \ref{lm:singleshapeintersections} using canonical maps.
	\begin{definition}[Canonical Maps]
		For each shape $\alpha$ and each ribbon $R$ of shape $\alpha$, we arbitrarily choose a canonical map $\phi_R: V(\alpha) \to V(R)$ such that $\phi_R(H_{\alpha}) = H_R$, $\phi_{R}(U_{\alpha}) = A_R$, and $\phi_{R}(V_{\alpha}) = B_R$. Note that there are $|Aut(\alpha)|$ possible choices for this map.
	\end{definition}
	\subsection{Proof of Lemma \ref{lm:morthsimplereexpression}}
	\begin{lemma}
		\[
		M^{orth}_{\tau}(H) = \sum_{\sigma \in Row(H),\sigma' \in Col(H)}{H(\sigma,\sigma')|Decomp(\sigma,\tau,{\sigma'}^T)|M_{\sigma \circ \tau \circ {\sigma'}^T}}
		\]
	\end{lemma}
	\begin{proof}
		Observe that there is a bijection between ribbons $R$ with shape $\sigma \circ \tau \circ {\sigma'}^T$ together with an element $\pi \in Decomp(\sigma,\tau,\sigma')$ and triples of ribbons $(R_1,R_2,R_3)$ such that
		\begin{enumerate}
			\item $R_1,R_2,R_3$ have shapes $\sigma$, $\tau$, and ${\sigma'}^T$, respectively.
			\item $V(R_1) \cap V(R_2) = A_{R_2} = B_{R_1}$,  $V(R_2) \cap V(R_3) = A_{R_3} = B_{R_2}$, and $V(R_1) \cap V(R_3) = A_{R_2} \cap B_{R_2}$
		\end{enumerate}
		To see this, note that given such ribbons $R_1,R_2,R_3$, the ribbon $R = R_1 \circ R_2 \circ R_3$ has shape $\sigma \circ \tau \circ {\sigma'}^T$. Further note that we have two bijective maps from $V(\sigma \circ \tau \circ {\sigma'}^T)$ to $V(R)$. The first map is $\phi_R$. The second map is $\phi_{R_1} \circ \phi_{R_2} \circ \phi_{R_3}$. Using this, we can take $\pi = \phi^{-1}_R(\phi_{R_1} \circ \phi_{R_2} \circ \phi_{R_3})$
		
		Conversely, given a ribbon $R$ of shape $\sigma \circ \tau \circ {\sigma'}^T$ and an element $\pi \in Decomp(\sigma,\tau,\sigma')$, let $R_1 = \phi_R(\pi(\sigma))$, let $R_2 = \phi_R(\pi(\tau))$, and let $R_3 = \phi_R(\pi({\sigma'}^T))$. Note that this is well defined because for any element $\pi' \in Aut(\sigma) \times Aut(\tau) \times Aut({\sigma'}^T)$, $\phi_R(\pi\pi'(\sigma)) = \phi_R(\pi(\pi'(\sigma))) = \phi_R(\pi(\sigma))$. Similarly, $\phi_R(\pi\pi'(\tau)) = \phi_R(\pi(\tau))$ and $\phi_R(\pi\pi'({\sigma'}^T)) = \phi_R(\pi({\sigma'}^T))$.
		
		To confirm that this is bijection, we have to show that these two maps are inverses of each other. Given $R_1$, $R_2$, and $R_3$, applying these two maps gives us ribbons $R'_1 = \phi_R\phi^{-1}_R(\phi_{R_1} \circ \phi_{R_2} \circ \phi_{R_3})(H_{\sigma}) = R_1$, $R'_2 = \phi_R\phi^{-1}_R(\phi_{R_1} \circ \phi_{R_2} \circ \phi_{R_3})(H_{\tau}) = R_2$, and $R'_3 = \phi_R\phi^{-1}_R(\phi_{R_1} \circ \phi_{R_2} \circ \phi_{R_3})(H_{{\sigma'}^T}) = R_3$. Conversely, given $R$ and an element $\pi \in Decomp(\sigma,\tau,\sigma')$ (which we represent by an element $\pi \in Aut(\sigma \circ \tau \circ {\sigma'}^T)$), applying these two maps gives us the ribbon 
		\[
		R' = \phi_R(\pi(\sigma)) \circ \phi_R(\pi(\tau)) \circ \phi_R(\pi({\sigma'}^T)) = {\phi_R}\pi(\sigma \circ \tau \circ {\sigma'}^T) = R
		\]
		and gives us the map 
		\[
		\phi^{-1}_R(\phi_{\phi_R(\pi(\sigma))} \circ \phi_{\phi_R(\pi(\tau))} \circ \phi_{\phi_R(\pi({\sigma'}^T))})
		\]
		Now observe that both ${\phi_R}\pi$ and $\phi_{\phi_R(\pi(\sigma))}$ give bijective maps from $\sigma$ to the ribbon ${\phi_R}\pi(\sigma)$ so $\phi^{-1}_{\phi_R(\pi(\sigma))}{\phi_R}\pi \in Aut(\sigma)$. Following similar logic for $\tau$ and ${\sigma'}^T$, in $Decomp(\sigma,\tau,\sigma')$ this map is equivalent to $
		\phi^{-1}_R({\phi_R}\pi) = \pi$
	\end{proof}
	\subsection{Proof of Lemma \ref{lm:singleshapeintersections}}
	\begin{definition}[Rigorous definition of intersection patterns]
		We define an intersection pattern $P$ on composable shapes $\gamma,\tau,{\gamma'}^T$ to consist of the shape $\gamma \circ \tau \circ {\gamma'}^T$ together with a non-empty set of constraint edges $E(P)$ on $V(\gamma \circ \tau \circ {\gamma'}^T)$ such that:
		\begin{enumerate}
			\item For all vertices $u,v,w \in V(\gamma \circ \tau \circ {\gamma'}^T)$, if $(u,v),(v,w) \in E(P)$ then $(u,w) \in E(P)$
			\item $E(P)$ does not contain a path between two vertices of $\gamma$, two vertices of $\tau$, or two vertices of ${\gamma'}^T$. This ensures that when we consider $\gamma,\tau,\gamma'$ individually, their vertices are distinct.
			\item Defining $V_{*}(\gamma) \subseteq V(\gamma)$ to be the vertices of $\gamma$ which are incident to an edge in $E(P)$, $U_{\gamma}$ is the unique minimum-weight vertex separator between $U_{\gamma}$ and $V_{*}(\gamma) \cup V_{\gamma}$
			\item Similarly, defining $V_{*}({\gamma'}^T) \subseteq V({\gamma'}^T)$ to be the vertices of ${\gamma'}^T$ which are incident to an edge in $E(P)$, $V_{{\gamma'}^T}$ is the unique minimum-weight vertex separator between $V_{*}({\gamma'}^T) \cup U_{{\gamma'}^T}$ and $V_{U_{{\gamma'}^T}}$
			\item[5.*] All edges in $E(P)$ are between vertices of the same type.
		\end{enumerate}
	\end{definition}
	\begin{definition}
		We say that two intersection patterns $P,P'$ on shapes $\gamma,\tau,{\gamma'}^T$ are equivalent (which we write as $P \equiv P'$) if there is an automorphism $\pi \in Aut(\gamma) \times Aut(\tau) \times Aut({\gamma'}^T)$ such that $\pi(P) = P'$ (i.e. if $E(P)$ and $E(P')$ are the constraint edges for $P$ and $P'$ respectively then $\pi(E(P)) = E(P')$).
	\end{definition}
	\begin{definition}
		Given composable shapes $\gamma,\tau,{\gamma'}^T$, we define $\mathcal{P}_{\gamma,\tau,{\gamma'}^T}$ to be the set of all possible intersection patterns $P$ on $\gamma,\tau,{\gamma'}^T$ (up to equivalence)
	\end{definition}
	\begin{definition}
		Given composable (but not properly composable) ribbons $R_1$, $R_2$, $R_3$ of shapes $\gamma, \tau, {\gamma'}$, we define the intersection pattern $P \in \mathcal{P}_{\gamma,\tau,{\gamma'}^T}$ induced by $R_1,R_2,R_3$ as follows:
		\begin{enumerate}
			\item Take the canonical maps $\phi_{R_1}: V(\gamma) \to V(R_1)$, $\phi_{R_2}: V(\tau) \to V(R_2)$, and $\phi_{R_3}: V({\gamma'}^T) \to V(R_3)$
			\item Given vertices $u \in V(\gamma)$ and $v \in V(\tau)$, add a constraint edge between $u$ and $v$ if and only if $\phi_{R_1}(u) = \phi_{R_2}(v)$. Similarly, given vertices $u \in V(\gamma)$ and $w \in V({\gamma'}^T)$, add a constraint edge between $u$ and $w$ if and only if $\phi_{R_1}(u) = \phi_{R_3}(w)$ and given vertices $v \in V(\tau)$ and $w \in V({\gamma'}^T)$, add a constraint edge between $v$ and $w$ if and only if $\phi_{R_2}(v) = \phi_{R_3}(w)$.
		\end{enumerate}
	\end{definition}
	\begin{definition}
		Given an intersection pattern $P \in \mathcal{P}_{\gamma,\tau,{\gamma'}^T}$, we define $V(\gamma \circ \tau \circ {\gamma'}^T)/E(P)$ to be $V(\gamma \circ \tau \circ {\gamma'}^T)$ where all of the edges in $E(P)$ are contracted (i.e. if $(u,v) \in E(P)$ then $u = v$ and $u = v$ only appears once).
	\end{definition}
	\begin{definition}
		Given an intersection pattern $P \in \mathcal{P}_{\gamma,\tau,{\gamma'}^T}$, we define $\tau_{P}$ to be the shape such that:
		\begin{enumerate}
			\item $V(H_{\tau_P}) = V(\gamma \circ \tau \circ {\gamma'}^T)/E(P)$
			\item $E(H_{\tau_P}) = E(\gamma) \cup E(\tau) \cup E({\gamma'}^{T})$
			\item $U_{\tau_P} = U_{\gamma}$
			\item $V_{\tau_P} = V_{{\gamma'}^T}$
		\end{enumerate}
	\end{definition}
	\begin{definition}
		Given an intersection pattern $P \in \mathcal{P}_{\gamma,\tau,{\gamma'}^T}$, we make the following definitions:
		\begin{enumerate}
			\item We define $Aut(P) = \{\pi \in Aut(\gamma \circ \tau \circ {\gamma'}^T): \pi(E(P)) = E(P)\}$
			\item We define $Aut_{pieces}(P) = \{\pi \in Aut(U_{\gamma}) \times Aut(\tau) \times Aut({\gamma'}^T): \pi(E(P)) = E(P)\}$
			\item We define $N(P) = |Aut(P)/Aut_{pieces}(P)|$
		\end{enumerate}
	\end{definition}
	\begin{lemma}
		For all composable $\sigma$, $\tau$, and ${\sigma'}^T$ (inclulding improper $\tau$), 
		\begin{align*}
			&M^{fact}_{\tau}(e_{\sigma}e^T_{\sigma'}) - M^{orth}_{\tau}(e_{\sigma}e^T_{\sigma'}) = \sum_{\sigma_2, \gamma: \gamma \text{ is non-trivial }, \atop \sigma_2 \cup \gamma = \sigma}{\frac{1}{|Aut(U_{\gamma})|}\sum_{P \in \mathcal{P}_{\gamma,\tau,Id_{V_{\tau}}}}N(P)M^{orth}_{\tau_P}(e_{\sigma_2}e^T_{\sigma'})} \\
			&+ \sum_{\sigma'_2, \gamma': \gamma' \text{ is non-trivial }, \atop \sigma'_2 \cup \gamma' = \sigma'}{\frac{1}{|Aut(U_{\gamma'})|}\sum_{P \in \mathcal{P}_{Id_{U_{\tau}},\tau,{\gamma'}^T}}N(P)M^{orth}_{\tau_P}(e_{\sigma}e^T_{\sigma'_2})} \\
			&+ \sum_{\sigma_2, \gamma: \gamma \text{ is non-trivial }, \atop \sigma_2 \cup \gamma = \sigma}{\sum_{\sigma'_2, \gamma': \gamma' \text{ is non-trivial }, \atop \sigma'_2 \cup \gamma' = \sigma'}{
					\frac{1}{|Aut(U_{\gamma})|\cdot|Aut(U_{\gamma'})|}\sum_{P \in \mathcal{P}_{\gamma,\tau,{\gamma'}^T}}N(P)M^{orth}_{\tau_P}(e_{\sigma_2}e^T_{\sigma'_2})}}
		\end{align*}
	\end{lemma}
	\begin{proof}
		This lemma follows from the following bijection. Consider the third term
		\[
		\sum_{\sigma_2, \gamma: \gamma \text{ is non-trivial }, \atop \sigma_2 \cup \gamma = \sigma}{\sum_{\sigma'_2, \gamma': \gamma' \text{ is non-trivial }, \atop \sigma'_2 \cup \gamma' = \sigma'}{
				\frac{1}{|Aut(U_{\gamma})|\cdot|Aut(U_{\gamma'})|}\sum_{P \in \mathcal{P}_{\gamma,\tau,{\gamma'}^T}}N(P)M^{orth}_{\tau_P}(e_{\sigma_2}e^T_{\sigma'_2})}}
		\]
		On one side, we have the following data:
		\begin{enumerate}
			\item Ribbons $R_1$, $R_2$, and $R_3$ such that 
			\begin{enumerate}
				\item $R_1,R_2,R_3$ have shapes $\sigma$, $\tau$, and ${\sigma'}^T$, respectively.
				\item $A_{R_2} = B_{R_1}$ and $A_{R_3} = B_{R_2}$
				\item $\left(V(R_1) \cup V(R_2)\right) \cap V(R_3) \neq A_{R_3}$ and $\left(V(R_2) \cup V(R_3)\right) \cap V(R_1) \neq B_{R_1}$
			\end{enumerate}
			\item An ordering $O_{S'}$ on the leftmost minimum vertex separator $S'$ between $A_{R_1}$ and $V_{*} \cup B_{R_1}$.
			\item An ordering $O_{T'}$ on the rightmost minimum vertex separator $S'$ between $V_{*} \cup A_{R_3}$ and $B_{R_3}$.
		\end{enumerate}
		On the other side, we have the following data
		\begin{enumerate}
			\item An intersection pattern $P \in \mathcal{P}_{\gamma,\tau,{\gamma'}^T}$ where $\gamma$ and ${\gamma'}^T$ are non-trivial.
			\item Ribbons $R'_1$, $R'_2$, $R'_3$ of shapes $\sigma_2$, $\tau_P$, ${\sigma'_2}^T$ such that $V(R'_1) \cap V(R'_2) = A_{R'_2} = B_{R'_1}$, $V(R'_2) \cap V(R'_3) = B_{R'_2} = A_{R'_3}$, and $V(R'_1) \cap V(R'_3) = A_{R'_2} \cap B_{R'_2}$
			\item An element $\pi \in Aut(P)/Aut_{pieces}(P)$
		\end{enumerate}
		To see this bijection, given $R_1,R_2,R_3$, we again implement our strategy for analyzing intersection terms. Recall that $V_{*}$ is the set of vertices in $V(R_1) \cup V(R_2) \cup V(R_3)$ which have an unexpected equality with another vertex, $S'$ is the leftmost minimum vertex separator between $A_{R_1}$ and $B_{R_1} \cup V_{*}$, and $T'$ is the rightmost minimum vertex separator between $A_{R_3} \cup V_{*}$ and $B_{R_3}$.
		\begin{enumerate}
			\item Decompose $R_1$ as $R_1 = {R'}_1 \circ R_4$ where ${R'}_1$ is the part of $R_1$ between $A_{R_1}$ and $(S',O_{S'})$ and $R_4$ is the part of $R_1$ between $(S',O_{S'})$ and $B_{R_1} = A_{R_2}$. Decompose $R_3$ as $R_5 \cup R'_3$ where $R_5$ is the part of $R_3$ between $A_{R_3}$ and $(T',O_{T'})$ and $R'_3$ is the part of $R_3$ between $(T',O_{T'})$ and $B_{R_3}$
			\item Take the intersection pattern $P$ and the ribbon $R'_2$ induced by $R_4$, $R_2$, and $R_5$.
			\item Observe that we have two bijective maps from $V(\gamma \circ \tau \circ {\gamma'}^T)/E(P)$ to $V(R_4) \cup V(R_2) \cup V(R_5)$. The first map is $\phi_{R_4} \circ \phi_{R_2} \circ \phi_{R_5}$ and the second map is $\phi_{R'_2}$. We take $\pi = \phi^{-1}_{R'_2}(\phi_{R_4} \circ \phi_{R_2} \circ \phi_{R_5})$.
		\end{enumerate}
		Conversely, given an intersection pattern $P \in \mathcal{P}_{\gamma,\tau,{\gamma'}^T}$, $R'_1$, $R'_2$, $R'_3$, and an element $\pi \in Aut(P)/Aut_{pieces}(P)$:
		\begin{enumerate}
			\item Take $R_4 = \phi_{R'_2}\pi(V(\gamma))$, $R_2 = \phi_{R'_2}\pi(V(\tau))$, and $R_5 = \phi_{R'_2}\pi(V({\gamma'}^T))$.
			\item Take $R_1 = R'_1 \cup R_4$ and take $R_3 = R_5 \cup R'_3$.
			\item Take $O_S$ and $O_T$ based on $B_{R'_1} = A_{R_4}$ and $B_{R_5} = A_{R'_3}$.
		\end{enumerate}
		To confirm that this is a bijection, we need to show that these maps are inverses of each other.
		
		If we apply the first map and then the second, we obtain the following:
		\begin{enumerate}
			\item We obtain the ribbons 
			\begin{enumerate}
				\item $R''_1 = R'_1 \circ \phi_{R'_2}\phi^{-1}_{R'_2}(\phi_{R_4} \circ \phi_{R_2} \circ \phi_{R_5})(V(\gamma))$
				\item $R''_2 = \phi_{R'_2}\phi^{-1}_{R'_2}(\phi_{R_4} \circ \phi_{R_2} \circ \phi_{R_5})(V(\tau))$
				\item $R''_3 = \phi_{R'_2}\phi^{-1}_{R'_2}(\phi_{R_4} \circ \phi_{R_2} \circ \phi_{R_5})(V({\gamma'}^T)) \circ R'_3$
			\end{enumerate} 
			where 
			\begin{enumerate}
				\item $R'_1$ is the part of $R_1$ between $A_{R_1}$ and $(S',O_{S'})$ where $S'$ is the minimum vertex separator between $A_{R_1}$ and $V_{*} \cup B_{R_1}$.
				\item $R_4$ is the part of $R_1$ between $(S',O_{S'})$ and $B_{R_1}$
				\item $R'_2$ is the ribbon of shape $\tau_{P}$ induced (along with the intersection pattern $P$) by $R_1$, $R_2$, and $R_3$.
				\item $R_5$ is the part of $R_3$ between $A_{R_3}$ and $(T',O_{T'})$.
				\item $R'_3$ is the part of $R_3$ between $(T',O_{T'})$ and $B_{R_3}$
			\end{enumerate}
			This implies that $R''_1 = R'_1 \circ R_4 = R_1$, $R''_2 = R_2$, and $R''_3 = R_5 \circ R'_3 = R_3$. Since the second map leaves $R'_1$ and $R'_3$ unchanged, we recover the orderings $O_S$ and $O_T$ as well.
		\end{enumerate}
		
		Conversely, if we apply the second map, we have that $R_1 = R'_1 \circ \phi_{R'_2}\pi(V(\gamma))$, $R_2 = \phi_{R'_2}\pi(V(\tau))$, and $R_3 = \phi_{R'_2}\pi(V({\gamma'}^T)) \circ R'_3$ and we have the orderings $O_S$ and $O_T$ corresponding to $B_{R'_1}$ and $A_{R'_3}$ respectively. If we apply the first map, 
		\begin{enumerate}
			\item $R'_1$ and $R'_3$ are preserved.
			\item $R''_2$ and $P''$ are the ribbon and intersection pattern induced by the ribbons $\phi_{R'_2}\pi(\gamma)$, $\phi_{R'_2}\pi(\tau)$, and $\phi_{R'_2}\pi({\gamma'}^T)$. To see that $R''_2 = R'_2$, observe that 
			\[
			R''_2 = \phi_{R'_2}\pi(V(\gamma)) \circ \phi_{R'_2}\pi(V(\tau)) \circ \phi_{R'_2}\pi(V({\gamma'}^T)) = \phi_{R'_2}{\pi(\gamma \circ \tau \circ {\gamma'}^T)} = \phi_{R_2}(\gamma \circ \tau \circ {\gamma'}^T) = R'_2
			\]
			To see that $P'' \equiv P$, observe that:
			\begin{enumerate}
				\item We have two bijective maps from $V(\gamma)$ to $V(\phi_{R'_2}\pi(\gamma))$. These two maps are $\phi_{R'_2}\pi$ and $\phi_{\phi_{R'_2}\pi(\gamma)}$.
				\item We have two bijective maps from $V(\tau)$ to $V(\phi_{R'_2}\pi(\tau))$. These two maps are $\phi_{R'_2}\pi$ and $\phi_{\phi_{R'_2}\pi(\tau)}$.
				\item We have two bijective maps from $V({\gamma'}^T)$ to $V(\phi_{R'_2}\pi({\gamma'}^T))$. These two maps are $\phi_{R'_2}\pi$ and $\phi_{\phi_{R'_2}\pi({\gamma'}^T)}$.
				\item For $P''$, the constraint edges are 
				\[
				\left(\phi^{-1}_{\phi_{R'_2}\pi(\gamma)}\phi_{R'_2}\pi \circ \phi^{-1}_{\phi_{R'_2}\pi(\tau)}\phi_{R'_2}\pi \circ \phi^{-1}_{\phi_{R'_2}\pi({\gamma'}^T))}\phi_{R'_2}\pi\right)(E(P))
				\]
			\end{enumerate}
			\item We have that 
			\[
			\pi'' = \phi^{-1}_{R'_2}(\phi_{\phi_{R'_2}\pi(V(\gamma))} \circ \phi_{\phi_{R'_2}\pi(V(\tau))} \circ \phi_{\phi_{R'_2}\pi(V({\gamma'}^T))})
			\]
			To see that $\pi'' \equiv \pi$, note that 
			\[
			\pi = \pi''\left(\phi^{-1}_{\phi_{R'_2}\pi(V(\gamma))}\phi_{R'_2}\pi \circ \phi^{-1}_{\phi_{R'_2}\pi(V(\tau))}\phi_{R'_2}\pi \circ \phi^{-1}_{\phi_{R'_2}\pi(V({\gamma'}^T))}\phi_{R'_2}\pi)\right)
			\]
		\end{enumerate}
		The analysis for the the first term is the same except that when $\gamma'$ is trivial, we always take $\gamma'$ to be the identity so $T = V(V_{\tau}) = V(U_{{\sigma'}^T})$ and the ordering $O_{T}$ is given by $V_{\tau} = U_{{\sigma'}^T}$. 
		Similarly, the analysis for the the second term is the same except that when $\gamma$ is trivial, we always take $\gamma$ to be the identity so $S = V(V_{\sigma}) = V(U_{\tau})$ and the ordering $O_{S}$ is given by $V_{\sigma} = U_{\tau}$.
	\end{proof}

\section{Degree 4 Planted Clique Analysis}\label{sec: deg_4_planted_clique}
For this example, we name the shapes based on what they look like to make them easier to keep track of. With the exception of $Id_{U}$, these names only appear in this section.
\subsection{The shapes $\alpha$ and coefficients $\lambda_{\alpha}$}
After several preprocessing steps, the moment matrix which needs to be analyzed is $M \approx \sum_{\alpha}{\lambda_{\alpha}M_{\alpha}}$ for the following shapes $\alpha$ and coefficients $\lambda_{\alpha}$
\begin{definition} \ 
\begin{enumerate}
\item Given $E \subseteq \{(u_1,v_1), (u_1,v_2), (u_2,v_1), (u_2,v_2)\}$, we define $\alpha_{E}$ to be the shape where $U_{\alpha_E} = (u_1,u_2)$, $V_{\alpha_E} = (v_1,v_2)$, and $E(\alpha) = E$.
\item Given $E \subseteq \{(u_1,v_1), (u_1,v_2), (u_2,v_1), (u_2,v_2)\}$, we define $\alpha_{X,E}$ to be the shape where $U_{\alpha_E} = (u_1,u_2)$, $V_{\alpha_E} = (v_1,v_2)$, there is one additional vertex $w_1$, and $E(\alpha) = E \cup \{(u_1,w_1), (u_2,w_1), (w_1,v_1), (w_1,v_2)\}$.
\item Given $i,j \in \{1,2\}$, we define $\alpha_{u_i = v_j,e}$ to be the shape where $U_{\alpha_{u_i = v_j,e}} = (u_1,u_2)$, $V_{\alpha_{u_i = v_j,e}} = (v_1,v_2)$, $u_i = v_j$, and $E(\alpha_{u_i = v_j,e}) = \{(u_{2-i},v_{2-j})\}$.
\item Given $i,j \in \{1,2\}$, we define $\alpha_{u_i = v_j,\emptyset}$ to be the shape where $U_{\alpha_E} = (u_1,u_2)$, $V_{\alpha_E} = (v_1,v_2)$, $u_i = v_j$, and $E(\alpha) = \emptyset$.
\item We define $\alpha_{Id:} = Id_{(u_1,u_2)}$ to be the shape where $U_{Id_{(u_1,u_2)}} = V_{Id_{(u_1,u_2)}} = (u_1,u_2)$ and $E(Id_{(u_1,u_2)}) = \emptyset$.
\item We define $\alpha_{swap}$ to be the shape where $U_{\alpha_{swap}} = (u_1,u_2)$, $V_{\alpha_{swap}} = (u_2,u_1)$ and $E(\alpha_{swap}) = \emptyset$.
\end{enumerate}
\end{definition}
For illustrations of these shapes $\alpha$, see Figures \ref{zerointersectionalphasfigure}.\\
\begin{figure}[ht]\label{zerointersectionalphasfigure}
\centerline{\includegraphics[height=4cm]{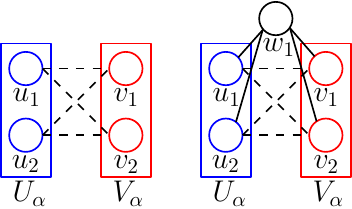}}
\caption{This figure shows the shapes $\alpha$ where $|U_{\alpha} \cap V_{\alpha}| = 0$. On the left we have $\alpha_{E}$ and on the right we have $\alpha_{X,E}$.}
\end{figure}
\begin{figure}[ht]\label{oneintersectionalphasfigure}
\centerline{\includegraphics[height=4cm]{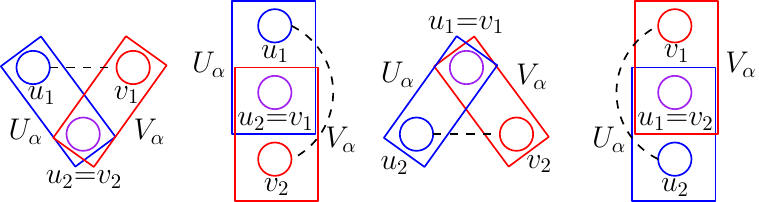}}
\caption{This figure shows the shapes $\alpha$ where $|U_{\alpha} \cap V_{\alpha}| = 1$. From left to right, we have $\alpha_{u_2 = v_2,\emptyset}$ and $\alpha_{u_2 = v_2,e}$, $\alpha_{u_2 = v_1,\emptyset}$ and $\alpha_{u_2 = v_1,e}$, $\alpha_{u_1 = v_1,\emptyset}$ and $\alpha_{u_1 = v_1,e}$, and $\alpha_{u_1 = v_2,\emptyset}$ and $\alpha_{u_1 = v_2,e}$.}
\end{figure}
\begin{figure}[ht]\label{twointersectionalphasfigure}
\centerline{\includegraphics[height=4cm]{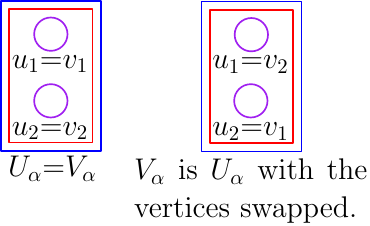}}
\caption{This figure shows the shapes $\alpha$ where $|U_{\alpha} \cap V_{\alpha}| = 2$. On the left we have $\alpha_{Id:}$ and on the right we have $\alpha_{swap}$.}
\end{figure}
We have the following coefficients on these shapes.
\begin{enumerate}
\item For each $E \subseteq \{(u_1,v_1), (u_1,v_2), (u_2,v_1), (u_2,v_2)\}$, $\lambda_{\alpha_E} = \frac{k^4}{n^4}$.
\item For each $E \subseteq \{(u_1,v_1), (u_1,v_2), (u_2,v_1), (u_2,v_2)\}$, $\lambda_{\alpha_{X,E}} = C\frac{k^5}{n^5}$ for some constant $C > 1$. These coefficients are the ad-hoc fix to the candidate pseudo-expectation values for planted clique in \cite{meka2015sum}.
\item For each $i,j \in \{1,2\}$, $\lambda_{\alpha_{u_i = v_j,e}} = \lambda_{\alpha_{u_i = v_j,\emptyset}} = \frac{k^3}{n^3}$.
\item $\lambda_{\alpha_{Id:}} = \lambda_{\alpha_{swap}} = \frac{k^2}{n^2}$
\end{enumerate}
\subsection{Decomposing $\alpha$ and coefficient matrices}
To find the coefficient matrices $H_{Id_{\emptyset}}$, $H_{Id_{(u_1)}}$, $H_{Id_{(u_1,u_2)}}$, and $H_{\tau}$, we need to decompose each $\alpha$ into a left part $\sigma$, a proper middle part $\tau$, and a right part ${\sigma'}^T$. 

The following left shapes will appear in these decompositions
\begin{definition} \ 
\begin{enumerate}
\item Define $\sigma_{Id:} = Id_{(u_1,u_2)}$. Note that $\sigma_{Id:} = \alpha_{Id:}$ but it is playing a different role.
\item Define $\sigma_{swap} = \alpha_{swap}$.
\item Define $\sigma_7$ to be the shape where $U_{\sigma_{7}} = (u_1,u_2)$, $V_{\sigma_{7}} = (v_1)$, and $E(\sigma_{7}) = \{(u_1,v_1),(u_2,v_1)\}$.
\item Define $\sigma_{u_1,u_2 \to u_1}$ to be the shape where $U_{\sigma_{u_1,u_2 \to u_1}} = (u_1,u_2)$, $V_{\sigma_{u_1,u_2 \to u_1}} = (u_1)$, and $E(\sigma_{u_1,u_2 \to u_1}) = \emptyset$.
\item Similarly, define $\sigma_{u_1,u_2 \to u_2}$ to be the shape where $U_{\sigma_{u_1,u_2 \to u_2}} = (u_1,u_2)$, $V_{\sigma_{u_1,u_2 \to u_2}} = (u_2)$, and $E(\sigma_{u_1,u_2 \to u_2}) = \emptyset$.
\item Define $\sigma_{u_1,u_2 \to \emptyset}$ to be the shape where $U_{\sigma_{u_1,u_2 \to \emptyset}} = (u_1,u_2)$, $V_{\sigma_{u_1,u_2 \to \emptyset}} = \emptyset$, and $E(\sigma_{u_1,u_2 \to \emptyset}) = \emptyset$.
\end{enumerate}
\end{definition}
These left shapes are illustrated in Figure \ref{onerightsidevertexsigmasfigure}. \\
\begin{figure}[ht]\label{onerightsidevertexsigmasfigure}
\centerline{\includegraphics[height=4cm]{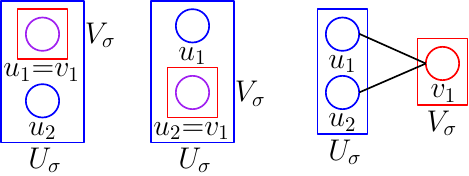}}
\caption{This figure shows the left shapes $\sigma$ where $|V_{sigma}| = 1$. From left to right we have $\sigma_{u_1,u_2 \to u_1}$, $\sigma_{u_1,u_2 \to u_2}$, and $\sigma_7$.}
\end{figure}
\begin{figure}[ht]\label{zeroortworightvertexsigmasfigure}
\centerline{\includegraphics[height=4cm]{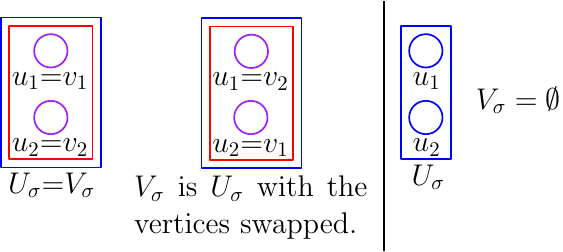}}
\caption{This figure shows the left shapes $\sigma$ where $|V_{\sigma}| = 2$ or $|V_{\sigma}| = 0$. On the left we have $\sigma_{Id:}$ and $\sigma_{swap}$. On the right we have $\sigma_{u_1,u_2 \to \emptyset}$.}
\end{figure}

The following proper middle shapes will appear in these decompositions.
\begin{definition} \ 
\begin{enumerate}
\item Define $\tau_{Id:} = Id_{(u_1,u_2)}$
\item Given $E \subseteq \{(u_1,v_1), (u_1,v_2), (u_2,v_1), (u_2,v_2)\}$ such that all four vertices $u_1,u_2,v_1,v_2$ are incident to at least one edge in $E$, we define $\tau_{E} = \alpha_{E}$.
\item Given $E \subseteq \{(u_1,v_1), (u_1,v_2), (u_2,v_1), (u_2,v_2)\}$ such that $E \neq \emptyset$, we define $\tau_{X,E} = \alpha_{X,E}$.
\item Given $i,j \in \{1,2\}$, we define $\tau_{u_i = v_j,e} = \alpha_{u_i = v_j,e}$.
\item Define $\tau_{Id\cdot} = Id_{(u_1)}$ to be the shape where $U_{Id_{(u_1)}} = V_{Id_{(u_1)}} = (u_1)$ and $E(Id_{(u_1)}) = \emptyset$.
\item Define $\tau_{e}$ to be the shape where $U_{\tau_{e}} = (u_1)$, $V_{\tau_{e}} = (v_1)$, and $E(\tau_{e}) = \{(u_1,v_1)\}$.
\item Define $\tau_{\emptyset}$ to be the empty shape with no vertices. 
\end{enumerate}
\end{definition}
These proper middle shapes (except for $\tau_{\emptyset}$) are illustrated in Figure \ref{twovertexseparatortausfigure}. \\
\begin{figure}[ht]\label{twovertexseparatortausfigure}
\centerline{\includegraphics[height=8cm]{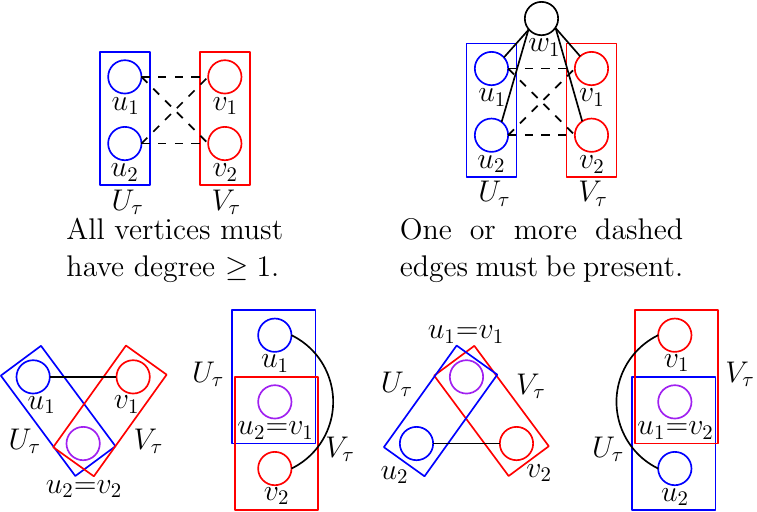}}
\caption{This figure shows the proper middle shapes $\tau$ where $|U_{tau}| = |V_{tau}| = 2$. In the upper row, we have $\tau_{Id:}$, $\tau_{E}$, and $\tau_{X,E}$. In the bottom row, we have  $\tau_{u_2 = v_2,e}$, $\tau_{u_2 = v_1,e}$, $\tau_{u_1 = v_1,e}$, and $\tau_{u_1 = v_2,e}$.}
\end{figure}
\begin{figure}[ht]\label{onevertexseparatortausfigure}
\centerline{\includegraphics[height=3cm]{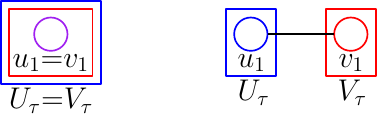}}
\caption{This figure shows the proper middle shapes $\tau$ where $|U_{tau}| = |V_{tau}| = 1$. On the left we have $\tau_{Id\cdot}$ and on the right we have $\tau_{e}$.}
\end{figure}

Some decompositions are as follows:
\begin{example} \ 
\begin{enumerate}
\item $\alpha_{X,\emptyset} = \sigma_7 \circ \tau_{Id\cdot} \circ \sigma_7^T$ (see Figure \ref{decompositiononefigure}).
\item $\alpha_{\{(u_1,v_1),(u_2,v_1)\}} = \sigma_7 \circ \tau_{Id\cdot} \circ \sigma_{u_1,u_2 \to u_1}^T$ (see Figure \ref{decompositiontwofigure}).
\item $\alpha_{\{(u_2,v_1)\}} = \sigma_{u_1,u_2 \to u_2} \circ \tau_{e} \circ \sigma_{u_1,u_2 \to u_1}^T$ (see Figure \ref{decompositionthreefigure}).
\item \begin{align*}
\alpha_{\{(u_1,v_1),(u_2,v_2)\}} &= \sigma_{Id:} \circ \alpha_{\{(u_1,v_1),(u_2,v_2)\}} \circ \sigma_{Id:}^T = \sigma_{Id:} \circ \alpha_{\{(u_1,v_2),(u_2,v_1)\}} \circ \sigma_{swap}^T \\
&= \sigma_{swap}^T \circ \alpha_{\{(u_1,v_2),(u_2,v_1)\}} \circ \sigma_{Id:} = \sigma_{swap} \circ \alpha_{\{(u_1,v_1),(u_2,v_2)\}} \circ \sigma_{swap}^T
\end{align*}
\end{enumerate}
\end{example}
\begin{remark}
Since there are $4$ different ways to decompose $\alpha_{\{(u_1,v_1),(u_2,v_2)\}}$, we split the coefficient $\lambda_{\alpha_{\{(u_1,v_1),(u_2,v_2)\}}}$ among these four decompositions. This is the reason for the factor of $4$ in the denominator in the entries of the matrix $H_{\tau}$.
\end{remark}
\begin{figure}[ht]\label{decompositiononefigure}
\centerline{\includegraphics[height=4cm]{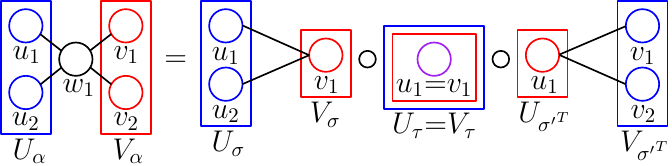}}
\caption{This figure shows the decomposition $\alpha_{X,\emptyset} = \sigma_7 \circ \tau_{Id\cdot} \circ \sigma_7^T$.}
\end{figure}
\begin{figure}[ht]\label{decompositiontwofigure}
\centerline{\includegraphics[height=4cm]{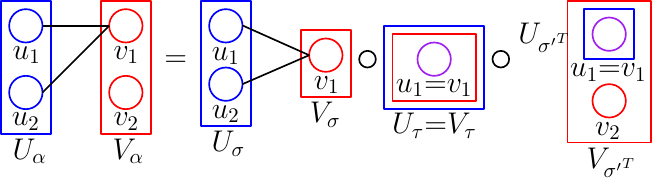}}
\caption{This figure shows the decomposition $\alpha_{\{(u_1,v_1),(u_2,v_1)\}} = \sigma_7 \circ \tau_{Id\cdot} \circ \sigma_{u_1,u_2 \to u_1}^T$.}
\end{figure}
\begin{figure}[ht]\label{decompositionthreefigure}
\centerline{\includegraphics[height=4cm]{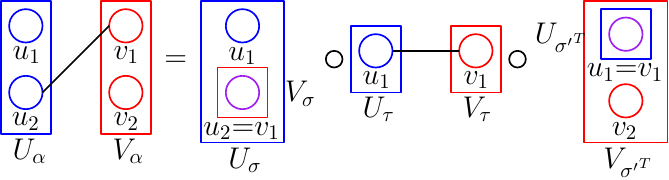}}
\caption{This figure shows the decomposition $\alpha_{\{(u_2,v_1)\}} = \sigma_{u_1,u_2 \to u_2} \circ \tau_{e} \circ \sigma_{u_1,u_2 \to u_1}^T$.}
\end{figure}

Our coefficient matrices are as follows (ignoring zero rows and columns):
\begin{enumerate}
\item $H_{Id_{(u_1,u_2)}}$ has two rows and columns indexed by $\sigma_{Id:}$ and $\sigma_{swap}$ and has entries $\left(\begin{matrix}
\frac{k^2}{2n^2} & \frac{k^2}{2n^2}\\
\frac{k^2}{2n^2} & \frac{k^2}{2n^2}\\
\end{matrix}\right)$.
\item $H_{Id_{(u_1)}}$ has rows and columns indexed by $\sigma_{u_1,u_2 \to u_1}$, $\sigma_{u_1,u_2 \to u_2}$, and $\sigma_7$ and has entries
\[
\left(\begin{matrix}
\frac{k^3}{n^3} & \frac{k^3}{n^3} & \frac{k^4}{n^4}\\
\frac{k^3}{n^3} & \frac{k^3}{n^3} & \frac{k^4}{n^4}\\
\frac{k^4}{n^4} & \frac{k^4}{n^4} & C\frac{k^5}{n^5}\\
\end{matrix}\right) \]
\item $H_{Id_{\emptyset}}$ has a single row and column indexed by $\sigma_{u_1,u_2 \to \emptyset}$ and has a single entry which is $\frac{k^4}{n^4}$.
\item For all $E \subseteq \{(u_1,v_1), (u_1,v_2), (u_2,v_1), (u_2,v_2)\}$ such that all four vertices $u_1,u_2,v_1,v_2$ are incident to at least one edge in $E$, $H_{\tau_{E}}$ has two rows and columns indexed by $\sigma_{Id:}$ and $\sigma_{swap}$ and has entries $\left(\begin{matrix}
\frac{k^4}{4n^4} & \frac{k^4}{4n^4}\\
\frac{k^4}{4n^4} & \frac{k^4}{4n^4}\\
\end{matrix}\right)$.
\item For all $E \subseteq \{(u_1,v_1), (u_1,v_2), (u_2,v_1), (u_2,v_2)\}$ such that $E \neq \emptyset$,  $H_{\tau_{X,E}}$ has two rows and columns indexed by $\sigma_{Id:}$ and $\sigma_{swap}$ and has entries $\left(\begin{matrix}
\frac{k^5}{4n^5} & \frac{k^5}{4n^5}\\
\frac{k^5}{4n^5} & \frac{k^5}{4n^5}\\
\end{matrix}\right)$.
\item For all $i,j \in \{1,2\}$, $H_{\tau_{u_i = v_j,e}}$ has two rows and columns indexed by $\sigma_{Id:}$ and $\sigma_{swap}$ and has entries $\left(\begin{matrix}
\frac{k^3}{4n^3} & \frac{k^3}{4n^3}\\
\frac{k^3}{4n^3} & \frac{k^3}{4n^3}\\
\end{matrix}\right)$.
\item $H_{\tau_{e}}$ has rows and columns indexed by $\sigma_{u_1,u_2 \to u_1}$ and $\sigma_{u_1,u_2 \to u_2}$ and has entries 
\[
\left(\begin{matrix}
\frac{k^4}{n^4} & \frac{k^4}{n^4}\\
\frac{k^4}{n^4} & \frac{k^4}{n^4}\\
\end{matrix}\right)\]
\end{enumerate}
\subsection{Verifying the first and second conditions of the machinery}
We can verify the fist and second conditions of the machinery as follows.
\begin{enumerate}
\item $H_{Id_{(u_1,u_2)}} \succeq 0$ and $H_{Id_{\emptyset}} \succeq 0$
\item As long as $C \geq 1$, $H_{Id_{(u_1)}} \succeq 0$. This condition is the reason why we need to add this term in.
\item For all $E \subseteq \{(u_1,v_1), (u_1,v_2), (u_2,v_1), (u_2,v_2)\}$ such that all four vertices $u_1,u_2,v_1,v_2$ are incident to at least one edge in $E$, $||M_{\tau_{E}}||$ is $\tilde{O}(n)$ so $||M_{\tau_{E}}||H_{\tau_{E}} \preceq H_{Id_{(u_1,u_2)}}$ as long as $k << \sqrt{n}$.
\item For all $E \subseteq \{(u_1,v_1), (u_1,v_2), (u_2,v_1), (u_2,v_2)\}$ such that $E \neq \emptyset$, 
$||M_{\tau_{X,E}}||$ is $\tilde{O}(n^{3/2})$ so $||M_{\tau_{X,E}}||H_{\tau_{X,E}} \preceq H_{Id_{(u_1,u_2)}}$ as long as $k << \sqrt{n}$.
\item For all $i,j \in \{1,2\}$, $||M_{\tau_{u_i = v_j,e}}||$ is $\tilde{O}(\sqrt{n})$ so $||M_{\tau_{u_i = v_j,e}}||H_{\tau_{u_i = v_j,e}} \preceq H_{Id_{(u_1,u_2)}}$ as long as $k << \sqrt{n}$.
\item Since $\left(\begin{matrix}
\frac{k^3}{n^3} & \frac{k^3}{n^3} & \frac{k^4}{n^4}\\
\frac{k^3}{n^3} & \frac{k^3}{n^3} & \frac{k^4}{n^4}\\
\frac{k^4}{n^4} & \frac{k^4}{n^4} & C\frac{k^5}{n^5}\\
\end{matrix}\right) \succeq (1 - \frac{1}{C})\left(\begin{matrix}
\frac{k^3}{n^3} & \frac{k^3}{n^3} & 0\\
\frac{k^3}{n^3} & \frac{k^3}{n^3} & 0\\
0 & 0 & 0\\
\end{matrix}\right)$ and $||M_{\tau_{e}}||$ is $\tilde{O}(\sqrt{n})$, $||M_{\tau_{e}}||H_{\tau_{e}} \preceq H_{Id_{(u_1)}}$ as long as $C > 1$ and $k << \sqrt{n}$. Note that for pseudo-calibration we take $C = 1$. We can do this because we have more terms which allows us to have a more delicate factorization.
\end{enumerate}
\subsection{Verifying the third condition of the machinery}
The following left shapes $\gamma$ appear.
\begin{definition} \ 
\begin{enumerate}
\item Define $\gamma_7 = \sigma_7$.
\item Define $\gamma_{u_1,u_2 \to u_1} = \sigma_{u_1,u_2 \to u_1}$.
\item Define $\gamma_{u_1,u_2 \to u_2} = \sigma_{u_1,u_2 \to u_2}$.
\item Define $\gamma_{u_1,u_2 \to \emptyset} = \sigma_{u_1,u_2 \to \emptyset}$.
\item Define $\gamma_{u_1 \to \emptyset}$ to be the shape with $U_{\gamma_{u_1 \to \emptyset}} = (u_1)$ and $V_{\gamma_{u_1 \to \emptyset}} = \emptyset$
\end{enumerate}
\end{definition}
For illustrations of these gammas, see Figure \ref{possiblegammasfigure}. \\
\begin{figure}[ht]\label{possiblegammasfigure}
\centerline{\includegraphics[height=4cm]{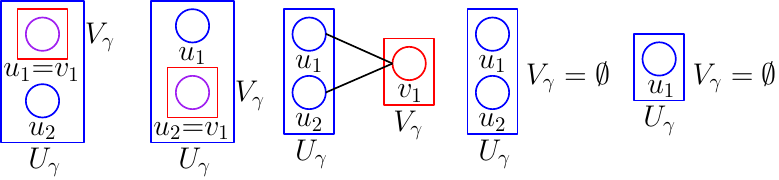}}
\caption{This figure shows the left shapes $\gamma$ which appear in the analysis. From left to right we have $\gamma_{u_1,u_2 \to u_1}$, $\gamma_{u_1,u_2 \to u_2}$, $\gamma_7$, $\gamma_{u_1,u_2 \to \emptyset}$, and $\gamma_{u_1 \to \emptyset}$}
\end{figure}
We have the following compositions.
\begin{enumerate}
\item $\sigma_{Id:} \circ \gamma_7 = \sigma_{swap} \circ \gamma_7 = \sigma_7$ (see Figure \ref{compositiononefigure}).
\item $\sigma_{Id:} \circ \gamma_{u_1,u_2 \to u_1} = \sigma_{u_1,u_2 \to u_1}$ and $\sigma_{swap} \circ \gamma_{u_1,u_2 \to u_1} = \sigma_{u_1,u_2 \to u_2}$ (see Figure \ref{compositiontwofigure}).
\item Similarly, $\sigma_{Id:} \circ \gamma_{u_1,u_2 \to u_2} = \sigma_{u_1,u_2 \to u_2}$ and $\sigma_{swap} \circ \gamma_{u_1,u_2 \to u_2} = \sigma_{u_1,u_2 \to u_1}$.
\item $\sigma_{Id:} \circ \gamma_{u_1,u_2 \to \emptyset} = \sigma_{swap} \circ \gamma_{u_1,u_2 \to \emptyset} = \sigma_{u_1,u_2 \to \emptyset}$.
\item $\sigma_{u_1,u_2 \to u_1} \circ \gamma_{u_1 \to \emptyset} = \sigma_{u_1,u_2 \to u_2} \circ \gamma_{u_1 \to \emptyset} = \sigma_{u_1,u_2 \to \emptyset}$ (see Figure \ref{compositionthreefigure}).
\end{enumerate}
\begin{figure}[ht]\label{compositiononefigure}
\centerline{\includegraphics[height=4cm]{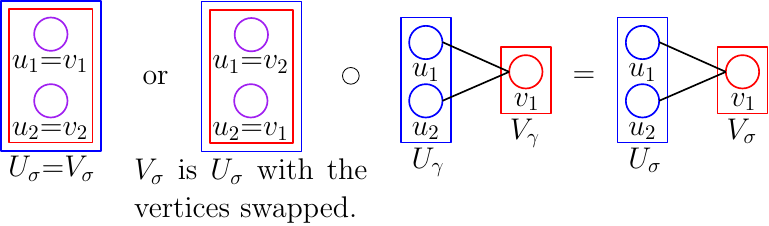}}
\caption{This figure shows the compositions $\sigma_{Id:} \circ \gamma_{7} = \sigma_{7}$ and $\sigma_{swap} \circ \gamma_{7} = \sigma_{7}$.}
\end{figure}
\begin{figure}[ht]\label{compositiontwofigure}
\centerline{\includegraphics[height=8cm]{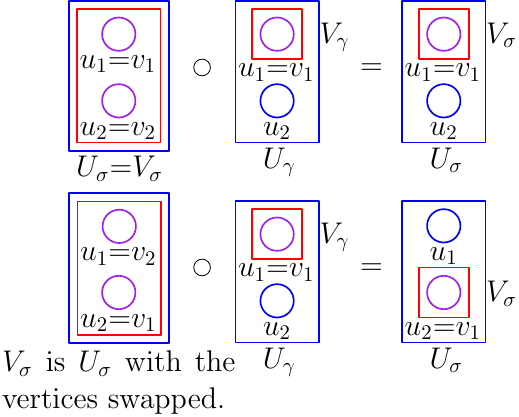}}
\caption{This figure shows the compositions $\sigma_{Id:} \circ \gamma_{u_1,u_2 \to u_1} = \sigma_{u_1,u_2 \to u_1}$ and $\sigma_{swap} \circ \gamma_{u_1,u_2 \to u_1} = \sigma_{u_1,u_2 \to u_2}$ .}
\end{figure}
\begin{figure}[ht]\label{compositionthreefigure}
\centerline{\includegraphics[height=4cm]{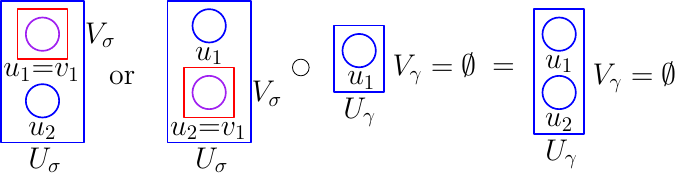}}
\caption{This figure shows the compositions $\sigma_{u_1,u_2 \to u_1} \circ \gamma_{u_1 \to \emptyset} = \sigma_{u_1,u_2 \to \emptyset}$ and $\sigma_{u_1,u_2 \to u_2} \circ \gamma_{u_1 \to \emptyset} = \sigma_{u_1,u_2 \to \emptyset}$.}
\end{figure}

Based on these compositions, we have the following matrices:
\begin{enumerate}
\item $H_{Id_{(u_1)}}^{-\gamma_7,\gamma_7}$ has two rows and columns indexed by $\sigma_{Id:}$ and $\sigma_{swap}$ and has entries $\left(\begin{matrix}
\frac{k^5}{n^5} & \frac{k^5}{n^5}\\
\frac{k^5}{n^5} & \frac{k^5}{n^5}\\
\end{matrix}\right)$.
\item $H_{Id_{(u_1)}}^{-\gamma_{u_1,u_2 \to u_1},\gamma_{u_1,u_2 \to u_1}}$ has two rows and columns indexed by $\sigma_{Id:}$ and $\sigma_{swap}$ and has entries $\left(\begin{matrix}
\frac{k^3}{n^3} & \frac{k^3}{n^3}\\
\frac{k^3}{n^3} & \frac{k^3}{n^3}\\
\end{matrix}\right)$.
\item $H_{Id_{(u_1)}}^{-\gamma_{u_1,u_2 \to u_2},\gamma_{u_1,u_2 \to u_2}} = H_{Id_{(u_1)}}^{-\gamma_{u_1,u_2 \to u_1},\gamma_{u_1,u_2 \to u_1}}$.
\item $H_{Id_{\emptyset}}^{-\gamma_{u_1,u_2 \to \emptyset},\gamma_{u_1,u_2 \to \emptyset}}$  has two rows and columns indexed by $\sigma_{Id:}$ and $\sigma_{swap}$ and has entries $\left(\begin{matrix}
\frac{k^4}{n^4} & \frac{k^4}{n^4}\\
\frac{k^4}{n^4} & \frac{k^4}{n^4}\\
\end{matrix}\right)$.
\item $H_{Id_{\emptyset}}^{-\gamma_{u_1 \to \emptyset},\gamma_{u_1 \to \emptyset}}$ has two rows and columns indexed by $\sigma_{u_1,u_2 \to u_1}$ and $\sigma_{u_1,u_2 \to u_2}$ and has entries $\left(\begin{matrix}
\frac{k^4}{n^4} & \frac{k^4}{n^4}\\
\frac{k^4}{n^4} & \frac{k^4}{n^4}\\
\end{matrix}\right)$.
\end{enumerate}
We can qualitatively verify the third condition of the machinery as follows
\begin{enumerate}
\item $B(\gamma_7)$ is $\tilde{O}(n^{\frac{|V(\gamma_7) \setminus U_{\gamma_7}|}{2}}) = \tilde{O}(\sqrt{n})$ so ${B(\gamma_7)^2}H_{Id_{(u_1)}}^{-\gamma_7,\gamma_7} \preceq H_{Id_{(u_1,u_2)}}$ as long as $k << n^{\frac{2}{3}}$.
\item $B(\gamma_{u_1,u_2 \to u_1})$ is $\tilde{O}(n^{\frac{|V(\gamma_{u_1,u_2 \to u_1}) \setminus U_{\gamma_{u_1,u_2 \to u_1}}|}{2}}) = \tilde{O}(1)$ so ${B(\gamma_{u_1,u_2 \to u_1})^2}H_{Id_{(u_1)}}^{-\gamma_{u_1,u_2 \to u_1},\gamma_{u_1,u_2 \to u_1}} \preceq H_{Id_{(u_1,u_2)}}$. Following the same logic, ${B(\gamma_{u_1,u_2 \to u_2})^2}H_{Id_{(u_1)}}^{-\gamma_{u_1,u_2 \to u_2},\gamma_{u_1,u_2 \to u_2}} \preceq H_{Id_{(u_1,u_2)}}$.
\item $B(\gamma_{u_1,u_2 \to \emptyset})$ is $\tilde{O}(n^{\frac{|V(\gamma_{u_1,u_2 \to \emptyset}) \setminus U_{\gamma_{u_1,u_2 \to \emptyset}}|}{2}}) = \tilde{O}(1)$ so ${B(\gamma_{u_1,u_2 \to \emptyset})^2}H_{Id_{\emptyset}}^{-\gamma_{u_1,u_2 \to \emptyset},\gamma_{u_1,u_2 \to \emptyset}} \preceq H_{Id_{(u_1,u_2)}}$.
\item $B(\gamma_{u_1 \to \emptyset})$ is $\tilde{O}(n^{\frac{|V(\gamma_{u_1 \to \emptyset}) \setminus U_{\gamma_{u_1 \to \emptyset}}|}{2}}) = \tilde{O}(1)$ so ${B(\gamma_{u_1 \to \emptyset})^2}H_{Id_{\emptyset}}^{-\gamma_{u_1 \to \emptyset},\gamma_{u_1 \to \emptyset}} \preceq H_{Id_{(u_1)}}$ as 
\[
\left(\begin{matrix}
\frac{k^3}{n^3} & \frac{k^3}{n^3} & \frac{k^4}{n^4}\\
\frac{k^3}{n^3} & \frac{k^3}{n^3} & \frac{k^4}{n^4}\\
\frac{k^4}{n^4} & \frac{k^4}{n^4} & C\frac{k^5}{n^5}\\
\end{matrix}\right) \succeq (1 - \frac{1}{C})\left(\begin{matrix}
\frac{k^3}{n^3} & \frac{k^3}{n^3} & 0\\
\frac{k^3}{n^3} & \frac{k^3}{n^3} & 0\\
0 & 0 & 0\\
\end{matrix}\right) \succeq \tilde{O}(1)\left(\begin{matrix}
\frac{k^4}{n^4} & \frac{k^4}{n^4} & 0\\
\frac{k^4}{n^4} & \frac{k^4}{n^4} & 0\\
0 & 0 & 0\\
\end{matrix}\right)
\].
\end{enumerate}	

\end{appendix}